\documentclass[a4paper,12pt, reqno,accepted=2024-03-06,onecolumn,allowfontchangeintitle]{quantumarticle}
\pdfoutput=1 

\usepackage{amssymb,amsmath,amsthm,amsfonts,amscd,bbm}
\usepackage{enumitem}
\usepackage{caption}
\usepackage{bm}

\usepackage{geometry} 

\usepackage{ifpdf}
\ifpdf
\usepackage[pdftex]{graphicx}
\else
\usepackage[dvips]{graphicx}
\fi
\usepackage[all]{xy}
\usepackage{hyperref}
\usepackage[dvipsnames]{xcolor}
\usepackage{graphicx}
\usepackage[outdir=./]{epstopdf}
\usepackage{enumitem}
\usepackage{tensor}
\usepackage{tikz-cd}
\usepackage{multicol}
\usepackage{mathtools}
\usepackage{tocvsec2}
\usepackage{bbm}
\usepackage{longtable}
\usepackage{fullpage}

\newcommand{\doi}[1]{\href{https://dx.doi.org/#1}{{\small DOI:#1}}}

\newcommand{\googlebooks}[1]{(preview at \href{https://books.google.com/books?id=#1}{google books})}

\newcommand{\numdam}[1]{}
\usepackage{mathrsfs}

\usepackage{mathrsfs}
\DeclareMathAlphabet{\mathpzc}{OT1}{pzc}{m}{it}

\usepackage{imakeidx}
\makeindex[options=-s dotted, title=Index of Symbols, columns=1]

\newcommand{\IN}[1]{\index{#1|BH}}

\def\semicolon{;}
\def\applytolist#1{
    \expandafter\def\csname multi#1\endcsname##1{
        \def\multiack{##1}\ifx\multiack\semicolon
            \def\next{\relax}
        \else
            \csname #1\endcsname{##1}
            \def\next{\csname multi#1\endcsname}
        \fi
        \next}
    \csname multi#1\endcsname}

\def\calc#1{\expandafter\def\csname c#1\endcsname{{\mathcal #1}}}
\applytolist{calc}QWERTYUIOPLKJHGFDSAZXCVBNM;
\def\bbc#1{\expandafter\def\csname bb#1\endcsname{{\mathbb #1}}}
\applytolist{bbc}QWERTYUIOPLKJHGFDSAZXCVBNM;
\def\bfc#1{\expandafter\def\csname bf#1\endcsname{{\mathbf #1}}}
\applytolist{bfc}QWERTYUIOPLKJHGFDSAZXCVBNM;
\def\sfc#1{\expandafter\def\csname s#1\endcsname{{\sf #1}}}
\applytolist{sfc}QWERTYUIOPLKJHGFDSAZXCVBNM;
\def\fc#1{\expandafter\def\csname f#1\endcsname{{\mathfrak #1}}}
\applytolist{fc}QWERTYUIOPLKJHGFDSAZXCVBNM;
\def\rmc#1{\expandafter\def\csname rm#1\endcsname{{\mathrm #1}}}
\applytolist{rmc}QWERTYUIOPLKJHGFDSAZXCVBNM;

\usepackage{tikz}
\usepackage{tikz-cd}

\usetikzlibrary{arrows,backgrounds,patterns.meta}
\usetikzlibrary{positioning,shadings,cd}
\usetikzlibrary{shapes}
\usetikzlibrary{backgrounds}
\usetikzlibrary{decorations,decorations.pathreplacing,decorations.markings,decorations.pathmorphing}
\usetikzlibrary{fit,calc,through}
\usetikzlibrary{external}
\usetikzlibrary{arrows}
\tikzset{vertex/.style = {shape=circle,draw,fill=black,inner sep=0pt,minimum size=5pt}}
\tikzset{edge/.style = {->,> = latex', bend right}}
\tikzset{
	super thick/.style={line width=3pt}
}
\tikzset{
    quadruple/.style args={[#1] in [#2] in [#3] in [#4]}{
        #1,preaction={preaction={preaction={draw,#4},draw,#3}, draw,#2}
    }
}
\tikzset{squiggly/.style={decorate, decoration=snake}}
\tikzstyle{knot}=[preaction={super thick, white, draw}]
\tikzstyle{shaded}=[fill=red!10!blue!20!gray!30!white]
\tikzstyle{unshaded}=[fill=white]
\tikzstyle{empty box}=[circle, draw, thick, fill=white, opaque, inner sep=2mm]
\tikzstyle{annular}=[scale=.7, inner sep=1mm, baseline]
\tikzstyle{rectangular}=[scale=.75, inner sep=1mm, baseline=-.1cm]
\tikzstyle{mid>}=[decoration={markings, mark=at position 0.5 with {\arrow{>}}}, postaction={decorate}]
\tikzstyle{mid<}=[decoration={markings, mark=at position 0.5 with {\arrow{<}}}, postaction={decorate}]
\tikzstyle{over}=[double, draw=white, super thick, double=]
\tikzstyle{snake}=[decorate, decoration={snake, segment length=1mm, amplitude=.3mm}]
\tikzstyle{saw}=[decorate, decoration={saw, segment length=.7mm, amplitude=.25mm}]

\tikzstyle{coupon}=[draw, very thick, rectangle, rounded corners=5pt]
\tikzset{Rightarrow/.style={double equal sign distance,>={Implies},->},
triplecd/.style={-,preaction={draw,Rightarrow}},
quadruplecd/.style={preaction={draw,Rightarrow,
shorten >=0pt
},
shorten >=1pt,
-,double,double
distance=0.2pt}}
\tikzset{
    tripleline/.style args={[#1] in [#2] in [#3]}{
        #1,preaction={preaction={draw,#3},draw,#2}
    }
}
\tikzstyle{triple}=[tripleline={[line width=.15mm,black] in
      [line width=.7mm,white] in
      [line width=1mm,black]}] 
\tikzset{
    quadrupleline/.style args={[#1] in [#2] in [#3] in [#4]}{
        #1,preaction={preaction={preaction={draw,#4},draw,#3}, draw,#2}
    }
}
\tikzstyle{quadruple}=[quadrupleline={[line width=.3mm,white] in
      [line width=.6mm,black] in
      [line width=1.2mm,white] in
      [line width=1.5mm,black]}]

\newcommand{\roundNbox}[6]{
	\draw[rounded corners=5pt, very thick, #1] ($#2+(-#3,-#3)+(-#4,0)$) rectangle ($#2+(#3,#3)+(#5,0)$);
	\coordinate (ZZa) at ($#2+(-#4,0)$);
	\coordinate (ZZb) at ($#2+(#5,0)$);
	\node at ($1/2*(ZZa)+1/2*(ZZb)$) {#6};
}

\newcommand{\halfDome}[4]{
        \draw[thick] ($ #1 + (#2,0) $) ellipse ({#2} and {#3});
        \draw[thick, dotted] ($ #1 $) arc(-180:0:#2);
	\node at ($ #1 + (#2,0) $) {#4};
}

\newcommand{\halfDomeNoDots}[4]{
        \draw[thick] ($ #1 + (#2,0) $) ellipse ({#2} and {#3});
        \draw[thick] ($ #1 $) arc(-180:0:#2);
	\node at ($ #1 + (#2,0) $) {#4};
}

\newcommand{\tikzmath}[2][]
     {\vcenter{\hbox{\begin{tikzpicture}[#1]#2
                     \end{tikzpicture}}}
     }

\theoremstyle{plain}
\newtheorem{thm}{Theorem}[section]
\newtheorem*{thm*}{Theorem}

\newtheorem*{cor*}{Corollary}

\newtheorem*{conj*}{Conjecture}
\newtheorem{lem}[thm]{Lemma}
\newtheorem*{lem*}{Lemma}

\newtheorem{prop}[thm]{Proposition}

\newtheorem*{quest*}{Question}
\newtheorem*{claim*}{Claim}

\theoremstyle{definition}
\newtheorem{defn}[thm]{Definition}
\newtheorem{fact}[thm]{Fact}

\newtheorem{ex}[thm]{Example}
\newtheorem{sub-ex}[thm]{Sub-Example}
\newtheorem{counter-ex}[thm]{Counter-Example}
\newtheorem{rem}[thm]{Remark}
\newtheorem*{rem*}{Remark}

\usepackage{xcolor}
\definecolor{dark-red}{rgb}{0.7,0.25,0.25}
\definecolor{dark-blue}{rgb}{0.15,0.15,0.55}
\definecolor{medium-blue}{rgb}{0,0,.8}
\definecolor{DarkGreen}{RGB}{0,150,0}
\definecolor{rho}{named}{red}
\definecolor{OIorange}{HTML}{E69F00}
\definecolor{OIskyblue}{HTML}{56B4E9}
\definecolor{OIbluishgreen}{HTML}{009E73}
\definecolor{OIyellow}{HTML}{F0E442}
\definecolor{OIblue}{HTML}{0072B2}
\definecolor{OIvermillion}{HTML}{D55E00}
\definecolor{OIreddishpurple}{HTML}{CC79A7}
\definecolor{OIblack}{HTML}{000000}
\hypersetup{
   colorlinks, linkcolor={purple},
   citecolor={medium-blue}, urlcolor={medium-blue}
}

\newcommand{\coev}{\operatorname{coev}}
\newcommand{\Dome}{\operatorname{Dome}}
\newcommand{\End}{\operatorname{End}}
\newcommand{\ev}{\operatorname{ev}}
\newcommand{\eval}{\operatorname{eval}}
\newcommand{\Forget}{\operatorname{Forget}}
\newcommand{\id}{\operatorname{id}}
\newcommand{\im}{\operatorname{im}}
\newcommand{\Irr}{\operatorname{Irr}}

\newcommand{\rev}{\operatorname{rev}}
\newcommand{\Sphere}{\operatorname{Sphere}}
\newcommand{\tot}{\operatorname{tot}}
\newcommand{\tr}{\operatorname{tr}}
\newcommand{\Tube}{\operatorname{Tube}}

\newcommand{\Hilb}{\mathsf{Hilb}}
\newcommand{\Mod}{\mathsf{Mod}}
\newcommand{\Rep}{\mathsf{Rep}}
\newcommand{\Vect}{\mathsf{Vect}}

\newcommand{\MR}[1]{}

\newcommand{\set}[2]{\left\{#1 \middle| #2\right\}}

\hypersetup{pdftitle=Enriched string-net models and their excitations,} 

\makeatletter
\renewcommand{\@printtitletextwithappropriatefontsize}{\@titleatfontsize{\LARGE}}
\makeatother

\title{Enriched string-net models and their excitations}
\author[1]{David Green}
\author[2]{Peter Huston}
\author[1]{Kyle Kawagoe}
\author[1]{David Penneys}
\author[1]{Anup Poudel}
\author[1]{Sean Sanford}
\affiliation[1]{The Ohio State University}
\affiliation[2]{Vanderbilt University}
\date{March 6, 2024}

\begin{document}

\maketitle
\begin{abstract}
Boundaries of Walker-Wang models have been used to construct commuting projector models which realize chiral unitary modular tensor categories (UMTCs) as boundary excitations.
Given a UMTC $\mathcal{A}$ representing the Witt class of an anomaly, the article [arXiv:2208.14018] gave a commuting projector model associated to an $\mathcal{A}$-enriched unitary fusion category $\cX$ on a 2D boundary of the 3D Walker-Wang model associated to $\mathcal{A}$.
That article claimed that the boundary excitations were given by the enriched center/M\"uger centralizer $Z^\mathcal{A}(\cX)$ of $\cA$ in $Z(\cX)$.

In this article, we give a rigorous treatment of this 2D boundary model, and we verify this assertion using topological quantum field theory (TQFT) techniques, including skein modules and a certain semisimple algebra whose representation category describes boundary excitations.
We also use TQFT techniques to show the 3D bulk point excitations of the Walker-Wang bulk are given by the M\"uger center $Z_2(\mathcal{A})$, and we construct bulk-to-boundary hopping operators $Z_2(\mathcal{A})\to Z^{\mathcal{A}}(\cX)$ reflecting how the UMTC of boundary excitations $Z^{\mathcal{A}}(\cX)$ is symmetric-braided enriched in $Z_2(\mathcal{A})$.

This article also includes a self-contained comprehensive review of the Levin-Wen string net model from a unitary tensor category viewpoint, as opposed to the skeletal 6j symbol viewpoint.

\end{abstract}

\section{Introduction}

One of the most powerful features of topologically ordered phases of matter \cite{Wen_2017B} is the robustness of their relevant properties under perturbation \cite{Wen_1990A,Wen_1990B}.  
As an example, despite the fact that Laughlin's wavefunction \cite{Laughlin_1983} is not a precise solution to the problem of an electron gas in the presence of a magnetic field, it accurately captures the phenomena of fractionalized charge and fractionally quantized Hall conductance. 
Furthering this philosophy, Levin and Wen wrote down their famous string-net models \cite{PhysRevB.71.045110}, which made concrete the connection between fusion categories and the physics they represent. 

String net models condense topological order from the vacuum; indeed, a fusion category $\cX$ is a condensable algebra in the symmetric monoidal 2-category $2\Vect$ \cite{1905.09566,MR4654609}.
The procedure to condense $\cX$ \cite{2206.02611} corresponds to passing to the ground state space of the commuting projector Levin-Wen local Hamiltonian \cite{PhysRevB.71.045110,PhysRevB.103.195155}, whose localized excitations can be described by the Drinfeld center $Z(\cX)$\IN{$Z(\cX)$ the Drinfeld center of $X$} \cite{MR2726654,1106.6033}, which is Witt-equivalent to the trivial topological order $\Vect$ \cite{MR3039775}.

It is widely accepted that (2+1)D commuting projector models produce exactly those unitary modular tensor categories (UMTCs) which are Drinfeld centers as their categories of localized anyonic excitations.
In particular, chiral UMTCs should only appear as boundaries of invertible (3+1)D topological quantum field theories (TQFTs) \cite{MR4444089,MR4334249}, which correspond to the anomaly \cite{MR3702386}.
This anomaly can be represented by a Witt-class \cite{MR3039775} of UMTCs \cite{MR4302495}. 
This is a manifestation of a bulk-boundary correspondence; such correspondences appear throughout topological physics \cite{Wen_1995,Wen_2013,PhysRevX.3.021009,1702.00673}.

To build a (3+1)D commuting projector model which realizes a chiral UMTC $\cC$ as boundary excitations, we first fix a UMTC $\cA$ representing the anomaly.
We then choose an $\cA$-enriched unitary fusion category $(\cX,\Phi^Z)$, i.e. a unitary fusion category (UFC) $\cX$ equipped with a unitary braided central functor $\Phi^Z:\cA\to Z(\cX)$ \cite[\S~II.B]{MR4640433}, so that the centralizer $Z^{\cA}(\cX)$ of $\Phi^Z(\cA)$ in $Z(\cX)$ is $\cC$.
The functor $\Phi^Z$ is the data necessary to attach a (2+1)D Levin-Wen model \cite{PhysRevB.71.045110} for $\cX$ as the boundary of a 
(3+1)D Walker-Wang model \cite{1104.2632} for $\cA$ (in red below).
We write $\Phi:= F\circ \Phi^Z: \cA\to \cX$, where $F: Z(\cX)\to \cX$ is the forgetful functor.
$$
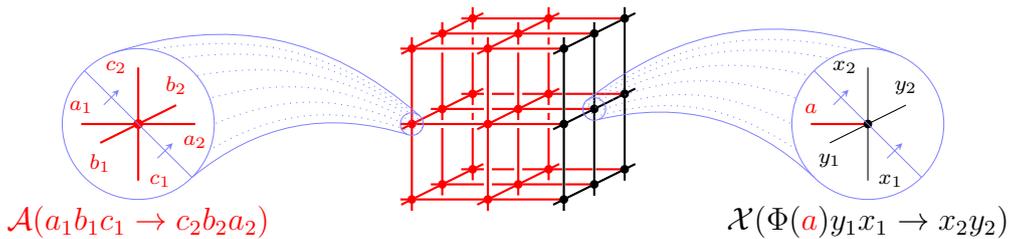
\begin{tikzpicture}
	\pgfmathsetmacro{\lattice}{1};	
	\pgfmathsetmacro{\xoffset}{.4};	
	\pgfmathsetmacro{\yoffset}{.2};	
	\pgfmathsetmacro{\extra}{.15};	
	\pgfmathsetmacro{\zextra}{.35};	
	\pgfmathsetmacro{\zoom}{6};	
	\coordinate (z1) at ($ -.6*(\zoom,0) + (0,\lattice) $);
	\coordinate (z2) at ($ (\zoom,\lattice) $);
	\coordinate (aaa) at (0,0);
	\coordinate (baa) at ($ (aaa) + (\lattice,0) $);
	\coordinate (caa) at ($ (aaa) + 2*(\lattice,0) $);
	\coordinate (aba) at ($ (aaa) + (0,\lattice) $);
	\coordinate (bba) at ($ (aaa) + (\lattice,0) + (0,\lattice) $);
	\coordinate (cba) at ($ (aaa) + 2*(\lattice,0) + (0,\lattice) $);
	\coordinate (aca) at ($ (aaa) + 2*(0,\lattice) $);
	\coordinate (bca) at ($ (aaa) + (\lattice,0) + 2*(0,\lattice) $);
	\coordinate (cca) at ($ (aaa) + 2*(\lattice,0) + 2*(0,\lattice) $);
	\coordinate (aab) at ($ (aaa) + (\xoffset,\yoffset) $);
	\coordinate (bab) at ($ (aaa) + (\lattice,0) + (\xoffset,\yoffset) $);
	\coordinate (cab) at ($ (aaa) + 2*(\lattice,0) + (\xoffset,\yoffset) $);
	\coordinate (abb) at ($ (aaa) + (0,\lattice) + (\xoffset,\yoffset) $);
	\coordinate (bbb) at ($ (aaa) + (\lattice,0) + (0,\lattice) + (\xoffset,\yoffset) $);
	\coordinate (cbb) at ($ (aaa) + 2*(\lattice,0) + (0,\lattice) + (\xoffset,\yoffset) $);
	\coordinate (acb) at ($ (aaa) + 2*(0,\lattice) + (\xoffset,\yoffset) $);
	\coordinate (bcb) at ($ (aaa) + (\lattice,0) + 2*(0,\lattice) + (\xoffset,\yoffset) $);
	\coordinate (ccb) at ($ (aaa) + 2*(\lattice,0) + 2*(0,\lattice) + (\xoffset,\yoffset) $);
	\coordinate (aac) at ($ (aaa) + 2*(\xoffset,\yoffset) $);
	\coordinate (bac) at ($ (aaa) + (\lattice,0) + 2*(\xoffset,\yoffset) $);
	\coordinate (cac) at ($ (aaa) + 2*(\lattice,0) + 2*(\xoffset,\yoffset) $);
	\coordinate (abc) at ($ (aaa) + (0,\lattice) + 2*(\xoffset,\yoffset) $);
	\coordinate (bbc) at ($ (aaa) + (\lattice,0) + (0,\lattice) + 2*(\xoffset,\yoffset) $);
	\coordinate (cbc) at ($ (aaa) + 2*(\lattice,0) + (0,\lattice) + 2*(\xoffset,\yoffset) $);
	\coordinate (acc) at ($ (aaa) + 2*(0,\lattice) + 2*(\xoffset,\yoffset) $);
	\coordinate (bcc) at ($ (aaa) + (\lattice,0) + 2*(0,\lattice) + 2*(\xoffset,\yoffset) $);
	\coordinate (ccc) at ($ (aaa) + 2*(\lattice,0) + 2*(0,\lattice) + 2*(\xoffset,\yoffset) $);
	\draw[thick, red] ($ (aac) - \extra*(\lattice,0) $) -- (cac);
	\draw[thick, red] ($ (abc) - \extra*(\lattice,0) $) -- (cbc);
	\draw[thick, red] ($ (acc) - \extra*(\lattice,0) $) -- (ccc);
	\draw[thick, red] ($ (aac) - \extra*(0,\lattice) $) -- ($ (acc) + \extra*(0,\lattice) $);
	\draw[thick, red] ($ (bac) - \extra*(0,\lattice) $) -- ($ (bcc) + \extra*(0,\lattice) $);
	\draw[thick] ($ (cac) - \extra*(0,\lattice) $) -- ($ (ccc) + \extra*(0,\lattice) $);
	\draw[thick, red, knot] ($ (aab) - \extra*(\lattice,0) $) -- (cab);
	\draw[thick, red, knot] ($ (abb) - \extra*(\lattice,0) $) -- (cbb);
	\draw[thick, red, knot] ($ (acb) - \extra*(\lattice,0) $) -- (ccb);
	\draw[thick, red, knot] ($ (aab) - \extra*(0,\lattice) $) -- ($ (acb) + \extra*(0,\lattice) $);
	\draw[thick, red, knot] ($ (bab) - \extra*(0,\lattice) $) -- ($ (bcb) + \extra*(0,\lattice) $);
	\draw[very thick, white] ($ (cab) - \extra*(0,\lattice) $) -- ($ (ccb) + \extra*(0,\lattice) $);
	\draw[thick] ($ (cab) - \extra*(0,\lattice) $) -- ($ (ccb) + \extra*(0,\lattice) $);
	\draw[thick, red, knot] ($ (aaa) - \extra*(\lattice,0) $) -- (caa);
	\draw[thick, red, knot] ($ (aba) - \extra*(\lattice,0) $) -- (cba);
	\draw[thick, red, knot] ($ (aca) - \extra*(\lattice,0) $) -- (cca);
	\draw[thick, red, knot] ($ (aaa) - \extra*(0,\lattice) $) -- ($ (aca) + \extra*(0,\lattice) $);
	\draw[thick, red, knot] ($ (baa) - \extra*(0,\lattice) $) -- ($ (bca) + \extra*(0,\lattice) $);
	\draw[very thick, white] ($ (caa) - \extra*(0,\lattice) $) -- ($ (cca) + \extra*(0,\lattice) $);
	\draw[thick] ($ (caa) - \extra*(0,\lattice) $) -- ($ (cca) + \extra*(0,\lattice) $);
	\draw[thick, red] ($ (aaa) - \zextra*(\xoffset,\yoffset) $) -- ($ (aac) + \zextra*(\xoffset,\yoffset) $);
	\draw[thick, red] ($ (aba) - \zextra*(\xoffset,\yoffset) $) -- ($ (abc) + \zextra*(\xoffset,\yoffset) $);
	\draw[thick, red] ($ (aca) - \zextra*(\xoffset,\yoffset) $) -- ($ (acc) + \zextra*(\xoffset,\yoffset) $);
	\draw[thick, red] ($ (baa) - \zextra*(\xoffset,\yoffset) $) -- ($ (bac) + \zextra*(\xoffset,\yoffset) $);
	\draw[thick, red] ($ (bba) - \zextra*(\xoffset,\yoffset) $) -- ($ (bbc) + \zextra*(\xoffset,\yoffset) $);
	\draw[thick, red] ($ (bca) - \zextra*(\xoffset,\yoffset) $) -- ($ (bcc) + \zextra*(\xoffset,\yoffset) $);
	\draw[thick] ($ (caa) - \zextra*(\xoffset,\yoffset) $) -- ($ (cac) + \zextra*(\xoffset,\yoffset) $);
	\draw[thick] ($ (cba) - \zextra*(\xoffset,\yoffset) $) -- ($ (cbc) + \zextra*(\xoffset,\yoffset) $);
	\draw[thick] ($ (cca) - \zextra*(\xoffset,\yoffset) $) -- ($ (ccc) + \zextra*(\xoffset,\yoffset) $);
	\filldraw[red] (aaa) circle (.05cm);
	\filldraw[red] (aba) circle (.05cm);
	\filldraw[red] (aca) circle (.05cm);
	\filldraw[red] (aab) circle (.05cm);
	\filldraw[red] (abb) circle (.05cm);
	\filldraw[red] (acb) circle (.05cm);
	\filldraw[red] (aac) circle (.05cm);
	\filldraw[red] (abc) circle (.05cm);
	\filldraw[red] (acc) circle (.05cm);
	\filldraw[red] (baa) circle (.05cm);
	\filldraw[red] (bab) circle (.05cm);
	\filldraw[red] (bac) circle (.05cm);
	\filldraw[red] (bba) circle (.05cm);
	\filldraw[red] (bbb) circle (.05cm);
	\filldraw[red] (bbc) circle (.05cm);
	\filldraw[red] (bca) circle (.05cm);
	\filldraw[red] (bcb) circle (.05cm);
	\filldraw[red] (bcc) circle (.05cm);
	\filldraw (caa) circle (.05cm);
	\filldraw (cba) circle (.05cm);
	\filldraw (cca) circle (.05cm);
	\filldraw (cab) circle (.05cm);
	\filldraw (cbb) circle (.05cm);
	\filldraw (ccb) circle (.05cm);
	\filldraw (cac) circle (.05cm);
	\filldraw (cbc) circle (.05cm);
	\filldraw (ccc) circle (.05cm);
	\draw[blue!50, very thin] (cbb) circle (.15cm);
		\foreach \x/\y/\s in {155/100/24, 175/120/23, 185/140/22, 200/160/23, 220/180/25, 240/200/26} 
		{\draw[dotted, blue!50] ($(cbb)+(\x:\extra)$) to[bend left=\s] ($(z2) + (\y:\lattice)$);}
		\foreach \x/\y/\s in {135/70/30,290/225/28}
		{\draw[blue!50, very thin] ($ (cbb) + (\x:\extra) $) to[bend left=\s] ($ (z2) + (\y:\lattice) $);}
			\draw[blue!50, very thin] (z2) circle (\lattice);
			\draw[thick, red] ($ (z2) - .75*(\lattice,0) $) node[above] {\scriptsize{$a$}} -- (z2);
			\draw ($ (z2) - .75*(0,\lattice) $) node[right] {\scriptsize{$x_1$}} -- ($ (z2) + .75*(0,\lattice) $) node[left] {\scriptsize{$x_2$}};
			\draw ($ (z2) - 1.25*(\xoffset,\yoffset) $) node[below] {\scriptsize{$y_1$}} -- ($ (z2) + 1.25*(\xoffset,\yoffset) $) node[above] {\scriptsize{$y_2$}};
			\filldraw (z2) circle (.05cm);
			\node at ($ (z2) -1.3*(0,\lattice) $) {$\cX(\Phi(\textcolor{red}{a})y_1x_1 \to x_2y_2)$};
			\draw[blue!50, very thin] ($ (z2) + (135:\lattice) $) -- ($ (z2) +(-45:\lattice) $);
			\foreach \x in {135,-45}
			{\draw[blue!50, very thin, -stealth] ($ (z2) + (\x:.5*\lattice) - (.1,.1)$) to ($ (z2) + (\x:.5*\lattice) + (.1,.1)$);}
	\pgftransformxscale{-1}
	\draw[blue!50, very thin] (aba) circle (.15cm);
		\foreach \x/\y/\s in {155/100/24, 175/120/23, 185/140/22, 200/160/23, 220/180/25, 240/200/26} 
		{\draw[dotted, blue!50] ($(aba)+(\x:\extra)$) to[bend left=\s] ($(z1) + (\y:\lattice)$);}
		\foreach \x/\y/\s in {135/70/30,290/225/28}
		{\draw[blue!50, very thin] ($ (aba) + (\x:\extra) $) to[bend left=\s] ($ (z1) + (\y:\lattice) $);}
			\draw[blue!50, very thin] (z1) circle (\lattice);
			\draw[thick, red] ($ (z1) - .75*(\lattice,0) $) node[below] {\scriptsize{$a_2$}} -- ($ (z1) + .75*(\lattice,0) $) node[above] {\scriptsize{$a_1$}};
			\draw[thick, red] ($ (z1) - .75*(0,\lattice) $) node[right] {\scriptsize{$c_1$}} -- ($ (z1) + .75*(0,\lattice) $) node[left] {\scriptsize{$c_2$}};
			\draw[thick, red] ($ (z1) - 1.25*(\xoffset,0) + 1.25*(0,\yoffset) $) node[above] {\scriptsize{$b_2$}} -- ($ (z1) + 1.25*(\xoffset,0) - 1.25*(0,\yoffset) $) node[below] {\scriptsize{$b_1$}};
			\filldraw[thick, red] (z1) circle (.05cm);
			\node at ($ (z1) -1.3*(0,\lattice) $) {\textcolor{red}{$\cA(a_1b_1c_1 \to c_2b_2a_2)$}};
			\draw[blue!50, very thin] ($ (z1) + (225:\lattice) $) -- ($ (z1) + (45:\lattice) $);
			\foreach \x in {225,45}
			{\draw[blue!50, very thin, -stealth] ($ (z1) + (\x:.5*\lattice) + (.1,-.1)$) to ($ (z1) + (\x:.5*\lattice) - (.1,-.1)$);}
\end{tikzpicture}
$$
The local Hamiltonian is made up of 1D edge terms and 2D plaquette terms for every 2D square in the lattice. The half-braiding for $\Phi(a)$ with $\cX$ is used in a crucial way to define the plaquette terms at the interface of the 3D bulk with the 2D boundary. 
We refer the reader to \S\ref{sec:Enriched} below for more details.
We remark that related (2+1)D boundary theories for (3+1)D Walker-Wang models also appear in \cite{PhysRevB.90.245122,PhysRevB.87.045107}.

These $\cA$-enriched string net models condense topological order Witt equivalent to $\cA$ from $\cA$-topological order;
indeed an $\cA$-enriched fusion category is a condensable algebra in the fusion 2-category $\Mod(\cA)$ \cite[Ex.~1.4.6]{MR4561584}.
Condensing $\cX$ (cf.~\cite{2206.02611}) corresponds to passing to the ground state space of our $\cA$-enriched string net local Hamiltonian, whose localized excitations can be described by the enriched center $Z^\cA(\cX):=\cA'\subset Z(\cX)$\IN{$Z^\cA(\cX)$ the enriched center of $\cX$} \cite{MR3725882},
which is Witt-equivalent to $\cA$ \cite{MR3039775}.

It was claimed in \cite[\S~{II}.B]{MR4640433} that the topological excitations which live on the boundary of the $\cA$-enriched string net model for $\cX$ are described by the enriched center $Z^{\cA}(\cX)$, i.e., the M\"uger centralizer of $\cA$ in $Z(\cX)$.
We give a rigorous proof of this fact using
TQFT techniques, including $\cA$-enriched skein modules for $\cX$ and the \emph{dome algebra} $\Dome^{\cA}(\cX)$, a certain quotient of the tube algebra $\Tube(\cX)$ \cite{MR1832764,MR1966525}.
In an ordinary (2+1)D string-net model, the algebra $\Tube(\cX)$ acts as local operators near the site of a topological excitation, so that anyon types correspond to (isomorphism classes of) indecomposable representations of $\Tube(\cX)$.
The presence of the (3+1)D $\cA$-bulk induces additional requirements on topological local operators, which cuts down $\Tube(\cX)$ to $\Dome^{\cA}(\cX)$.

The structure of this paper is as follows. 
First, \S\ref{sec:unenriched} is a self-contained primer on the Levin-Wen string net model \cite{PhysRevB.71.045110} from a unitary tensor category viewpoint.
For simplicity, we use a square lattice, and all degrees of freedom occur at vertices, similar to \cite{MR3204497}.
In \S\ref{sec:SkeinModules}, we define the disk skein module associated to a UFC $\cX$, and prove basic results about skein modules which will be used throughout the paper.
In \S\ref{sec:StringNet}, we provide some helpful conceptual proofs of well-known facts in the literature, including the matrix coefficients for the plaquette operators in Lemma \ref{lem:PlaquetteCoefficients} and the proof that the product of all plaquette operators is a scalar multiple of the projection to the skein module in Theorem \ref{thm:ProjectToSkeinModule}.
We also give an elegant description of the ground state space of the edge terms in terms of hom spaces in the UFC $\cX$ in Proposition \ref{prop:DescriptionOfIm(PA)}.
In \S\ref{sec:tubealgebra}, we review the basics of the tube algebra $\Tube(\cX)$ and give the well-known correspondence between its category of representations and objects in $Z(\cX)$ \cite{MR1832764}.
In \S\ref{sec:StringOperators}, we fully describe excitations localized on one edge and its two adjacent plaquettes, together with a shortened and simplified presentation of general string operators and their connection to $\Tube(\cX)$-representations from \cite{MR4642306}.
Finally, in \S\ref{sec:Hoppingoperators}, we give a shortened and simplified presentation of hopping operators from \cite{MR4642306}.

Next, we analyze the $\cA$-enriched string net model for an $\cA$-enriched UFC introduced in \cite[II.B]{MR4640433} in \S\ref{sec:Enriched}.
We do so in more generality, assuming that $\cA$ is unitary braided (but not necessarily modular), but still requiring $\Phi^Z: \cA\to Z(\cX)$ to be fully faithful.
In \S\ref{sec:EnrichedStringNets}, we give explicit detail on the plaquette operators and their connection to $\cA$-enriched skein modules.
In \S\ref{sec:DomeAlgebra}, we introduce the \emph{dome algebra} $\Dome^\cA(\cX)$, communicated to us by Corey Jones and inspired by Kevin Walker, which is a quotient of $\Tube(\cX)$ by an ideal $\cJ^\cA$ coming from the $\cA$-enrichment.
The point of $\Dome^\cA(\cX)$ is Theorem \ref{thm:DomeReps=EnrichedCenter} which shows that $\Rep(\Dome^\cA(\cX))\cong Z^\cA(\cX)$, the enriched center \cite{MR3725882} of $\cX$.
We then show that our $\Tube(\cX)$-action on the $\cX$ string net model also gives an action on our $\cA$-enriched model which descends to an action of $\Dome^\cA(\cX)$. 
Our ideal $\cJ^\cA$ acts as zero on $\Tube(\cX)$-representations from precisely those string operators corresponding to anyons in $\Irr(Z^{\cA}(\cX))$, so localized enriched excitations exactly correspond to the enriched center $Z^\cA(\cX)$ as claimed in \cite{MR4640433}.
Finally, in \S\ref{sec:SphereAlgebra}, we analyze bulk excitations using the \emph{sphere algebra} of $\cA$, again communicated to us by Corey Jones and inspired by Kevin Walker.
Our analysis is consistent with \cite[Lem.~2.16]{MR4654609}; the point excitations in the bulk are $\Omega Z(\Sigma \cA)= Z_2(\cA)$.
Under the additional assumption that the braided central functor $\Phi^Z:\cA\to Z(\cX)$ is full, $Z_2(Z^\cA(\cX))=Z_2(\cA)$ by \cite[Prop.~4.3]{MR3022755}, and the canonical map is exactly bringing a bulk excitation to the boundary via a modified hopping operator.

\subsection*{Acknowledgements}
The authors are indebted to Corey Jones for numerous ideas in this article, including the statement of Proposition \ref{prop:DescriptionOfIm(PA)}, the definitions of the dome and sphere algebras, and a modified string net model which can host excitations at vertices in Remark \ref{rem:HostExcitationsAtVertices}.
We also would like to further thank Kevin Walker for his generosity with his ideas and for sharing his insight in many helpful conversations over the years.
The authors would like to further thank
Fiona Burnell,
Daniel Wallick,
and
Dominic Williamson
for helpful discussions.
KK was partially supported by the NSF grant DMS 1654159 as well as the Center for Emergent Materials, an NSF-funded MRSEC, under Grant No. DMR-2011876.
DG and DP were supported by NSF grants DMS 1654159 and 2154389.
PH was partially supported by US ARO grant W911NF2310026. 

\section{String net models for unitary fusion categories}\label{sec:unenriched}

Here we review the commuting projector Levin-Wen string net model built from a unitary fusion category (UFC) $\cX$, including the braided category of excitations and string operators. 
General references include
\cite{MR3242743} for fusion categories,
\cite{MR3663592,MR4133163} for UFCs, and 
\cite{PhysRevB.71.045110,MR3204497,PhysRevB.103.195155} for 2D Levin-Wen string net models. 
In particular, we always equip our UFC with its canonical unitary spherical structure, allowing us to form closed diagrams on a sphere which are invariant under isotopy. 

We view these 2D string net models as condensing topological order from the vacuum, which is mathematically supported by \cite{1905.09566,2206.02611,MR4444089} and informs our diagrammatic calculus of shaded regions. 
Starting from the vacuum, from any UFC $\cX$ one can attach a finite dimensional Hilbert space built from hom spaces in $\cX$ to every vertex of a square lattice as in the first picture below. 
Enforcing the consistency condition that labels of simple objects on edges/links (the $A_\ell$ terms) has the effect of condensing a 1-skeleton. 
Passing to the image of the plaquette/face operators (the $B_p$ terms) allows isotopy across each plaquette, effectively condensing 2-cells onto the 1-skeleton.
$$
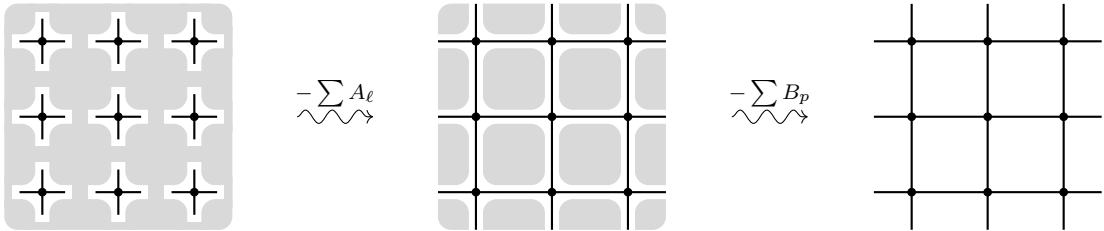
\tikzmath{
\clip[rounded corners=5pt] (-.5,-.5) rectangle (2.5,2.5);
\foreach \x in {0,1,2,3}{
\foreach \y in {0,1,2}{
\filldraw (\x,\y) circle (.05cm);
\draw[thick] ($ (\x,\y) - (.3,0) $) -- ($ (\x,\y) + (.3,0) $);
\draw[thick] ($ (\x,\y) - (0,.3) $) -- ($ (\x,\y) + (0,.3) $);
\fill[gray!30] ($ (\x,\y) + (-.4,-.2) $) rectangle ($ (\x,\y) + (-.6,.2) $);
\fill[gray!30] ($ (\x,\y) + (.4,-.2) $) rectangle ($ (\x,\y) + (.6,.2) $);
\fill[gray!30] ($ (\x,\y) + (-.2,-.4) $) rectangle ($ (\x,\y) + (.2,-.6) $);
\fill[gray!30] ($ (\x,\y) + (-.2,.4) $) rectangle ($ (\x,\y) + (.2,.6) $);
}}
\foreach \x in {-.5,...,3.5}{
\foreach \y in {-.5,...,2.5}{
\filldraw[gray!30, rounded corners=5pt] ($ (\x,\y) - (.4,.4)$) rectangle ($ (\x,\y) + (.4,.4)$);
}}
}
\qquad
\overset{-\sum A_\ell}{
\tikzmath{\draw[squiggly,->] (0,0) -- (1,0);}
}
\qquad
\tikzmath{
\clip[rounded corners=5pt] (-.5,-.5) rectangle (2.5,2.5);
\draw[thick] (-.5,-.5) grid (3.5,2.5);
\foreach \x in {0,1,2,3}{
\foreach \y in {0,1,2}{
\filldraw (\x,\y) circle (.05cm);
}}
\foreach \x in {-.5,...,3.5}{
\foreach \y in {-.5,...,2.5}{
\filldraw[gray!30, rounded corners=5pt] ($ (\x,\y) - (.4,.4)$) rectangle ($ (\x,\y) + (.4,.4)$);
}}
}
\qquad
\overset{-\sum B_p}{
\tikzmath{\draw[squiggly,->] (0,0) -- (1,0);}
}
\qquad
\tikzmath{
\clip[rounded corners=5pt] (-.5,-.5) rectangle (2.5,2.5);
\draw[thick] (-.5,-.5) grid (3.5,2.5);
\foreach \x in {0,1,2,3}{
\foreach \y in {0,1,2}{
\filldraw (\x,\y) circle (.05cm);
}}
\foreach \x in {-.5,...,3.5}{
\foreach \y in {-.5,...,2.5}{
}}
}
$$

\subsection{Skein modules and string nets}
\label{sec:SkeinModules}

For this section, $\cX$\IN{$\cX$ a unitary fusion category} is a unitary fusion category and $\Irr(\cX)$ is the set of isomorphism classes of simple objects of $\cX$. \IN{$\Irr(\cX)$ the set of isomorphism classes of simple objects of $\cX$}

\begin{defn}
The disk skein module $\cS_\cX(\bbD,n)$\IN{$\cS_\cX(\bbD,n)$ the disk skein module} for $\cX$ with $n$ boundary points is the Hilbert space orthogonal direct sum
$$
\bigoplus_{x_1,\dots, x_n \in \Irr(\cX)} \cX(1\to x_1\otimes \cdots \otimes x_n)
=
\left\{
\tikzmath{
\roundNbox{}{(0,0)}{.3}{.2}{.2}{$f$}
\draw (-.3,.3) --node[left]{$\scriptstyle x_1$} (-.3,.7) ;
\draw (.3,.3) --node[right]{$\scriptstyle x_n$} (.3,.7) ;
\node at (90:.5cm) {$\cdots$};
}
\right\}
$$
where each $\cX(1\to x_1\otimes \cdots \otimes x_n)$ is equipped with the inner product
$$
\langle f | g\rangle
:=
\frac{1}{\sqrt{d_{x_1}\cdots d_{x_n}}} 
f^\dag \circ g
=
\frac{1}{\sqrt{d_{x_1}\cdots d_{x_n}}} 
\cdot
\tikzmath{
\roundNbox{}{(0,1)}{.3}{.2}{.2}{$f^\dag$}
\roundNbox{}{(0,0)}{.3}{.2}{.2}{$g$}
\draw (-.3,.3) --node[left]{$\scriptstyle x_1$} (-.3,.7) ;
\draw (.3,.3) --node[right]{$\scriptstyle x_n$} (.3,.7) ;
\node at (90:.5cm) {$\cdots$};
}\,.
$$
\end{defn}

\begin{ex}[Trivalent skein module]
Using the canonical unitary spherical structure of $\cX$, the skein module $\cS_\cX(\bbD,3)$ can also be viewed as the orthogonal direct sum $\bigoplus_{a,b,c\in\Irr(\cC)}\cX(a b \to c)$ with inner product
$$
\left\langle 
\tikzmath{
\draw (0,.3) --node[left]{$\scriptstyle z$} (0,.7);
\draw (-.15,-.3) --node[left]{$\scriptstyle x$} (-.15,-.7);
\draw (.15,-.3) --node[right]{$\scriptstyle y$} (.15,-.7);
\roundNbox{}{(0,0)}{.3}{0}{0}{$\xi$}
}
\middle| 
\tikzmath{
\draw (0,.3) --node[left]{$\scriptstyle z'$} (0,.7);
\draw (-.15,-.3) --node[left]{$\scriptstyle x'$} (-.15,-.7);
\draw (.15,-.3) --node[right]{$\scriptstyle y'$} (.15,-.7);
\roundNbox{}{(0,0)}{.3}{0}{0}{$\xi'$}
}
\right\rangle
:=
\delta_{x=x'}
\delta_{y=y'}
\delta_{z=z'}
\frac{1}{\sqrt{d_xd_yd_z}}
\cdot
\tikzmath{
\draw (0,1.3) arc (180:0:.3cm) --node[right]{$\scriptstyle \overline{z}$} (.6,-.3) arc (0:-180:.3cm);
\draw (-.15,.3) --node[left]{$\scriptstyle x$} (-.15,.7);
\draw (.15,.3) --node[right]{$\scriptstyle y$} (.15,.7);
\roundNbox{}{(0,1)}{.3}{0}{0}{$\xi'$}
\roundNbox{}{(0,0)}{.3}{0}{0}{$\xi^\dag$}
}\,.
$$
We draw a pair of shaded nodes 
$$
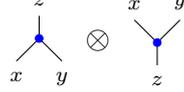
\tikzmath{
\draw (-.3,-.3) node[below]{$\scriptstyle x$} -- (0,0); 
\draw (.3,-.3) node[below]{$\scriptstyle y$} -- (0,0); 
\draw (0,.3) node[above]{$\scriptstyle z$} -- (0,0); 
\filldraw[blue] (0,0) circle (.05cm);
}
\otimes
\tikzmath{
\draw (-.3,.3) node[above]{$\scriptstyle x$} -- (0,0); 
\draw (.3,.3) node[above]{$\scriptstyle y$} -- (0,0); 
\draw (0,-.3) node[below]{$\scriptstyle z$} -- (0,0); 
\filldraw[blue] (0,0) circle (.05cm);
}
$$
to indicate summing over an orthonormal basis (ONB) for $\cS_\cX(\bbD,3)$ and its dagger.
We have the following relations for trivalent shaded nodes as in \cite[\S2.5]{MR3663592}:
\begin{align}
\begin{tikzpicture}[baseline=-.1cm]
	\draw (0,-.6) -- (0,-.3);
	\draw (0,-.3) .. controls ++(45:.2cm) and ++(-45:.2cm) .. (0,.3);
	\draw (0,-.3) .. controls ++(135:.2cm) and ++(225:.2cm) .. (0,.3);
	\draw (0,.6) -- (0,.3);
	\fill[fill=blue] (0,-.3) circle (.05cm);
	\fill[fill=blue] (0,.3) circle (.05cm);	
	\node at (0,-.8) {\scriptsize{$z$}};
	\node at (.3,0) {\scriptsize{$y$}};
	\node at (-.3,0) {\scriptsize{$x$}};
	\node at (0,.8) {\scriptsize{$z$}};
\end{tikzpicture}
=& \,\,\,
\sqrt{d_xd_yd_z^{-1}}\cdot N_{x,y}^z\,
\begin{tikzpicture}[baseline=-.1cm]
	\draw (-.2,-.6) -- (-.2,.6);
	\node at (-.2,-.8) {\scriptsize{$z$}};
\end{tikzpicture}
&&
\tag{Bigon 1}
\label{eq:Bigon 1}
\\
\begin{tikzpicture}[baseline=-.1cm]
	\draw (0,-.6) -- (0,-.3);
	\draw (0,-.3) .. controls ++(45:.2cm) and ++(-45:.2cm) .. (0,.3);
	\draw (0,-.3) .. controls ++(135:.2cm) and ++(225:.2cm) .. (0,.3);
	\draw (0,.6) -- (0,.3);
	\fill[fill=blue] (0,-.3) circle (.05cm);
	\fill[fill=orange] (0,.3) circle (.05cm);	
	\node at (0,-.8) {\scriptsize{$z$}};
	\node at (.3,0) {\scriptsize{$y$}};
	\node at (-.3,0) {\scriptsize{$x$}};
	\node at (0,.8) {\scriptsize{$z$}};
\end{tikzpicture}
\otimes
\begin{tikzpicture}[baseline=.2cm]
	\draw (.2,.6) -- (0,.3) -- (-.2,.6);
	\draw (0,0) -- (0,.3);
	\fill[fill=orange] (0,.3) circle (.05cm);	
	\node at (-.2,.8) {\scriptsize{$x$}};
	\node at (.2,.8) {\scriptsize{$y$}};
	\node at (0,-.2) {\scriptsize{$z$}};
\end{tikzpicture}
\otimes
\begin{tikzpicture}[baseline=-.4cm]
	\draw (.2,-.6) -- (0,-.3) -- (-.2,-.6);
	\draw (0,0) -- (0,-.3);
	\fill[fill=blue] (0,-.3) circle (.05cm);	
	\node at (-.2,-.8) {\scriptsize{$x$}};
	\node at (.2,-.8) {\scriptsize{$y$}};
	\node at (0,.2) {\scriptsize{$z$}};
\end{tikzpicture}
&\,= \,\,
\sqrt{d_xd_yd_z^{-1}}\cdot
\begin{tikzpicture}[baseline=-.1cm]
	\draw (-.2,-.6) -- (-.2,.6);
	\node at (-.2,-.8) {\scriptsize{$z$}};
\end{tikzpicture}
\otimes
\begin{tikzpicture}[baseline=.2cm]
	\draw (.2,.6) -- (0,.3) -- (-.2,.6);
	\draw (0,0) -- (0,.3);
	\fill[fill=blue] (0,.3) circle (.05cm);	
	\node at (-.2,.8) {\scriptsize{$x$}};
	\node at (.2,.8) {\scriptsize{$y$}};
	\node at (0,-.2) {\scriptsize{$z$}};
\end{tikzpicture}
\otimes
\begin{tikzpicture}[baseline=-.4cm]
	\draw (.2,-.6) -- (0,-.3) -- (-.2,-.6);
	\draw (0,0) -- (0,-.3);
	\fill[fill=blue] (0,-.3) circle (.05cm);	
	\node at (-.2,-.8) {\scriptsize{$x$}};
	\node at (.2,-.8) {\scriptsize{$y$}};
	\node at (0,.2) {\scriptsize{$z$}};
\end{tikzpicture}
&&
\tag{Bigon 2}
\label{eq:Bigon 2}
\\
\sum_{z\in \Irr(\cX)}
\sqrt{d_z}&
\begin{tikzpicture}[baseline=-.1cm]
	\draw (.2,-.6) -- (0,-.3) -- (-.2,-.6);
	\draw (.2,.6) -- (0,.3) -- (-.2,.6);
	\draw (0,-.3) -- (0,.3);
	\fill[fill=blue] (0,-.3) circle (.05cm);
	\fill[fill=blue] (0,.3) circle (.05cm);	
	\node at (-.2,-.8) {\scriptsize{$x$}};
	\node at (.2,-.8) {\scriptsize{$y$}};
	\node at (-.2,.8) {\scriptsize{$x$}};
	\node at (.2,.8) {\scriptsize{$y$}};
	\node at (.2,0) {\scriptsize{$z$}};
\end{tikzpicture}
\,=\,\sqrt{d_xd_y}\cdot
\begin{tikzpicture}[baseline=-.1cm]
	\draw (.2,-.6) -- (.2,.6);
	\draw (-.2,-.6) -- (-.2,.6);
	\node at (-.2,-.8) {\scriptsize{$x$}};
	\node at (.2,-.8) {\scriptsize{$y$}};
\end{tikzpicture}
&&
\tag{Fusion}
\label{eq:Fusion}
\\
\qquad
\sum_{v\in \Irr(\cX)}
\begin{tikzpicture}[baseline=-.1cm]
	\draw (.2,-.6) -- (0,-.3) -- (-.2,-.6);
	\draw (.2,.6) -- (0,.3) -- (-.2,.6);
	\draw (0,-.3) -- (0,.3);
	\fill[fill=blue] (0,-.3) circle (.05cm);
	\fill[fill=orange] (0,.3) circle (.05cm);
	\node at (-.2,-.8) {\scriptsize{$x$}};
	\node at (.2,-.8) {\scriptsize{$w$}};
	\node at (-.2,.8) {\scriptsize{$y$}};
	\node at (.2,.8) {\scriptsize{$z$}};
	\node at (.2,0) {\scriptsize{$v$}};
\end{tikzpicture}
\,\otimes
\begin{tikzpicture}[baseline=-.1cm]
	\draw (.2,-.6) -- (0,-.3) -- (-.2,-.6);
	\draw (.2,.6) -- (0,.3) -- (-.2,.6);
	\draw (0,-.3) -- (0,.3);
	\fill[fill=blue] (0,-.3) circle (.05cm);
	\fill[fill=orange] (0,.3) circle (.05cm);
	\node at (-.2,-.8) {\scriptsize{$\overline w$}};
	\node at (.2,-.8) {\scriptsize{$\overline x$}};
	\node at (-.2,.8) {\scriptsize{$\overline z$}};
	\node at (.2,.8) {\scriptsize{$\overline y$}};
	\node at (.2,0) {\scriptsize{$\overline v$}};
\end{tikzpicture}
&\,=
\sum_{u\in \Irr(\cX)}
\begin{tikzpicture}[baseline=-.1cm, rotate=90]
	\draw (.4,-.5) -- (-.03,-.3) -- (-.4,-.5);
	\draw (.4,.5) -- (.03,.3) -- (-.4,.5);
	\draw (-.03,-.3) -- (.03,.3);
	\fill[fill=blue] (-.03,-.3) circle (.05cm);
	\fill[fill=orange] (.03,.3) circle (.05cm);
	\node at (-.6,-.5) {\scriptsize{$w$}};
	\node at (.6,-.5) {\scriptsize{$z$}};
	\node at (-.6,.5) {\scriptsize{$x$}};
	\node at (.6,.5) {\scriptsize{$y$}};
	\node at (.2,0) {\scriptsize{$u$}};
\end{tikzpicture}
\otimes
\begin{tikzpicture}[baseline=-.1cm, rotate=90]
	\draw (.4,-.5) -- (.03,-.3) -- (-.4,-.5);
	\draw (.4,.5) -- (-.03,.3) -- (-.4,.5);
	\draw (.03,-.3) -- (-.03,.3);
	\fill[fill=orange] (.03,-.3) circle (.05cm);
	\fill[fill=blue] (-.03,.3) circle (.05cm);
	\node at (-.6,-.5) {\scriptsize{$\overline x$}};
	\node at (.6,-.5) {\scriptsize{$\overline y$}};
	\node at (-.6,.5) {\scriptsize{$\overline w$}};
	\node at (.6,.5) {\scriptsize{$\overline z$}};
	\node at (.2,0) {\scriptsize{$\overline u$}};
\end{tikzpicture}
&&
\tag{I=H}
\label{eq:I=H}
\end{align}
\end{ex}

\begin{fact}
\label{fact:VaccuumSeparate}
Given $f: 1\to yx$ and $g: yx\to 1$,
$$
\tikzmath{
\draw (-.2,.3) -- node[left]{$\scriptstyle y$} (-.2,1.3);
\draw (.2,.3) -- node[right]{$\scriptstyle x$} (.2,1.3);
\roundNbox{}{(0,0)}{.3}{.2}{.2}{$f$}
\roundNbox{}{(0,1.6)}{.3}{.2}{.2}{$g$}
}
\underset{\text{\eqref{eq:Fusion}}}{=}
\sum_{y\in \Irr(\cX)}
\frac{\sqrt{d_z}}{\sqrt{d_xd_y}}\cdot
\tikzmath{
\draw (-.2,.3) -- node[left]{$\scriptstyle y$} (0,.6);
\draw (-.2,1.3) -- node[left]{$\scriptstyle y$} (0,1);
\draw (.2,.3) -- node[right]{$\scriptstyle x$} (0,.6);
\draw (.2,1.3) -- node[right]{$\scriptstyle x$} (0,1);
\draw (0,.6) -- node[right]{$\scriptstyle z$} (0,1);
\filldraw[blue] (0,.6) circle (.05cm);
\filldraw[blue] (0,1) circle (.05cm);
\roundNbox{}{(0,0)}{.3}{.2}{.2}{$f$}
\roundNbox{}{(0,1.6)}{.3}{.2}{.2}{$g$}
}
=
\frac{\delta_{y=\overline{x}}}{d_x}\cdot
\tikzmath{
\draw (-.2,.3) node[above, xshift=-.15cm]{$\scriptstyle \overline{x}$} arc (180:0:.2cm) node[above, xshift=.15cm]{$\scriptstyle x$};
\draw (-.2,1.3) node[below, xshift=-.15cm]{$\scriptstyle \overline{x}$} arc (-180:0:.2cm) node[below, xshift=.15cm]{$\scriptstyle x$};
\roundNbox{}{(0,0)}{.3}{.2}{.2}{$f$}
\roundNbox{}{(0,1.6)}{.3}{.2}{.2}{$g$}
}\,.
$$
Indeed, there is a non-zero map $1_\cX\to z$ if and only if $z=1_\cX$.
In this case,
$\overline{y}=x$ and $\{\ev_x\}$ is an ONB of $\cX(\overline{x}x\to 1_\cX)\subset \cS_\cX(\bbD,3)$.
\end{fact}

We define a gluing map $\operatorname{gl}:\cS_\cX(\bbD,m+1)\otimes \cS_\cX(\bbD,n+1)\to \cS_\cX(\bbD,m+n)$
by
\begin{equation}
\label{eq:GluingMap}
\tikzmath{
\roundNbox{}{(0,0)}{.3}{.6}{.6}{$f$}
\draw (-.7,.3) -- (-.7,.7) node[above]{$\scriptstyle x_1$};
\draw (0,.3) -- (0,.7) node[above]{$\scriptstyle x_m$};
\draw (.7,.3) -- (.7,.7) node[above]{$\scriptstyle x_{m+1}$};
\node at (-.3,.55) {$\cdots$};
}
\otimes
\tikzmath{
\roundNbox{}{(0,0)}{.3}{.6}{.6}{$g$}
\draw (-.7,.3) -- (-.7,.7) node[above]{$\scriptstyle y_1$};
\draw (0,.3) -- (0,.7) node[above]{$\scriptstyle y_2$};
\draw (.7,.3) -- (.7,.7) node[above]{$\scriptstyle y_{n+1}$};
\node at (.35,.55) {$\cdots$};
}
\longmapsto
\delta_{x_{m+1}=\overline{y_1}}
\cdot
\tikzmath{
\roundNbox{}{(2,0)}{.3}{.6}{.6}{$g$}
\roundNbox{}{(0,0)}{.3}{.6}{.6}{$f$}
\draw (-.7,.3) -- (-.7,.7) node[above]{$\scriptstyle x_1$};
\draw (0,.3) -- (0,.7) node[above]{$\scriptstyle x_m$};
\draw (.7,.3) node[above, xshift=-.15cm]{$\scriptstyle \overline{y_1}$} arc (180:0:.3cm) node[above, xshift=.15cm]{$\scriptstyle y_1$};
\draw (2,.3) -- (2,.7) node[above]{$\scriptstyle y_2$};
\draw (2.7,.3) -- (2.7,.7) node[above]{$\scriptstyle y_n$};
\node at (-.3,.55) {$\cdots$};
\node at (2.35,.55) {$\cdots$};
}\,.
\end{equation}

\begin{lem}
\label{lem:GluingUnitary}
The gluing map \eqref{eq:GluingMap} restricts to a unitary on the subspace of $\cS_\cX(\bbD,m+1)\otimes \cS_\cX(\bbD,n+1)$ where $x_{m+1}=\overline{y_1}$.
In particular, gluing two ONBs for $\cS_\cX(\bbD,m+1)$ and $\cS_\cX(\bbD,n+1)$ 
whose string labels match appropriately
gives an ONB for $\cS_\cX(\bbD,m+n)$.
\end{lem}
\begin{proof}
When the appropriate labels of $f,g$ match ($x_{m+1}=\overline{y_1}$)
as well as those of $f',g'$ ($x_{m+1}'=\overline{y_1'}$),
applying Fact \ref{fact:VaccuumSeparate} in the sideways direction yields 
\begin{align*}
\langle \operatorname{gl}(f\otimes g)|\operatorname{gl}(f'\otimes g')\rangle
&=
\left(
\prod_{i=1}^{m}
\frac{\delta_{x_i=x_i'}}{\sqrt{d_{x_i}}}
\prod_{j=2}^{n+1}
\frac{\delta_{y_j=y_j'}}{\sqrt{d_{y_j}}}
\right)
\cdot
\tikzmath{
\roundNbox{}{(2,0)}{.3}{.6}{.6}{$g'$}
\roundNbox{}{(0,0)}{.3}{.6}{.6}{$f'$}
\roundNbox{}{(2,2)}{.3}{.6}{.6}{$g^\dag$}
\roundNbox{}{(0,2)}{.3}{.6}{.6}{$f^\dag$}
\draw (-.7,.3) -- node[left]{$\scriptstyle x_1$} (-.7,1.7);
\draw (0,.3) -- node[right]{$\scriptstyle x_m$} (0,1.7);
\draw (.7,.3) node[above, xshift=-.15cm]{$\scriptstyle \overline{y_1'}$} arc (180:0:.3cm) node[above, xshift=.15cm]{$\scriptstyle y_1'$};
\draw (.7,1.7) node[below, xshift=-.15cm]{$\scriptstyle \overline{y_1}$} arc (-180:0:.3cm) node[below, xshift=.15cm]{$\scriptstyle y_1$};
\draw (2,.3) -- node[left]{$\scriptstyle y_2$} (2,1.7) ;
\draw (2.7,.3) -- node[right]{$\scriptstyle y_n$} (2.7,1.7);
\node at (-.3,1) {$\cdots$};
\node at (2.35,1) {$\cdots$};
}
\\&=
\left(
\prod_{i=1}^{m+1}
\frac{\delta_{x_i=x_i'}}{\sqrt{d_{x_i}}}
\prod_{j=1}^{n+1}
\frac{\delta_{y_j=y_j'}}{\sqrt{d_{y_j}}}
\right)
\cdot
\tikzmath{
\roundNbox{}{(3,0)}{.3}{.6}{.6}{$g'$}
\roundNbox{}{(0,0)}{.3}{.6}{.6}{$f'$}
\roundNbox{}{(3,2)}{.3}{.6}{.6}{$g^\dag$}
\roundNbox{}{(0,2)}{.3}{.6}{.6}{$f^\dag$}
\draw (-.7,.3) -- node[left]{$\scriptstyle x_1$} (-.7,1.7);
\draw (0,.3) -- node[right]{$\scriptstyle x_m$} (0,1.7);
\draw (.7,.3) -- node[right]{$\scriptstyle x_{m+1}$} (.7,1.7);
\draw (2.3,.3) -- node[left]{$\scriptstyle y_1$} (2.3,1.7);
\draw (3,.3) -- node[left]{$\scriptstyle y_2$} (3,1.7) ;
\draw (3.7,.3) -- node[right]{$\scriptstyle y_n$} (3.7,1.7);
\node at (-.3,1) {$\cdots$};
\node at (3.35,1) {$\cdots$};
}
\\&=
\langle f|f'\rangle\cdot \langle g|g'\rangle
=
\langle f\otimes g | g\otimes g'\rangle.
\qedhere
\end{align*}
\end{proof}

We end with the following observation about
how morphism spaces act on skein modules.
First, using the unitary spherical structure of $\cX$, we may identify 
$$
\cS_\cX(\bbD,m+n) \cong 
\bigoplus_{
\substack{
x_1,\dots, x_m
\\y_1,\dots, y_n
\\
\in\Irr(\cX)
}}
\cX(x_1\otimes \cdots \otimes x_m \to y_1\otimes \cdots \otimes y_n).
$$
Given a map $f: y_1,\dots, y_n\to z_1,\dots, z_p$,
we may consider the map $\gamma_f$ that post-composes with $f$ when the string labels agree and gives zero otherwise.
One computes that for $g\in \cS_\cX(\bbD,m+n)$ and $h\in \cS_\cX(\bbD,m+p)$
\begin{align*}
\langle h | \gamma_f g\rangle
&=
\frac{1}{\sqrt{d_{x_1}\cdots d_{x_m}d_{z_1}\cdots d_{z_p}}} 
\cdot
\tr_\cX(h^\dag \circ f\circ g)
\\
\langle \gamma_{f^\dag} h |  g\rangle
&=
\frac{1}{\sqrt{d_{x_1}\cdots d_{x_m}d_{y_1}\cdots d_{y_n}}} 
\cdot
\tr_\cX((f^\dag\circ h)^\dag \circ g).
\end{align*}
Thus we see that
$$
\gamma_f^\dag 
= 
\frac{\sqrt{d_{y_1}\cdots d_{y_n}}}{\sqrt{d_{z_1}\cdots d_{z_p}}}
\gamma_{f^\dag}.
$$

\begin{defn}
We define $\Gamma_f$ to be the operator which post-composes with $f$ and multiplies by a scalar factor as in \cite{MR4642306} so that $\Gamma_f^\dag = \Gamma_{f^\dag}$ on the nose:
\begin{equation}
\label{eq:GluingOperatorWithScalars}
\Gamma_f := 
\left(\frac{d_{z_1}\cdots d_{z_p}}{d_{y_1}\cdots d_{y_n}}\right)^{1/4}
\cdot \gamma_f.
\end{equation}
\end{defn}

\subsection{String nets}
\label{sec:StringNet}

As before, $\cX$ is a unitary fusion category.
We now consider a rectangular lattice and a collection of Hilbert spaces defined in terms of $\cX$ and the lattice. 
$$
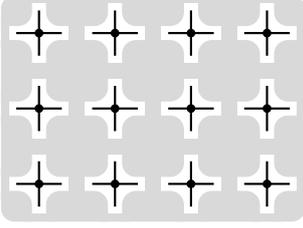
\tikzmath{
\clip[rounded corners=5pt] (-.5,-.5) rectangle (3.5,2.5);
\foreach \x in {0,1,2,3}{
\foreach \y in {0,1,2}{
\filldraw (\x,\y) circle (.05cm);
\draw[thick] ($ (\x,\y) - (.3,0) $) -- ($ (\x,\y) + (.3,0) $);
\draw[thick] ($ (\x,\y) - (0,.3) $) -- ($ (\x,\y) + (0,.3) $);
\fill[gray!30] ($ (\x,\y) + (-.4,-.2) $) rectangle ($ (\x,\y) + (-.6,.2) $);
\fill[gray!30] ($ (\x,\y) + (.4,-.2) $) rectangle ($ (\x,\y) + (.6,.2) $);
\fill[gray!30] ($ (\x,\y) + (-.2,-.4) $) rectangle ($ (\x,\y) + (.2,-.6) $);
\fill[gray!30] ($ (\x,\y) + (-.2,.4) $) rectangle ($ (\x,\y) + (.2,.6) $);
}}
\foreach \x in {-.5,...,3.5}{
\foreach \y in {-.5,...,2.5}{
\filldraw[gray!30, rounded corners=5pt] ($ (\x,\y) - (.4,.4)$) rectangle ($ (\x,\y) + (.4,.4)$);
}}
}
\qquad\qquad
\tikzmath{
\draw[thick] (-.5,0) node[left]{$\scriptstyle w$} -- (.5,0) node[right]{$\scriptstyle z$};
\draw[thick] (0,-.5) node[below]{$\scriptstyle x$} -- (0,.5) node[above]{$\scriptstyle y$};
\filldraw (0,0) circle (.05cm);
\node at (-.2,-.2) {$\scriptstyle v$};
\draw[blue!50, very thin] (-.5,.5) -- (.5,-.5);
{\draw[blue!50, very thin, -stealth] ($ (-.3,.3) - (.1,.1)$) to ($ (-.3,.3) + (.1,.1)$);}
{\draw[blue!50, very thin, -stealth] ($ (.3,-.3) - (.1,.1)$) to ($ (.3,-.3) + (.1,.1)$);}
}
\leftrightarrow
\cH_v
=
\bigoplus_{w,x,y,z\in \Irr(\cX)} \cX(w\otimes x \to y\otimes z).
$$
The inner product on $\cH_v$\IN{$\cH_v$ the local hilbert space at the node $v$} is the inner product on $\cS_{\cX}(\bbD, 4)$ up to isotopy of boundary strings:
$$
\langle f | g\rangle
:=
\frac{1}{\sqrt{d_{w}d_xd_yd_z} }
\tr_\cC(f^\dag \circ g)
=
\frac{1}{\sqrt{d_{w}d_xd_yd_z} }
\cdot
\tikzmath{
\roundNbox{}{(0,1)}{.3}{.1}{.1}{$f^\dag$}
\roundNbox{}{(0,0)}{.3}{.1}{.1}{$g$}
\draw (-.2,.3) --node[left]{$\scriptstyle y$} (-.2,.7) ;
\draw (.2,.3) --node[left]{$\scriptstyle z$} (.2,.7) ;
\draw (-.2,1.3) to[out=90,in=90] (1,1.3) --node[right]{$\scriptstyle \overline{w}$} (1,-.3) to[out=270,in=270] (-.2,-.3);
\draw (.2,1.3) arc (180:0:.2cm) (.6,1.3) --node[right]{$\scriptstyle \overline{x}$} (.6,-.3) arc (0:-180:.2cm);
}.
$$

The total Hilbert space is $\cH_{\tot}:=\bigotimes_v \cH_v$\IN{$\cH_{\tot}$ the total Hilbert space}, represented by the disconnected lattice above, in which the gray region is vacuum. 
The lattice is directed by the blue arrows in the above figure, so that every edge has an incoming and outgoing vertex. 
The next step is to define the low-energy Hilbert space.  
We do this by defining the two terms that make up the Hamiltonian $H$. 

First,
for each edge $\ell$, the term $A_\ell$\IN{$A_\ell$ the edge term associated to the edge $\ell$} is a projector onto states where the types of simple objects labeling the two halves of $\ell$ match.
Let $P_A$\IN{$P_A$ the projection onto the edge ground state} denote the projection onto the ground state space of $-\sum_\ell A_\ell$.
Passing to $P_A\cH_{\tot}$ has the effect of connecting the edges to condense the 1-skeleton of the lattice above. 
$$
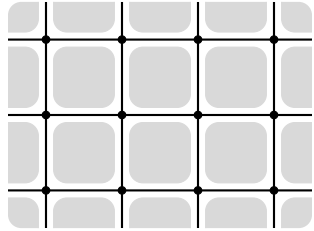
\tikzmath{
\clip[rounded corners=5pt] (-.5,-.5) rectangle (3.5,2.5);
\draw[thick] (-.5,-.5) grid (3.5,2.5);
\foreach \x in {0,1,2,3}{
\foreach \y in {0,1,2}{
\filldraw (\x,\y) circle (.05cm);
}}
\foreach \x in {-.5,...,3.5}{
\foreach \y in {-.5,...,2.5}{
\filldraw[gray!30, rounded corners=5pt] ($ (\x,\y) - (.4,.4)$) rectangle ($ (\x,\y) + (.4,.4)$);
}}
}
$$
There is an elegant description of $P_A\cH_{\tot}$ in terms of the UFC $\cX$.
Let $I: \cX\to Z(\cX)$ be the adjoint of the forgetful functor $F: Z(\cX)\to \cX$ \cite[Prop. 8.1]{MR1966525}, for which there are canonical natural isomorphisms
$$
\cX(F(X)\to y) \cong Z(\cX)(X\to I(y))
\qquad\qquad
\qquad\qquad
\forall\, X\in Z(\cX)\quad\forall\,y\in\cX.
$$
In particular, we have the formula $FI(x)=\bigoplus_{y\in\Irr(\cX)} yx\overline{y}$.

\begin{prop}
\label{prop:DescriptionOfIm(PA)}
On a contractible patch $\Lambda$ of lattice that contains $\# p=\# p(\Lambda)$ many plaquettes,
$$
P_A\bigotimes_{v\in\Lambda} \cH_v
=
\bigoplus_{\vec{x}\in \partial \Lambda}
\cX(\vec{x} \to FI(1_\cX)^{\otimes\# p}),
$$
where $\vec{x}$ denotes a boundary condition, i.e., a tensor product over the simple labels $x_i$ on the boundary of $\Lambda$, and we sum over the set $\partial \Lambda$ of all such boundary conditions.
\end{prop}
\begin{proof}
We begin with the analysis of a single plaquette, which should be viewed as having a puncture in it corresponding to the vacuum.
Using Fact \ref{fact:VaccuumSeparate} three times, we glue the 4 copies of $\cS_\cX(\bbD,4)$ along matched boundary conditions to obtain
\begin{align*}
\bigoplus_{\substack{g,h,i,j
\\
s,\cdots, z
\\
\in \Irr(\cX)}}
\tikzmath{
\draw[step=1.0,black,thin] (0.5,0.5) grid (2.5,2.5);
\filldraw[gray!30, rounded corners=5pt] (1.1,1.1) rectangle (1.9,1.9);
\node at (1.5,.85) {$\scriptstyle g$};
\node at (2.15,1.5) {$\scriptstyle h$};
\node at (1.5,2.15) {$\scriptstyle i$};
\node at (.85,1.5) {$\scriptstyle j$};
\node at (.3,1) {$\scriptstyle s$};
\node at (1,.3) {$\scriptstyle t$};
\node at (2,.3) {$\scriptstyle u$};
\node at (2.7,1) {$\scriptstyle v$};
\node at (2.7,2) {$\scriptstyle w$};
\node at (2,2.7) {$\scriptstyle x$};
\node at (1,2.7) {$\scriptstyle y$};
\node at (.3,2) {$\scriptstyle z$};
}
&:=
P_A
\bigotimes_{i=1}^4 \cH_v
=
\bigoplus_{\substack{g,h,i,j
\\
s,\cdots, z
\\
\in \Irr(\cX)}}
\parbox{5cm}{
$\cX(zj\to yi)\otimes \cX(ih\to xw)$
\\
$\qquad\qquad\otimes \cX(gu\to hv)\otimes \cX(st\to jg)$
}
\\&\cong
\bigoplus_{\substack{
s,\cdots, z
\\
\in \Irr(\cX)}}
\cX\left(
s\cdot t\cdot u\cdot\overline{v}\cdot\overline{w}\cdot\overline{x}\cdot\overline{y}\cdot z
\to
\underbrace{\bigoplus_j j\cdot\overline{j}}_{FI(1_\cX)} 
\right).
\end{align*}
Here, we add the dots to disambiguate the bar of the tensor product and the tensor product of the bars.
As it was obtained from gluing skein modules, the above isomorphism is unitary when the last space is equipped with the skein module inner product.
Indeed,
up to rotations of our $\cS_\cX(\bbD,4)$ ONB, using Fact \ref{fact:VaccuumSeparate} three times,
\begin{align*}
\langle \eta| \xi\rangle
&=
\left\langle
\tikzmath{
\draw (0,0) to[out=180,in=-90] (-.4,.5) node[above]{$\scriptstyle j$};
\draw (3,0) to[out=0,in=-90] (3.4,.5) node[above]{$\scriptstyle \overline{j}$};
\draw (0,0) -- (3,0);
\draw (-.4,-.5) node[below]{$\scriptstyle s$} -- (0,0) -- (.4,-.5) node[below]{$\scriptstyle t$};
\draw (.6,-.5) node[below]{$\scriptstyle u$} -- (1,0) -- (1.4,-.5) node[below]{$\scriptstyle \overline{v}$};
\draw (1.6,-.5) node[below]{$\scriptstyle \overline{w}$} -- (2,0) -- (2.4,-.5) node[below]{$\scriptstyle \overline{x}$};
\draw (2.6,-.5) node[below]{$\scriptstyle \overline{y}$} -- (3,0) -- (3.4,-.5) node[below]{$\scriptstyle z$};
\node at (.5,-.2) {$\scriptstyle g'$};
\node at (1.5,-.2) {$\scriptstyle h'$};
\node at (2.5,-.2) {$\scriptstyle \overline{i'}$};
\node at (0,.2) {$\scriptstyle \eta_{1,1}$};
\node at (1,.2) {$\scriptstyle \eta_{2,1}$};
\node at (2,.2) {$\scriptstyle \eta_{2,2}$};
\node at (3,.2) {$\scriptstyle \eta_{1,2}$};
}
\middle|
\tikzmath{
\draw (0,0) to[out=180,in=-90] (-.4,.5) node[above]{$\scriptstyle j$};
\draw (3,0) to[out=0,in=-90] (3.4,.5) node[above]{$\scriptstyle \overline{j}$};
\draw (0,0) -- (3,0);
\draw (-.4,-.5) node[below]{$\scriptstyle s$} -- (0,0) -- (.4,-.5) node[below]{$\scriptstyle t$};
\draw (.6,-.5) node[below]{$\scriptstyle u$} -- (1,0) -- (1.4,-.5) node[below]{$\scriptstyle \overline{v}$};
\draw (1.6,-.5) node[below]{$\scriptstyle \overline{w}$} -- (2,0) -- (2.4,-.5) node[below]{$\scriptstyle \overline{x}$};
\draw (2.6,-.5) node[below]{$\scriptstyle \overline{y}$} -- (3,0) -- (3.4,-.5) node[below]{$\scriptstyle z$};
\node at (.5,-.2) {$\scriptstyle g$};
\node at (1.5,-.2) {$\scriptstyle h$};
\node at (2.5,-.2) {$\scriptstyle \overline{i}$};
\node at (0,.2) {$\scriptstyle \xi_{1,1}$};
\node at (1,.2) {$\scriptstyle \xi_{2,1}$};
\node at (2,.2) {$\scriptstyle \xi_{2,2}$};
\node at (3,.2) {$\scriptstyle \xi_{1,2}$};
}
\right\rangle
\\&=
\frac{1}{d_j\sqrt{d_s\cdots d_z}}
\cdot
\tikzmath{
\draw (0,0) -- (3,0);
\draw (0,1) -- (3,1);
\draw (0,0) arc (270:90:.5cm);
\draw (3,0) arc (-90:90:.5cm);
\node at (-.6,.5) {$\scriptstyle \overline{j}$};
\node at (3.6,.5) {$\scriptstyle j$};
\draw (0,0) to[out=135,in=-135] node[left,xshift=.1cm]{$\scriptstyle s$} (0,1);
\draw (0,0) to[out=45,in=-45] node[right,xshift=-.1cm]{$\scriptstyle t$} (0,1);
\draw (1,0) to[out=135,in=-135] node[left,xshift=.1cm]{$\scriptstyle u$} (1,1);
\draw (1,0) to[out=45,in=-45] node[right,xshift=-.1cm]{$\scriptstyle \overline{v}$} (1,1);
\draw (2,0) to[out=135,in=-135] node[left,xshift=.1cm]{$\scriptstyle \overline{w}$} (2,1);
\draw (2,0) to[out=45,in=-45] node[right,xshift=-.1cm]{$\scriptstyle \overline{x}$} (2,1);
\draw (3,0) to[out=135,in=-135] node[left,xshift=.1cm]{$\scriptstyle \overline{y}$} (3,1);
\draw (3,0) to[out=45,in=-45] node[right,xshift=-.1cm]{$\scriptstyle z$} (3,1);
\node at (.5,.2) {$\scriptstyle g'$};
\node at (1.5,.2) {$\scriptstyle h'$};
\node at (2.5,.2) {$\scriptstyle \overline{i'}$};
\node at (0,-.2) {$\scriptstyle \eta^\dag_{1,1}$};
\node at (1,-.2) {$\scriptstyle \eta^\dag_{2,1}$};
\node at (2,-.2) {$\scriptstyle \eta^\dag_{2,2}$};
\node at (3,-.2) {$\scriptstyle \eta^\dag_{1,2}$};
\node at (.5,.8) {$\scriptstyle g$};
\node at (1.5,.8) {$\scriptstyle h$};
\node at (2.5,.8) {$\scriptstyle \overline{i}$};
\node at (0,1.2) {$\scriptstyle \xi_{1,1}$};
\node at (1,1.2) {$\scriptstyle \xi_{2,1}$};
\node at (2,1.2) {$\scriptstyle \xi_{2,2}$};
\node at (3,1.2) {$\scriptstyle \xi_{1,2}$};
}
\\&=
\frac{\delta_{g=g'}\delta_{h=h'}\delta_{i=i'}}{d_id_jd_gd_h\sqrt{d_s\cdots d_z}}
\cdot
\tikzmath{
\draw (0,0) arc (270:90:.5cm);
\draw (0,0) arc (-90:90:.5cm);
\node at (-.6,.5) {$\scriptstyle \overline{j}$};
\node at (.6,.5) {$\scriptstyle g$};
\draw (0,0) to[out=135,in=-135] node[left,xshift=.1cm]{$\scriptstyle s$} (0,1);
\draw (0,0) to[out=45,in=-45] node[right,xshift=-.1cm]{$\scriptstyle t$} (0,1);
\node at (0,-.2) {$\scriptstyle \eta^\dag_{1,1}$};
\node at (0,1.2) {$\scriptstyle \xi_{1,1}$};
}
\cdot
\tikzmath{
\draw (0,0) arc (270:90:.5cm);
\draw (0,0) arc (-90:90:.5cm);
\node at (-.6,.5) {$\scriptstyle \overline{g}$};
\node at (.6,.5) {$\scriptstyle h$};
\draw (0,0) to[out=135,in=-135] node[left,xshift=.1cm]{$\scriptstyle u$} (0,1);
\draw (0,0) to[out=45,in=-45] node[right,xshift=-.1cm]{$\scriptstyle \overline{v}$} (0,1);
\node at (0,-.2) {$\scriptstyle \eta^\dag_{2,1}$};
\node at (0,1.2) {$\scriptstyle \xi_{2,1}$};
}
\cdot
\tikzmath{
\draw (0,0) arc (270:90:.5cm);
\draw (0,0) arc (-90:90:.5cm);
\node at (-.6,.5) {$\scriptstyle \overline{h}$};
\node at (.6,.5) {$\scriptstyle \overline{i}$};
\draw (0,0) to[out=135,in=-135] node[left,xshift=.1cm]{$\scriptstyle \overline{w}$} (0,1);
\draw (0,0) to[out=45,in=-45] node[right,xshift=-.1cm]{$\scriptstyle \overline{x}$} (0,1);
\node at (0,-.2) {$\scriptstyle \eta^\dag_{2,2}$};
\node at (0,1.2) {$\scriptstyle \xi_{2,2}$};
}
\cdot
\tikzmath{
\draw (0,0) arc (270:90:.5cm);
\draw (0,0) arc (-90:90:.5cm);
\node at (-.6,.5) {$\scriptstyle i$};
\node at (.6,.5) {$\scriptstyle j$};
\draw (0,0) to[out=135,in=-135] node[left,xshift=.1cm]{$\scriptstyle \overline{y}$} (0,1);
\draw (0,0) to[out=45,in=-45] node[right,xshift=-.1cm]{$\scriptstyle z$} (0,1);
\node at (0,-.2) {$\scriptstyle \eta^\dag_{1,2}$};
\node at (0,1.2) {$\scriptstyle \xi_{1,2}$};
}
\end{align*}
which is exactly the inner product on $P_\cA \bigotimes_{i = 1}^4 \cH_v$.

We now consider two neighboring plaquettes:
\begin{equation}
\label{eq:GlueSkeinModsForI(1)}
\bigoplus
\tikzmath{
\draw[step=1.0,black,thin] (0.5,0.5) grid (3.5,2.5);
\filldraw[gray!30, rounded corners=5pt] (1.1,1.1) rectangle (1.9,1.9);
\filldraw[gray!30, rounded corners=5pt] (2.1,1.1) rectangle (2.9,1.9);
\node at (1.5,.85) {$\scriptstyle g$};
\node at (2.5,.85) {$\scriptstyle h$};
\node at (2.15,1.5) {$\scriptstyle m$};
\node at (3.15,1.5) {$\scriptstyle i$};
\node at (2.5,2.15) {$\scriptstyle j$};
\node at (1.5,2.15) {$\scriptstyle k$};
\node at (.85,1.5) {$\scriptstyle \ell$};
\node at (.3,1) {$\scriptstyle q$};
\node at (1,.3) {$\scriptstyle r$};
\node at (2,.3) {$\scriptstyle s$};
\node at (3,.3) {$\scriptstyle t$};
\node at (3.7,1) {$\scriptstyle u$};
\node at (3.7,2) {$\scriptstyle v$};
\node at (3,2.7) {$\scriptstyle w$};
\node at (2,2.7) {$\scriptstyle x$};
\node at (1,2.7) {$\scriptstyle y$};
\node at (.3,2) {$\scriptstyle z$};
\draw[thick, blue, dashed, rounded corners=5pt] (1.7,.7) rectangle (2.3,2.3);
}
\quad
\xrightarrow[\cong]{F}
\quad
\bigoplus
\tikzmath{
\filldraw[gray!30, rounded corners=5pt] (2.1,1.1) rectangle (2.9,1.9);
\filldraw[gray!30, rounded corners=5pt] (4.1,2.1) rectangle (4.9,2.9);
\draw (2,.5) -- (2,3.5);
\draw[rounded corners=5pt] (1.5,1) -- (3,1) -- (3,3.5);
\draw (1.5,2) -- (5.5,2); 
\draw (5,.5) -- (5,3.5); 
\draw[rounded corners=5pt] (4,.5) -- (4,3) -- (5.5,3);
\node at (2.5,.85) {$\scriptstyle g$};
\node at (4.5,1.85) {$\scriptstyle h$};
\node at (3.5,1.8) {$\scriptstyle n$};
\node at (5.15,2.5) {$\scriptstyle i$};
\node at (4.5,3.15) {$\scriptstyle j$};
\node at (2.5,2.15) {$\scriptstyle k$};
\node at (1.85,1.5) {$\scriptstyle \ell$};
\node at (1.3,1) {$\scriptstyle q$};
\node at (2,.3) {$\scriptstyle r$};
\node at (4,.3) {$\scriptstyle s$};
\node at (5,.3) {$\scriptstyle t$};
\node at (5.7,2) {$\scriptstyle u$};
\node at (5.7,3) {$\scriptstyle v$};
\node at (5,3.7) {$\scriptstyle w$};
\node at (3,3.7) {$\scriptstyle x$};
\node at (2,3.7) {$\scriptstyle y$};
\node at (1.3,2) {$\scriptstyle z$};
}\,.
\end{equation}
We get a unitary above on $P_A\bigotimes_{i = 1}^4  \cH_v$ by applying a unitary $F$-symbol from $\cX$.
Indeed, using Lemma \ref{lem:GluingUnitary}, gluing ONBs for $\cS_\cX(\bbD,4)$ along the $m$-string (as in the dotted blue line above) gives an ONB for $\cS_\cX(\bbD,6)$, and thus there is a unitary isomorphism
$$
\bigoplus
\tikzmath{
\draw[step=1.0,black,thin] (1.5,0.5) grid (2.5,2.5);
\node at (1.5,.85) {$\scriptstyle g$};
\node at (2.5,.85) {$\scriptstyle h$};
\node at (2.15,1.5) {$\scriptstyle m$};
\node at (2.5,2.15) {$\scriptstyle j$};
\node at (1.5,2.15) {$\scriptstyle k$};
\node at (2,.3) {$\scriptstyle s$};
\node at (2,2.7) {$\scriptstyle x$};
}
\quad
\xrightarrow[\cong]{F}
\quad
\bigoplus
\tikzmath{
\draw[rounded corners=5pt] (2.5,1) -- (3,1) -- (3,3.5);
\draw (2.5,2) -- (4.5,2); 
\draw[rounded corners=5pt] (4,.5) -- (4,3) -- (4.5,3);
\node at (2.5,.85) {$\scriptstyle g$};
\node at (4.5,1.85) {$\scriptstyle h$};
\node at (3.5,1.8) {$\scriptstyle n$};
\node at (4.5,3.15) {$\scriptstyle j$};
\node at (2.5,2.15) {$\scriptstyle k$};
\node at (4,.3) {$\scriptstyle s$};
\node at (3,3.7) {$\scriptstyle x$};
}\,.
$$
Now looking at the right hand side of \eqref{eq:GlueSkeinModsForI(1)},
we have glued together 
$$
\bigoplus
\tikzmath{
\draw (2,.5) -- (2,2.5);
\draw[rounded corners=5pt] (1.5,1) -- (3,1) -- (3,2.5);
\filldraw[gray!30, rounded corners=5pt] (2.1,1.9) rectangle (2.9,1.1);
\draw (1.5,2) -- (3.5,2); 
\node at (2.5,.85) {$\scriptstyle g$};
\node at (3.5,1.8) {$\scriptstyle n$};
\node at (2.5,2.15) {$\scriptstyle k$};
\node at (1.85,1.5) {$\scriptstyle \ell$};
\node at (1.3,1) {$\scriptstyle q$};
\node at (2,.3) {$\scriptstyle r$};
\node at (3,2.7) {$\scriptstyle x$};
\node at (2,2.7) {$\scriptstyle y$};
\node at (1.3,2) {$\scriptstyle z$};
}
\qquad\text{and}\qquad
\bigoplus
\tikzmath{
\draw (3.5,2) -- (5.5,2); 
\draw (5,1.5) -- (5,3.5); 
\draw[rounded corners=5pt] (4,1.5) -- (4,3) -- (5.5,3);
\filldraw[gray!30, rounded corners=5pt] (4.1,2.9) rectangle (4.9,2.1);
\node at (4.5,1.85) {$\scriptstyle h$};
\node at (3.5,1.8) {$\scriptstyle n$};
\node at (5.15,2.5) {$\scriptstyle i$};
\node at (4.5,3.15) {$\scriptstyle j$};
\node at (4,1.3) {$\scriptstyle s$};
\node at (5,1.3) {$\scriptstyle t$};
\node at (5.7,2) {$\scriptstyle u$};
\node at (5.7,3) {$\scriptstyle v$};
\node at (5,3.7) {$\scriptstyle w$};
}
$$
along the $n$-strand.
Similar to the argument above for a single plaquette, these spaces correspond to the following hom spaces with their skein-module inner products:
$$
\bigoplus
\cX\left(
q\cdot r\cdot\overline{n}\cdot\overline{x}\cdot\overline{y}\cdot z
\to
\underbrace{\bigoplus_\ell \ell \overline{\ell}}_{FI(1_\cX)}
\right)
\qquad\text{and}\qquad
\bigoplus
\cX\left(
n\cdot s\cdot t\cdot\overline{u}\cdot\overline{v}\cdot\overline{w}
\to
\underbrace{\bigoplus_j j \overline{j}}_{FI(1_\cX)}
\right).
$$
Now using $I(1_\cX)\in Z(\cX)$ and Lemma \ref{lem:GluingUnitary}, the result is the glued skein module
$$
\cX\left(
q\cdot r\cdot s\cdot t\cdot \overline{u}\cdot\overline{v}\cdot\overline{w}\cdot\overline{x}\cdot\overline{y}\cdot z
\to
FI(1_\cX)^{\otimes 2}
\right)
$$
as claimed via the following unitary isomorphism:
$$
\tikzmath{
\draw[thick, blue] (0,.3) -- (0,.7) node[above]{$\scriptstyle FI(1_\cX)$};
\draw (-.75,-.3) -- (-.75,-.7) node[below]{$\scriptstyle q$};
\draw (-.45,-.3) -- (-.45,-.7) node[below]{$\scriptstyle r$};
\draw (-.15,-.3) -- (-.15,-.7) node[below]{$\scriptstyle \overline{n}$};
\draw (.15,-.3) -- (.15,-.7) node[below]{$\scriptstyle \overline{x}$};
\draw (.45,-.3) -- (.45,-.7) node[below]{$\scriptstyle \overline{y}$};
\draw (.75,-.3) -- (.75,-.7) node[below]{$\scriptstyle z$};
\roundNbox{}{(0,0)}{.3}{.6}{.6}{$f$}
}
\otimes
\tikzmath{
\draw[thick, blue] (0,.3) -- (0,.7) node[above]{$\scriptstyle FI(1_\cX)$};
\draw (-.75,-.3) -- (-.75,-.7) node[below]{$\scriptstyle n$};
\draw (-.45,-.3) -- (-.45,-.7) node[below]{$\scriptstyle s$};
\draw (-.15,-.3) -- (-.15,-.7) node[below]{$\scriptstyle t$};
\draw (.15,-.3) -- (.15,-.7) node[below]{$\scriptstyle \overline{u}$};
\draw (.45,-.3) -- (.45,-.7) node[below]{$\scriptstyle \overline{v}$};
\draw (.75,-.3) -- (.75,-.7) node[below]{$\scriptstyle \overline{w}$};
\roundNbox{}{(0,0)}{.3}{.6}{.6}{$g$}
}
\longmapsto
\tikzmath{
\draw[thick, blue] (0,.3) to[out=90,in=-90] (1,2.2) node[above]{$\scriptstyle FI(1_\cX)$};
\draw[thick, blue] (-1,1.8) to[out=90,in=-90] (-1,2.2) node[above]{$\scriptstyle FI(1_\cX)$};
\draw[knot] (-1.5,1.2) to[out=-90,in=90] (-1.85,-.7) node[below]{$\scriptstyle q$};
\draw[knot] (-1.2,1.2) to[out=-90,in=90] (-1.55,-.7) node[below]{$\scriptstyle r$};
\draw[knot] (-.9,1.2) to[out=-90,in=90] (-1.25,-.3) node[below, xshift=-.1cm]{$\scriptstyle \overline{n}$} arc (-180:0:.25cm);
\draw[knot] (-.6,1.2) to[out=-90,in=90] (1.4,-.2) -- (1.4,-.7) node[below]{$\scriptstyle \overline{x}$};
\draw[knot] (-.3,1.2) to[out=-90,in=90] (1.7,-.1) -- (1.7,-.7) node[below]{$\scriptstyle \overline{y}$};
\draw[knot] (0,1.2) to[out=-90,in=90] (2,0) -- (2,-.7) node[below]{$\scriptstyle z$};
\draw (-.75,-.3) node[below, xshift=.1cm]{$\scriptstyle n$};
\draw (-.45,-.3) -- (-.45,-.7) node[below]{$\scriptstyle s$};
\draw (-.15,-.3) -- (-.15,-.7) node[below]{$\scriptstyle t$};
\draw (.15,-.3) -- (.15,-.7) node[below]{$\scriptstyle \overline{u}$};
\draw (.45,-.3) -- (.45,-.7) node[below]{$\scriptstyle \overline{v}$};
\draw (.75,-.3) -- (.75,-.7) node[below]{$\scriptstyle \overline{w}$};
\roundNbox{}{(0,0)}{.3}{.6}{.6}{$g$}
\roundNbox{}{(-.75,1.5)}{.3}{.6}{.6}{$f$}
}\,.
$$
We now leave the general argument for many plaquettes to the reader.
\end{proof}

Second, 
on $P_A\cH_{\tot}$, for each plaquette/face $p$, 
we have an operator $B_p$\IN{$B_p$ the plaquette operator associated to the plaquette $p$} as follows:
\begin{align*}
\frac{1}{D_\cX}\sum_{r\in \Irr(\cX)}
d_r\cdot
\tikzmath{
\draw[step=1.0,black,thin] (0.5,0.5) grid (2.5,2.5);
\node at (2.3,.8) {$\scriptstyle \xi_{2,1}$};
\node at (.7,2.2) {$\scriptstyle \xi_{1,2}$};
\node at (.7,.8) {$\scriptstyle \xi_{1,1}$};
\node at (2.3,2.2) {$\scriptstyle \xi_{2,2}$};
\node at (1.5,.85) {$\scriptstyle g$};
\node at (2.15,1.5) {$\scriptstyle h$};
\node at (1.5,2.15) {$\scriptstyle i$};
\node at (.85,1.5) {$\scriptstyle j$};
\filldraw[knot, thick, blue, rounded corners=5pt, fill=gray!30] (1.15,1.15) rectangle (1.85,1.85);
\node[blue] at (1.3,1.5) {$\scriptstyle r$};
}
&=
\frac{1}{D_\cX}\sum_{r,s,t,u,v \in \Irr(\cX)}
\frac{\sqrt{d_kd_\ell d_md_n}}{d_r\sqrt{d_gd_hd_id_j}}
\tikzmath{
\filldraw[gray!30, rounded corners=5pt] (1.15,1.15) rectangle (1.85,1.85);
\draw[step=1.0,black,thin] (0.5,0.5) grid (2.5,2.5);
\draw[thick, blue] (1.3,1) -- (1,1.3);
\draw[thick, blue] (1.7,1) -- (2,1.3);
\draw[thick, blue] (1.3,2) -- (1,1.7);
\draw[thick, blue] (1.7,2) -- (2,1.7);
\fill[fill=green] (1.3,1) circle (.05cm);
\fill[fill=green] (1.7,1) circle (.05cm);
\fill[fill=red] (2,1.3) circle (.05cm);
\fill[fill=red] (2,1.7) circle (.05cm);
\fill[fill=yellow] (1.3,2) circle (.05cm);
\fill[fill=yellow] (1.7,2) circle (.05cm);
\fill[fill=orange] (1,1.3) circle (.05cm);
\fill[fill=orange] (1,1.7) circle (.05cm);
\node at (2.3,.8) {$\scriptstyle \xi_{2,1}$};
\node at (.7,2.2) {$\scriptstyle \xi_{1,2}$};
\node at (.7,.8) {$\scriptstyle \xi_{1,1}$};
\node at (2.3,2.2) {$\scriptstyle \xi_{2,2}$};
\node at (1.15,.85) {$\scriptstyle g$};
\node at (1.85,.85) {$\scriptstyle g$};
\node at (2.15,1.15) {$\scriptstyle h$};
\node at (2.15,1.85) {$\scriptstyle h$};
\node at (1.15,2.15) {$\scriptstyle i$};
\node at (1.85,2.15) {$\scriptstyle i$};
\node at (.85,1.15) {$\scriptstyle j$};
\node at (.85,1.85) {$\scriptstyle j$};
\node at (1.5,.85) {$\scriptstyle k$};
\node at (2.15,1.5) {$\scriptstyle \ell$};
\node at (1.5,2.15) {$\scriptstyle m$};
\node at (.85,1.5) {$\scriptstyle n$};
}
\\&=
\sum_{\eta}
C(\xi,\eta)
\tikzmath{
\draw[step=1.0,black,thin] (0.5,0.5) grid (2.5,2.5);
\filldraw[gray!30, rounded corners=5pt] (1.1,1.1) rectangle (1.9,1.9);
\node at (2.3,.8) {$\scriptstyle \eta_{2,1}$};
\node at (.7,1.8) {$\scriptstyle \eta_{1,2}$};
\node at (.7,.8) {$\scriptstyle \eta_{1,1}$};
\node at (2.3,1.8) {$\scriptstyle \eta_{2,2}$};
}\,.
\end{align*}
Here, $D_\cX=\sum_{x\in \Irr(\cX)} d_x^2$\IN{$D_\cX$ the global dimension of $\cX$} is the global dimension of $\cX$.
In the final sum, we write $\xi$ and $\eta$ for simple tensors over our orthonormal basis of $\cH_v$ whose edge labels match, thus defining elements of $P_A\cH_{\tot}$.

We will see below that passing to the ground state space of $B_p$ from $P_A\cH_{\tot}$ has the effect of gluing/condensing a 2-cell onto the plaquette $p$.

\begin{fact}
\label{fact:VaccuumUnzip}
For all $f:1\to xyz$ and $g: xyz\to 1$ for $x,y,z\in \Irr(\cX)$,
we can apply the \eqref{eq:Fusion} relation twice to obtain
$$
\tikzmath{
\draw (-.4,.3) -- node[left]{$\scriptstyle x$} (-.4,1.9);
\draw (0,.3) -- node[left]{$\scriptstyle y$} (0,1.9);
\draw (.4,.3) -- node[right]{$\scriptstyle z$} (.4,1.9);
\roundNbox{}{(0,0)}{.3}{.3}{.3}{$f$}
\roundNbox{}{(0,2.2)}{.3}{.3}{.3}{$g$}
}
=
\frac{\sqrt{d_w}}{\sqrt{d_xd_y}}
\cdot
\tikzmath{
\draw (-.2,.6) -- node[left]{$\scriptstyle w$} (-.2,1.6);
\draw (-.4,.3) -- node[left]{$\scriptstyle x$} (-.2,.6);
\draw (0,.3) -- node[right]{$\scriptstyle y$} (-.2,.6);
\draw (-.4,1.9) -- node[left]{$\scriptstyle x$} (-.2,1.6);
\draw (0,1.9) -- node[right]{$\scriptstyle y$} (-.2,1.6);
\draw (.4,.3) -- node[right]{$\scriptstyle z$} (.4,1.9);
\filldraw[blue] (-.2,.6) circle (.05cm);
\filldraw[blue] (-.2,1.6) circle (.05cm);
\roundNbox{}{(0,0)}{.3}{.3}{.3}{$f$}
\roundNbox{}{(0,2.2)}{.3}{.3}{.3}{$g$}
}
=
\frac{\sqrt{d_v}}{\sqrt{d_xd_yd_z}}
\cdot
\tikzmath{
\draw (.1,.9) -- node[right]{$\scriptstyle v$} (.1,1.3);
\draw (-.2,.6) -- node[left]{$\scriptstyle w$} (.1,.9);
\draw (-.2,1.6) -- node[left]{$\scriptstyle w$} (.1,1.3);
\draw (-.4,.3) -- node[left]{$\scriptstyle x$} (-.2,.6);
\draw (0,.3) -- node[right]{$\scriptstyle y$} (-.2,.6);
\draw (-.4,1.9) -- node[left]{$\scriptstyle x$} (-.2,1.6);
\draw (0,1.9) -- node[right]{$\scriptstyle y$} (-.2,1.6);
\draw (.4,.3) -- node[right]{$\scriptstyle z$} (.1,.9);
\draw (.4,1.9) -- node[right]{$\scriptstyle z$} (.1,1.3);
\filldraw[blue] (-.2,.6) circle (.05cm);
\filldraw[blue] (-.2,1.6) circle (.05cm);
\filldraw[red] (.1,.9) circle (.05cm);
\filldraw[red] (.1,1.3) circle (.05cm);
\roundNbox{}{(0,0)}{.3}{.3}{.3}{$f$}
\roundNbox{}{(0,2.2)}{.3}{.3}{.3}{$g$}
}
=
\frac{\sqrt{1}}{\sqrt{d_xd_yd_z}}
\cdot
\tikzmath{
\draw (-.4,.3) -- node[left]{$\scriptstyle x$} (0,.9);
\draw (0,.3) -- node[left, xshift=.1cm,yshift=-.1cm]{$\scriptstyle y$} (0,.9);
\draw (.4,.3) -- node[right]{$\scriptstyle z$} (0,.9);
\draw (-.4,1.9) -- node[left]{$\scriptstyle x$} (0,1.3);
\draw (0,1.9) -- node[left, xshift=.1cm,yshift=.1cm]{$\scriptstyle y$} (0,1.3);
\draw (.4,1.9) -- node[right]{$\scriptstyle z$} (0,1.3);
\filldraw[blue] (0,.9) circle (.05cm);
\filldraw[blue] (0,1.3) circle (.05cm);
\roundNbox{}{(0,0)}{.3}{.3}{.3}{$f$}
\roundNbox{}{(0,2.2)}{.3}{.3}{.3}{$g$}
}
\,.
$$
In the final equality above, we see that we have zero unless $v=1_\cX$ (so $w=\overline{z}$) as there are no non-zero maps between distinct simples in $\cX$.
Moreover, we may choose the red vertices to be $\coev_z$ and its dagger, as
$\{\ev_z\}$ is an ONB for $\cX(\overline{z}z\to 1)$.
\end{fact}

\begin{lem}[{\cite[\S5]{0907.2204} and \cite[\S5]{YanbaiZhang}}]
\label{lem:PlaquetteCoefficients}
For
$$
\xi=
\tikzmath{
\draw[step=1.0,black,thin] (0.5,0.5) grid (2.5,2.5);
\filldraw[gray!30, rounded corners=5pt] (1.1,1.1) rectangle (1.9,1.9);
\node at (2.3,.8) {$\scriptstyle \xi_{2,1}$};
\node at (.7,1.8) {$\scriptstyle \xi_{1,2}$};
\node at (.7,.8) {$\scriptstyle \xi_{1,1}$};
\node at (2.3,1.8) {$\scriptstyle \xi_{2,2}$};
\node at (.3,1) {$\scriptstyle s$};
\node at (1,.3) {$\scriptstyle t$};
\node at (2,.3) {$\scriptstyle u$};
\node at (2.7,1) {$\scriptstyle v$};
\node at (2.7,2) {$\scriptstyle w$};
\node at (2,2.7) {$\scriptstyle x$};
\node at (1,2.7) {$\scriptstyle y$};
\node at (.3,2) {$\scriptstyle z$};
}
\qquad\qquad
\eta=
\tikzmath{
\draw[step=1.0,black,thin] (0.5,0.5) grid (2.5,2.5);
\filldraw[gray!30, rounded corners=5pt] (1.1,1.1) rectangle (1.9,1.9);
\node at (2.3,.8) {$\scriptstyle \eta_{2,1}$};
\node at (.7,1.8) {$\scriptstyle \eta_{1,2}$};
\node at (.7,.8) {$\scriptstyle \eta_{1,1}$};
\node at (2.3,1.8) {$\scriptstyle \eta_{2,2}$};
\node at (.3,1) {$\scriptstyle s$};
\node at (1,.3) {$\scriptstyle t$};
\node at (2,.3) {$\scriptstyle u$};
\node at (2.7,1) {$\scriptstyle v$};
\node at (2.7,2) {$\scriptstyle w$};
\node at (2,2.7) {$\scriptstyle x$};
\node at (1,2.7) {$\scriptstyle y$};
\node at (.3,2) {$\scriptstyle z$};
}\,,
$$
the coefficient $C(\xi,\eta)$\IN{$C(\xi,\eta)$ the matrix coefficients of a $B_p$ operator.} above is given by
\begin{equation}
\label{eq:CoefficientOfBp}
C(\xi,\eta)
=
\frac{1}{D_\cX\sqrt{d_sd_td_ud_vd_wd_xd_yd_z}}
\tikzmath{
\draw[knot, mid>] (1,1) -- node[above, xshift=-.1cm]{$\scriptstyle g$} (2,1);
\draw[mid>] (2,1) -- node[left]{$\scriptstyle h$} (2,2);
\draw[mid>] (1,2) -- node[above]{$\scriptstyle i$} (2,2);
\draw[mid>] (1,1) -- node[left]{$\scriptstyle j$} (1,2);
\draw[knot, mid<] (0,0) to[out=-90,in=-90] node[above]{$\scriptstyle k$} (3,0);
\draw[mid<] (3,0) to[out=0,in=0] node[left]{$\scriptstyle\ell$} (3,3);
\draw[mid<] (0,3) to[out=90,in=90] node[below]{$\scriptstyle m$} (3,3);
\draw[mid<] (0,0) to[out=180,in=180] node[right]{$\scriptstyle n$} (0,3);
\draw[mid>] (0,0) to[out=90,in=180] node[above]{$\scriptstyle s$} (1,1);
\draw[mid>] (0,0) to[out=0,in=-90] node[below]{$\scriptstyle t$} (1,1);
\draw[mid>] (3,0) to[out=180,in=-90] node[below]{$\scriptstyle u$} (2,1);
\draw[mid<] (3,0) to[out=90,in=0] node[right]{$\scriptstyle v$} (2,1);
\draw[mid<] (3,3) to[out=180,in=90] node[left]{$\scriptstyle w$} (2,2);
\draw[mid<] (3,3) to[out=-90,in=0] node[right]{$\scriptstyle x$} (2,2);
\draw[mid<] (0,3) to[out=0,in=90] node[right]{$\scriptstyle y$} (1,2);
\draw[mid>]  (0,3) to[out=-90,in=180] node[below]{$\scriptstyle z$} (1,2);
\node at (1.3,1.8) {$\scriptstyle \xi_{1,2}$};
\node at (1.3,.8) {$\scriptstyle \xi_{1,1}$};
\node at (2.3,1.2) {$\scriptstyle \xi_{2,1}$};
\node at (2.3,1.8) {$\scriptstyle \xi_{2,2}$};
\node at (-.3,-.2) {$\scriptstyle \eta^\dag_{1,1}$};
\node at (3.6,-.2) {$\scriptstyle \overline{\eta_{2,1}}$};
\node at (3.3,3.2) {$\scriptstyle \eta^\dag_{2,2}$};
\node at (-.3,3.2) {$\scriptstyle \overline{\eta_{1,2}}$};
}\,.
\end{equation}
When not every external edge/link label of $\eta,\xi$ match, $C(\xi,\eta)=0$.
In particular, $B_p$ is a self-adjoint projection.
\end{lem}
We give a short conceptual proof for convenience and completeness.
We will prove an enriched version in \S\ref{sec:Enriched} below.
In the diagram on the right hand side of \eqref{eq:CoefficientOfBp} above, two different involutions are applied to morphisms in $\cX$, namely the adjoint $\dag$ and the conjugate $\overline{\,\cdot\,}$.
The adjoint is the usual one in the UFC $\cX$, and the conjugate is given by the adjoint of the rotation by $180^\circ$.
That is, if $f: x\to y$, then
$$
\tikzmath{
\draw (0,-.7) --node[right]{$\scriptstyle\overline{x}$} (0,-.3);
\draw (0,.7) --node[right]{$\scriptstyle\overline{y}$} (0,.3);
\roundNbox{}{(0,0)}{.3}{0}{0}{$\overline{f}$}
}
:=
\tikzmath{
\draw (0,.3) node[right,yshift=.2cm]{$\scriptstyle x$} arc(0:180:.3cm) --node[left]{$\scriptstyle\overline{x}$} (-.6,-.7);
\draw (0,-.3) node[left,yshift=-.2cm]{$\scriptstyle y$} arc(-180:0:.3cm) --node[right]{$\scriptstyle\overline{y}$} (.6,.7);
\roundNbox{}{(0,0)}{.3}{0}{0}{$f^\dag$}
}
=
\tikzmath{
\draw (0,.3) node[left,yshift=.2cm]{$\scriptstyle x$} arc(180:0:.3cm) --node[right]{$\scriptstyle\overline{x}$} (.6,-.7);
\draw (0,-.3) node[right,yshift=-.2cm]{$\scriptstyle y$} arc(0:-180:.3cm) --node[left]{$\scriptstyle\overline{y}$} (-.6,.7);
\roundNbox{}{(0,0)}{.3}{0}{0}{$f^\dag$}
}
\,.
$$

\begin{proof}
To ease the notation, we set
$$
K
:=
\frac{1}{D_\cX\sqrt{d_sd_td_ud_vd_wd_xd_yd_z}}.
$$
Under the usual Fourier expansion of each of the four corners in the $\{\eta\}$ ONB of $\cH_v=\cS_\cX(\bbD,4)$ (assuming the internal edges match giving a vector in $P_A\cH_{\tot}$),
$$
C(\xi,\eta)
=
K
\sum_{\substack{r\in \Irr(\cX)
}}
\frac{1}{ d_r\sqrt{d_gd_hd_id_jd_kd_\ell d_md_n}}
\tikzmath{
\clip (-1,-1) rectangle (5,5);
\draw[mid>] (1,1) -- node[below]{$\scriptstyle g$} (1.5,1);
\draw[mid>] (1.5,1) .. controls ++(0:1.5cm) and ++(-90:2cm) .. node[right]{$\scriptstyle k$} (0,0);
\draw[mid>] (1,1) -- node[left]{$\scriptstyle j$} (1,1.5);
\draw[mid>] (1,1.5) .. controls ++(90:1.5cm) and ++(180:2cm) .. node[above]{$\scriptstyle n$} (0,0);
\draw[mid>] (1,2.5) -- node[left]{$\scriptstyle j$} (1,3);
\draw[mid>] (0,4) .. controls ++(180:2cm) and ++(-90:1.5cm) .. node[below]{$\scriptstyle n$} (1,2.5);
\draw[mid>] (1,3) -- node[above]{$\scriptstyle i$} (1.5,3);
\draw[mid>] (1.5,3) .. controls ++(0:1.5cm) and ++(90:2cm) .. node[right]{$\scriptstyle m$} (0,4);
\draw[mid>] (2.5,1) -- node[below]{$\scriptstyle g$} (3,1);
\draw[mid<] (2.5,1) .. controls ++(180:1.5cm) and ++(-90:2cm) .. node[below]{$\scriptstyle k$} (4,0);
\draw[mid>] (3,1) -- node[right]{$\scriptstyle h$} (3,1.5);
\draw[mid>] (3,1.5) .. controls ++(90:1.5cm) and ++(0:2cm) .. node[above]{$\scriptstyle \ell$} (4,0);
\draw[mid>] (3,2.5) -- node[right]{$\scriptstyle h$} (3,3);
\draw[mid>] (4,4) .. controls ++(0:2cm) and ++(-90:1.5cm) .. node[below]{$\scriptstyle \ell$} (3,2.5);
\draw[mid>] (2.5,3) -- node[above]{$\scriptstyle i$} (3,3);
\draw[mid>] (4,4) .. controls ++(90:2cm) and ++(180:1.5cm) .. node[left]{$\scriptstyle m$} (2.5,3);
\draw[mid>] (0,0) to[out=90,in=180] node[above]{$\scriptstyle s$} (1,1);
\draw[mid>] (0,0) to[out=0,in=-90] node[below]{$\scriptstyle t$} (1,1);
\draw[mid>] (4,0) to[out=180,in=-90] node[below]{$\scriptstyle u$} (3,1);
\draw[mid<] (4,0) to[out=90,in=0] node[right]{$\scriptstyle v$} (3,1);
\draw[mid<] (4,4) to[out=-90,in=0] node[right]{$\scriptstyle w$} (3,3);
\draw[mid<] (4,4) to[out=180,in=90] node[left]{$\scriptstyle x$} (3,3);
\draw[mid<] (0,4) to[out=0,in=90] node[right]{$\scriptstyle y$} (1,3);
\draw[mid>]  (0,4) to[out=-90,in=180] node[below]{$\scriptstyle z$} (1,3);
\node at (.5,2.85) {$\scriptstyle \xi_{1,2}$};
\node at (.7,.7) {$\scriptstyle \xi_{1,1}$};
\node at (3.25,.8) {$\scriptstyle \xi_{2,1}$};
\node at (3.4,3.3) {$\scriptstyle \xi_{2,2}$};
\node at (-.3,-.2) {$\scriptstyle \eta^\dag_{1,1}$};
\node at (4.6,-.2) {$\scriptstyle \overline{\eta_{2,1}}$};
\node at (4.3,4.2) {$\scriptstyle \eta^\dag_{2,2}$};
\node at (-.3,4.2) {$\scriptstyle \overline{\eta_{1,2}}$};
\draw[thick, blue, mid>] (1.5,1) --node[above]{$\scriptstyle r$} (1,1.5);
\draw[thick, blue, mid>] (1,2.5) --node[right]{$\scriptstyle r$} (1.5,3);
\draw[thick, blue, mid>] (3,1.5) --node[left]{$\scriptstyle r$} (2.5,1);
\draw[thick, blue, mid>] (2.5,3) --node[below]{$\scriptstyle r$} (3,2.5);
\fill[fill=green] (1.5,1) circle (.05cm);
\fill[fill=green] (2.5,1) circle (.05cm);
\fill[fill=purple] (3,1.5) circle (.05cm);
\fill[fill=purple] (3,2.5) circle (.05cm);
\fill[fill=yellow] (1.5,3) circle (.05cm);
\fill[fill=yellow] (2.5,3) circle (.05cm);
\fill[fill=orange] (1,1.5) circle (.05cm);
\fill[fill=orange] (1,2.5) circle (.05cm);
}\,.
$$
Indeed, each of the four closed diagrams represents the inner product in the skein module (up to coefficients) of one of the four corners of \eqref{eq:EnrichedBpAction} with an appropriate $\eta_{i,j}$.
Observe here that there is no longer a sum over $k,\ell,m,n$ as picking particular $\eta_{i,j}$ determines these labels.
We now apply Fact \ref{fact:VaccuumUnzip} to three places amongst these four closed diagrams corresponding to pairs of green, purple, and yellow nodes
to simplify
$$
C(\xi,\eta)
=
K
\sum_{\substack{r\in \Irr(\cX)
}}
\frac{\sqrt{d_r}}{\sqrt{d_jd_n}}
\tikzmath{
\clip (-1.5,-1.5) rectangle (5.5,5.5);
\draw[mid>] (4,0) to[out=-90,in=-90] node[above]{$\scriptstyle k$} (0,0);
\draw[mid>] (4,4) to[out=0,in=0] node[left]{$\scriptstyle \ell$} (4,0);
\draw[mid>] (4,4) to[out=90,in=90] node[below]{$\scriptstyle m$} (0,4);
\draw[mid>] (1,1.5) .. controls ++(90:1.5cm) and ++(180:2cm) .. node[above]{$\scriptstyle n$} (0,0);
\draw[mid>] (0,4) .. controls ++(180:2cm) and ++(-90:1.5cm) .. node[below]{$\scriptstyle n$} (1,2.5);
\draw[mid>] (1,1) -- node[below]{$\scriptstyle g$} (3,1);
\draw[mid>] (3,1) -- node[right]{$\scriptstyle h$} (3,3);
\draw[mid>] (1,3) -- node[above]{$\scriptstyle i$} (3,3);
\draw[mid>] (1,1) -- node[left]{$\scriptstyle j$} (1,1.5);
\draw[mid>] (1,2.5) -- node[left]{$\scriptstyle j$} (1,3);
\draw[mid>] (0,0) to[out=90,in=180] node[above]{$\scriptstyle s$} (1,1);
\draw[mid>] (0,0) to[out=0,in=-90] node[below]{$\scriptstyle t$} (1,1);
\draw[mid>] (4,0) to[out=180,in=-90] node[below]{$\scriptstyle u$} (3,1);
\draw[mid<] (4,0) to[out=90,in=0] node[right]{$\scriptstyle v$} (3,1);
\draw[mid<] (4,4) to[out=-90,in=0] node[right]{$\scriptstyle w$} (3,3);
\draw[mid<] (4,4) to[out=180,in=90] node[left]{$\scriptstyle x$} (3,3);
\draw[mid<] (0,4) to[out=0,in=90] node[right]{$\scriptstyle y$} (1,3);
\draw[mid>]  (0,4) to[out=-90,in=180] node[below]{$\scriptstyle z$} (1,3);
\node at (.5,2.85) {$\scriptstyle \xi_{1,2}$};
\node at (.7,.7) {$\scriptstyle \xi_{1,1}$};
\node at (3.25,.8) {$\scriptstyle \xi_{2,1}$};
\node at (3.4,3.3) {$\scriptstyle \xi_{2,2}$};
\node at (-.3,-.2) {$\scriptstyle \eta^\dag_{1,1}$};
\node at (4.6,-.2) {$\scriptstyle \overline{\eta_{2,1}}$};
\node at (4.3,4.2) {$\scriptstyle \eta^\dag_{2,2}$};
\node at (-.3,4.2) {$\scriptstyle \overline{\eta_{1,2}}$};
\draw[thick, blue, mid>] (1,2.5) to[out=45, in=180] (1.5,2.8) -- (2.5,2.8) to[out=0,in=90] (2.8,2.5) -- (2.8,1.5) to[out=-90,in=0] (2.5,1.2) -- node[above]{$\scriptstyle r$} (1.5,1.2) to[out=180, in=-45] (1,1.5);
\fill[fill=orange] (1,1.5) circle (.05cm);
\fill[fill=orange] (1,2.5) circle (.05cm);
}\,.
$$
As this diagram is a closed diagram in $\cX$, we can apply the \eqref{eq:Fusion} relation again to `unzip' along the $r$ string to obtain \eqref{eq:CoefficientOfBp} as claimed.

To see that $B_p$ is self-adjoint, we simply observe $C(\eta, \xi) = \overline{C(\xi,\eta)}$, and thus the matrix representation of $B_p$ is self-adjoint.
That $B_p^2=B_p$ and $[B_p,B_q]=0$ for distinct plaquettes $p,q$ follows from \eqref{eq:I=H} and \eqref{eq:Bigon 2} as depicted in \cite[\S5.4]{YanbaiZhang}.
\end{proof}

The string-net local Hamiltonian is given by
\begin{equation}
\label{eq:XHamiltonian}
H := -\sum_\ell A_\ell - \sum_p B_p.
\end{equation}
We can describe the ground state space of $H$ in terms of skein modules for $\cX$ \cite{MR3204497}.
There is an obvious \emph{evaluation map} $\eval: P_A \cH_{\tot}\to \cS(\bbD,N)$ given by evaluating a morphism in the skein module, where $N$ is the number of external edges of our lattice.
The following result is claimed in \cite{MR3204497}; we give a short conceptual proof for convenience and completeness.
We will prove an enriched version in \S\ref{sec:Enriched} below.

\begin{thm}[\cite{MR3204497}]
\label{thm:ProjectToSkeinModule}
On $P_A \cH_{\tot}$, $\prod_p B_p = D_\cX^{-\#p}\eval^\dag \circ \eval$,
where $\#p$ is the number of plaquettes in the lattice.
\end{thm}
\begin{proof}
The proof proceeds by adjoining adjacent plaquettes by induction.
We provide the detailed argument for 2 adjacent plaquettes $o,p$, and we leave the remaining details to the reader.
We will show that if $\xi,\xi'$ are choices of 6 vectors from our distinguished ONB of $\cH_v$,
$$
\langle \xi | B_oB_p \xi'\rangle_{\bigotimes_{v\in \cB} \cH_v} 
= 
D_\cX^{-2}
\langle \eval\xi | \eval\xi'\rangle_{\cS_\cX(\bbD,10)}
\qquad\Longleftrightarrow
\qquad
B_oB_p=D_\cX^{-2}\eval^\dag\circ \eval.
$$
To ease the notation, we set
\begin{align*}
K&:= \frac{1}{\sqrt{d_qd_rd_sd_td_ud_vd_wd_xd_yd_z}},
\displaybreak[1]\\
K_o&:=\frac{\sqrt{d_jd_hd_md_g}}{\sqrt{d_{j'}d_{h''}d_{m'}d_{g'}}}
&
K_o'&:=\frac{1}{\sqrt{d_jd_{j'}d_{h}d_md_{m'}d_gd_{g'}}},
\displaybreak[1]\\
K_p(h'')&:=\frac{\sqrt{d_kd_id_\ell d_{h''}}}{\sqrt{d_{k'}d_{i'}d_{\ell'}d_{h'}}}
&
K_p'&:=\frac{1}{\sqrt{d_kd_{k'}d_id_{i'}d_\ell d_{\ell'}d_{h'}}}.
\end{align*}
In the calculation below, we label strings and vertices as much as possible, but sometimes we omit labels that can be determined from the previous diagram when the diagram is getting very intricate.
Moreover, since we have fixed $\eta_{i,j}\in \operatorname{ONB}$, some of the sums over simples from the plaquette operators collapse as in the proof of Lemma \ref{lem:PlaquetteCoefficients}.
We suppress shadings on squares for readability.
\begin{align*}
C_{o,p}(\xi',\xi)
&=
\left\langle
\tikzmath{
\begin{scope}
\draw[step=1.0,black] (.5,.5) grid (3.5,2.5);
\end{scope}
\foreach \i in {1,2,3} {
\foreach \j in {1, 2} {
        \node at ($ (\i,\j) + (.25,-.15) $) {$\scriptstyle \xi_{\i,\j}$};
}
}
\node at (.3,1) {$\scriptstyle q$};
\node at (1,.3) {$\scriptstyle r$};
\node at (2,.3) {$\scriptstyle s$};
\node at (3,.3) {$\scriptstyle t$};
\node at (3.7,1) {$\scriptstyle u$};
\node at (3.7,2) {$\scriptstyle v$};
\node at (3,2.7) {$\scriptstyle w$};
\node at (2,2.7) {$\scriptstyle x$};
\node at (1,2.7) {$\scriptstyle y$};
\node at (.3,2) {$\scriptstyle z$};
\node at (1.5,1.5) {$\scriptstyle o$};
\node at (2.5,1.5) {$\scriptstyle p$};
\node at (.85,1.5) {$\scriptstyle g$};
\node at (2.15,1.5) {$\scriptstyle h$};
\node at (3.15,1.5) {$\scriptstyle i$};
\node at (1.7,1.15) {$\scriptstyle j$};
\node at (2.7,1.1) {$\scriptstyle k$};
\node at (2.7,2.15) {$\scriptstyle \ell$};
\node at (1.7,2.1) {$\scriptstyle m$};
}
\Bigg|
B_o B_p
\tikzmath{
\begin{scope}
\draw[step=1.0,black] (.5,.5) grid (3.5,2.5);
\end{scope}
\foreach \i in {1,2,3} {
\foreach \j in {1,2} {
        \node at ($ (\i,\j) + (.25,-.15) $) {$\scriptstyle \xi'_{\i,\j}$};
}
}
\node at (.3,1) {$\scriptstyle q$};
\node at (1,.3) {$\scriptstyle r$};
\node at (2,.3) {$\scriptstyle s$};
\node at (3,.3) {$\scriptstyle t$};
\node at (3.7,1) {$\scriptstyle u$};
\node at (3.7,2) {$\scriptstyle v$};
\node at (3,2.7) {$\scriptstyle w$};
\node at (2,2.7) {$\scriptstyle x$};
\node at (1,2.7) {$\scriptstyle y$};
\node at (.3,2) {$\scriptstyle z$};
\node at (1.5,1.5) {$\scriptstyle o$};
\node at (2.5,1.5) {$\scriptstyle p$};
\node at (.85,1.5) {$\scriptstyle g'$};
\node at (2.15,1.5) {$\scriptstyle h'$};
\node at (3.15,1.5) {$\scriptstyle i'$};
\node at (1.7,1.15) {$\scriptstyle j'$};
\node at (2.7,1.1) {$\scriptstyle k'$};
\node at (2.7,2.15) {$\scriptstyle \ell'$};
\node at (1.7,2.15) {$\scriptstyle m'$};
}
\right\rangle
\displaybreak[1]\\&=
\frac{1}{\cD_\cX}
\sum_{n',h''\in \Irr(\cX)} 
\frac{K_p(h'')}{d_{n'}}
\left\langle
\tikzmath{
\begin{scope}
\draw[step=1.0,black] (.5,.5) grid (3.5,2.5);
\end{scope}
\node at (.3,1) {$\scriptstyle q$};
\node at (1,.3) {$\scriptstyle r$};
\node at (2,.3) {$\scriptstyle s$};
\node at (3,.3) {$\scriptstyle t$};
\node at (3.7,1) {$\scriptstyle u$};
\node at (3.7,2) {$\scriptstyle v$};
\node at (3,2.7) {$\scriptstyle w$};
\node at (2,2.7) {$\scriptstyle x$};
\node at (1,2.7) {$\scriptstyle y$};
\node at (.3,2) {$\scriptstyle z$};
\node at (1.15,1.5) {$\scriptstyle g$};
\node at (2.15,1.5) {$\scriptstyle h$};
\node at (3.15,1.5) {$\scriptstyle i$};
\node at (1.5,.9) {$\scriptstyle j$};
\node at (2.5,.9) {$\scriptstyle k$};
\node at (2.5,1.9) {$\scriptstyle \ell$};
\node at (1.5,1.9) {$\scriptstyle m$};
}
\Bigg|
B_o
\tikzmath{
\begin{scope}
\draw[step=1.0,black] (.5,.5) grid (3.5,2.5);
\end{scope}
\draw[thick, cyan] (2.3,1) -- (2,1.3);
\draw[thick, cyan] (2.7,1) -- (3,1.3);
\draw[thick, cyan] (2.3,2) -- (2,1.7);
\draw[thick, cyan] (2.7,2) -- (3,1.7);
\node[cyan] at (2.3,1.3) {$\scriptstyle n'$};
\fill[fill=green] (2.3,1) circle (.05cm);
\fill[fill=green] (2.7,1) circle (.05cm);
\fill[fill=purple] (3,1.3) circle (.05cm);
\fill[fill=purple] (3,1.7) circle (.05cm);
\fill[fill=yellow] (2.3,2) circle (.05cm);
\fill[fill=yellow] (2.7,2) circle (.05cm);
\fill[fill=orange] (2,1.3) circle (.05cm);
\fill[fill=orange] (2,1.7) circle (.05cm);
\node at (.3,1) {$\scriptstyle q$};
\node at (1,.3) {$\scriptstyle r$};
\node at (2,.3) {$\scriptstyle s$};
\node at (3,.3) {$\scriptstyle t$};
\node at (3.7,1) {$\scriptstyle u$};
\node at (3.7,2) {$\scriptstyle v$};
\node at (3,2.7) {$\scriptstyle w$};
\node at (2,2.7) {$\scriptstyle x$};
\node at (1,2.7) {$\scriptstyle y$};
\node at (.3,2) {$\scriptstyle z$};
\node at (1.85,1.15) {$\scriptstyle h'$};
\node at (1.85,1.5) {$\scriptstyle h''$};
\node at (1.85,1.85) {$\scriptstyle h'$};
\node at (3.15,1.85) {$\scriptstyle i'$};
\node at (3.15,1.5) {$\scriptstyle i$};
\node at (3.15,1.15) {$\scriptstyle i'$};
\node at (2.15,.9) {$\scriptstyle k'$};
\node at (2.5,.9) {$\scriptstyle k$};
\node at (2.85,.9) {$\scriptstyle k'$};
\node at (2.15,2.15) {$\scriptstyle \ell'$};
\node at (2.5,2.1) {$\scriptstyle \ell$};
\node at (2.85,2.15) {$\scriptstyle \ell'$};
}
\right\rangle
\displaybreak[1]\\&=
\frac{1}{\cD_\cX^2}
\sum_{n,n',h''\in \Irr(\cX)} 
\frac{K_o K_p(h'')}{d_nd_{n'}}
\left\langle
\tikzmath{
\begin{scope}
\draw[step=1.0,black] (.5,.5) grid (3.5,2.5);
\end{scope}
\node at (.3,1) {$\scriptstyle q$};
\node at (1,.3) {$\scriptstyle r$};
\node at (2,.3) {$\scriptstyle s$};
\node at (3,.3) {$\scriptstyle t$};
\node at (3.7,1) {$\scriptstyle u$};
\node at (3.7,2) {$\scriptstyle v$};
\node at (3,2.7) {$\scriptstyle w$};
\node at (2,2.7) {$\scriptstyle x$};
\node at (1,2.7) {$\scriptstyle y$};
\node at (.3,2) {$\scriptstyle z$};
\node at (1.15,1.5) {$\scriptstyle g$};
\node at (2.15,1.5) {$\scriptstyle h$};
\node at (3.15,1.5) {$\scriptstyle i$};
\node at (1.5,.9) {$\scriptstyle j$};
\node at (2.5,.9) {$\scriptstyle k$};
\node at (2.5,1.9) {$\scriptstyle \ell$};
\node at (1.5,1.9) {$\scriptstyle m$};
}
\Bigg|
\tikzmath{
\begin{scope}
\draw[step=1.0,black] (.5,.5) grid (3.5,2.5);
\end{scope}
\draw[thick, blue] (1.3,1) -- (1,1.3);
\draw[thick, blue] (1.7,1) -- (2,1.3);
\draw[thick, blue] (1.3,2) -- (1,1.7);
\draw[thick, blue] (1.7,2) -- (2,1.7);
\fill[fill=green] (1.3,1) circle (.05cm);
\fill[fill=green] (1.7,1) circle (.05cm);
\fill[fill=purple] (2,1.3) circle (.05cm);
\fill[fill=purple] (2,1.7) circle (.05cm);
\fill[fill=yellow] (1.3,2) circle (.05cm);
\fill[fill=yellow] (1.7,2) circle (.05cm);
\fill[fill=orange] (1,1.3) circle (.05cm);
\fill[fill=orange] (1,1.7) circle (.05cm);
\node[blue] at (1.3,1.3) {$\scriptstyle n$};
\draw[thick, cyan] (2.3,1) -- (2,1.2);
\draw[thick, cyan] (2.7,1) -- (3,1.3);
\draw[thick, cyan] (2.3,2) -- (2,1.8);
\draw[thick, cyan] (2.7,2) -- (3,1.7);
\node[cyan] at (2.3,1.3) {$\scriptstyle n'$};
\node at (.3,1) {$\scriptstyle q$};
\node at (1,.3) {$\scriptstyle r$};
\node at (2,.3) {$\scriptstyle s$};
\node at (3,.3) {$\scriptstyle t$};
\node at (3.7,1) {$\scriptstyle u$};
\node at (3.7,2) {$\scriptstyle v$};
\node at (3,2.7) {$\scriptstyle w$};
\node at (2,2.7) {$\scriptstyle x$};
\node at (1,2.7) {$\scriptstyle y$};
\node at (.3,2) {$\scriptstyle z$};
\node at (.85,1.85) {$\scriptstyle g'$};
\node at (.85,1.5) {$\scriptstyle g$};
\node at (.85,1.15) {$\scriptstyle g'$};
\node at (1.85,1.5) {$\scriptstyle h$};
\node at (1.15,.9) {$\scriptstyle j'$};
\node at (1.5,.9) {$\scriptstyle j$};
\node at (1.85,.9) {$\scriptstyle j'$};
\node at (1.15,2.15) {$\scriptstyle m'$};
\node at (1.5,2.1) {$\scriptstyle m$};
\node at (1.85,2.15) {$\scriptstyle m'$};
\node at (3.15,1.85) {$\scriptstyle i'$};
\node at (3.15,1.5) {$\scriptstyle i$};
\node at (3.15,1.15) {$\scriptstyle i'$};
\node at (2.15,.9) {$\scriptstyle k'$};
\node at (2.5,.9) {$\scriptstyle k$};
\node at (2.85,.9) {$\scriptstyle k'$};
\node at (2.15,2.15) {$\scriptstyle \ell'$};
\node at (2.5,2.1) {$\scriptstyle \ell$};
\node at (2.85,2.15) {$\scriptstyle \ell'$};
}
\right\rangle
\displaybreak[1]\\&=
\frac{K}{\cD_\cX^2}
\sum_{n,n',h''\in \Irr(\cX)} 
\frac{K_o'K_p'}{d_nd_{n'}}
\tikzmath{
\clip (-1,-1) rectangle (9,5);
\draw[mid>] (1,1) -- node[below]{$\scriptstyle j'$} (1.5,1);
\draw[mid>] (1.5,1) .. controls ++(0:1.5cm) and ++(-90:2cm) .. node[right]{$\scriptstyle j$} (0,0);
\draw[mid>] (1,1) -- node[left]{$\scriptstyle g'$} (1,1.5);
\draw[mid>] (1,1.5) .. controls ++(90:1.5cm) and ++(180:2cm) .. node[above]{$\scriptstyle g$} (0,0);
\draw[mid>] (1,2.5) -- node[left]{$\scriptstyle g'$} (1,3);
\draw[mid>] (0,4) .. controls ++(180:2cm) and ++(-90:1.5cm) .. node[below]{$\scriptstyle g$} (1,2.5);
\draw[mid>] (1,3) -- node[above]{$\scriptstyle m'$} (1.5,3);
\draw[mid>] (1.5,3) .. controls ++(0:1.5cm) and ++(90:2cm) .. node[right]{$\scriptstyle m$} (0,4);
\draw[mid>] (2.5,1) -- node[below]{$\scriptstyle j'$} (3,1);
\draw[mid<] (2.5,1) .. controls ++(180:1.5cm) and ++(-90:2cm) .. node[below]{$\scriptstyle j$} (4,0);
\draw[mid>] (3,1) -- (3,1.5); 
\draw[mid>] (3,1.5) .. controls ++(90:1.5cm) and ++(0:2cm) .. node[above]{$\scriptstyle h$} (4,0);
\draw[mid>] (3,2.5) -- (3,3); 
\draw[mid>] (4,4) .. controls ++(0:2cm) and ++(-90:1.5cm) .. node[below]{$\scriptstyle h$} (3,2.5);
\draw[mid>] (2.5,3) -- node[above]{$\scriptstyle m'$} (3,3);
\draw[mid>] (4,4) .. controls ++(90:2cm) and ++(180:1.5cm) .. node[left]{$\scriptstyle m$} (2.5,3);
\draw[mid>] (0,0) to[out=90,in=180] node[above]{$\scriptstyle q$} (1,1);
\draw[mid>] (0,0) to[out=0,in=-90] node[below]{$\scriptstyle r$} (1,1);
\draw[mid>] (4,0) to[out=180,in=-90] node[below]{$\scriptstyle s$} (3,1);
\draw[mid<] (4,0) to[out=90,in=0] node[right]{$\scriptstyle k$} (3.5,1) -- (3,1);
\draw[mid<] (4,4) to[out=180,in=90] node[left]{$\scriptstyle x$} (3,3);
\draw[mid<] (4,4) to[out=-90,in=0] node[right]{$\scriptstyle \ell$} (3.5,3) -- (3,3);
\draw[mid<] (0,4) to[out=0,in=90] node[right]{$\scriptstyle y$} (1,3);
\draw[mid>]  (0,4) to[out=-90,in=180] node[below]{$\scriptstyle z$} (1,3);
\node at (.4,2.85) {$\scriptstyle \xi_{1,2}$};
\node at (.7,.7) {$\scriptstyle \xi_{1,1}$};
\node at (3.25,.8) {$\scriptstyle \xi_{2,1}$};
\node at (3.4,3.3) {$\scriptstyle \xi_{2,2}$};
\node at (-.3,-.2) {$\scriptstyle \eta^\dag_{1,1}$};
\node at (4.3,-.2) {$\scriptstyle \overline{\eta_{2,1}}$};
\node at (4.3,4.2) {$\scriptstyle \eta^\dag_{2,2}$};
\node at (-.3,4.2) {$\scriptstyle \overline{\eta_{1,2}}$};
\draw[thick, blue, mid>] (1.5,1) --node[above]{$\scriptstyle n$} (1,1.5);
\draw[thick, blue, mid>] (1,2.5) --node[right]{$\scriptstyle n$} (1.5,3);
\draw[thick, blue, mid>] (3,1.5) --node[left]{$\scriptstyle n$} (2.5,1);
\draw[thick, blue, mid>] (2.5,3) --node[below]{$\scriptstyle n$} (3,2.5);
\draw[thick, cyan, mid>] (3.5,1) --node[above]{$\scriptstyle n'$} (3,1.3);
\draw[thick, cyan, mid>] (3,2.7) --node[right]{$\scriptstyle n'$} (3.5,3);
\fill[fill=green] (1.5,1) circle (.05cm);
\fill[fill=green] (2.5,1) circle (.05cm);
\fill[fill=purple] (3,1.5) circle (.05cm);
\fill[fill=purple] (3,2.5) circle (.05cm);
\fill[fill=yellow] (1.5,3) circle (.05cm);
\fill[fill=yellow] (2.5,3) circle (.05cm);
\fill[fill=orange] (1,1.5) circle (.05cm);
\fill[fill=orange] (1,2.5) circle (.05cm);
\draw[mid>] (5.5,1) -- node[below]{$\scriptstyle k'$} (6,1);
\draw[mid<] (5.5,1) .. controls ++(180:1.5cm) and ++(-90:2cm) .. node[below]{$\scriptstyle k$} (7,0);
\draw[mid>] (6,1) -- node[right]{$\scriptstyle i'$} (6,1.5);
\draw[mid>] (6,1.5) .. controls ++(90:1.5cm) and ++(0:2cm) .. node[above]{$\scriptstyle i$} (7,0);
\draw[mid>] (6,2.5) -- node[right]{$\scriptstyle i'$} (6,3);
\draw[mid>] (7,4) .. controls ++(0:2cm) and ++(-90:1.5cm) .. node[below]{$\scriptstyle i$} (6,2.5);
\draw[mid>] (5.5,3) -- node[above]{$\scriptstyle \ell'$} (6,3);
\draw[mid>] (7,4) .. controls ++(90:2cm) and ++(180:1.5cm) .. node[left]{$\scriptstyle \ell$} (5.5,3);
\draw[mid>] (7,0) to[out=180,in=-90] node[below]{$\scriptstyle t$} (6,1);
\draw[mid<] (7,0) to[out=90,in=0] node[right]{$\scriptstyle u$} (6.5,1) -- (6,1);
\draw[mid<] (7,4) to[out=-90,in=0] node[right]{$\scriptstyle v$} (6.5,3) -- (6,3);
\draw[mid<] (7,4) to[out=180,in=90] node[left]{$\scriptstyle w$} (6,3);
\draw[thick, cyan, mid>] (6,1.5) --node[left]{$\scriptstyle n'$} (5.5,1);
\draw[thick, cyan, mid>] (5.5,3) --node[below]{$\scriptstyle n'$} (6,2.5);
\node at (6.25,.8) {$\scriptstyle \xi_{3,1}$};
\node at (6.4,3.3) {$\scriptstyle \xi_{3,2}$};
\node at (7.3,-.2) {$\scriptstyle \overline{\eta_{3,1}}$};
\node at (7.3,4.2) {$\scriptstyle \eta^\dag_{3,2}$};
}
\displaybreak[1]\\&=
\frac{K}{\cD_\cX^2}
\sum_{n',h''\in \Irr(\cX)} 
\frac{K_p'\sqrt{d_{h''}}}{d_{n'}}
\cdot
\tikzmath{
\clip (-1,-1.5) rectangle (9,4.5);
\draw[mid>] (1,1) -- node[below]{$\scriptstyle j'$} (2,1);
\draw[mid>] (2,1) -- (2,2); 
\draw[mid>] (1,2) -- node[above]{$\scriptstyle m'$} (2,2);
\draw[mid>] (1,1) -- node[left]{$\scriptstyle g'$} (1,2);
\draw[mid<] (0,0) to[out=-90,in=-90] node[above]{$\scriptstyle j$} (3,0);
\draw[mid<] (3,0) to[out=0,in=0] node[left]{$\scriptstyle h$} (3,3);
\draw[mid<] (0,3) to[out=90,in=90] node[below]{$\scriptstyle m$} (3,3);
\draw[mid<] (0,0) to[out=180,in=180] node[right]{$\scriptstyle g$} (0,3);
\draw[mid>] (0,0) to[out=90,in=180] node[above]{$\scriptstyle q$} (1,1);
\draw[mid>] (0,0) to[out=0,in=-90] node[below]{$\scriptstyle r$} (1,1);
\draw[mid>] (3,0) to[out=180,in=-90] node[below]{$\scriptstyle s$} (2,1);
\draw[mid<] (3,0) to[out=90,in=0] node[right]{$\scriptstyle k$} (2.5,1) -- (2,1);
\draw[mid<] (3,3) to[out=180,in=90] node[left]{$\scriptstyle x$} (2,2);
\draw[mid<] (3,3) to[out=-90,in=0] node[right]{$\scriptstyle \ell$} (2.5,2) -- (2,2);
\draw[mid<] (0,3) to[out=0,in=90] node[right]{$\scriptstyle y$} (1,2);
\draw[mid>]  (0,3) to[out=-90,in=180] node[below]{$\scriptstyle z$} (1,2);
\draw[thick, cyan, mid>] (2.5,1) --node[above]{$\scriptstyle n'$} (2,1.3);
\draw[thick, cyan, mid>] (2,1.7) --node[right]{$\scriptstyle n'$} (2.5,2);
\node at (.7,1.8) {$\scriptstyle \xi_{1,2}$};
\node at (.7,.8) {$\scriptstyle \xi_{1,1}$};
\node at (2.3,.8) {$\scriptstyle \xi_{2,1}$};
\node at (2.3,2.2) {$\scriptstyle \xi_{2,2}$};
\node at (-.3,-.2) {$\scriptstyle \eta^\dag_{1,1}$};
\node at (3.3,-.2) {$\scriptstyle \overline{\eta_{2,1}}$};
\node at (3.3,3.2) {$\scriptstyle \eta^\dag_{2,2}$};
\node at (-.3,3.2) {$\scriptstyle \overline{\eta_{1,2}}$};
\node at (1.8,1.5) {$\scriptstyle h''$};
\node at (1.8,1.15) {$\scriptstyle h'$};
\node at (1.8,1.85) {$\scriptstyle h'$};
\begin{scope}[yshift=-.5cm, xshift=-.5cm]
\draw[mid>] (5.5,1) -- node[below]{$\scriptstyle k'$} (6,1);
\draw[mid<] (5.5,1) .. controls ++(180:1.5cm) and ++(-90:2cm) .. node[below]{$\scriptstyle k$} (7,0);
\draw[mid>] (6,1) -- node[right]{$\scriptstyle i'$} (6,1.5);
\draw[mid>] (6,1.5) .. controls ++(90:1.5cm) and ++(0:2cm) .. node[above]{$\scriptstyle i$} (7,0);
\draw[mid>] (6,2.5) -- node[right]{$\scriptstyle i'$} (6,3);
\draw[mid>] (7,4) .. controls ++(0:2cm) and ++(-90:1.5cm) .. node[below]{$\scriptstyle i$} (6,2.5);
\draw[mid>] (5.5,3) -- node[above]{$\scriptstyle \ell'$} (6,3);
\draw[mid>] (7,4) .. controls ++(90:2cm) and ++(180:1.5cm) .. node[left]{$\scriptstyle \ell$} (5.5,3);
\draw[mid>] (7,0) to[out=180,in=-90] node[below]{$\scriptstyle t$} (6,1);
\draw[mid<] (7,0) to[out=90,in=0] node[right]{$\scriptstyle u$} (6.5,1) -- (6,1);
\draw[mid<] (7,4) to[out=-90,in=0] node[right]{$\scriptstyle v$} (6.5,3) -- (6,3);
\draw[mid<] (7,4) to[out=180,in=90] node[left]{$\scriptstyle w$} (6,3);
\draw[thick, cyan, mid>] (6,1.5) --node[left]{$\scriptstyle n'$} (5.5,1);
\draw[thick, cyan, mid>] (5.5,3) --node[below]{$\scriptstyle n'$} (6,2.5);
\node at (6.25,.8) {$\scriptstyle \xi_{3,1}$};
\node at (6.4,3.3) {$\scriptstyle \xi_{3,2}$};
\node at (7.3,-.2) {$\scriptstyle \overline{\eta_{3,1}}$};
\node at (7.3,4.2) {$\scriptstyle \eta^\dag_{3,2}$};
\end{scope}
}
\displaybreak[1]\\&=
\frac{K}{\cD_\cX^2}
\sum_{n',h''\in \Irr(\cX)} 
\frac{K_p'\sqrt{d_{h''}}}{d_{n'}}
\cdot
\tikzmath{
\clip (-1.7,-1.7) rectangle (9,4.7);
\draw[mid>] (1,1) -- node[below]{$\scriptstyle j'$} (2,1);
\draw[mid>] (2,1) -- (2,2); 
\draw[mid>] (1,2) -- node[above]{$\scriptstyle m'$} (2,2);
\draw[mid>] (1,1) -- node[left]{$\scriptstyle g'$} (1,2);
\draw[mid<] (0,0) to[out=-90,in=-90] node[above]{$\scriptstyle j$} (3,0);
\draw[mid<] (3,0) .. controls ++(0:1cm) and ++(0:1cm) .. (3,-1.5) -- (0,-1.5) to[out=180,in=-90] (-1.5,1.5) node[right]{$\scriptstyle h$} to[out=90,in=180] (0,4.5) -- (3,4.5) .. controls ++(0:1cm) and ++(0:1cm) .. (3,3);
\draw[mid<] (0,3) to[out=90,in=90] node[below]{$\scriptstyle m$} (3,3);
\draw[mid<] (0,0) to[out=180,in=180] node[right]{$\scriptstyle g$} (0,3);
\draw[mid>] (0,0) to[out=90,in=180] node[above]{$\scriptstyle q$} (1,1);
\draw[mid>] (0,0) to[out=0,in=-90] node[below]{$\scriptstyle r$} (1,1);
\draw[mid>] (3,0) to[out=180,in=-90] node[below]{$\scriptstyle s$} (2,1);
\draw[mid<] (3,0) to[out=90,in=0] node[right]{$\scriptstyle k$} (2.5,1) -- (2,1);
\draw[mid<] (3,3) to[out=180,in=90] node[left]{$\scriptstyle x$} (2,2);
\draw[mid<] (3,3) to[out=-90,in=0] node[right]{$\scriptstyle \ell$} (2.5,2) -- (2,2);
\draw[mid<] (0,3) to[out=0,in=90] node[right]{$\scriptstyle y$} (1,2);
\draw[mid>]  (0,3) to[out=-90,in=180] node[below]{$\scriptstyle z$} (1,2);
\draw[thick, cyan, mid>] (2.5,1) --node[above]{$\scriptstyle n'$} (2,1.3);
\draw[thick, cyan, mid>] (2,1.7) --node[right]{$\scriptstyle n'$} (2.5,2);
\node at (.7,1.8) {$\scriptstyle \xi_{1,2}$};
\node at (.7,.8) {$\scriptstyle \xi_{1,1}$};
\node at (2.3,.8) {$\scriptstyle \xi_{2,1}$};
\node at (2.3,2.2) {$\scriptstyle \xi_{2,2}$};
\node at (-.3,-.2) {$\scriptstyle \eta^\dag_{1,1}$};
\node at (3.3,-.2) {$\scriptstyle \overline{\eta_{2,1}}$};
\node at (3.3,3.2) {$\scriptstyle \eta^\dag_{2,2}$};
\node at (-.3,3.2) {$\scriptstyle \overline{\eta_{1,2}}$};
\node at (1.8,1.5) {$\scriptstyle h''$};
\node at (1.8,1.15) {$\scriptstyle h'$};
\node at (1.8,1.85) {$\scriptstyle h'$};
\fill[fill=green] (2.5,1) circle (.05cm);
\fill[fill=yellow] (2.5,2) circle (.05cm);
\fill[fill=orange] (2,1.3) circle (.05cm);
\fill[fill=orange] (2,1.7) circle (.05cm);
\begin{scope}[yshift=-.5cm, xshift=-1cm]
\draw[mid>] (5.5,1) -- node[below]{$\scriptstyle k'$} (6,1);
\draw[mid<] (5.5,1) .. controls ++(180:1.5cm) and ++(-90:2cm) .. node[below]{$\scriptstyle k$} (7,0);
\draw[mid>] (6,1) -- node[right]{$\scriptstyle i'$} (6,1.5);
\draw[mid>] (6,1.5) .. controls ++(90:1.5cm) and ++(0:2cm) .. node[above]{$\scriptstyle i$} (7,0);
\draw[mid>] (6,2.5) -- node[right]{$\scriptstyle i'$} (6,3);
\draw[mid>] (7,4) .. controls ++(0:2cm) and ++(-90:1.5cm) .. node[below]{$\scriptstyle i$} (6,2.5);
\draw[mid>] (5.5,3) -- node[above]{$\scriptstyle \ell'$} (6,3);
\draw[mid>] (7,4) .. controls ++(90:2cm) and ++(180:1.5cm) .. node[left]{$\scriptstyle \ell$} (5.5,3);
\draw[mid>] (7,0) to[out=180,in=-90] node[below]{$\scriptstyle t$} (6,1);
\draw[mid<] (7,0) to[out=90,in=0] node[right]{$\scriptstyle u$} (6.5,1) -- (6,1);
\draw[mid<] (7,4) to[out=-90,in=0] node[right]{$\scriptstyle v$} (6.5,3) -- (6,3);
\draw[mid<] (7,4) to[out=180,in=90] node[left]{$\scriptstyle w$} (6,3);
\draw[thick, cyan, mid>] (6,1.5) --node[left]{$\scriptstyle n'$} (5.5,1);
\draw[thick, cyan, mid>] (5.5,3) --node[below]{$\scriptstyle n'$} (6,2.5);
\fill[fill=green] (5.5,1) circle (.05cm);
\fill[fill=purple] (6,1.5) circle (.05cm);
\fill[fill=purple] (6,2.5) circle (.05cm);
\fill[fill=yellow] (5.5,3) circle (.05cm);
\node at (6.25,.8) {$\scriptstyle \xi_{3,1}$};
\node at (6.4,3.3) {$\scriptstyle \xi_{3,2}$};
\node at (7.3,-.2) {$\scriptstyle \overline{\eta_{3,1}}$};
\node at (7.3,4.2) {$\scriptstyle \eta^\dag_{3,2}$};
\end{scope}
}
\displaybreak[1]\\&=
\frac{K}{\cD_\cX^2}
\cdot
\tikzmath{
\clip (-1.7,-1.7) rectangle (9,4.7);
\draw[mid>] (1,1) -- node[above, xshift=-.1cm]{$\scriptstyle j'$} (2,1);
\draw[mid>] (2,1) -- node[left]{$\scriptstyle h'$} (2,2);
\draw[mid>] (1,2) -- node[above]{$\scriptstyle m'$} (2,2);
\draw[mid>] (1,1) -- node[left]{$\scriptstyle g'$} (1,2);
\draw[mid>] (2,2) -- node[above]{$\scriptstyle \ell'$} (5,2);
\draw[mid>] (2,1) -- node[above]{$\scriptstyle k'$} (5,1);
\draw[mid>] (5,1) -- node[left]{$\scriptstyle i'$} (5,2);
\draw[mid<] (0,0) to[out=-90,in=-90] node[above]{$\scriptstyle j$} (3,0);
\draw[mid<] (3,0) .. controls ++(0:1cm) and ++(0:1cm) .. (3,-1.5) -- (0,-1.5) to[out=180,in=-90] (-1.5,1.5) node[right]{$\scriptstyle h$} to[out=90,in=180] (0,4.5) -- (3,4.5) .. controls ++(0:1cm) and ++(0:1cm) .. (3,3);
\draw[mid<] (0,3) to[out=90,in=90] node[below]{$\scriptstyle m$} (3,3);
\draw[mid<] (0,0) to[out=180,in=180] node[right]{$\scriptstyle g$} (0,3);
\draw[mid>] (6,3) to[out=0,in=0] node[right]{$\scriptstyle i$} (6,0);
\draw[mid>] (0,0) to[out=90,in=180] node[above]{$\scriptstyle q$} (1,1);
\draw[mid>] (0,0) to[out=0,in=-90] node[below]{$\scriptstyle r$} (1,1);
\draw[mid>] (3,0) to[out=180,in=-90] node[below]{$\scriptstyle s$} (2,1);
\draw[mid<] (3,0) to[out=90,in=90] (4.5,0) to[out=-90,in=-90] node[below]{$\scriptstyle k$} (6,0);
\draw[mid<] (3,3) to[out=180,in=90] node[left]{$\scriptstyle x$} (2,2);
\draw[mid<] (3,3) to[out=-90,in=-90] (4,3) to[out=90, in=90] node[above]{$\scriptstyle \ell$} (6,3);
\draw[mid>] (6,0) to[out=180,in=-90] node[below]{$\scriptstyle t$} (5,1);
\draw[mid<] (6,0) to[out=90,in=0] node[right]{$\scriptstyle u$} (5,1);
\draw[mid<] (6,3) to[out=-90,in=0] node[right]{$\scriptstyle v$} (5,2);
\draw[mid<] (6,3) to[out=180,in=90] node[left]{$\scriptstyle w$} (5,2);
\draw[mid<] (0,3) to[out=0,in=90] node[right]{$\scriptstyle y$} (1,2);
\draw[mid>]  (0,3) to[out=-90,in=180] node[below]{$\scriptstyle z$} (1,2);
\node at (1.3,1.8) {$\scriptstyle \xi_{1,2}$};
\node at (1.3,.8) {$\scriptstyle \xi_{1,1}$};
\node at (2.3,.8) {$\scriptstyle \xi_{2,1}$};
\node at (2.3,2.2) {$\scriptstyle \xi_{2,2}$};
\node at (5.3,.8) {$\scriptstyle \xi_{3,1}$};
\node at (5.3,2.2) {$\scriptstyle \xi_{3,2}$};
\node at (-.3,-.2) {$\scriptstyle \eta^\dag_{1,1}$};
\node at (-.3,3.2) {$\scriptstyle \overline{\eta_{1,2}}$};
\node at (2.8,.2) {$\scriptstyle \overline{\eta_{2,1}}$};
\node at (2.8,2.7) {$\scriptstyle \eta^\dag_{2,2}$};
\node at (6.3,3.2) {$\scriptstyle \eta^\dag_{3,2}$};
\node at (6.3,-.2) {$\scriptstyle \overline{\eta_{3,1}}$};
}
\end{align*}
This final diagram is exactly the $D_\cX^{-2}$ times the skein module $\cS_\cX(\bbD,10)$ inner product.
\end{proof}

\begin{rem}
Under the unitary isomorphism from Proposition \ref{prop:DescriptionOfIm(PA)}, the operator $B_p$ corresponds to a \textit{Kirby element} of $\cX$ (see \cite[Def.~2.1]{MR2251160}) using the strings on tubes graphical calculus from \cite{MR3578212,MR4528312}.
Here, we require that $I:\cX\to Z(\cX)$ is the right unitary adjoint of $F: Z(\cX)\to \cX$ from \cite[\S2.1]{2301.11114}.
Recall that $\End_{Z(\cX)}(I(1_\cX)) \cong K_0(\cX)$, the fusion algebra of $\cX$, as
$$
\End_{Z(\cX)}(I(1_\cX)) 
\cong
\cX(FI(1_\cX)\to 1_\cX) 
\cong
\bigoplus_{x\in \Irr(\cX)}\cX(\overline{x}x\to 1_\cX).
$$
We write $i: 1_\cX\to I(1_\cX)=IF(1_{Z(\cX)})$ for the unit of the adjunction, and we have $i^\dag \circ i = 1_\bbC\in \End_\cX(1_\cX)$ by \cite[Lem.~3.8]{2301.11114}.
By \cite[Prop.~3.3.6]{MR3242743}, we see that
$$
\tikzmath{
\draw[thick] (0,0) arc (-180:0:.3 and .1);
\draw[thick, dotted] (0,0) arc (180:0:.3 and .1);
\draw[thick] (0,0) -- (0,.1) arc (180:0:.3) -- (.6,0);
\draw[thick] (0,1) arc (180:0:.3 and .1);
\draw[thick] (0,1) arc (-180:0:.3 and .1);
\draw[thick] (0,1) -- (0,.9) arc (-180:0:.3) -- (.6,1);
}
\quad
=
\quad
\frac{1}{D_\cX}
\sum_{x\in \cX}d_x\cdot
\tikzmath{
\draw[thick] (0,0) arc (-180:0:.3 and .1);
\draw[thick, dotted] (0,0) arc (180:0:.3 and .1);
\draw[thick] (0,1) arc (180:0:.3 and .1);
\draw[thick] (0,1) arc (-180:0:.3 and .1);
\draw[thick] (0,0) -- (0,1);
\draw[thick] (.6,0) -- (.6,1);
\draw[thick, blue] (0,.5) arc (-180:0:.3 and .1);
\draw[thick, blue, dotted] (0,.5) arc (180:0:.3 and .1);
\node[blue] at (.2,.3) {$\scriptstyle x$};
}
$$
as both are orthogonal projections $p\in K_0(\cX)$ (using \cite[Prop.~3.15]{2301.11114}) satisfying
$$
\tikzmath{
\draw[thick] (0,0) arc (-180:0:.3 and .1);
\draw[thick, dotted] (0,0) arc (180:0:.3 and .1);
\draw[thick] (0,1) arc (180:0:.3 and .1);
\draw[thick] (0,1) arc (-180:0:.3 and .1);
\draw[thick] (0,0) -- (0,1);
\draw[thick] (.6,0) -- (.6,1);
\draw[thick, blue] (0,.5) arc (-180:0:.3 and .1);
\draw[thick, blue, dotted] (0,.5) arc (180:0:.3 and .1);
\node[blue] at (.2,.3) {$\scriptstyle x$};
}
\cdot p
= 
d_x p
\qquad\qquad\qquad
\forall x\in \Irr(\cX).
$$
\end{rem}

Hence when our total Hilbert space $\cH_{\tot}$ has $N$ external edges, 
the ground state space of our Hamlitonian is unitarily isomorphic to $\cS_\cX(\bbD,N)$. 
$$
\tikzmath{
\clip[rounded corners=5pt] (-.5,-.5) rectangle (2.5,2.5);
\draw[thick] (-.5,-.5) grid (3.5,2.5);
\foreach \x in {0,1,2,3}{
\foreach \y in {0,1,2}{
\filldraw (\x,\y) circle (.05cm);
}}
\foreach \x in {-.5,...,3.5}{
\foreach \y in {-.5,...,2.5}{
}}
}
\qquad
\overset{\cong}{
\tikzmath{\draw[squiggly,->] (0,0) -- (1,0);}
}
\qquad
\cS_\cX(\bbD,N)
=
\left\{
\tikzmath{
\roundNbox{}{(0,0)}{.3}{.2}{.2}{$f$}
\draw (-.3,.3) --node[left]{$\scriptstyle x_1$} (-.3,.7) ;
\draw (.3,.3) --node[right]{$\scriptstyle x_N$} (.3,.7) ;
\node at (90:.5cm) {$\cdots$};
}
\right\}
$$

\subsection{The tube algebra}\label{sec:tubealgebra}

Suppose $\cX$ is a unitary fusion category.
We review Ocneanu's \emph{tube algebra}, which is a powerful tool for calculating the Drinfeld center $Z(\cX)$ \cite{MR1832764,MR1966525}.
Let $\Irr(\cX)$ denote a set of representatives of the simple objects of $\cX$.
We suppress tensor symbols, associators, and unitors whenever possible.

\begin{defn}
The \emph{tube algebra}\IN{$\Tube(\cX)$ the tube algebra of $\cX$} of $\cX$ is the finite dimensional $\rm C^*$-algebra
$$
\Tube(\cX)
:=
\bigoplus_{v,x,y\in \Irr(\cX)} \cX(vx \to yv)
$$
with multiplication and adjoint given respectively by
$$
\tikzmath{
\roundNbox{}{(0,0)}{.3}{.2}{.2}{$\phi$}
\draw (-.3,-.3) --node[left]{$\scriptstyle u$} (-.3,-.7);
\draw (.3,-.3) --node[right]{$\scriptstyle y'$} (.3,-.7);
\draw (-.3,.3) --node[left]{$\scriptstyle z$} (-.3,.7);
\draw (.3,.3) --node[right]{$\scriptstyle u$} (.3,.7);
}
\cdot
\tikzmath{
\roundNbox{}{(0,0)}{.3}{.2}{.2}{$\psi$}
\draw (-.3,-.3) --node[left]{$\scriptstyle v$} (-.3,-.7);
\draw (.3,-.3) --node[right]{$\scriptstyle x$} (.3,-.7);
\draw (-.3,.3) --node[left]{$\scriptstyle y$} (-.3,.7);
\draw (.3,.3) --node[right]{$\scriptstyle v$} (.3,.7);
}
:=
\delta_{y=y'}
\sum_{w\in \Irr(\cC)}
\frac{\sqrt{d_w}}{\sqrt{d_ud_v}}
\tikzmath{
\draw (.5,-.2) --node[right]{$\scriptstyle v$} (.5,.8) to[out=90,in=-45] (.3,1.3);
\draw (-.5,.2) --node[left]{$\scriptstyle u$} (-.5,-.8) to[out=270,in=135] (-.3,-1.3);
\draw (.1,.8) to[out=90,in=-135] node[left]{$\scriptstyle u$} (.3,1.3);
\draw (-.1,-.8) to[out=270,in=45] node[right]{$\scriptstyle v$} (-.3,-1.3);
\draw (.3,1.3) --node[right]{$\scriptstyle w$} (.3,1.7);
\draw (-.3,-1.3) --node[right]{$\scriptstyle w$} (-.3,-1.7);
\draw (.5,-.8) --node[right]{$\scriptstyle x$} (.5,-1.7);
\draw (0,-.2) --node[right]{$\scriptstyle y$} (0,.2);
\draw (-.5,.8) --node[left]{$\scriptstyle z$} (-.5,1.7);
\roundNbox{}{(.2,-.5)}{.3}{.2}{.2}{$\psi$}
\roundNbox{}{(-.2,.5)}{.3}{.2}{.2}{$\phi$}
\fill[fill=red] (.3,1.3) circle (.05cm);
\fill[fill=red] (-.3,-1.3) circle (.05cm);
}
\qquad\qquad
\left(
\tikzmath{
\roundNbox{}{(0,0)}{.3}{.2}{.2}{$\psi$}
\draw (-.3,-.3) --node[left]{$\scriptstyle v$} (-.3,-.7);
\draw (.3,-.3) --node[right]{$\scriptstyle x$} (.3,-.7);
\draw (-.3,.3) --node[left]{$\scriptstyle y$} (-.3,.7);
\draw (.3,.3) --node[right]{$\scriptstyle v$} (.3,.7);
}\right)^*
:=
\tikzmath{
\roundNbox{}{(.2,-.5)}{.3}{.2}{.2}{$\psi^\dag$}
\draw (0,-.2) arc (0:180:.3cm) -- node[left, yshift=-.35cm]{$\bar{\scriptstyle v}$}(-.6,-1.7);
\draw (0,-.8) --node[right, yshift=-.1cm]{$\scriptstyle y$} (0,-1.7);
\draw (0.4,-.2) --node[left, yshift=0.3cm]{$\scriptstyle x$} (0.4,1.1);
\draw (0.4,-0.8) arc (-180:0:.3cm) -- node[right, yshift=0.65cm]{$\bar{\scriptstyle v}$}(1.0,1.1);
}
$$
One uses \eqref{eq:I=H} to show $\Tube(\cX)$ is associative.
\end{defn}


Elements of $\Tube(\cX)$ can also be visualized as annuli (or tubes), where the composition is given by nesting (or stacking) and using the fusion relation for parallel strands. 
$$
\tikzmath{
\roundNbox{}{(0,0)}{.3}{.2}{.2}{$\psi$}
\draw (-.3,-.3) --node[left]{$\scriptstyle v$} (-.3,-.7);
\draw (.3,-.3) --node[right]{$\scriptstyle x$} (.3,-.7);
\draw (-.3,.3) --node[left]{$\scriptstyle y$} (-.3,.7);
\draw (.3,.3) --node[right]{$\scriptstyle v$} (.3,.7);
}
=\,\,
\tikzmath{
\draw[very thick] (0,0) circle (.3cm);
\draw (0,0) circle (1cm);
\node at (1.3,0) {$\scriptstyle \overline{v}$};
\node at (-1.3,0) {$\scriptstyle v$};
\draw (0,.3) --node[right]{$\scriptstyle x$} (0,.7);
\draw (0,1.3) --node[right]{$\scriptstyle y$} (0,1.7);
\draw[very thick] (0,0) circle (1.7cm);
\roundNbox{fill=white}{(90:1cm)}{.3}{0}{0}{$\psi$}
}
\,\,=\,\,
\tikzmath{
\draw (0,-.3) --node[right]{$\scriptstyle x$} (0,.5);
\draw (0,1.1) --node[right]{$\scriptstyle y$} (0,1.7);
\draw (-.9,1.1) arc (-180:0:.9 and .3);
\draw[dotted] (-.9,1.1) arc (180:0:.9 and .3);
\node at (.7,.7) {$\scriptstyle v$};
\node at (-.7,.7) {$\scriptstyle v$};
\draw[very thick] (-.9,0) -- (-.9,2);
\draw[very thick] (.9,0) -- (.9,2);
\draw[very thick] (0,2) ellipse (.9 and .3);
\draw[very thick] (-.9,0) arc(-180:0:.9 and .3);
\draw[very thick, dotted] (-.9,0) arc(180:0:.9 and .3);
\roundNbox{fill=white}{(0,.8)}{.3}{.2}{.2}{$\psi$}
}
$$

There is a well-known equivalence of categories $\Rep(\Tube(\cX)) \simeq Z(\cX)$. 
Our convention for $Z(\cX)$ is that objects $(X,\sigma_X)\in Z(\cX)$\IN{$\sigma_X$ the half braiding for an object $(X,\sigma_X)$ in the Drinfeld center} will be denoted by pairs consisting of an upper-case Roman letter and a lower-case Greek letter, usually $\sigma$ or $\tau$ representing the half-braiding. Our graphical calculus for $\sigma_X$ is that the string for $X$, which is drawn in color, goes \emph{under} to signify that $\sigma_X$ is natural with respect to all morphisms in $\cX$, but only special morphisms in $\cX$ define morphisms in $Z(\cX)$.
\begin{equation}
\label{eq:HalfBraidConvention}
\sigma_{X,y}
=
\tikzmath{
\draw[thick, orange] (.4,-.6) node[below]{$\scriptstyle X$} to[out=90, in=270] (-.4,.6) node[above]{$\scriptstyle X$};
\draw[knot] (-.4,-.6) node[below]{$\scriptstyle y$} to[out=90, in=270] (.4,.6) node[above]{$\scriptstyle y$};
}\,.
\end{equation}

Given $(X,\sigma_X)\in Z(\cX)$, we get a $*$-representation of $\Tube(\cX)$ on the Hilbert space
$\cH_X:= \bigoplus_{y\in \Irr(\cX)} \cX(X\to y)$\IN{$\cH_X$ the $*$-representation associated to the object $(X,\sigma_X)\in Z(\cC)$}
by
$$
\tikzmath{
\roundNbox{}{(0,0)}{.3}{.2}{.2}{$\phi$}
\draw (-.3,-.3) --node[left]{$\scriptstyle v$} (-.3,-.7);
\draw (.3,-.3) --node[right]{$\scriptstyle y$} (.3,-.7);
\draw (-.3,.3) --node[left]{$\scriptstyle z$} (-.3,.7);
\draw (.3,.3) --node[right]{$\scriptstyle v$} (.3,.7);
}
\rhd \,\,
\tikzmath{
\draw[thick,orange] (0,-.7) --node[right]{$\scriptstyle X$} (0,-.3);
\draw (0,.7) --node[right]{$\scriptstyle y$} (0,.3);
\roundNbox{}{(0,0)}{.3}{0}{0}{$m$}
}
:=
\tikzmath{
\draw[thick, orange] (.3,-2) node[below]{$\scriptstyle X$} -- (.3,-1.3);
\draw (.3,-.7) -- node[right]{$\scriptstyle y$} (.3,-.3);
\draw (-.3,.3) -- (-.3,.7) node[above]{$\scriptstyle z$};
\draw[knot] (.3,.3) node[left, yshift=.15cm]{$\scriptstyle v$} arc (180:0:.3cm) -- node[right]{$\scriptstyle \overline{v}$} (.9,-1) arc (0:-180:.6cm) -- node[left]{$\scriptstyle v$} (-.3,-.3);
\roundNbox{fill=white}{(0,0)}{.3}{.3}{.3}{$\phi$}
\roundNbox{fill=white}{(.3,-1)}{.3}{0}{0}{$m$}
}
\qquad\qquad\qquad
\forall\,m\in \cX(X\to y).
$$
Here, $\cX(X\to y)$ has the isometry inner product determined by $\langle f|g\rangle \id_y = f^\dag\circ g$.
That this representation is a $*$-representation is easily verified using the identity
\begin{equation}
\label{eq:TraceInnerProduct}
\langle n | \phi\rhd m \rangle_{\cX(X\to z)}
=
\frac{1}{d_X}
\tr^\cX_{X}\left(
\tikzmath{
\draw[thick, orange] (.3,-2) -- node[right]{$\scriptstyle X$} (.3,-1.6) --(.3,-1.3);
\draw[thick, orange] (-.3,1.7) -- node[right]{$\scriptstyle X$} (-.3,1.3);
\draw (.3,-.7) -- node[right]{$\scriptstyle y$} (.3,-.3);
\draw (-.3,.3) -- node[left]{$\scriptstyle z$} (-.3,.7);
\draw[knot] (.3,.3) node[left, yshift=.15cm]{$\scriptstyle v$} arc (180:0:.3cm) -- node[right]{$\scriptstyle \overline{v}$} (.9,-1) arc (0:-180:.6cm) -- node[left]{$\scriptstyle v$} (-.3,-.3);
\roundNbox{fill=white}{(0,0)}{.3}{.3}{.3}{$\phi$}
\roundNbox{fill=white}{(.3,-1)}{.3}{0}{0}{$m$}
\roundNbox{fill=white}{(-.3,1)}{.3}{0}{0}{$n^\dag$}
}
\right).
\end{equation}
By \cite{MR1832764,MR1966525}, all finite dimensional representations of $\Tube(\cX)$ are of this form.

We now explain how the above formula of the action completely determines the half-braiding $\sigma_X$.
Indeed, given this representation, we can recover `matrix elements' for the half-braiding $\sigma_X$.
First, the vector $\phi \rhd m$ is determined by its inner product against vectors $n\in \cX(X\to z)$
as in \eqref{eq:TraceInnerProduct} above.
Now, we can always factor $\phi$ through simples in $\Irr(\cX)$ to write it as a linear combination of the form
\begin{equation}
\label{eq:FactorTubeThroughSimples}
\tikzmath{
\draw (.3,-.7) -- node[right]{$\scriptstyle y$} (.3,-.3);
\draw (-.3,.3) -- node[left]{$\scriptstyle z$} (-.3,.7);
\draw (.3,.3) --  node[right]{$\scriptstyle v$} (.3,.7);
\draw (-.3,-.3) -- node[left]{$\scriptstyle v$} (-.3,-.7);
\roundNbox{fill=white}{(0,0)}{.3}{.3}{.3}{$\phi$}
}
=
\sum_{x\in \Irr(\cX)}
\sum_{i=1}^{n_x}
\tikzmath{
\draw (.3,-1.7) -- node[right]{$\scriptstyle y$} (.3,-1.3);
\draw (-.3,.3) -- node[left]{$\scriptstyle z$} (-.3,.7);
\draw (.3,.3) --  node[right]{$\scriptstyle v$} (.3,.7);
\draw (-.3,-1.3) -- node[left]{$\scriptstyle v$} (-.3,-1.7);
\draw (0,-.7) --node[right]{$\scriptstyle x$} (0,-.3);
\roundNbox{fill=white}{(0,0)}{.3}{.3}{.3}{$g_{x,i}^\dag$}
\roundNbox{fill=white}{(0,-1)}{.3}{.3}{.3}{$f_{x,i}$}
}
\end{equation}
Using the tracial property of $\tr^\cX$, we see that \eqref{eq:TraceInnerProduct} is equal to
$$
\sum_{x\in \Irr(\cX)}\sum_{i=1}^{n_x}
\frac{1}{d_X}\tr^\cX_x\left(
\tikzmath{
\draw[thick, orange] (-.3,-.45) node[left, yshift=.2cm]{$\scriptstyle X$} .. controls ++(90:.3cm) and ++(-90:.3cm) .. (.3,.45);
\draw[knot] (.2,-1.4) -- (.2,-.45) .. controls ++(90:.3cm) and ++(-90:.3cm) .. (-.2,.45) --node[left]{$\scriptstyle v$} (-.2,1.4);
\draw (.3,.95) --node[right]{$\scriptstyle y$} (.3,1.4);
\draw (-.3,-.95) --node[left]{$\scriptstyle z$} (-.3,-1.4);
\draw (0,2) --node[right]{$\scriptstyle x$} (0,2.4);
\draw (0,-2) --node[right]{$\scriptstyle x$} (0,-2.4);
\roundNbox{fill=white}{(0,1.7)}{.3}{.2}{.2}{$f_{x,i}$}
\roundNbox{fill=white}{(.3,.7)}{.25}{0}{0}{$m$}
\roundNbox{fill=white}{(-.3,-.7)}{.25}{0}{0}{$n^\dag$}
\roundNbox{fill=white}{(0,-1.7)}{.3}{.2}{.2}{$g^\dag_{x,i}$}
}
\right)
$$
Since $x\in \Irr(\cX)$, the above number is completely determined by the `matrix elements' for the half-braiding $\sigma_{X,v}$
$$
\left\langle
(g,n)
\bigg|
\tikzmath{
\draw[thick, orange] (-.4,-.6) node[below]{$\scriptstyle X$} to[out=90, in=270] (.4,.6) node[above]{$\scriptstyle X$};
\draw[knot] (.4,-.6) node[below]{$\scriptstyle v$} to[out=90, in=270] (-.4,.6) node[above]{$\scriptstyle v$};
}
\bigg|
(f,m)
\right\rangle
=
\tikzmath{
\draw[thick, orange] (-.3,-.45) node[left, yshift=.2cm]{$\scriptstyle X$} .. controls ++(90:.3cm) and ++(-90:.3cm) .. (.3,.45);
\draw[knot] (.2,-1.4) -- (.2,-.45) .. controls ++(90:.3cm) and ++(-90:.3cm) .. (-.2,.45) --node[left]{$\scriptstyle v$} (-.2,1.4);
\draw (.3,.95) --node[right]{$\scriptstyle y$} (.3,1.4);
\draw (-.3,-.95) --node[left]{$\scriptstyle z$} (-.3,-1.4);
\draw (0,2) --node[right]{$\scriptstyle x$} (0,2.4);
\draw (0,-2) --node[right]{$\scriptstyle x$} (0,-2.4);
\roundNbox{fill=white}{(0,1.7)}{.3}{.2}{.2}{$f$}
\roundNbox{fill=white}{(.3,.7)}{.25}{0}{0}{$m$}
\roundNbox{fill=white}{(-.3,-.7)}{.25}{0}{0}{$n^\dag$}
\roundNbox{fill=white}{(0,-1.7)}{.3}{.2}{.2}{$g^\dag$}
}
=
\frac{1}{d_x}\tr^\cX_x
\left(
\tikzmath{
\draw[thick, orange] (-.3,-.45) node[left, yshift=.2cm]{$\scriptstyle X$} .. controls ++(90:.3cm) and ++(-90:.3cm) .. (.3,.45);
\draw[knot] (.2,-1.4) -- (.2,-.45) .. controls ++(90:.3cm) and ++(-90:.3cm) .. (-.2,.45) --node[left]{$\scriptstyle v$} (-.2,1.4);
\draw (.3,.95) --node[right]{$\scriptstyle y$} (.3,1.4);
\draw (-.3,-.95) --node[left]{$\scriptstyle z$} (-.3,-1.4);
\draw (0,2) --node[right]{$\scriptstyle x$} (0,2.4);
\draw (0,-2) --node[right]{$\scriptstyle x$} (0,-2.4);
\roundNbox{fill=white}{(0,1.7)}{.3}{.2}{.2}{$f$}
\roundNbox{fill=white}{(.3,.7)}{.25}{0}{0}{$m$}
\roundNbox{fill=white}{(-.3,-.7)}{.25}{0}{0}{$n^\dag$}
\roundNbox{fill=white}{(0,-1.7)}{.3}{.2}{.2}{$g^\dag$}
}
\right)
$$
indexed by ONBs 
\begin{align*}
&\set{(f,m)}{f\in \cX(vy\to x),\,m\in \cX(X\to y) }
&&
\text{and}
\\
&\set{(g,n)}{g\in \cX(zv\to x),\,n\in \cX(X\to z) }
\intertext{for}
&\cX(vX\to x)
\cong 
\bigoplus_{y\in \Irr(\cX)} \cX(vy\to x)\otimes \cX(X\to y)
&&\text{and}
\\
&\cX(Xv\to x)
\cong 
\bigoplus_{z\in \Irr(\cX)} \cX(zv\to x)\otimes \cX(X\to z).
\end{align*}
Conversely, these matrix elements completely determine $\sigma_{X,v}$.
Thus, given a $*$-representation $\cH$ of $\Tube(\cX)$,
we use the minimal central idempotents $p_x\in \Tube(\cX)$ for $x\in \Irr(\cX)$ to get an object in $Z(\cX)$ by
$$
X_\cH:= \bigoplus_{y\in \Irr(\cX)} p_y\cH \rhd y
$$
using the canonical $\Hilb$-module structure of $\cX$.
Observe that we have a canonical identification $\cX(X_{\cH}\to y)=p_y\cH$ by the Yoneda Lemma.
The half-braiding $\sigma_{\cH_X}$ is defined by the matrix-unit formula above.

\subsection{Localized excitations, punctured skein modules, and string operators}
\label{sec:StringOperators}
Violations of the terms in the Hamiltonian \eqref{eq:XHamiltonian} 
restricted to a small region of the lattice represent localized excitations. 
In this section, 
we consider only the excitations that are localized on one edge and its two adjacent plaquettes.
Such excitations violate (at most) 3 terms in the Hamiltonian, and can be viewed as elements in a punctured skein module $\cS_\cX(\bbA,n,2)$\IN{$\cS_\cX(\bbA,n,m)$ the annular skein module with $n$ external/$m$ internal boundary points} (where $\bbA$ denotes the annulus) with $n$ external boundary points and 2 internal boundary points.
$$
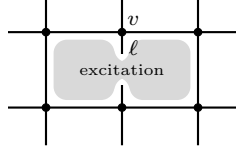
\tikzmath{
\draw[thick] (.5,.5) grid (3.5,2.5);
\fill[white] (1.8,1.3) rectangle (2.2,1.7);
\fill[gray!30, rounded corners=5pt] (1.5,1.1) -- (1.9,1.1) -- (1.9,1.4) -- (2.1,1.4) -- (2.1,1.1) -- (2.9,1.1) -- (2.9,1.9) -- (2.1, 1.9) -- (2.1,1.6) -- (1.9,1.6) -- (1.9,1.9) -- (1.1,1.9) -- (1.1,1.1) -- (1.5,1.1);
\foreach \x in {1,2,3}{
\foreach \y in {1,2}{
\filldraw (\x,\y) circle (.05cm);
}}
\foreach \x in {0.5,1.5,2.5}{
\foreach \y in {0.5,1.5}{
}}
\node at (2,1.5) {\tiny excitation};
\node at (2.15,2.15) {$\scriptstyle v$};
\node at (2.15,1.8) {$\scriptstyle \ell$};
}
$$
This space carries an action of the $\rm C^*$-algebra 
\begin{equation}
\label{eq:BigTube}
\fA
:=
\bigoplus_{w,x,y,x\in \Irr(\cX)}
\cX(\overline{r}w\to xs)\otimes \cX(ry\to z\overline{s})
=
\left\{
\tikzmath{
\draw[very thick] (0,0) circle (.3cm);
\draw (0,0) circle (1cm);
\node at (1.3,0) {$\scriptstyle s$};
\node at (-1.3,0) {$\scriptstyle r$};
\draw (0,-1.3) --node[right]{$\scriptstyle w$} (0,-1.7);
\draw (0,-.3) --node[right]{$\scriptstyle x$} (0,-.7);
\draw (0,.3) --node[right]{$\scriptstyle y$} (0,.7);
\draw (0,1.3) --node[right]{$\scriptstyle z$} (0,1.7);
\draw[very thick] (0,0) circle (1.7cm);
\roundNbox{fill=white}{(90:1cm)}{.3}{0}{0}{$\psi$}
\roundNbox{fill=white}{(-90:1cm)}{.3}{0}{0}{$\phi$}
}
\right\}
\end{equation}
whose $*$-algebra structure is given by
\begin{align*}
\tikzmath{
\draw[very thick] (0,0) circle (.3cm);
\draw (0,0) circle (1cm);
\node at (1.3,0) {$\scriptstyle s'$};
\node at (-1.3,0) {$\scriptstyle r'$};
\draw (0,-1.3) --node[right]{$\scriptstyle w'$} (0,-1.7);
\draw (0,-.3) --node[right]{$\scriptstyle x'$} (0,-.7);
\draw (0,.3) --node[right]{$\scriptstyle y'$} (0,.7);
\draw (0,1.3) --node[right]{$\scriptstyle z'$} (0,1.7);
\draw[very thick] (0,0) circle (1.7cm);
\roundNbox{fill=white}{(90:1cm)}{.3}{0}{0}{$\psi'$}
\roundNbox{fill=white}{(-90:1cm)}{.3}{0}{0}{$\phi'$}
}
\cdot\tikzmath{
\draw[very thick] (0,0) circle (.3cm);
\draw (0,0) circle (1cm);
\node at (1.3,0) {$\scriptstyle s$};
\node at (-1.3,0) {$\scriptstyle r$};
\draw (0,-1.3) --node[right]{$\scriptstyle w$} (0,-1.7);
\draw (0,-.3) --node[right]{$\scriptstyle x$} (0,-.7);
\draw (0,.3) --node[right]{$\scriptstyle y$} (0,.7);
\draw (0,1.3) --node[right]{$\scriptstyle z$} (0,1.7);
\draw[very thick] (0,0) circle (1.7cm);
\roundNbox{fill=white}{(90:1cm)}{.3}{0}{0}{$\psi$}
\roundNbox{fill=white}{(-90:1cm)}{.3}{0}{0}{$\phi$}
}
&:=
\sum_{r'',s''\in \Irr(\cX)}
\delta_{x'=w}
\delta_{y'=z}\,
\tikzmath{
\draw[very thick] (0,0) circle (.3cm);
\draw (-90:1cm) to[out=0,in=-135] (-20:1.5cm);
\draw (90:1cm) to[out=0,in=135] (20:1.5cm);
\draw (-20:1.5cm) arc (-20:20:1.5cm);
\draw (-90:2cm) to[out=0,in=-45] (-20:1.5cm);
\draw (90:2cm) to[out=0,in=45] (20:1.5cm);
\draw (-90:1cm) to[out=180,in=-45] (200:1.5cm);
\draw (90:1cm) to[out=180,in=45] (160:1.5cm);
\draw (160:1.5cm) arc (160:200:1.5cm);
\draw (-90:2cm) to[out=180,in=-135] (200:1.5cm);
\draw (90:2cm) to[out=180,in=135] (160:1.5cm);
\node at (50:1cm) {$\scriptstyle s$};
\node at (-50:1cm) {$\scriptstyle s$};
\node at (60:1.7cm) {$\scriptstyle s'$};
\node at (-60:1.7cm) {$\scriptstyle s'$};
\node at (1.3,0) {$\scriptstyle s''$};
\node at (130:1cm) {$\scriptstyle r$};
\node at (-130:1cm) {$\scriptstyle r$};
\node at (120:1.7cm) {$\scriptstyle r'$};
\node at (-120:1.7cm) {$\scriptstyle r'$};
\node at (-1.3,0) {$\scriptstyle r''$};
\draw (0,-2.3) --node[right]{$\scriptstyle w'$} (0,-2.7);
\draw (0,-1.3) --node[right]{$\scriptstyle w$} (0,-1.7);
\draw (0,-.3) --node[right]{$\scriptstyle x$} (0,-.7);
\draw (0,.3) --node[right]{$\scriptstyle y$} (0,.7);
\draw (0,1.3) --node[right]{$\scriptstyle z$} (0,1.7);
\draw (0,2.3) --node[right]{$\scriptstyle z'$} (0,2.7);
\draw[very thick] (0,0) ellipse (2cm and 2.7cm);
\filldraw[red] (20:1.5cm) circle (.05cm);
\filldraw[red] (-20:1.5cm) circle (.05cm);
\filldraw[blue] (160:1.5cm) circle (.05cm);
\filldraw[blue] (200:1.5cm) circle (.05cm);
\roundNbox{fill=white}{(90:1cm)}{.3}{0}{0}{$\psi$}
\roundNbox{fill=white}{(90:2cm)}{.3}{0}{0}{$\psi'$}
\roundNbox{fill=white}{(-90:1cm)}{.3}{0}{0}{$\phi$}
\roundNbox{fill=white}{(-90:2cm)}{.3}{0}{0}{$\phi'$}
}
\\
\tikzmath{
\draw[very thick] (0,0) circle (.3cm);
\draw (0,0) circle (1cm);
\node at (1.3,0) {$\scriptstyle s$};
\node at (-1.3,0) {$\scriptstyle r$};
\draw (0,-1.3) --node[right]{$\scriptstyle w$} (0,-1.7);
\draw (0,-.3) --node[right]{$\scriptstyle x$} (0,-.7);
\draw (0,.3) --node[right]{$\scriptstyle y$} (0,.7);
\draw (0,1.3) --node[right]{$\scriptstyle z$} (0,1.7);
\draw[very thick] (0,0) circle (1.7cm);
\roundNbox{fill=white}{(90:1cm)}{.3}{0}{0}{$\psi$}
\roundNbox{fill=white}{(-90:1cm)}{.3}{0}{0}{$\phi$}
}^*
&:=
\tikzmath{
\draw[very thick] (0,0) circle (.3cm);
\draw (0,0) circle (1cm);
\node at (1.3,0) {$\scriptstyle \overline{s}$};
\node at (-1.3,0) {$\scriptstyle \overline{r}$};
\draw (0,-1.3) --node[right]{$\scriptstyle x$} (0,-1.7);
\draw (0,-.3) --node[right]{$\scriptstyle w$} (0,-.7);
\draw (0,.3) --node[right]{$\scriptstyle z$} (0,.7);
\draw (0,1.3) --node[right]{$\scriptstyle y$} (0,1.7);
\draw[very thick] (0,0) circle (1.7cm);
\roundNbox{fill=white}{(90:1cm)}{.3}{0}{0}{$\psi^\dag$}
\roundNbox{fill=white}{(-90:1cm)}{.3}{0}{0}{$\phi^\dag$}
}\,.
\end{align*}
We will not give the explicit action here, as it is a straightforward generalization of the $\Tube(\cX)$-action on $\cS_\cX(\bbA,n,1)$ that we reduce to below.

Observe that Ocneanu's tube algebra $\Tube(\cX)$ sits inside this $\rm C^*$-algebra as the corner with $w=x=1_\cX$; i.e., $\Tube(\cX) = p\fA p$ for the orthogonal projection
$$
p=
\sum_{y\in \Irr(\cX)}
\tikzmath{
\draw[very thick] (0,0) circle (.3cm);
\draw (0,.3) --node[right, xshift=-.1cm]{$\scriptstyle y$} (0,.7);
\draw[very thick] (0,0) circle (.7cm);
}
\in \fA.
$$
Moreover, this corner is \emph{full}, i.e., $\fA = \fA\Tube(\cX) \fA$.
This immediately implies that $\fA$ and $\Tube(\cX)$ are Morita equivalent via the bimodule $\fA\Tube(\cX)$, and thus they have the same representation theory.
Indeed, this can be shown by applying an isotopy which contracts the $s$-string:
$$
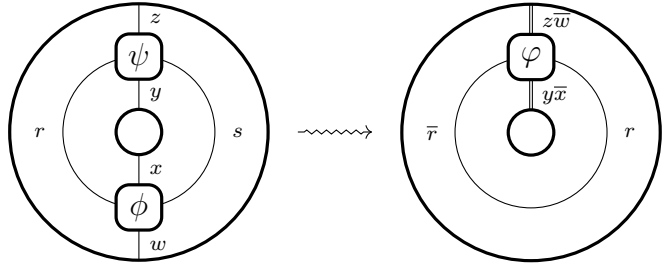
\tikzmath{
\draw[very thick] (0,0) circle (.3cm);
\draw (0,0) circle (1cm);
\node at (1.3,0) {$\scriptstyle s$};
\node at (-1.3,0) {$\scriptstyle r$};
\draw (0,-1.3) --node[right]{$\scriptstyle w$} (0,-1.7);
\draw (0,-.3) --node[right]{$\scriptstyle x$} (0,-.7);
\draw (0,.3) --node[right]{$\scriptstyle y$} (0,.7);
\draw (0,1.3) --node[right]{$\scriptstyle z$} (0,1.7);
\draw[very thick] (0,0) circle (1.7cm);
\roundNbox{fill=white}{(90:1cm)}{.3}{0}{0}{$\psi$}
\roundNbox{fill=white}{(-90:1cm)}{.3}{0}{0}{$\phi$}
}
\begin{tikzcd}
\mbox{}
\arrow[r,squiggly]
&
\mbox{}
\end{tikzcd}
\tikzmath{
\draw (0,0) circle (1cm);
\node at (1.3,0) {$\scriptstyle r$};
\node at (-1.3,0) {$\scriptstyle \overline{r}$};
\draw[double] (0,.3) --node[right]{$\scriptstyle y\overline{x}$} (0,.7);
\draw[double] (0,1.3) --node[right]{$\scriptstyle z\overline{w}$} (0,1.7);
\draw[very thick] (0,0) circle (.3cm);
\draw[very thick] (0,0) circle (1.7cm);
\roundNbox{fill=white}{(90:1cm)}{.3}{0}{0}{$\varphi$}
}\,.
$$
Now applying the fusion relation to both $\id_{z\overline{w}}$ and $\id_{y\overline{x}}$ exhibits $\Tube(\cX)$ as a full corner.

By the above discussion, without loss of generality, we may focus our attention on the subspace of $\cS_\cX(\bbA,n,1)\subset \cS_\cX(\bbA,n,2)$ where the top section of the broken edge is labelled by $\textbf{1}_X$.
That is, we project to the image of the operator $\pi^1_{\ell,v}$ which enforces that the edge label on $\ell$ adjacent to the vertex $v$ is $1_\cX$:
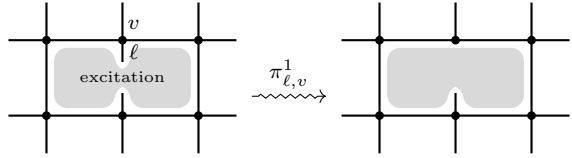
$$
\begin{tikzcd}
\tikzmath{
\draw[thick] (.5,.5) grid (3.5,2.5);
\fill[white] (1.8,1.3) rectangle (2.2,1.7);
\fill[gray!30, rounded corners=5pt] (1.5,1.1) -- (1.9,1.1) -- (1.9,1.4) -- (2.1,1.4) -- (2.1,1.1) -- (2.9,1.1) -- (2.9,1.9) -- (2.1, 1.9) -- (2.1,1.6) -- (1.9,1.6) -- (1.9,1.9) -- (1.1,1.9) -- (1.1,1.1) -- (1.5,1.1);
\foreach \x in {1,2,3}{
\foreach \y in {1,2}{
\filldraw (\x,\y) circle (.05cm);
}}
\foreach \x in {0.5,1.5,2.5}{
\foreach \y in {0.5,1.5}{
}}
\node at (2,1.45) {\tiny excitation};
\node at (2.15,2.15) {$\scriptstyle v$};
\node at (2.15,1.75) {$\scriptstyle \ell$};
}
\arrow[r, squiggly, "\pi^1_{\ell,v}"]
&
\tikzmath{
\draw[thick] (.5,.5) grid (3.5,2.5);
\fill[white] (1.8,1.3) rectangle (2.2,1.95);
\fill[gray!30, rounded corners=5pt] (1.5,1.1) -- (1.9,1.1) -- (1.9,1.4) -- (2.1,1.4) -- (2.1,1.1) -- (2.9,1.1) -- (2.9,1.9) -- (1.1,1.9) -- (1.1,1.1) -- (1.5,1.1);
\foreach \x in {1,2,3}{
\foreach \y in {1,2}{
\filldraw (\x,\y) circle (.05cm);
}}
\foreach \x in {0.5,1.5,2.5}{
\foreach \y in {0.5,1.5}{
}}
}
\end{tikzcd}
$$

We now describe the action of $\Tube(\cX)$ on $\cS_\cX(\bbA, n,1)$ and construct a unitary isomorphism (as a tube algebra representation) to a direct sum of regular representations of $\Tube(\cX)$.
Suppose we have a state $|\Omega\rangle$ which is in the ground state of every $-A_k$ for $k\neq \ell$ and every $-B_p$ except the two plaquettes adjacent to $\ell$.
We further assume that $|\Omega\rangle = \pi^1_{\ell,v}|\Omega\rangle$.
The action of $f\in \cX(sy\to zs)\subset \Tube(\cX)$ on $|\Omega\rangle$ is given similar to a plaquette operator:
\begin{equation}
\label{eq:TubeAction}
f\rhd
\tikzmath{
\draw[thick] (.5,.5) grid (3.5,2.5);
\fill[white] (1.8,1.3) rectangle (2.2,1.95);
\fill[gray!30, rounded corners=5pt] (1.5,1.1) -- (1.9,1.1) -- (1.9,1.4) -- (2.1,1.4) -- (2.1,1.1) -- (2.9,1.1) -- (2.9,1.9) -- (1.1,1.9) -- (1.1,1.1) -- (1.5,1.1);
\foreach \x in {1,2,3}{
\foreach \y in {1,2}{
\filldraw (\x,\y) circle (.05cm);
}}
\foreach \x in {0.5,1.5,2.5}{
\foreach \y in {0.5,1.5}{
}}
\node at (2.15,1.2) {$\scriptstyle x$};
}
:=
\delta_{x=y}
\left(\frac{d_z}{d_x}\right)^{1/4}
\cdot
\tikzmath{
\draw[thick] (.5,.5) grid (3.5,2.5);
\fill[white] (1.8,1.6) rectangle (2.2,1.95);
\filldraw[red, thick, fill=gray!30] (1.5,1.1) -- (1.8,1.1) to[out=0,in=180] (2,1.3) to[out=0,in=180] (2.2,1.1) -- (2.8,1.1) arc (-90:0:.1cm) -- (2.9,1.8) arc (0:90:.1cm) -- (1.2,1.9) arc (90:180:.1cm) -- (1.1,1.2) arc (180:270:.1cm) -- (1.5,1.1);
\draw[thick] (2,1.3) -- (2,1.5);
\filldraw[red] (2,1.3) node[right, yshift=.1cm]{$\scriptstyle f$} circle (.05cm);
\foreach \x in {1,2,3}{
\foreach \y in {1,2}{
\filldraw (\x,\y) circle (.05cm);
}}
\foreach \x in {0.5,1.5,2.5}{
\foreach \y in {0.5,1.5}{
}}
\node[red] at (1.5,1.2) {$\scriptstyle s$};
\node at (2,1.65) {$\scriptstyle z$};
}
=
\delta_{x=y}
\cdot
\tikzmath{
\draw[thick] (.5,.5) grid (3.5,2.5);
\fill[white] (1.8,1.6) rectangle (2.2,1.95);
\fill[fill=gray!30] (1.5,1.1) to[out=0,in=180] (2,1.4) to[out=0,in=180] (2.5,1.1) to[out=0,in=-90] (2.9,1.5) to[out=90,in=0] (2.5,1.9) to[out=180,in=0] (2,1.7) to[out=180,in=0] (1.5,1.9) to[out=180,in=90] (1.1,1.5) to[out=-90,in=180] (1.5,1.1);
\draw[thick, red] (1.3,1) -- (1,1.3);
\draw[thick, red] (1.7,1) to[out=45,in=180] (2,1.3) to[out=0, in=135] (2.3,1);
\draw[thick, red] (2.7,1) -- (3,1.3);
\draw[thick, red] (2.7,2) -- (3,1.7);
\draw[thick, red] (1.7,2) to[out=-45,in=180] (2,1.8) to[out=0, in=-135] (2.3,2);
\draw[thick, red] (1.3,2) -- (1,1.7);
\filldraw[cyan] (1,1.3) circle (.05 cm);
\filldraw[cyan] (1,1.7) circle (.05 cm);
\filldraw[yellow] (1.3,1) circle (.05 cm);
\filldraw[yellow] (1.7,1) circle (.05 cm);
\filldraw[blue] (2.3,1) circle (.05 cm);
\filldraw[blue] (2.7,1) circle (.05 cm);
\filldraw[orange] (2.7,2) circle (.05 cm);
\filldraw[orange] (2.3,2) circle (.05 cm);
\filldraw[green] (1.3,2) circle (.05 cm);
\filldraw[green] (1.7,2) circle (.05 cm);
\filldraw[purple] (3,1.3) circle (.05 cm);
\filldraw[purple] (3,1.7) circle (.05 cm);
\draw[thick] (2,1.3) -- (2,1.5);
\filldraw[red] (2,1.3) node[right, yshift=.1cm]{$\scriptstyle f$} circle (.05cm);
\foreach \x in {1,2,3}{
\foreach \y in {1,2}{
\filldraw (\x,\y) circle (.05cm);
}}
\foreach \x in {0.5,1.5,2.5}{
\foreach \y in {0.5,1.5}{
}}
}
\end{equation}
where suppress the sum over simples and the resulting scalars in the final diagram.
We can thus classify all localized excitations as irreducible representations of $\Tube(\cX)$, which correspond to the anyon types in $Z(\cX)$.
The scalar $(d_z/d_x)^{1/4}$ is needed (see \eqref{eq:GluingOperatorWithScalars} above) to make this representation of $\Tube(\cX)$ a $*$-representation.

To see that all such anyons arise as excitations, we now define our \emph{string operators} as operators from the space of ground states to a twice punctured skein module living inside $\cH_{\tot}$.
Given $(X,\sigma_X)\in \Irr(Z(\cX))$, $x,y\in \Irr(\cX)$, $\psi\in \cX(X\to x)$, $\phi\in \cX(X\to y)$, two distinct edges $k,\ell$ of our lattice (which we assume are separated by several plaquettes), and a path $\gamma:\ell\to k$ on our lattice,
we get the string operator $S^{(X,\sigma_X)}_{\gamma}(\phi,\psi)$ 
on a ground state $|\Omega\rangle$
by first applying $\pi_\ell^1\pi_k^1$ and then `inserting' the $(X,\sigma_X)$ string, where $\phi$ and $\psi$ are used at the endpoints of the string to obtain an element in $\cH_{\tot}$.
Here, for an edge $\ell$, $\pi_\ell^1$ is the projection on $\cH_{\tot}$ which enforces both edge labels on $\ell$ to be $1_\cX$.
We include a cartoon below of one side of the string operator:
$$
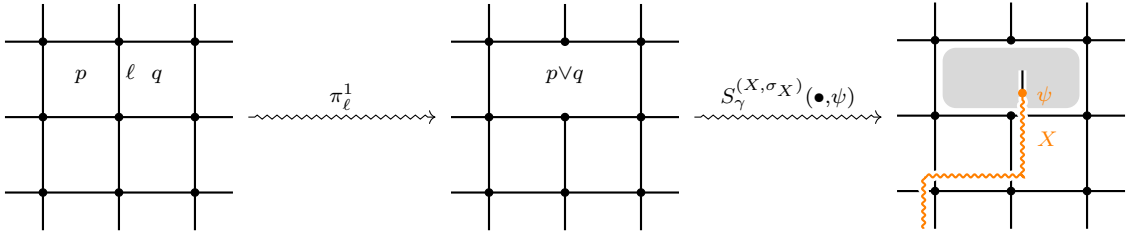
\begin{tikzcd}[column sep = 6em]
\tikzmath{
\draw[thick] (.5,-.5) grid (3.5,2.5);
\foreach \x in {1,2,3}{
\foreach \y in {0,1,2}{
\filldraw (\x,\y) circle (.05cm);
}}
\foreach \x in {0.5,1.5,2.5}{
\foreach \y in {0.5,1.5}{
}}
\node at (2.15,1.5) {$\scriptstyle \ell$};
\node at (1.5,1.5) {$\scriptstyle p$};
\node at (2.5,1.5) {$\scriptstyle q$};
}
\arrow[r, squiggly, "\pi_\ell^1"]
&
\tikzmath{
\draw[thick] (.5,-.5) grid (3.5,2.5);
\fill[white] (1.8,1.05) rectangle (2.2,1.95);
\foreach \x in {1,2,3}{
\foreach \y in {0,1,2}{
\filldraw (\x,\y) circle (.05cm);
}}
\foreach \x in {0.5,1.5,2.5}{
\foreach \y in {0.5,1.5}{
}}
\node at (2,1.5) {$\scriptstyle p\vee q$};
}
\arrow[r, squiggly, "{S^{(X,\sigma_X)}_{\gamma}(\bullet,\psi)}"]
&
\tikzmath{
\draw[thick] (.5,-.5) grid (3.5,2.5);
\fill[white] (1.8,1.05) rectangle (2.2,1.95);
\fill[gray!30, rounded corners=5pt] (1.1,1.1) rectangle (2.9,1.9);
\foreach \x in {1,2,3}{
\foreach \y in {0,1,2}{
\filldraw (\x,\y) circle (.05cm);
}}
\foreach \x in {0.5,1.5,2.5}{
\foreach \y in {0.5,1.5}{
}}
\draw[thick, knot, orange, decorate, decoration={snake, segment length=1mm, amplitude=.2mm}] (.85,-.5) -- (.85,.15) -- (2.15,.15) --node[right]{$\scriptstyle X$} (2.15,1.25);
\draw[knot, thick] (2.15,1.3) -- (2.15,1.6);
\filldraw[orange] (2.15,1.3) node[right]{$\scriptstyle \psi$} circle (.05cm);
}
\end{tikzcd}
$$

\begin{rem}
Recall that our convention for half-braidings in \eqref{eq:HalfBraidConvention} is that we always write the object in $Z(\cX)$ as passing \emph{below} objects in $\cX$.
Here, we have instead drawn our excitations as living \emph{above} the $\cX$-lattice.
That is, we are using the convention
$$
\tikzmath{
\draw (-.4,-.6) node[below]{$\scriptstyle y$} to[out=90, in=270] (.4,.6) node[above]{$\scriptstyle y$};
\draw[thick, knot,orange,decorate, decoration={snake, segment length=1mm, amplitude=.2mm}] (.4,-.6) node[below]{$\scriptstyle X$} to[out=90, in=270] (-.4,.6) node[above]{$\scriptstyle X$};
}
:=
\tikzmath{
\draw[thick, orange] (.4,-.6) node[below]{$\scriptstyle X$} to[out=90, in=270] (-.4,.6) node[above]{$\scriptstyle X$};
\draw[knot] (-.4,-.6) node[below]{$\scriptstyle y$} to[out=90, in=270] (.4,.6) node[above]{$\scriptstyle y$};
}
$$
and we are thus describing the excitations as $Z(\cX)^{\rev}$, the \emph{reverse} of $Z(\cX)$, i.e.,
$$
\beta^{Z^\cA(\cX)^{\rev}}_{X,Y}
=
\tikzmath{
\draw[thick, orange, decorate, decoration={snake, segment length=1mm, amplitude=.2mm}] (.4,-.6) node[below]{$\scriptstyle X$} to[out=90, in=270] (-.4,.6) node[above]{$\scriptstyle X$};
\draw[thick, cyan, knot, decorate, decoration={snake, segment length=1mm, amplitude=.2mm}] (-.4,-.6) node[below]{$\scriptstyle Y$} to[out=90, in=270] (.4,.6) node[above]{$\scriptstyle Y$};
}
:=
\left(
\tikzmath{
\draw[thick, cyan] (.4,-.6) node[below]{$\scriptstyle Y$} to[out=90, in=270] (-.4,.6) node[above]{$\scriptstyle Y$};
\draw[thick, orange, knot] (-.4,-.6) node[below]{$\scriptstyle X$} to[out=90, in=270] (.4,.6) node[above]{$\scriptstyle X$};
}
\right)^{-1}
=
(\beta^{Z^\cA(\cX)}_{Y,X})^{-1}
=
\tikzmath{
\draw[thick, cyan] (-.4,-.6) node[below]{$\scriptstyle Y$} to[out=90, in=270] (.4,.6) node[above]{$\scriptstyle Y$};
\draw[thick, orange, knot] (.4,-.6) node[below]{$\scriptstyle X$} to[out=90, in=270] (-.4,.6) node[above]{$\scriptstyle X$};
}\,.
$$
This convention will be useful in the enriched setting in \S\ref{sec:Enriched} below.
There, we use an $\cA$-enriched UFC $\cX$ to construct a (2+1)D boundary for the (3+1)D $\cA$-Walker-Wang model.
We show in Theorem \ref{thm:DomeReps=EnrichedCenter} below that the boundary excitations are given by the enriched center $Z^\cA(\cX):= \cA'\subset Z(\cX)$.
We draw these excitations on the opposite side of the enriching lattice for clarity, viewing $\cX$ as an $\cA- Z^\cA(\cX)^{\rev}$ bimodule tensor category as $Z(\cX)\cong \cA\boxtimes Z^\cA(\cX)$.
\end{rem}

As with plaquette operators, we must resolve this final picture into $\cH_{\tot}$;
to do so, we decompose $F(X,\sigma_X)=X\in \cX$ into simples, where $F: Z(\cX)\to \cX$ is the forgetful functor, and use \eqref{eq:Fusion}.
We include a cartoon below, where we suppress sums over simples and ONBs and scalars to ease the notation.
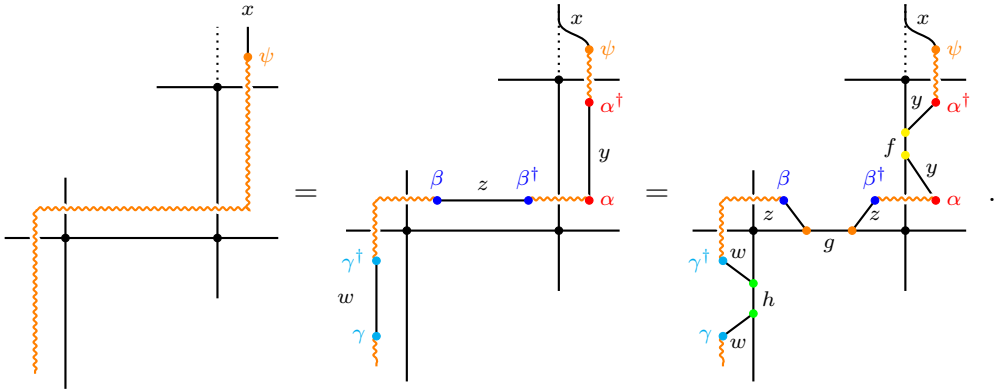
\begin{align*}
\tikzmath{
\draw[thick] (-.8,0) -- (.8,0);
\draw[thick] (0,-2) -- (0,0);
\draw[thick] (-2,-2) -- (-2,-1.2);
\draw[thick] (-2,-2) -- (-2.8,-2);
\draw[thick] (0,-2) -- (.8,-2);
\draw[thick] (0,-2) -- (0,-2.8);
\draw[thick, knot, orange, decorate, decoration={snake, segment length=1mm, amplitude=.2mm}] (.4,.4) -- (.4,-1.6) -- (-2.4, -1.6) -- (-2.4,-3.8);
\draw[thick] (.4,.4) -- (.4,.8) node[above]{$\scriptstyle x$};
\draw[thick,dotted] (0,0) -- (0,.8);
\filldraw (0,0) circle (.05cm);
\filldraw (-2,-2) circle (.05cm);
\filldraw (0,-2) circle (.05cm);
\draw[thick] (0,-2) -- (-2,-2) -- (-2,-4);
\filldraw[orange] (.4,.4) node[right]{$\scriptstyle \psi$} circle (.05cm);
}
&=
\tikzmath{
\draw[thick] (-.8,0) -- (.8,0);
\draw[thick] (0,-2) -- (0,0);
\draw[thick] (-2,-2) -- (-2,-1.2);
\draw[thick] (-2,-2) -- (-2.8,-2);
\draw[thick] (0,-2) -- (.8,-2);
\draw[thick] (0,-2) -- (0,-2.8);
\draw[thick] (.4,-.3) -- (.4,-1.6);
\draw[thick] (-.4,-1.6) -- (-1.6,-1.6);
\draw[thick] (-2.4,-2.4) -- (-2.4,-3.4);
\draw[thick, knot, orange, decorate, decoration={snake, segment length=1mm, amplitude=.2mm}] (.4,.4) -- (.4,-.3) (.4,-1.6) -- (-.4,-1.6) (-1.6,-1.6)-- (-2.4, -1.6) -- (-2.4, -2.4) (-2.4,-3.4) -- (-2.4,-3.8);
\draw[thick] (.4,.4) to[out=90,in=-90] (0,.8) node[right]{$\scriptstyle x$} -- (0,1);
\draw[thick,dotted] (0,0) -- (0,1);
\filldraw (0,0) circle (.05cm);
\filldraw (-2,-2) circle (.05cm);
\filldraw (0,-2) circle (.05cm);
\draw[thick] (0,-2) -- (-2,-2) -- (-2,-4);
\filldraw[orange] (.4,.4) node[right]{$\scriptstyle \psi$} circle (.05cm);
\filldraw[red] (.4,-.3) node[right]{$\scriptstyle \alpha^\dagger$} circle (.05 cm);
\filldraw[red] (.4,-1.6) node[right]{$\scriptstyle \alpha$} circle (.05 cm);
\filldraw[blue] (-.4,-1.6) node[above]{$\scriptstyle \beta^\dagger$} circle (.05 cm);
\filldraw[blue] (-1.6,-1.6) node[above]{$\scriptstyle \beta$} circle (.05 cm);
\filldraw[cyan] (-2.4,-2.4) node[left]{$\scriptstyle \gamma^\dagger$} circle (.05 cm);
\filldraw[cyan] (-2.4,-3.4) node[left]{$\scriptstyle \gamma$} circle (.05 cm);
\draw (.6,-1) node{$\scriptstyle y$};
\draw (-1,-1.4) node{$\scriptstyle z$};
\draw (-2.8,-2.9) node{$\scriptstyle w$};
}
=
\tikzmath{
\draw[thick] (-.8,0) -- (.8,0);
\draw[thick] (0,-2) -- (0,0);
\draw[thick] (-2,-2) -- (-2,-1.2);
\draw[thick] (-2,-2) -- (-2.8,-2);
\draw[thick] (0,-2) -- (.8,-2);
\draw[thick] (0,-2) -- (0,-2.8);
\draw[thick] (.4,-.3) -- (0,-.7);
\draw[thick] (.4,-1.6) -- (0,-1);
\draw[thick] (-.4,-1.6) -- (-.7,-2);
\draw[thick] (-1.6,-1.6) -- (-1.3,-2);
\draw[thick] (-2.4,-2.4) -- (-2,-2.7);
\draw[thick] (-2.4,-3.4) -- (-2,-3.1);
\draw[thick, knot, orange, decorate, decoration={snake, segment length=1mm, amplitude=.2mm}] (.4,.4) -- (.4,-.3) (.4,-1.6) -- (-.4,-1.6) (-1.6,-1.6)-- (-2.4, -1.6) -- (-2.4, -2.4) (-2.4,-3.4) -- (-2.4,-3.8);
\draw[thick] (.4,.4) to[out=90,in=-90] (0,.8) node[right]{$\scriptstyle x$} -- (0,1);
\draw[thick,dotted] (0,0) -- (0,1);
\filldraw (0,0) circle (.05cm);
\filldraw (-2,-2) circle (.05cm);
\filldraw (0,-2) circle (.05cm);
\draw[thick] (0,-2) -- (-2,-2) -- (-2,-4);
\filldraw[orange] (.4,.4) node[right]{$\scriptstyle \psi$} circle (.05cm);
\filldraw[red] (.4,-.3) node[right]{$\scriptstyle \alpha^\dagger$} circle (.05 cm);
\filldraw[red] (.4,-1.6) node[right]{$\scriptstyle \alpha$} circle (.05 cm);
\filldraw[blue] (-.4,-1.6) node[above]{$\scriptstyle \beta^\dagger$} circle (.05 cm);
\filldraw[blue] (-1.6,-1.6) node[above]{$\scriptstyle \beta$} circle (.05 cm);
\filldraw[cyan] (-2.4,-2.4) node[left]{$\scriptstyle \gamma^\dagger$} circle (.05 cm);
\filldraw[cyan] (-2.4,-3.4) node[left]{$\scriptstyle \gamma$} circle (.05 cm);
\draw (.15,-.3) node{$\scriptstyle y$};
\draw (.35,-1.2) node{$\scriptstyle y$};
\draw (-.2,-.9) node{$\scriptstyle f$};
\draw (-.4,-1.8) node{$\scriptstyle z$};
\draw (-1.8,-1.8) node{$\scriptstyle z$};
\draw (-1,-2.2) node{$\scriptstyle g$};
\draw (-1.8,-2.9) node{$\scriptstyle h$};
\draw (-2.2,-2.3) node{$\scriptstyle w$};
\draw (-2.2,-3.5) node{$\scriptstyle w$};
\filldraw[yellow] (0,-.7) circle (.05 cm);
\filldraw[yellow] (0,-1) circle (.05 cm);
\filldraw[orange] (-1.3,-2) circle (.05 cm);
\filldraw[orange] (-.7,-2) circle (.05 cm);
\filldraw[green] (-2,-2.7) circle (.05 cm);
\filldraw[green] (-2,-3.1) circle (.05 cm);
}
\,.
\end{align*}
This decomposition occurs along the entire string, up to the other endpoint.
The trivalent vertices represent an ONB for $\cS_\cX(\bbD,3)$, while the bivalent vertices represent ONBs for the orthogonal direct sum $\bigoplus_{x\in\Irr(\cX)}\cX(x\to X)$ with respect to the \emph{isometry inner product} determined by the equation
\begin{equation} 
\label{eq:IsometryInnerProdDefn}
\left\langle
\tikzmath{
\draw[thick, orange, decorate, decoration={snake, segment length=1mm, amplitude=.2mm}] (0,.3) --node[right]{$\scriptstyle X$} (0,.7);
\draw (0,-.7) --node[right]{$\scriptstyle x$} (0,-.3);
\roundNbox{}{(0,0)}{.3}{0}{0}{$\alpha$}
}
\middle|
\tikzmath{
\draw[thick, orange, decorate, decoration={snake, segment length=1mm, amplitude=.2mm}] (0,.3) --node[right]{$\scriptstyle X$} (0,.7);
\draw (0,-.7) --node[right]{$\scriptstyle x$} (0,-.3);
\roundNbox{}{(0,0)}{.3}{0}{0}{$\beta$}
}
\right\rangle
\id_x
=
\alpha^\dag \circ \beta.
\end{equation}

We now analyze the tube algebra representation arising from the image of a string operator $S^{(X,\sigma_X)}_{k,\ell}(\phi,\psi)$.
By \cite[Cor.~2.10]{MR4642306}, on the space of ground states $\{|\Omega\rangle\}$ for our Hamiltonian $H$ from \eqref{eq:XHamiltonian}, 
\begin{equation}
\label{eq:BpvqForFree}
\pi^1_\ell=D_\cX\pi^1_\ell B_pB_q\pi^1_\ell=B_{p\vee q}\pi^1_\ell
\end{equation}
where $B_{p\vee q}$ is a plaquette operator on the $p\vee q$ plaquette of the mutated lattice obtained by deleting $\ell$.
$$
\pi^1_\ell|\Omega\rangle
=
B_{p\vee q}\pi^1_\ell|\Omega\rangle
=
\tikzmath{
\draw[thick] (.5,.5) grid (3.5,2.5);
\fill[white] (1.8,1.05) rectangle (2.2,1.95);
\draw[draw=blue, thick, rounded corners=5pt] (1.1,1.1) rectangle (2.9,1.9);
\node[blue] at (1.5,1.25) {$\scriptstyle r$};
\foreach \x in {1,2,3}{
\foreach \y in {1,2}{
\filldraw (\x,\y) circle (.05cm);
}}
\foreach \x in {0.5,1.5,2.5}{
\foreach \y in {0.5,1.5}{
}}
}
=
\frac{1}{D_\cX}
\sum_{r\in \Irr(\cX)} d_r\cdot K \cdot
\tikzmath{
\draw[thick] (.5,.5) grid (3.5,2.5);
\fill[white] (1.8,1.05) rectangle (2.2,1.95);
\fill[fill=gray!30] (1.5,1.1) to[out=0,in=180] (2,1.3) to[out=0,in=180] (2.5,1.1) to[out=0,in=-90] (2.9,1.5) to[out=90,in=0] (2.5,1.9) to[out=180,in=0] (2,1.7) to[out=180,in=0] (1.5,1.9) to[out=180,in=90] (1.1,1.5) to[out=-90,in=180] (1.5,1.1);
\draw[thick, blue] (1.3,1) -- (1,1.3);
\draw[thick, blue] (1.7,1) to[out=45,in=180] (2,1.2) to[out=0, in=135] (2.3,1);
\draw[thick, blue] (2.7,1) -- (3,1.3);
\draw[thick, blue] (2.7,2) -- (3,1.7);
\draw[thick, blue] (1.7,2) to[out=-45,in=180] (2,1.8) to[out=0, in=-135] (2.3,2);
\draw[thick, blue] (1.3,2) -- (1,1.7);
\filldraw[cyan] (1,1.3) circle (.05 cm);
\filldraw[cyan] (1,1.7) circle (.05 cm);
\filldraw[yellow] (1.3,1) circle (.05 cm);
\filldraw[yellow] (1.7,1) circle (.05 cm);
\filldraw[blue] (2.3,1) circle (.05 cm);
\filldraw[blue] (2.7,1) circle (.05 cm);
\filldraw[orange] (2.7,2) circle (.05 cm);
\filldraw[orange] (2.3,2) circle (.05 cm);
\filldraw[green] (1.3,2) circle (.05 cm);
\filldraw[green] (1.7,2) circle (.05 cm);
\filldraw[purple] (3,1.3) circle (.05 cm);
\filldraw[purple] (3,1.7) circle (.05 cm);
\foreach \x in {1,2,3}{
\foreach \y in {1,2}{
\filldraw (\x,\y) circle (.05cm);
}}
\foreach \x in {0.5,1.5,2.5}{
\foreach \y in {0.5,1.5}{
}}
\node[blue] at (1.5,1.2) {$\scriptstyle r$};
}
$$
again suppressing sums over simples and the resulting scalars $K$ in the final diagram.
This means $\pi^1_\ell|\Omega\rangle$ is a ground state on the mutated lattice with a modified local Hamiltonian.
This means that on the space of ground states,
$$
S^{(X,\sigma_X)}_{\gamma}(\phi,\psi)
=
S^{(X,\sigma_X)}_{\gamma}(\phi,\psi)\pi^1_\ell
=
S^{(X,\sigma_X)}_{\gamma}(\phi,\psi)B_{p\vee q}\pi^1_\ell
,$$
so one can encircle the end of the string operator by a $B_{p\vee q}$:
\begin{equation}
\label{eq:BpvqProtection}
\tikzmath{
\draw[thick] (.5,.5) grid (3.5,2.5);
\fill[white] (1.8,1.05) rectangle (2.2,1.95);
\fill[gray!30, rounded corners=5pt] (1.1,1.1) rectangle (2.9,1.9);
\foreach \x in {1,2,3}{
\foreach \y in {1,2}{
\filldraw (\x,\y) circle (.05cm);
}}
\foreach \x in {0.5,1.5,2.5}{
\foreach \y in {0.5,1.5}{
}}
\draw[thick, knot, orange, decorate, decoration={snake, segment length=1mm, amplitude=.2mm}](2.15,.5) --node[right]{$\scriptstyle X$} (2.15,1) -- (2.15,1.25);
\draw[knot, thick] (2.15,1.3) -- (2.15,1.6);
\filldraw[orange] (2.15,1.3) node[right]{$\scriptstyle \psi$} circle (.05cm);
}
=
\frac{1}{D_\cX}
\sum_{r\in \Irr(\cX)} d_r
\tikzmath{
\draw[thick] (.5,.5) grid (3.5,2.5);
\fill[white] (1.8,1.05) rectangle (2.2,1.95);
\filldraw[draw=blue, thick, fill=gray!30, rounded corners=5pt] (1.1,1.1) rectangle (2.9,1.9);
\node[blue] at (1.5,1.25) {$\scriptstyle r$};
\foreach \x in {1,2,3}{
\foreach \y in {1,2}{
\filldraw (\x,\y) circle (.05cm);
}}
\foreach \x in {0.5,1.5,2.5}{
\foreach \y in {0.5,1.5}{
}}
\draw[thick, knot, orange, decorate, decoration={snake, segment length=1mm, amplitude=.2mm}](2.15,.5) --node[right]{$\scriptstyle X$} (2.15,1) -- (2.15,1.25);
\draw[knot, thick] (2.15,1.3) -- (2.15,1.6);
\filldraw[orange] (2.15,1.3) node[right]{$\scriptstyle \psi$} circle (.05cm);
}
\end{equation}
Now we can use the \eqref{eq:Fusion} relation on the tube algebra action from \eqref{eq:TubeAction} above to see that the action of $f\in \cX(sx'\to ys)\subset \Tube(\cX)$ is given by
\begin{equation}
\label{eq:IzumiTubeAction}
\frac{\delta_{x=x'}}{D_\cX}
\left(\frac{d_y}{d_x}\right)^{1/4}
\sum_{r\in \Irr(\cX)} d_r
\tikzmath[xscale=1.25, yscale=1.5]{
\draw[thick] (.8,.8) grid (3.2,2.2);
\fill[white] (1.8,1.03) rectangle (2.2,1.97);
\draw[draw=blue, thick, rounded corners=5pt] (1.1,1.1) rectangle (2.9,1.9);
\node[blue] at (1.8,1.2) {$\scriptstyle r$};
\foreach \x in {1,2,3}{
\foreach \y in {1,2}{
\filldraw (\x,\y) circle (.03cm);
}}
\foreach \x in {0.5,1.5,2.5}{
\foreach \y in {0.5,1.5}{
}}
\draw[thick, knot, orange, decorate, decoration={snake, segment length=1mm, amplitude=.2mm}](2.15,.8) --node[right]{$\scriptstyle X$} (2.15,1) -- (2.15,1.2);
\draw[draw=red, thick, fill=gray!30, rounded corners=5pt] (1.2,1.35) rectangle (2.8,1.8);
\node[red] at (1.6,1.25) {$\scriptstyle s$};
\draw[thick] (2.15,1.2) -- (2.15,1.5);
\filldraw[red] (2.15,1.35) node[right, yshift=.15cm]{$\scriptstyle f$} circle (.03cm);
\filldraw[orange] (2.15,1.2) node[right]{$\scriptstyle \psi$} circle (.03cm);
}
=
\frac{\delta_{x=x'}}{D_\cX}
\left(\frac{d_y}{d_x}\right)^{1/4}
\sum_{t\in \Irr(\cX)} d_t
\tikzmath[scale=1.5]{
\draw[thick] (.8,.8) grid (3.2,2.2);
\fill[white] (1.8,1.03) rectangle (2.2,1.97);
\filldraw[draw=violet, thick, fill=gray!30, rounded corners=5pt] (1.1,1.1) rectangle (2.9,1.9);
\node[violet] at (1.8,1.2) {$\scriptstyle t$};
\foreach \x in {1,2,3}{
\foreach \y in {1,2}{
\filldraw (\x,\y) circle (.03cm);
}}
\foreach \x in {0.5,1.5,2.5}{
\foreach \y in {0.5,1.5}{
}}
\filldraw[draw=red, thick, fill=white, rounded corners=5pt] (2.15,1.35) circle (.2cm);
\draw[thick, knot, orange, decorate, decoration={snake, segment length=1mm, amplitude=.2mm}](2.15,.8) --node[right]{$\scriptstyle X$} (2.15,1) -- (2.15,1.35);
\node[red] at (2,1.6) {$\scriptstyle s$};
\draw[thick] (2.15,1.35) -- (2.15,1.7);
\filldraw[red] (2.15,1.55) node[right, yshift=.15cm]{$\scriptstyle f$} circle (.03cm);
\filldraw[orange] (2.15,1.35) node[right, xshift=-.1cm]{$\scriptstyle \psi$} circle (.03cm);
}\,.
\end{equation}
This is an amplification of Izumi's action of $\Tube(\cX)$ corresponding to $(X,\sigma_X)$; the scalar in front appears only because we use a different inner product for the punctured skein module than Izumi uses.
We therefore see that all simple objects in $Z(\cX)$ arise as localized excitations, and every excitation corresponds to some simple object in $Z(\cX)$.

\begin{rem}
\label{rem:TubeAlgebraMoritaCorrect}

Local operators cannot change the anyon type of a localized excitation.
Consider a twice-punctured sphere where the punctures are sufficiently separated.
We claim the space of ground-states on this punctured sphere is exactly $\Tube(\cX)$, which carries commuting left and right actions of $\Tube(\cX)$, one action at each puncture.
Indeed, by a variant of Proposition \ref{prop:DescriptionOfIm(PA)}, the ground state space of our twice punctured sphere with boundary strands that can host excitations is given by 
$$
\bigoplus_{x,y\in \cX} \cX(x\to FI(y))
\cong
\bigoplus_{v,x,y\in \cX} \cX(x\to \overline{v}yv)
\cong
\bigoplus_{v,x,y\in \cX} \cX(vx\to yv)
=
\Tube(\cX).
$$
Moreover, our two commuting $\Tube(\cX)$-actions correspond to the left and right regular representation.

Now as a $\Tube(\cX)-\Tube(\cX)$ bimodule, the regular representation $\Tube(\cX)$ has one irreducible summand for each type of irreducible representation of $\Tube(\cX)$, {i.e.}~each isomorphism class in $\Irr(Z(\cX))$.
We thus see that the anyon types of the two excitations at the punctures must agree, up to an opposite which occurs due to a left/right action of $\Tube(\cX)$, which translates into an anyon type and its dual at the end of the strings.

Now any local operator is localized away from at least one of the two excitations.
While the local operator can destroy the $\Tube(\cX)$ representation on one side, it is localized away from the other excitation, and therefore commutes with the action of $\Tube(\cX)$ at the site of the other excitation.
Hence the local operator still preserves the anyon type at the other puncture.
So if our local operator takes one excitation to another at the same (or nearby) site, then the above argument says it cannot change the anyon type.

We conclude that on a twice-punctured sphere, an irreducible $\Tube(\cX)$ representation type is a complete invariant of a superselection sector.
It is invariant under applying local operators, and any finer invariant would necessarily split irreducible summands of the regular $\Tube(\cX)$-$\Tube(\cX)$ bimodule, which disrespects the fact that $\Tube(\cX)$ acts by local operators.
\end{rem}

\subsection{Hopping operators}\label{sec:Hoppingoperators}

Hopping operators \cite{PhysRevB.97.195154} translate localized excitations in our system.
The following is an equivalent but more streamlined definition of hopping operator that appears in \cite[2.2.2]{MR4642306}:
$$
H^{(X,\sigma_X)}_{\delta}:= \sum_{y\in \Irr(\cX)}\sum_{\psi_y\in \operatorname{ONB}(X\to y)}B_p B_q T^{(X,\sigma_X)}_{\delta}(\psi_y,\psi_y)\pi^1_k,
$$
where $T^{(X,\sigma_X)}_{\delta}(\psi,\psi)$ is a modified string operator which does not apply $\pi^1_\ell$ at the start, but rather just $\pi^1_{\ell,v}$.
Here is a cartoon of the $H^{(X,\sigma_X)}_{\delta}$ action on an excited state.
$$
H^{(X,\sigma_X)}_{\delta}
\rhd
\tikzmath{
\draw[thick] (.5,.5) grid (4.5,4.5);
\fill[white] (1.8,1.3) rectangle (2.2,1.95);
\fill[gray!30, rounded corners=5pt] (1.5,1.1) -- (1.9,1.1) -- (1.9,1.4) -- (2.1,1.4) -- (2.1,1.1) -- (2.9,1.1) -- (2.9,1.9) -- (1.1,1.9) -- (1.1,1.1) -- (1.5,1.1);
\foreach \x in {1,2,3,4}{
\foreach \y in {1,2,3,4}{
\filldraw (\x,\y) circle (.05cm);
}}
\foreach \x in {0.5,1.5,2.5}{
\foreach \y in {0.5,1.5}{
}}
\node at (2.15,1.2) {$\scriptstyle x$};
\node at (3.15,3.5) {$\scriptstyle k$};
}
:=
\frac{1}{D_\cX^2}
\sum_{\psi \in \operatorname{ONB}(X\to x)}
\sum_{s,t\in \Irr(\cX)}
d_sd_t
\tikzmath{
\draw[thick] (.5,.5) grid (4.5,4.5);
\fill[white] (1.8,1.3) rectangle (2.2,1.95);
\fill[white] (2.8,3.05) rectangle (3.2,3.95);
\fill[gray!30, rounded corners=5pt] (2.1,3.1) rectangle (3.9,3.9);
\filldraw[blue, thick, fill=gray!30, rounded corners=5pt] (1.1,1.1) rectangle (1.7,1.9);
\filldraw[blue, thick, fill=gray!30, rounded corners=5pt] (2.3,1.1) rectangle (2.9,1.9);
\node[blue] at (1.25,1.5) {$\scriptstyle s$};
\node[blue] at (2.45,1.5) {$\scriptstyle t$};
\foreach \x in {1,2,3,4}{
\foreach \y in {1,2,3,4}{
\filldraw (\x,\y) circle (.05cm);
}}
\foreach \x in {0.5,1.5,2.5}{
\foreach \y in {0.5,1.5}{
}}
\draw[thick] (2,1.3) to[out=90,in=-90] (2.15,1.6);
\draw[thick, knot] (3.15,3.3) -- (3.15,3.6);
\draw[thick, knot, orange, decorate, decoration={snake, segment length=1mm, amplitude=.2mm}] (2.15,1.6) -- (2.15,2.15) --node[above]{$\scriptstyle X$} (2.85,2.15) -- (3.15,2.15) -- (3.15,3.3);
\filldraw[orange] (2.15,1.6) node[left, xshift=.1cm]{$\scriptstyle \psi^\dag$} circle (.05cm);
\filldraw[orange] (3.15,3.3) node[right]{$\scriptstyle \psi$} circle (.05cm);
\node at (2.15,1.2) {$\scriptstyle x$};
}
$$

We now explain how hopping operators and string operators are compatible.
Given anyons $(X,\sigma_X),(Y,\sigma_Y)\in Z(\cX)$, we can first apply the string operator $S^{(X,\sigma_X)}_{
\gamma}(\bullet,\phi)$ 
where $\phi: H\to x$ and we suppress the data at the other end of the string,
after which we can apply the hopping operator
$H^{(Y,\sigma_Y)}_{\delta}$.
Here is a cartoon of this setup ignoring the start of $\gamma$:
$$
\frac{1}{D_\cX^2}
\sum_{\psi \in \operatorname{ONB}(Y\to x)}
\sum_{s,t\in \Irr(\cX)}
d_sd_t
\tikzmath{
\draw[thick] (.5,.5) grid (4.5,4.5);
\fill[white] (1.8,1.05) rectangle (2.2,1.95);
\fill[white] (2.8,3.05) rectangle (3.2,3.95);
\fill[gray!30, rounded corners=5pt] (2.1,3.1) rectangle (3.9,3.9);
\filldraw[blue, thick, fill=gray!30, rounded corners=5pt] (1.1,1.1) rectangle (1.7,1.9);
\filldraw[blue, thick, fill=gray!30, rounded corners=5pt] (2.4,1.1) rectangle (2.9,1.9);
\node[blue] at (1.25,1.5) {$\scriptstyle s$};
\node[blue] at (2.55,1.5) {$\scriptstyle t$};
\foreach \x in {1,2,3,4}{
\foreach \y in {1,2,3,4}{
\filldraw (\x,\y) circle (.05cm);
}}
\foreach \x in {0.5,1.5,2.5}{
\foreach \y in {0.5,1.5}{
}}
\draw[thick] (2.15,1.3) --node[right, xshift=-.1cm]{$\scriptstyle x$} (2.15,1.7);
\draw[thick, knot] (3.15,3.3) -- (3.15,3.6);
\draw[thick, knot, orange, decorate, decoration={snake, segment length=1mm, amplitude=.2mm}] (2.15,.5) --node[right]{$\scriptstyle X$} (2.15,1) -- (2.15,1.3);
\draw[thick, knot, cyan, decorate, decoration={snake, segment length=1mm, amplitude=.2mm}] (2.15,1.7) -- (2.15,2.15) --node[above]{$\scriptstyle Y$} (2.85,2.15) -- (3.15,2.15) -- (3.15,3.3);
\filldraw[cyan] (2.15,1.7) node[left, xshift=.1cm]{$\scriptstyle \psi^\dag$} circle (.05cm);
\filldraw[orange] (2.15,1.3) node[left]{$\scriptstyle \phi$} circle (.05cm);
\filldraw[cyan] (3.15,3.3) node[right]{$\scriptstyle \psi$} circle (.05cm);
}
$$
If $Y$ does not contain $x$ as a simple summand, we clearly have $H^{(Y,\sigma_Y)}_{\delta}S^{(X,\sigma_X)}_{\gamma}(\bullet,\phi)=0$.
Now by \eqref{eq:BpvqProtection} above, we really have a $B_{p\vee q}$ when applying $S^{(X,\sigma_X)}_{\gamma}(\bullet,\phi)$, and looking closely at the edge $s(\delta)=t(\gamma)$ and ignoring the sum over $\psi$ for now, we get the following by \eqref{eq:Fusion}:
$$
\frac{1}{D_\cX^3}
\sum_{r,s,t\in \Irr(\cX)}
d_rd_sd_t
\tikzmath[scale=1.5]{
\draw[thick] (.8,.8) grid (3.2,2.2);
\fill[white] (1.8,1.03) rectangle (2.2,1.97);
\filldraw[blue, thick, fill=gray!30, rounded corners=5pt] (1.3,1.2) rectangle (1.8,1.8);
\draw[blue, thick, rounded corners=5pt] (1.1,1.1) rectangle (2.9,1.9);
\filldraw[blue, thick, fill=gray!30, rounded corners=5pt] (2.3,1.2) rectangle (2.8,1.8);
\node[blue] at (1.2,1.5) {$\scriptstyle r$};
\node[blue] at (1.4,1.5) {$\scriptstyle s$};
\node[blue] at (2.4,1.5) {$\scriptstyle t$};
\foreach \x in {1,2,3}{
\foreach \y in {1,2,}{
\filldraw (\x,\y) circle (.03cm);
}}
\foreach \x in {0.5,1.5,2.5}{
\foreach \y in {0.5,1.5}{
}}
\draw[thick] (2.15,1.3) --node[left]{$\scriptstyle x$} (2.15,1.7);
\draw[thick, knot, orange, decorate, decoration={snake, segment length=1mm, amplitude=.2mm}] (2.15,.8) --node[right]{$\scriptstyle X$} (2.15,1) -- (2.15,1.3);
\draw[thick, knot, cyan, decorate, decoration={snake, segment length=1mm, amplitude=.2mm}] (2.15,1.7) -- (2.15,2) --node[right]{$\scriptstyle Y$} (2.15,2.2);
\filldraw[cyan] (2.15,1.7) node[left, xshift=.1cm]{$\scriptstyle \psi^\dag$} circle (.05cm);
\filldraw[orange] (2.15,1.3) node[left]{$\scriptstyle \phi$} circle (.05cm);
}
=
\frac{1}{D_\cX^3}
\sum_{r,s,t\in \Irr(\cX)}
d_rd_sd_t
\tikzmath[scale=1.5]{
\draw[thick] (.8,.8) grid (3.2,2.2);
\fill[white] (1.8,1.03) rectangle (2.2,1.97);
\filldraw[blue, thick, fill=gray!30, rounded corners=5pt] (1.1,1.1) rectangle (1.6,1.9);
\draw[blue, thick, rounded corners=5pt] (1.7,1.1) rectangle (2.3,1.9);
\filldraw[blue, thick, fill=gray!30, rounded corners=5pt] (2.4,1.1) rectangle (2.9,1.9);
\node[blue] at (1.8,1.5) {$\scriptstyle r$};
\node[blue] at (1.2,1.5) {$\scriptstyle s$};
\node[blue] at (2.5,1.5) {$\scriptstyle t$};
\foreach \x in {1,2,3}{
\foreach \y in {1,2,}{
\filldraw (\x,\y) circle (.03cm);
}}
\foreach \x in {0.5,1.5,2.5}{
\foreach \y in {0.5,1.5}{
}}
\draw[thick] (2.15,1.3) --node[left]{$\scriptstyle x$} (2.15,1.7);
\draw[thick, knot, orange, decorate, decoration={snake, segment length=1mm, amplitude=.2mm}] (2.15,.8) --node[right]{$\scriptstyle X$} (2.15,1) -- (2.15,1.3);
\draw[thick, knot, cyan, decorate, decoration={snake, segment length=1mm, amplitude=.2mm}] (2.15,1.7) -- (2.15,2) --node[right]{$\scriptstyle Y$} (2.15,2.2);
\filldraw[cyan] (2.15,1.7) node[left, xshift=.1cm]{$\scriptstyle \psi^\dag$} circle (.05cm);
\filldraw[orange] (2.15,1.3) node[left]{$\scriptstyle \phi$} circle (.05cm);
}
\,.
$$
For each choice of $\psi$, we see that the encircled
$\psi^\dag \circ \phi$ in $\cX(X\to Y)$ actually lives in $Z(\cX)$:
\begin{equation}
\label{eq:MendStringOperators}
\frac{1}{D_\cX}
\sum_{r\in \Irr(\cX)}
d_r
\tikzmath[scale=1.5]{
\draw[thick, red] (1.5,.5) to[out=90, in=-90] (2.5,1.3) --node[right]{$\scriptstyle y$} (2.5,2.2);
\draw[blue, thick, rounded corners=5pt] (1.7,1.1) rectangle (2.3,1.9);
\node[blue] at (1.8,1.5) {$\scriptstyle r$};
\draw[thick] (2.15,1.3) --node[left]{$\scriptstyle x$} (2.15,1.7);
\draw[thick, knot, orange, decorate, decoration={snake, segment length=1mm, amplitude=.2mm}] (2.15,.5) --node[right]{$\scriptstyle X$} (2.15,1) -- (2.15,1.3);
\draw[thick, knot, cyan, decorate, decoration={snake, segment length=1mm, amplitude=.2mm}] (2.15,1.7) -- (2.15,2) --node[right]{$\scriptstyle Y$} (2.15,2.2);
\filldraw[cyan] (2.15,1.7) node[left, xshift=.1cm]{$\scriptstyle \psi^\dag$} circle (.05cm);
\filldraw[orange] (2.15,1.3) node[left]{$\scriptstyle \phi$} circle (.05cm);
}
=
\frac{1}{D_\cX}
\sum_{r\in \Irr(\cX)}
d_r
\tikzmath[scale=1.5]{
\draw[thick, red] (1.5,.8) --node[left]{$\scriptstyle y$} (1.5,1.7) to[out=90,in=-90] (2.5,2.5);
\draw[blue, thick, rounded corners=5pt] (1.7,1.1) rectangle (2.3,1.9);
\node[blue] at (1.8,1.5) {$\scriptstyle r$};
\draw[thick] (2.15,1.3) --node[left]{$\scriptstyle x$} (2.15,1.7);
\draw[thick, knot, orange, decorate, decoration={snake, segment length=1mm, amplitude=.2mm}] (2.15,.8) --node[right]{$\scriptstyle X$} (2.15,1) -- (2.15,1.3);
\draw[thick, knot, cyan, decorate, decoration={snake, segment length=1mm, amplitude=.2mm}] (2.15,1.7) -- (2.15,2) --node[left,yshift=.1cm]{$\scriptstyle Y$} (2.15,2.5);
\filldraw[cyan] (2.15,1.7) node[left, xshift=.1cm]{$\scriptstyle \psi^\dag$} circle (.05cm);
\filldraw[orange] (2.15,1.3) node[left]{$\scriptstyle \phi$} circle (.05cm);
}
\qquad\qquad
\forall y\in \cX.
\end{equation}
This immediately implies that
$$
H^{(Y,\sigma_Y)}_{\delta}S^{(X,\sigma_X)}_{\gamma}(\bullet,\phi)
=
\begin{cases}
S^{(X,\sigma_X)}_{\delta\cdot \gamma}(\bullet,\phi)
&\text{if } X=Y
\\
0 
&\text{else,}
\end{cases}
$$ 
where $\delta\cdot \gamma$ denotes the concatenation of $\delta$ and $\gamma$.
Indeed, summing over $\psi \in \operatorname{ONB}(Y\to x)$ yields
$$
\frac{1}{D_\cX^3}
\sum_{r,s,t\in \Irr(\cX)}
d_rd_sd_t
\tikzmath{
\draw[thick] (.5,.5) grid (4.5,4.5);
\fill[white] (1.8,1.05) rectangle (2.2,1.95);
\fill[white] (2.8,3.05) rectangle (3.2,3.95);
\fill[gray!30, rounded corners=5pt] (2.1,3.1) rectangle (3.9,3.9);
\draw[blue, thick, rounded corners=5pt] (1.8,1.1) rectangle (2.3,1.9);
\filldraw[blue, thick, fill=gray!30, rounded corners=5pt] (1.1,1.1) rectangle (1.7,1.9);
\filldraw[blue, thick, fill=gray!30, rounded corners=5pt] (2.4,1.1) rectangle (2.9,1.9);
\node[blue] at (1.25,1.5) {$\scriptstyle s$};
\node[blue] at (2.55,1.5) {$\scriptstyle t$};
\foreach \x in {1,2,3,4}{
\foreach \y in {1,2,3,4}{
\filldraw (\x,\y) circle (.05cm);
}}
\foreach \x in {0.5,1.5,2.5}{
\foreach \y in {0.5,1.5}{
}}
\draw[thick, knot] (3.15,3.3) -- (3.15,3.6);
\draw[thick, knot, orange, decorate, decoration={snake, segment length=1mm, amplitude=.2mm}] (2.15,.5) --node[right]{$\scriptstyle X$} (2.15,1) -- (2.15,2.15) -- (2.85,2.15) -- (3.15,2.15) -- (3.15,3.3);
\filldraw[orange] (3.15,3.3) node[right]{$\scriptstyle \phi$} circle (.05cm);
}
=
\tikzmath{
\draw[thick] (.5,.5) grid (4.5,4.5);
\fill[white] (2.8,3.05) rectangle (3.2,3.95);
\fill[gray!30, rounded corners=5pt] (2.1,3.1) rectangle (3.9,3.9);
\foreach \x in {1,2,3,4}{
\foreach \y in {1,2,3,4}{
\filldraw (\x,\y) circle (.05cm);
}}
\foreach \x in {0.5,1.5,2.5}{
\foreach \y in {0.5,1.5}{
}}
\draw[thick, knot] (3.15,3.3) -- (3.15,3.6);
\draw[thick, knot, orange, decorate, decoration={snake, segment length=1mm, amplitude=.2mm}] (2.15,.5) --node[right]{$\scriptstyle X$} (2.15,1) -- (2.15,2.15) -- (2.85,2.15) -- (3.15,2.15) -- (3.15,3.3);
\filldraw[orange] (3.15,3.3) node[right]{$\scriptstyle \phi$} circle (.05cm);
}\,.
$$

\section{String net models for enriched unitary fusion categories}
\label{sec:Enriched}

In this section we give a rigorous treatment of the enriched Levin-Wen string net model associated to an $\cA$-enriched UFC $\cX$ from \cite{MR4640433}, where $\cA$ is a UMTC representing the Witt class of the anomaly. 
In particular, we use TQFT techniques to prove the claim that the boundary excitations are the enriched center/M\"uger centeralizer $Z^\cA(\cX)$ of $\cA$ in $Z(\cX)$.
We do so in more generality, starting with an arbitrary unitary braided fusion category $\cA$ and a fully faithful unitary braided functor $\Phi^Z: \cA\to Z(\cX)$.

In the $\cA$-enriched model, objects of the form $\Phi(a)\in \cX$ for $a\in \cA$ can be `pulled off' the $\cX$-boundary into the $\cA$-bulk.
This process shows that the $\Tube(\cX)$-action induced from string operators on the boundary descends to a quotient of $\Tube(\cS)$ called the \emph{dome algebra}.
We analyze this algebra in \S\ref{sec:DomeAlgebra} and show that it classifies the boundary anyon types in \S\ref{sec:localEnrichedExcitations}.
We conclude with \S\ref{sec:SphereAlgebra} on point excitations in the bulk and their connection to point excitations on the boundary.

\subsection{Enriched skein modules and enriched string nets}
\label{sec:EnrichedStringNets}

Let $\cA$ \IN{$\cA$ a unitary braided fusion category (typically with braided functor $\cA \to Z(\cX)$)} be a unitary braided fusion category and $\cX$ be a UFC together with a fully faithful braided functor $\Phi^Z: \cA\to Z(\cX)$.
We denote by $\Phi: \cA\to \cX$ the composite $F\circ \Phi^Z$ where $F: Z(\cX) \to \cX$ is the forgetful functor.

\begin{defn}
The $\cA$-enriched disk skein module for $\cX$ with $m$ $\cA$-boundary points and $n$ $\cX$-boundary points denoted $\cS^{\textcolor{red}{\cA}}_\cX(\bbD, \textcolor{red}{m},n)$ is the Hilbert space orthogonal direct sum
$$
\bigoplus_{\substack{x_1,\dots, x_n \in \Irr(\cX)\\ \textcolor{red}{a_1,\dots, a_m\in \Irr(\cA)}}} \cX(\Phi(\textcolor{red}{a_1\otimes\dots \otimes a_m})\to x_1\otimes \cdots \otimes x_n)
=
\left\{
\tikzmath{
\draw[thick, red] (-120:.5cm and .3cm) -- (-120:1cm) node[below]{$\scriptstyle a_1$};
\draw[thick, red] (-60:.5cm and .3cm) -- (-60:1cm) node[below]{$\scriptstyle a_m$};
\draw[very thick, knot] (0,0) circle (1cm and .6cm);
\draw[very thick] (0,0) circle (.5cm and .3cm);
\draw[very thick] (-1,0) arc (-180:0:1cm);
\node at (0,0) {$f$};
\draw (140:.5cm and .3cm) -- (140:1cm and .6cm) node[left]{$\scriptstyle x_1$};
\draw (40:.5cm and .3cm) -- (40:1cm and .6cm) node[right]{$\scriptstyle x_n$};
\node at (70:.45cm) {$\cdot$};
\node at (90:.45cm) {$\cdot$};
\node at (110:.45cm) {$\cdot$};
\node[red] at (-70:.75cm) {$\cdot$};
\node[red] at (-90:.75cm) {$\cdot$};
\node[red] at (-110:.75cm) {$\cdot$};
}
\right\}.
$$
Here, each $\cX(\Phi(\textcolor{red}{a_1\otimes\dots \otimes a_m})\to x_1\otimes \cdots \otimes x_n)$ is equipped with the inner product
\begin{equation}
\label{eq:EnrichedSkeinModuleInnerProduct}
\langle f | g\rangle
:=
\frac{1}{\sqrt{\textcolor{red}{d_{a_1}\cdots d_{a_m}}d_{x_1}\cdots d_{x_n}}} 
\tr_\cC(f^\dag \circ g)
=
\frac{1}{\sqrt{\textcolor{red}{d_{a_1}\cdots d_{a_m}}d_{x_1}\cdots d_{x_n}}} 
\cdot
\cdot
\tikzmath{
\roundNbox{}{(0,1)}{.3}{.2}{.2}{$f^\dag$}
\roundNbox{}{(0,0)}{.3}{.2}{.2}{$g$}
\draw (-.3,.3) --node[left, xshift=.1cm]{$\scriptstyle x_1$} (-.3,.7) ;
\draw (.3,.3) --node[right, xshift=-.1cm]{$\scriptstyle x_n$} (.3,.7) ;
\draw[thick, red] (.3,1.3) arc (180:0:.2cm) --node[right, xshift=-.1cm]{$\scriptstyle \overline{\Phi(a_m)}$} (.7,-.3) arc (0:-180:.2cm);
\draw[thick, red] (-.3,1.3) arc (180:0:1cm and .5cm) -- node[right, xshift=-.1cm]{$\scriptstyle \overline{\Phi(a_1)}$} (1.7,-.3) arc (0:-180:1cm and .5cm);
\node[red] at (1.2,1cm) {$\cdots$};
\node at (0,.5cm) {$\cdots$};
}\,.
\end{equation}
\end{defn}

For simplicity we work with a cubic lattice, but this geometry is not necessary.
The 2D boundary corresponds to $\cX$, and the 3D bulk corresponds to $\cA$.
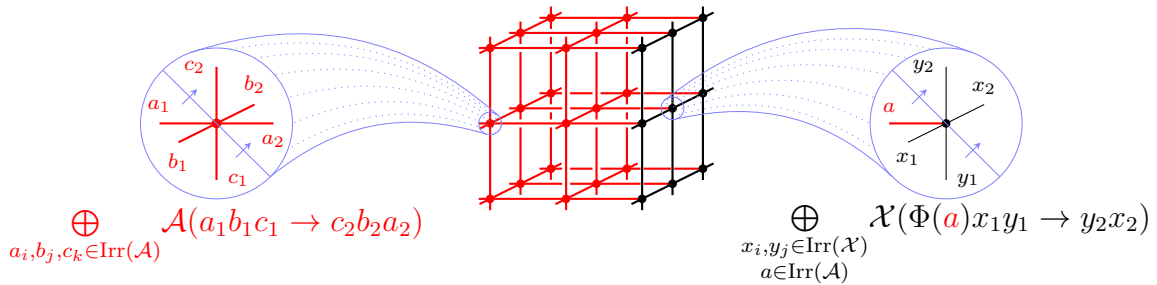
$$
\begin{tikzpicture}
	\pgfmathsetmacro{\lattice}{1};	
	\pgfmathsetmacro{\xoffset}{.4};	
	\pgfmathsetmacro{\yoffset}{.2};	
	\pgfmathsetmacro{\extra}{.15};	
	\pgfmathsetmacro{\zextra}{.35};	
	\pgfmathsetmacro{\zoom}{6};	
	\coordinate (z1) at ($ -.6*(\zoom,0) + (0,\lattice) $);
	\coordinate (z2) at ($ (\zoom,\lattice) $);
	\coordinate (aaa) at (0,0);
	\coordinate (baa) at ($ (aaa) + (\lattice,0) $);
	\coordinate (caa) at ($ (aaa) + 2*(\lattice,0) $);
	\coordinate (aba) at ($ (aaa) + (0,\lattice) $);
	\coordinate (bba) at ($ (aaa) + (\lattice,0) + (0,\lattice) $);
	\coordinate (cba) at ($ (aaa) + 2*(\lattice,0) + (0,\lattice) $);
	\coordinate (aca) at ($ (aaa) + 2*(0,\lattice) $);
	\coordinate (bca) at ($ (aaa) + (\lattice,0) + 2*(0,\lattice) $);
	\coordinate (cca) at ($ (aaa) + 2*(\lattice,0) + 2*(0,\lattice) $);
	\coordinate (aab) at ($ (aaa) + (\xoffset,\yoffset) $);
	\coordinate (bab) at ($ (aaa) + (\lattice,0) + (\xoffset,\yoffset) $);
	\coordinate (cab) at ($ (aaa) + 2*(\lattice,0) + (\xoffset,\yoffset) $);
	\coordinate (abb) at ($ (aaa) + (0,\lattice) + (\xoffset,\yoffset) $);
	\coordinate (bbb) at ($ (aaa) + (\lattice,0) + (0,\lattice) + (\xoffset,\yoffset) $);
	\coordinate (cbb) at ($ (aaa) + 2*(\lattice,0) + (0,\lattice) + (\xoffset,\yoffset) $);
	\coordinate (acb) at ($ (aaa) + 2*(0,\lattice) + (\xoffset,\yoffset) $);
	\coordinate (bcb) at ($ (aaa) + (\lattice,0) + 2*(0,\lattice) + (\xoffset,\yoffset) $);
	\coordinate (ccb) at ($ (aaa) + 2*(\lattice,0) + 2*(0,\lattice) + (\xoffset,\yoffset) $);
	\coordinate (aac) at ($ (aaa) + 2*(\xoffset,\yoffset) $);
	\coordinate (bac) at ($ (aaa) + (\lattice,0) + 2*(\xoffset,\yoffset) $);
	\coordinate (cac) at ($ (aaa) + 2*(\lattice,0) + 2*(\xoffset,\yoffset) $);
	\coordinate (abc) at ($ (aaa) + (0,\lattice) + 2*(\xoffset,\yoffset) $);
	\coordinate (bbc) at ($ (aaa) + (\lattice,0) + (0,\lattice) + 2*(\xoffset,\yoffset) $);
	\coordinate (cbc) at ($ (aaa) + 2*(\lattice,0) + (0,\lattice) + 2*(\xoffset,\yoffset) $);
	\coordinate (acc) at ($ (aaa) + 2*(0,\lattice) + 2*(\xoffset,\yoffset) $);
	\coordinate (bcc) at ($ (aaa) + (\lattice,0) + 2*(0,\lattice) + 2*(\xoffset,\yoffset) $);
	\coordinate (ccc) at ($ (aaa) + 2*(\lattice,0) + 2*(0,\lattice) + 2*(\xoffset,\yoffset) $);
	\draw[thick, red] ($ (aac) - \extra*(\lattice,0) $) -- (cac);
	\draw[thick, red] ($ (abc) - \extra*(\lattice,0) $) -- (cbc);
	\draw[thick, red] ($ (acc) - \extra*(\lattice,0) $) -- (ccc);
	\draw[thick, red] ($ (aac) - \extra*(0,\lattice) $) -- ($ (acc) + \extra*(0,\lattice) $);
	\draw[thick, red] ($ (bac) - \extra*(0,\lattice) $) -- ($ (bcc) + \extra*(0,\lattice) $);
	\draw[thick] ($ (cac) - \extra*(0,\lattice) $) -- ($ (ccc) + \extra*(0,\lattice) $);
	\draw[thick, red, knot] ($ (aab) - \extra*(\lattice,0) $) -- (cab);
	\draw[thick, red, knot] ($ (abb) - \extra*(\lattice,0) $) -- (cbb);
	\draw[thick, red, knot] ($ (acb) - \extra*(\lattice,0) $) -- (ccb);
	\draw[thick, red, knot] ($ (aab) - \extra*(0,\lattice) $) -- ($ (acb) + \extra*(0,\lattice) $);
	\draw[thick, red, knot] ($ (bab) - \extra*(0,\lattice) $) -- ($ (bcb) + \extra*(0,\lattice) $);
	\draw[very thick, white] ($ (cab) - \extra*(0,\lattice) $) -- ($ (ccb) + \extra*(0,\lattice) $);
	\draw[thick] ($ (cab) - \extra*(0,\lattice) $) -- ($ (ccb) + \extra*(0,\lattice) $);
	\draw[thick, red, knot] ($ (aaa) - \extra*(\lattice,0) $) -- (caa);
	\draw[thick, red, knot] ($ (aba) - \extra*(\lattice,0) $) -- (cba);
	\draw[thick, red, knot] ($ (aca) - \extra*(\lattice,0) $) -- (cca);
	\draw[thick, red, knot] ($ (aaa) - \extra*(0,\lattice) $) -- ($ (aca) + \extra*(0,\lattice) $);
	\draw[thick, red, knot] ($ (baa) - \extra*(0,\lattice) $) -- ($ (bca) + \extra*(0,\lattice) $);
	\draw[very thick, white] ($ (caa) - \extra*(0,\lattice) $) -- ($ (cca) + \extra*(0,\lattice) $);
	\draw[thick] ($ (caa) - \extra*(0,\lattice) $) -- ($ (cca) + \extra*(0,\lattice) $);
	\draw[thick, red] ($ (aaa) - \zextra*(\xoffset,\yoffset) $) -- ($ (aac) + \zextra*(\xoffset,\yoffset) $);
	\draw[thick, red] ($ (aba) - \zextra*(\xoffset,\yoffset) $) -- ($ (abc) + \zextra*(\xoffset,\yoffset) $);
	\draw[thick, red] ($ (aca) - \zextra*(\xoffset,\yoffset) $) -- ($ (acc) + \zextra*(\xoffset,\yoffset) $);
	\draw[thick, red] ($ (baa) - \zextra*(\xoffset,\yoffset) $) -- ($ (bac) + \zextra*(\xoffset,\yoffset) $);
	\draw[thick, red] ($ (bba) - \zextra*(\xoffset,\yoffset) $) -- ($ (bbc) + \zextra*(\xoffset,\yoffset) $);
	\draw[thick, red] ($ (bca) - \zextra*(\xoffset,\yoffset) $) -- ($ (bcc) + \zextra*(\xoffset,\yoffset) $);
	\draw[thick] ($ (caa) - \zextra*(\xoffset,\yoffset) $) -- ($ (cac) + \zextra*(\xoffset,\yoffset) $);
	\draw[thick] ($ (cba) - \zextra*(\xoffset,\yoffset) $) -- ($ (cbc) + \zextra*(\xoffset,\yoffset) $);
	\draw[thick] ($ (cca) - \zextra*(\xoffset,\yoffset) $) -- ($ (ccc) + \zextra*(\xoffset,\yoffset) $);
	\filldraw[red] (aaa) circle (.05cm);
	\filldraw[red] (aba) circle (.05cm);
	\filldraw[red] (aca) circle (.05cm);
	\filldraw[red] (aab) circle (.05cm);
	\filldraw[red] (abb) circle (.05cm);
	\filldraw[red] (acb) circle (.05cm);
	\filldraw[red] (aac) circle (.05cm);
	\filldraw[red] (abc) circle (.05cm);
	\filldraw[red] (acc) circle (.05cm);
	\filldraw[red] (baa) circle (.05cm);
	\filldraw[red] (bab) circle (.05cm);
	\filldraw[red] (bac) circle (.05cm);
	\filldraw[red] (bba) circle (.05cm);
	\filldraw[red] (bbb) circle (.05cm);
	\filldraw[red] (bbc) circle (.05cm);
	\filldraw[red] (bca) circle (.05cm);
	\filldraw[red] (bcb) circle (.05cm);
	\filldraw[red] (bcc) circle (.05cm);
	\filldraw (caa) circle (.05cm);
	\filldraw (cba) circle (.05cm);
	\filldraw (cca) circle (.05cm);
	\filldraw (cab) circle (.05cm);
	\filldraw (cbb) circle (.05cm);
	\filldraw (ccb) circle (.05cm);
	\filldraw (cac) circle (.05cm);
	\filldraw (cbc) circle (.05cm);
	\filldraw (ccc) circle (.05cm);
	\draw[blue!50, very thin] (cbb) circle (.15cm);
		\foreach \x/\y/\s in {155/100/24, 175/120/23, 185/140/22, 200/160/23, 220/180/25, 240/200/26} 
		{\draw[dotted, blue!50] ($(cbb)+(\x:\extra)$) to[bend left=\s] ($(z2) + (\y:\lattice)$);}
		\foreach \x/\y/\s in {135/70/30,290/225/28}
		{\draw[blue!50, very thin] ($ (cbb) + (\x:\extra) $) to[bend left=\s] ($ (z2) + (\y:\lattice) $);}
			\draw[blue!50, very thin] (z2) circle (\lattice);
			\draw[thick, red] ($ (z2) - .75*(\lattice,0) $) node[above] {\scriptsize{$a$}} -- (z2);
			\draw ($ (z2) - .75*(0,\lattice) $) node[right] {\scriptsize{$y_1$}} -- ($ (z2) + .75*(0,\lattice) $) node[left] {\scriptsize{$y_2$}};
			\draw ($ (z2) - 1.25*(\xoffset,\yoffset) $) node[below] {\scriptsize{$x_1$}} -- ($ (z2) + 1.25*(\xoffset,\yoffset) $) node[above] {\scriptsize{$x_2$}};
			\filldraw (z2) circle (.05cm);
			\node at ($ (z2) -1.6*(0,\lattice) $) {$\underset{\substack{x_i,y_j\in \Irr(\cX)\\a\in \Irr(\cA)}}{\bigoplus}\cX(\Phi(\textcolor{red}{a})x_1y_1 \to y_2x_2)$};
			\draw[blue!50, very thin] ($ (z2) + (135:\lattice) $) -- ($ (z2) +(-45:\lattice) $);
			\foreach \x in {135,-45}
			{\draw[blue!50, very thin, -stealth] ($ (z2) + (\x:.5*\lattice) - (.1,.1)$) to ($ (z2) + (\x:.5*\lattice) + (.1,.1)$);}
	\pgftransformxscale{-1}
	\draw[blue!50, very thin] (aba) circle (.15cm);
		\foreach \x/\y/\s in {155/100/24, 175/120/23, 185/140/22, 200/160/23, 220/180/25, 240/200/26} 
		{\draw[dotted, blue!50] ($(aba)+(\x:\extra)$) to[bend left=\s] ($(z1) + (\y:\lattice)$);}
		\foreach \x/\y/\s in {135/70/30,290/225/28}
		{\draw[blue!50, very thin] ($ (aba) + (\x:\extra) $) to[bend left=\s] ($ (z1) + (\y:\lattice) $);}
			\draw[blue!50, very thin] (z1) circle (\lattice);
			\draw[thick, red] ($ (z1) - .75*(\lattice,0) $) node[below] {\scriptsize{$a_2$}} -- ($ (z1) + .75*(\lattice,0) $) node[above] {\scriptsize{$a_1$}};
			\draw[thick, red] ($ (z1) - .75*(0,\lattice) $) node[right] {\scriptsize{$c_1$}} -- ($ (z1) + .75*(0,\lattice) $) node[left] {\scriptsize{$c_2$}};
			\draw[thick, red] ($ (z1) - 1.25*(\xoffset,0) + 1.25*(0,\yoffset) $) node[above] {\scriptsize{$b_2$}} -- ($ (z1) + 1.25*(\xoffset,0) - 1.25*(0,\yoffset) $) node[below] {\scriptsize{$b_1$}};
			\filldraw[thick, red] (z1) circle (.05cm);
			\node at ($ (z1) -1.5*(0,\lattice) $) {\textcolor{red}{$\underset{a_i,b_j,c_k\in \Irr(\cA)}{\bigoplus}\cA(a_1b_1c_1 \to c_2b_2a_2)$}};
			\draw[blue!50, very thin] ($ (z1) + (225:\lattice) $) -- ($ (z1) + (45:\lattice) $);
			\foreach \x in {225,45}
			{\draw[blue!50, very thin, -stealth] ($ (z1) + (\x:.5*\lattice) + (.1,-.1)$) to ($ (z1) + (\x:.5*\lattice) - (.1,-.1)$);}
\end{tikzpicture}
$$
Each vertex $v$ has a Hilbert space $\cH_v$ which is an orthogonal direct sum of hom spaces in $\cA$ or $\cX$, endowed with the following inner products:
\begin{align*}
\left\langle
\tikzmath{
\draw[thick, red] (-.3,-.3) -- (-.3,-.6) node[below]{$\scriptstyle a_1$};
\draw[thick, red] (0,-.3) -- (0,-.6) node[below]{$\scriptstyle b_1$};
\draw[thick, red] (.3,-.3) -- (.3,-.6) node[below]{$\scriptstyle c_1$};
\draw[thick, red] (-.3,.3) -- (-.3,.6) node[above]{$\scriptstyle c_2$};
\draw[thick, red] (0,.3) -- (0,.6) node[above]{$\scriptstyle b_2$};
\draw[thick, red] (.3,.3) -- (.3,.6) node[above]{$\scriptstyle a_2$};
\roundNbox{red}{(0,0)}{.3}{.2}{.2}{\textcolor{red}{$\eta$}}
}
\bigg|
\tikzmath{
\draw[thick, red] (-.3,-.3) -- (-.3,-.6) node[below]{$\scriptstyle a_1'$};
\draw[thick, red] (0,-.3) -- (0,-.6) node[below]{$\scriptstyle b_1'$};
\draw[thick, red] (.3,-.3) -- (.3,-.6) node[below]{$\scriptstyle c_1'$};
\draw[thick, red] (-.3,.3) -- (-.3,.6) node[above]{$\scriptstyle c_2'$};
\draw[thick, red] (0,.3) -- (0,.6) node[above]{$\scriptstyle b_2'$};
\draw[thick, red] (.3,.3) -- (.3,.6) node[above]{$\scriptstyle a_2'$};
\roundNbox{red}{(0,0)}{.3}{.2}{.2}{\textcolor{red}{$\eta'$}}
}
\right\rangle
&=
\delta_{a_i=a_i'}
\delta_{b_j=b_j'}
\delta_{c_k=c_k'}
\frac{1}{\sqrt{\textcolor{red}{d_{a_1}d_{a_2}d_{b_1}d_{b_2}d_{c_1}d_{c_2}}}}
\cdot
\tikzmath{
\draw[thick, red] (-.3,1.3) arc (180:0:.8cm) -- (1.3,-.3) arc (0:-180:.8cm);
\draw[thick, red] (0,1.3) arc (180:0:.5cm) -- (1,-.3) arc (0:-180:.5cm);
\draw[thick, red] (.3,1.3) arc (180:0:.2cm) -- (.7,-.3) arc (0:-180:.2cm);
\draw[thick, red] (-.3,.3) -- (-.3,.7); 
\draw[thick, red] (0,.3) -- (0,.7); 
\draw[thick, red] (.3,.3) -- (.3,.7); 
\roundNbox{red}{(0,0)}{.3}{.2}{.2}{\textcolor{red}{$\eta'$}}
\roundNbox{red}{(0,1)}{.3}{.2}{.2}{\textcolor{red}{$\eta^\dag$}}
}
\\
\left\langle
\tikzmath{
\draw[thick, red] (-.3,-.3) -- (-.3,-.6) node[left]{$\scriptstyle F(a)$};
\draw (0,-.3) -- (0,-.6) node[below]{$\scriptstyle x_1$};
\draw (.3,-.3) -- (.3,-.6) node[below]{$\scriptstyle y_1$};
\draw (0,.3) -- (0,.6) node[above]{$\scriptstyle y_2$};
\draw (.3,.3) -- (.3,.6) node[above]{$\scriptstyle x_2$};
\roundNbox{}{(0,0)}{.3}{.2}{.2}{$\xi$}
}
\bigg|
\tikzmath{
\draw[thick, red] (-.3,-.3) -- (-.3,-.6) node[left]{$\scriptstyle F(a')$};
\draw (0,-.3) -- (0,-.6) node[below]{$\scriptstyle x_1'$};
\draw (.3,-.3) -- (.3,-.6) node[below]{$\scriptstyle y_1'$};
\draw (0,.3) -- (0,.6) node[above]{$\scriptstyle y_2'$};
\draw (.3,.3) -- (.3,.6) node[above]{$\scriptstyle x_2'$};
\roundNbox{}{(0,0)}{.3}{.2}{.2}{$\xi'$}
}
\right\rangle
&=
\delta_{a=a'}
\delta_{x_i=x_i'}
\delta_{y_j=y_j'}
\frac{1}{\sqrt{\textcolor{red}{d_{a}}d_{x_1}d_{x_2}d_{y_1}d_{y_2}}}
\cdot
\tikzmath{
\draw[thick, red] (-.3,1.3) arc (180:0:.8cm) -- (1.3,-.3) arc (0:-180:.8cm);
\draw (0,1.3) arc (180:0:.5cm) -- (1,-.3) arc (0:-180:.5cm);
\draw (.3,1.3) arc (180:0:.2cm) -- (.7,-.3) arc (0:-180:.2cm);
\draw (0,.3) -- (0,.7); 
\draw (.3,.3) -- (.3,.7); 
\roundNbox{red}{(0,0)}{.3}{.2}{.2}{\textcolor{red}{$\xi'$}}
\roundNbox{red}{(0,1)}{.3}{.2}{.2}{\textcolor{red}{$\xi^\dag$}}
}\,.
\end{align*}
This second inner product is just the inner product on the $\cA$-enriched disk skein module $\cS_\cX^{\textcolor{red}{\cA}}(\bbD,\textcolor{red}{1},4)$.
We use the following shorthand notation for closed diagrams in $\cX$:
\begin{equation}
\label{eq:SkeinModuleIPShorthand}
\tikzmath{
\draw (-.5,-.25) node[below] {\scriptsize{$w$}} -- (.5,.25) node[right] {\scriptsize{$z$}};
\draw (-1.5,-.25) -- (-.5,.25) .. controls ++(30:.3cm) and ++(30:.3cm) .. (.5,.25);
\draw[knot] (0,-.5) node[right] {\scriptsize{$x$}} -- (0,.5) node[right, yshift=.2cm] {\scriptsize{$y$}} arc (0:180:.5cm);
\draw (0,-.5) arc (0:-180:.5cm) -- (-1,.5);
\draw[knot] (-1.5,-.25) .. controls ++(-150:.3cm) and ++(-150:.3cm) .. (-.5,-.25);
\draw[thick, red] (-1,0) -- node[above, xshift=.2cm, yshift=-.1cm] {\scriptsize{$a$}} (0,0);
\filldraw (0,0) node[right, yshift=-.1cm]{$\scriptstyle \xi$} circle (.05cm);
\filldraw (-1,0) node[left, yshift=.1cm]{$\scriptstyle \overline{\xi'}$} circle (.05cm);
}
:=
\tikzmath[xscale=-1]{
\draw (0,1) arc (90:270:.7cm and 1cm);
\draw (0,1) arc (90:270:.4cm and 1cm);
\draw (0,1) arc (90:-90:.3cm and 1cm);
\draw (0,1) arc (90:-90:.6cm and 1cm);
\draw[thick, red] (0,1) arc (90:-90:.9cm and 1cm);
\draw[thick] (0,0) circle (1cm);
\draw (-1,0) arc (-180:0:1cm and .3cm);
\draw[dotted] (-1,0) arc (180:0:1cm and .3cm);
\filldraw (0,-1) node[below]{$\scriptstyle \overline{\xi'}$} circle (.05cm);
\filldraw (0,1) node[above]{$\scriptstyle \xi$} circle (.05cm);
\node at (-.85,0) {$\scriptstyle \overline{y}$};
\node at (-.55,0) {$\scriptstyle \overline{z}$};
\node at (.45,0) {$\scriptstyle w$};
\node at (.15,0) {$\scriptstyle x$};
\node[red] at (1.3,0) {$\scriptstyle F(a)$};
}
=
\tikzmath{
\draw[thick, red] (0,1) --node[right]{$\scriptstyle a$} (0,-1);
\draw[knot] (-1,0) arc (-180:0:1cm and .3cm);
\draw (0,1) arc (90:270:.7cm and 1cm);
\draw (0,1) arc (90:270:.4cm and 1cm);
\draw (0,1) arc (90:-90:.4cm and 1cm);
\draw (0,1) arc (90:-90:.7cm and 1cm);
\draw[thick] (0,0) circle (1cm);
\draw[dotted] (-1,0) arc (180:0:1cm and .3cm);
\filldraw (0,-1) node[below]{$\scriptstyle \overline{\xi'}$} circle (.05cm);
\filldraw (0,1) node[above]{$\scriptstyle \xi$} circle (.05cm);
\node at (-.85,0) {$\scriptstyle w$};
\node at (-.55,0) {$\scriptstyle x$};
\node at (.85,0) {$\scriptstyle \overline{y}$};
\node at (.55,0) {$\scriptstyle \overline{z}$};
}
\end{equation}
The diagram on the left hand side above is shorthand, and the black crossings \emph{do not} represent any kind of braiding; rather the diagram is drawn on the surface of a 2-sphere.
As $\cX$ is $\cA$-enriched, we may draw the $\cA$-string as passing through the 3D bulk inside the 2-sphere.

The total Hilbert space is the tensor product over the local Hilbert spaces.
The Hamiltonian is made of 2 types of terms.
First, for every edge/link $\ell$ in our lattice, we have a projector $A_\ell$ enforcing that labels on edges match.
To define the face/plaquette terms $B_p$, we first pass to the image of the projector $P_A$ onto the ground state of $-\sum_{\ell\in \Lambda} A_\ell$.
There are three types of $B_p$ terms:
those in the $\cA$-bulk, those on plaquettes made of both $\cA$ and $\cX$-links, and those on the $\cX$ boundary. 
We will define the third type rigorously, and leave the first two to the reader.
On these plaquettes, the operator $B_p$ is given by 
\begin{align}
\frac{1}{D_\cX}\sum_{r\in \Irr(\cX)}
d_r\cdot
\tikzmath{
\draw[step=1.0,black,thin] (0.5,0.5) grid (2.5,2.5);
\draw[thick, red] (2,1) -- ($ (2,1) + (-.4,.4)$);
\draw[thick, red] (1,1) -- ($ (1,1) + (-.4,.4)$);
\draw[thick, red] (1,2) -- ($ (1,2) + (-.4,.4)$);
\draw[thick, red] (2,2) -- ($ (2,2) + (-.4,.4)$);
\node at (2.3,.8) {$\scriptstyle \xi_{2,1}$};
\node at (.5,2.2) {$\scriptstyle \xi_{1,2}$};
\node at (.7,.8) {$\scriptstyle \xi_{1,1}$};
\node at (2.3,2.2) {$\scriptstyle \xi_{2,2}$};
\node at (1.5,.85) {$\scriptstyle g$};
\node at (2.15,1.5) {$\scriptstyle h$};
\node at (1.5,2.15) {$\scriptstyle i$};
\node at (.85,1.5) {$\scriptstyle j$};
\draw[knot, thick, blue, rounded corners=5pt] (1.15,1.15) rectangle (1.85,1.85);
\fill[gray!60, rounded corners=5pt, opacity=.5] (1.16,1.16) rectangle (1.84,1.84);
\node[blue] at (1.3,1.5) {$\scriptstyle r$};
}
&=
\sum_{\substack{r, k,\ell,m,n\in \Irr(\cX)
\\
\eta \in \operatorname{ONB}
}}
\frac{\sqrt{d_kd_\ell d_md_n}}{d_r\sqrt{d_gd_hd_id_j}}
\tikzmath{
\draw[step=1.0,black,thin] (0.5,0.5) grid (2.5,2.5);
\draw[thick, red] (2,1) -- ($ (2,1) + (-.4,.4)$);
\draw[thick, red] (1,1) -- ($ (1,1) + (-.4,.4)$);
\draw[thick, red] (1,2) -- ($ (1,2) + (-.4,.4)$);
\draw[thick, red] (2,2) -- ($ (2,2) + (-.4,.4)$);
\draw[thick, blue] (1.3,1) -- (1,1.3);
\draw[knot, thick, blue] (1.7,1) -- (2,1.3);
\draw[thick, blue] (1.3,2) -- (1,1.7);
\draw[thick, blue] (1.7,2) -- (2,1.7);
\fill[fill=green] (1.3,1) circle (.05cm);
\fill[fill=green] (1.7,1) circle (.05cm);
\fill[fill=purple] (2,1.3) circle (.05cm);
\fill[fill=purple] (2,1.7) circle (.05cm);
\fill[fill=yellow] (1.3,2) circle (.05cm);
\fill[fill=yellow] (1.7,2) circle (.05cm);
\fill[fill=orange] (1,1.3) circle (.05cm);
\fill[fill=orange] (1,1.7) circle (.05cm);
\node at (2.3,.8) {$\scriptstyle \xi_{2,1}$};
\node at (.5,2.2) {$\scriptstyle \xi_{1,2}$};
\node at (.7,.8) {$\scriptstyle \xi_{1,1}$};
\node at (2.3,2.2) {$\scriptstyle \xi_{2,2}$};
\node at (1.5,.85) {$\scriptstyle k$};
\node at (2.15,1.5) {$\scriptstyle \ell$};
\node at (1.5,2.15) {$\scriptstyle m$};
\node at (.85,1.5) {$\scriptstyle n$};
\node at (1.15,.85) {$\scriptstyle g$};
\node at (1.85,.85) {$\scriptstyle g$};
\node at (2.15,1.15) {$\scriptstyle h$};
\node at (2.15,1.85) {$\scriptstyle h$};
\node at (1.15,2.15) {$\scriptstyle i$};
\node at (1.85,2.15) {$\scriptstyle i$};
\node at (.85,1.15) {$\scriptstyle j$};
\node at (.85,1.85) {$\scriptstyle j$};
\fill[gray!60, rounded corners=5pt, opacity=.5] (1.15,1.15) rectangle (1.85,1.85);
}
\label{eq:EnrichedBpAction}
\\&=
\sum_{\eta \in \operatorname{ONB}}
C(\xi,\eta)
\tikzmath{
\draw[step=1.0,black,thin] (0.5,0.5) grid (2.5,2.5);
\draw[thick, red] (2,1) -- ($ (2,1) + (-.4,.4)$);
\draw[thick, red] (1,1) -- ($ (1,1) + (-.4,.4)$);
\draw[thick, red] (1,2) -- ($ (1,2) + (-.4,.4)$);
\draw[thick, red] (2,2) -- ($ (2,2) + (-.4,.4)$);
\fill[gray!60, rounded corners=5pt, opacity=.5] (1.1,1.1) rectangle (1.9,1.9);
\node at (2.3,.8) {$\scriptstyle \eta_{2,1}$};
\node at (1.3,1.8) {$\scriptstyle \eta_{1,2}$};
\node at (1.3,.8) {$\scriptstyle \eta_{1,1}$};
\node at (2.3,1.8) {$\scriptstyle \eta_{2,2}$};
}
\notag
\end{align}
where $\xi,\eta$ are shorthand for choices of 4 vectors from a fixed ONB for the enriched skein module $\cS^{\textcolor{red}{\cA}}_\cX(\bbD, \textcolor{red}{1},4)$
such that the edges which meet are labelled by the same simple object,
and $D_\cX=\sum_{x\in \Irr(\cX)} d_x^2$ is the global dimension of $\cX$.

\begin{prop}
\label{prop:MatrixCoefficientsOfBp}
When all edge/link labels on simple tensors $\eta,\xi$ match
$$
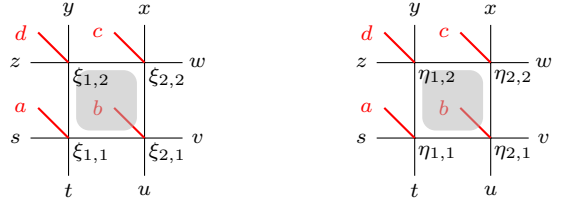
\tikzmath{
\draw[step=1.0,black,thin] (0.5,0.5) grid (2.5,2.5);
\draw[thick, red] (2,1) -- ($ (2,1) + (-.4,.4)$) node[left]{$\scriptstyle b$};
\draw[thick, red] (1,1) -- ($ (1,1) + (-.4,.4)$) node[left]{$\scriptstyle a$};
\draw[thick, red] (1,2) -- ($ (1,2) + (-.4,.4)$) node[left]{$\scriptstyle d$};
\draw[thick, red] (2,2) -- ($ (2,2) + (-.4,.4)$) node[left]{$\scriptstyle c$};
\fill[gray!60, rounded corners=5pt, opacity=.5] (1.1,1.1) rectangle (1.9,1.9);
\node at (2.3,.8) {$\scriptstyle \xi_{2,1}$};
\node at (1.3,1.8) {$\scriptstyle \xi_{1,2}$};
\node at (1.3,.8) {$\scriptstyle \xi_{1,1}$};
\node at (2.3,1.8) {$\scriptstyle \xi_{2,2}$};
\node at (.3,1) {$\scriptstyle s$};
\node at (1,.3) {$\scriptstyle t$};
\node at (2,.3) {$\scriptstyle u$};
\node at (2.7,1) {$\scriptstyle v$};
\node at (2.7,2) {$\scriptstyle w$};
\node at (2,2.7) {$\scriptstyle x$};
\node at (1,2.7) {$\scriptstyle y$};
\node at (.3,2) {$\scriptstyle z$};
}
\qquad\qquad
\tikzmath{
\draw[step=1.0,black,thin] (0.5,0.5) grid (2.5,2.5);
\draw[thick, red] (2,1) -- ($ (2,1) + (-.4,.4)$) node[left]{$\scriptstyle b$};
\draw[thick, red] (1,1) -- ($ (1,1) + (-.4,.4)$) node[left]{$\scriptstyle a$};
\draw[thick, red] (1,2) -- ($ (1,2) + (-.4,.4)$) node[left]{$\scriptstyle d$};
\draw[thick, red] (2,2) -- ($ (2,2) + (-.4,.4)$) node[left]{$\scriptstyle c$};
\fill[gray!60, rounded corners=5pt, opacity=.5] (1.1,1.1) rectangle (1.9,1.9);
\node at (2.3,.8) {$\scriptstyle \eta_{2,1}$};
\node at (1.3,1.8) {$\scriptstyle \eta_{1,2}$};
\node at (1.3,.8) {$\scriptstyle \eta_{1,1}$};
\node at (2.3,1.8) {$\scriptstyle \eta_{2,2}$};
\node at (.3,1) {$\scriptstyle s$};
\node at (1,.3) {$\scriptstyle t$};
\node at (2,.3) {$\scriptstyle u$};
\node at (2.7,1) {$\scriptstyle v$};
\node at (2.7,2) {$\scriptstyle w$};
\node at (2,2.7) {$\scriptstyle x$};
\node at (1,2.7) {$\scriptstyle y$};
\node at (.3,2) {$\scriptstyle z$};
}\,,
$$
the matrix coefficient $C(\xi,\eta)$ is given by
\begin{align*}
C(\xi,\eta)
&=
\frac{1}{D_\cX\sqrt{\textcolor{red}{d_ad_bd_cd_d}d_sd_td_ud_vd_wd_xd_yd_z}}
\cdot
\tikzmath{
\begin{scope}
\clip (1.2,1.8) rectangle (3.5,3.6);
\begin{scope}[shift={(-.3,.3)}]
\draw[step=1.0,black] (0.5,0.5) grid (4,4);
\end{scope}
\draw[step=1.0,black, knot] (1,1) grid (4,4);
\end{scope}
\foreach \j in {2,3} {
\draw (1.2,\j) arc (270:90:.15cm);
\draw (3.5,\j) arc (-90:90:.15cm);
};
\foreach \i in {2,3} {
    \draw (\i,3.6) arc (0:180:.15cm);
    \draw (\i,1.8) arc (0:-180:.15cm);
};
\foreach \i in {1,2} {
\foreach \j in {1,2} {
        \draw[thick, red] ($ (\i,\j) + (1,1)$) -- ($ (\i,\j) + (.7,1.3)$);
        \node at ($ (\i,\j) + (1.25,.85) $) {$\scriptstyle \xi_{\i,\j}$};
        \node at ($ (\i,\j) + (.45,1.45) $) {$\scriptstyle \overline{\eta_{\i,\j}}$};
}}
}
\\&=
\frac{1}{D_\cX\sqrt{\textcolor{red}{d_ad_bd_cd_d}d_sd_td_ud_vd_wd_xd_yd_z}}
\cdot
\tikzmath{
\clip (-1,-1) rectangle (4,4);
\draw[thick, red, mid<] (1,1) .. controls ++(135:1cm) and ++(135:1cm) .. node[above, yshift=.1cm]{$\scriptstyle a$} (0,0);
\draw[thick, red, mid<] (2,1) .. controls ++(135:1cm) and ++(135:1cm) .. (2,0) node[left, xshift=-.1cm]{$\scriptstyle b$} .. controls ++(-45:1cm) and ++(-45:1cm) .. (3,0);
\draw[thick, red, mid<] (2,2) .. controls ++(135:1cm) and ++(135:1cm) .. node[above, yshift=.1cm]{$\scriptstyle c$} (3,3);
\draw[thick, red, mid<] (1,2) .. controls ++(135:1cm) and ++(-45:1cm) .. node[right, xshift=.05cm]{$\scriptstyle d$} (0,3);
\draw[knot, mid>] (1,1) -- node[above, xshift=-.1cm]{$\scriptstyle g$} (2,1);
\draw[mid>] (2,1) -- node[left]{$\scriptstyle h$} (2,2);
\draw[mid>] (1,2) -- node[above]{$\scriptstyle i$} (2,2);
\draw[mid>] (1,1) -- node[left]{$\scriptstyle j$} (1,2);
\draw[knot, mid<] (0,0) to[out=-90,in=-90] node[above]{$\scriptstyle k$} (3,0);
\draw[mid<] (3,0) to[out=0,in=0] node[left]{$\scriptstyle\ell$} (3,3);
\draw[mid<] (0,3) to[out=90,in=90] node[below]{$\scriptstyle m$} (3,3);
\draw[mid<] (0,0) to[out=180,in=180] node[right]{$\scriptstyle n$} (0,3);
\draw[mid>] (0,0) to[out=90,in=180] node[above]{$\scriptstyle s$} (1,1);
\draw[mid>] (0,0) to[out=0,in=-90] node[below]{$\scriptstyle t$} (1,1);
\draw[mid>] (3,0) to[out=180,in=-90] node[below]{$\scriptstyle u$} (2,1);
\draw[mid<] (3,0) to[out=90,in=0] node[right]{$\scriptstyle v$} (2,1);
\draw[mid<] (3,3) to[out=-90,in=0] node[right]{$\scriptstyle w$} (2,2);
\draw[mid<] (3,3) to[out=180,in=90] node[left]{$\scriptstyle x$} (2,2);
\draw[mid<] (0,3) to[out=0,in=90] node[right]{$\scriptstyle y$} (1,2);
\draw[mid>]  (0,3) to[out=-90,in=180] node[below]{$\scriptstyle z$} (1,2);
\node at (1.3,1.8) {$\scriptstyle \xi_{1,2}$};
\node at (1.3,.8) {$\scriptstyle \xi_{1,1}$};
\node at (2.3,1.2) {$\scriptstyle \xi_{2,1}$};
\node at (2.3,1.8) {$\scriptstyle \xi_{2,2}$};
\node at (-.3,-.2) {$\scriptstyle \eta^\dag_{1,1}$};
\node at (3.6,-.2) {$\scriptstyle \overline{\eta_{2,1}}$};
\node at (3.3,3.2) {$\scriptstyle \eta^\dag_{2,2}$};
\node at (-.3,3.2) {$\scriptstyle \overline{\eta_{1,2}}$};
}\,.
\end{align*}
Otherwise, $C(\xi,\eta)=0$.
In particular, the $B_p$ are self-adjoint, commuting orthogonal projections.
\end{prop}
Again, this first inner product picture is meant to be interpreted as the inner product on the enriched skein module $\cS^{\textcolor{red}{\cA}}_\cX(\bbD, \textcolor{red}{4},8)$ which can be drawn on a sphere similar to \eqref{eq:SkeinModuleIPShorthand}.
\begin{proof}
To ease the notation, we set
$$
K
:=
\frac{1}{D_\cX\sqrt{\textcolor{red}{d_ad_bd_cd_d}d_sd_td_ud_vd_wd_xd_yd_z}}.
$$
Under the usual Fourier expansion of each of the four corners in the $\{\eta\}$ ONB of $\cH_v=\cS_\cX^{\textcolor{red}{\cA}}(\bbD,\textcolor{red}{1},4)$ (assuming the internal edges match giving a vector in $P_A\cH_{\tot}$),
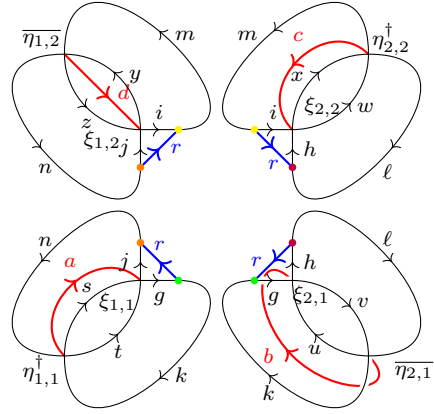
$$
C(\xi,\eta)
=
K
\sum_{\substack{r\in \Irr(\cX)
}}
\frac{1}{ d_r\sqrt{d_gd_hd_id_jd_kd_\ell d_md_n}}
\tikzmath{
\clip (-1,-1) rectangle (5,5);
\draw[thick, red, mid<] (1,1) .. controls ++(135:.8cm) and ++(135:.8cm) .. node[above, yshift=.1cm]{$\scriptstyle a$} (0,0);
\draw[thick, red, mid<] (3,1) .. controls ++(135:.7cm) and ++(135:.8cm) .. (3,0) node[left, xshift=-.1cm]{$\scriptstyle b$} .. controls ++(-45:.8cm) and ++(-45:.7cm) .. (4,0);
\draw[thick, red, mid<] (3,3) .. controls ++(135:.8cm) and ++(135:.8cm) .. node[above, yshift=.1cm]{$\scriptstyle c$} (4,4);
\draw[thick, red, mid<] (1,3) .. controls ++(135:.8cm) and ++(-45:.8cm) .. node[right, xshift=.05cm]{$\scriptstyle d$} (0,4);
\draw[mid>] (1,1) -- node[below]{$\scriptstyle g$} (1.5,1);
\draw[mid>] (1.5,1) .. controls ++(0:1.5cm) and ++(-90:2cm) .. node[right]{$\scriptstyle k$} (0,0);
\draw[mid>] (1,1) -- node[left]{$\scriptstyle j$} (1,1.5);
\draw[mid>] (1,1.5) .. controls ++(90:1.5cm) and ++(180:2cm) .. node[above]{$\scriptstyle n$} (0,0);
\draw[mid>] (1,2.5) -- node[left]{$\scriptstyle j$} (1,3);
\draw[mid>] (0,4) .. controls ++(180:2cm) and ++(-90:1.5cm) .. node[below]{$\scriptstyle n$} (1,2.5);
\draw[mid>] (1,3) -- node[above]{$\scriptstyle i$} (1.5,3);
\draw[mid>] (1.5,3) .. controls ++(0:1.5cm) and ++(90:2cm) .. node[right]{$\scriptstyle m$} (0,4);
\draw[knot, mid>] (2.5,1) -- node[below]{$\scriptstyle g$} (3,1);
\draw[knot, mid<] (2.5,1) .. controls ++(180:1.5cm) and ++(-90:2cm) .. node[below]{$\scriptstyle k$} (4,0);
\draw[mid>] (3,1) -- node[right]{$\scriptstyle h$} (3,1.5);
\draw[mid>] (3,1.5) .. controls ++(90:1.5cm) and ++(0:2cm) .. node[above]{$\scriptstyle \ell$} (4,0);
\draw[mid>] (3,2.5) -- node[right]{$\scriptstyle h$} (3,3);
\draw[mid>] (4,4) .. controls ++(0:2cm) and ++(-90:1.5cm) .. node[below]{$\scriptstyle \ell$} (3,2.5);
\draw[mid>] (2.5,3) -- node[above]{$\scriptstyle i$} (3,3);
\draw[mid>] (4,4) .. controls ++(90:2cm) and ++(180:1.5cm) .. node[left]{$\scriptstyle m$} (2.5,3);
\draw[mid>] (0,0) to[out=90,in=180] node[above]{$\scriptstyle s$} (1,1);
\draw[mid>] (0,0) to[out=0,in=-90] node[below]{$\scriptstyle t$} (1,1);
\draw[mid>] (4,0) to[out=180,in=-90] node[below]{$\scriptstyle u$} (3,1);
\draw[mid<] (4,0) to[out=90,in=0] node[right]{$\scriptstyle v$} (3,1);
\draw[mid<] (4,4) to[out=-90,in=0] node[right]{$\scriptstyle w$} (3,3);
\draw[mid<] (4,4) to[out=180,in=90] node[left]{$\scriptstyle x$} (3,3);
\draw[mid<] (0,4) to[out=0,in=90] node[right]{$\scriptstyle y$} (1,3);
\draw[mid>]  (0,4) to[out=-90,in=180] node[below]{$\scriptstyle z$} (1,3);
\node at (.5,2.85) {$\scriptstyle \xi_{1,2}$};
\node at (.7,.7) {$\scriptstyle \xi_{1,1}$};
\node at (3.25,.8) {$\scriptstyle \xi_{2,1}$};
\node at (3.4,3.3) {$\scriptstyle \xi_{2,2}$};
\node at (-.3,-.2) {$\scriptstyle \eta^\dag_{1,1}$};
\node at (4.6,-.2) {$\scriptstyle \overline{\eta_{2,1}}$};
\node at (4.3,4.2) {$\scriptstyle \eta^\dag_{2,2}$};
\node at (-.3,4.2) {$\scriptstyle \overline{\eta_{1,2}}$};
\draw[thick, blue, mid>] (1.5,1) --node[above]{$\scriptstyle r$} (1,1.5);
\draw[thick, blue, mid>] (1,2.5) --node[right]{$\scriptstyle r$} (1.5,3);
\draw[thick, blue, mid>] (3,1.5) --node[left]{$\scriptstyle r$} (2.5,1);
\draw[thick, blue, mid>] (2.5,3) --node[below]{$\scriptstyle r$} (3,2.5);
\fill[fill=green] (1.5,1) circle (.05cm);
\fill[fill=green] (2.5,1) circle (.05cm);
\fill[fill=purple] (3,1.5) circle (.05cm);
\fill[fill=purple] (3,2.5) circle (.05cm);
\fill[fill=yellow] (1.5,3) circle (.05cm);
\fill[fill=yellow] (2.5,3) circle (.05cm);
\fill[fill=orange] (1,1.5) circle (.05cm);
\fill[fill=orange] (1,2.5) circle (.05cm);
}\,.
$$
Indeed, each of the four closed diagrams represents the inner product in the skein module (up to coeffients) of one of the four corners of \eqref{eq:EnrichedBpAction} with an appropriate $\eta_{i,j}$.
Observe here that there is no longer a sum over $k,\ell,m,n$ as picking particular $\eta_{i,j}$ determines these labels as in the proofs of Lemma \ref{lem:PlaquetteCoefficients} and Theorem \ref{thm:ProjectToSkeinModule}.
We now apply Fact \ref{fact:VaccuumUnzip} to three places amongst these four closed diagrams corresponding to pairs of green, purple, and yellow nodes
to simplify
$$
C(\xi,\eta)
=
K
\sum_{\substack{r\in \Irr(\cX)
}}
\frac{\sqrt{d_r}}{\sqrt{d_jd_n}}
\tikzmath{
\clip (-1.5,-1.5) rectangle (5.5,5.5);
\draw[thick, red, mid<] (1,1) .. controls ++(135:.8cm) and ++(135:.8cm) .. node[above, yshift=.1cm]{$\scriptstyle a$} (0,0);
\draw[thick, red, mid<] (3,1) .. controls ++(135:.7cm) and ++(135:.8cm) .. (3,0) node[left, xshift=-.1cm]{$\scriptstyle b$} .. controls ++(-45:.8cm) and ++(-45:.7cm) .. (4,0);
\draw[thick, red, mid<] (3,3) .. controls ++(135:.8cm) and ++(135:.8cm) .. node[above, yshift=.1cm]{$\scriptstyle c$} (4,4);
\draw[thick, red, mid<] (1,3) .. controls ++(135:.8cm) and ++(-45:.8cm) .. node[right, xshift=.05cm]{$\scriptstyle d$} (0,4);
\draw[knot, mid>] (4,0) to[out=-90,in=-90] node[above]{$\scriptstyle k$} (0,0);
\draw[mid>] (4,4) to[out=0,in=0] node[left]{$\scriptstyle \ell$} (4,0);
\draw[mid>] (4,4) to[out=90,in=90] node[below]{$\scriptstyle m$} (0,4);
\draw[mid>] (1,1.5) .. controls ++(90:1.5cm) and ++(180:2cm) .. node[above]{$\scriptstyle n$} (0,0);
\draw[mid>] (0,4) .. controls ++(180:2cm) and ++(-90:1.5cm) .. node[below]{$\scriptstyle n$} (1,2.5);
\draw[knot, mid>] (1,1) -- node[below]{$\scriptstyle g$} (3,1);
\draw[mid>] (3,1) -- node[right]{$\scriptstyle h$} (3,3);
\draw[mid>] (1,3) -- node[above]{$\scriptstyle i$} (3,3);
\draw[mid>] (1,1) -- node[left]{$\scriptstyle j$} (1,1.5);
\draw[mid>] (1,2.5) -- node[left]{$\scriptstyle j$} (1,3);
\draw[mid>] (0,0) to[out=90,in=180] node[above]{$\scriptstyle s$} (1,1);
\draw[mid>] (0,0) to[out=0,in=-90] node[below]{$\scriptstyle t$} (1,1);
\draw[mid>] (4,0) to[out=180,in=-90] node[below]{$\scriptstyle u$} (3,1);
\draw[mid<] (4,0) to[out=90,in=0] node[right]{$\scriptstyle v$} (3,1);
\draw[mid<] (4,4) to[out=-90,in=0] node[right]{$\scriptstyle w$} (3,3);
\draw[mid<] (4,4) to[out=180,in=90] node[left]{$\scriptstyle x$} (3,3);
\draw[mid<] (0,4) to[out=0,in=90] node[right]{$\scriptstyle y$} (1,3);
\draw[mid>]  (0,4) to[out=-90,in=180] node[below]{$\scriptstyle z$} (1,3);
\node at (.5,2.85) {$\scriptstyle \xi_{1,2}$};
\node at (.7,.7) {$\scriptstyle \xi_{1,1}$};
\node at (3.25,.8) {$\scriptstyle \xi_{2,1}$};
\node at (3.4,3.3) {$\scriptstyle \xi_{2,2}$};
\node at (-.3,-.2) {$\scriptstyle \eta^\dag_{1,1}$};
\node at (4.6,-.2) {$\scriptstyle \overline{\eta_{2,1}}$};
\node at (4.3,4.2) {$\scriptstyle \eta^\dag_{2,2}$};
\node at (-.3,4.2) {$\scriptstyle \overline{\eta_{1,2}}$};
\draw[thick, blue, mid>] (1,2.5) to[out=45, in=180] (1.5,2.8) -- (2.5,2.8) to[out=0,in=90] (2.8,2.5) -- (2.8,1.5) to[out=-90,in=0] (2.5,1.2) -- node[above]{$\scriptstyle r$} (1.5,1.2) to[out=180, in=-45] (1,1.5);
\fill[fill=orange] (1,1.5) circle (.05cm);
\fill[fill=orange] (1,2.5) circle (.05cm);
}\,.
$$
As this diagram is a closed diagram in $\cX$, we can apply the \eqref{eq:Fusion} relation again to `unzip' along the $r$ string to obtain
the claimed formula for $C(\xi,\eta)$.

To see that $B_p$ is self-adjoint, we simply observe $C(\eta, \xi) = \overline{C(\xi,\eta)}$, and thus the matrix representation of $B_p$ is self-adjoint with respect to our distinguished ONB $\{\xi\}$ for the ground state space of $-\sum_\ell A_\ell$.
That $B_p^2=B_p$ and $[B_p,B_q]=0$ for distinct plaquettes $p,q$ follows from \eqref{eq:I=H} and \eqref{eq:Bigon 2} as depicted in \cite[\S5.4]{YanbaiZhang}.
\end{proof}

\begin{lem}[$5=6$]
\label{lem:5=6}
The product of any 5 plaquette operators around a bulk cube is equal to the product of all 6 plaquette operators around the bulk cube, with the exception that the deleted plaquette operator cannot be a boundary $\cX$-plaquette operator.
\end{lem}
\begin{proof}
Suppose we have applied $B_p$ to all but the front and back faces of a cube. 
To do the most complicated case, we assume the front face is a boundary $\cX$-plaquette. 
Since
$$
\frac{1}{D_\cX}
\sum_{r\in \Irr(\cX)}
d_r
\tikzmath{
\draw[thick, red] (1.5,.5) to[out=90, in=-90] (2.5,1.3) --node[right]{$\scriptstyle \Phi(a)$} (2.5,2.2);
\filldraw[blue, thick, fill=gray!30, rounded corners=5pt] (1.6,1.1) rectangle (2.3,1.9);
\node[blue] at (1.4,1.5) {$\scriptstyle r$};
}
=
\frac{1}{D_\cX}
\sum_{s\in \Irr(\cX)}
d_s
\tikzmath{
\draw[thick, red] (1.5,.8) --node[left]{$\scriptstyle \Phi(a)$} (1.5,1.7) to[out=90,in=-90] (2.5,2.5);
\filldraw[blue, thick, fill=gray!30, rounded corners=5pt] (1.7,1.1) rectangle (2.4,1.9);
\node[blue] at (2.6,1.5) {$\scriptstyle s$};
}
\qquad\qquad
\textcolor{red}{\forall a\in \cA},
$$
the $\cX$-plaquette operator absorbs the $\cA$-plaquette operator, i.e. $B_p^\cA B_p^\cX = B_p^\cX$.
Hence if we apply an $\cX$-plaquette operator to the front of the cube, we may insert an extra $\cA$-plaquette operator around it. 
The $\cA$-plaquette operator may then be homotoped around the sides of the bulk cube as follows (suppressing constants).
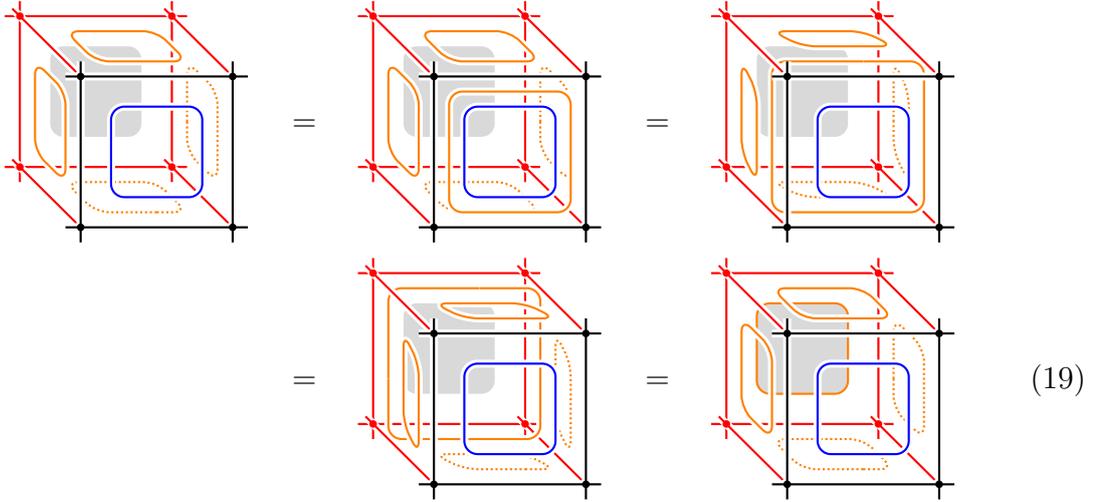
\begin{align}
\tikzmath[scale=2]{
\fill[gray!30, rounded corners=5pt] (.8,1.6) rectangle (1.4,2.2);
\draw[thick, red, xshift=-.4cm, yshift=.4cm] (.9,.9) grid (2.1,2.1);
\foreach \x in {1,2}{
\foreach \y in {1,2}{
\draw[thick, red, knot] (\x,\y) -- ($ (\x,\y) + (-.45,.45) $);
}}
\foreach \x in {.6,1.6}{
\foreach \y in {1.4,2.4}{
\fill[red] (\x,\y) circle (.025cm);
}}
\draw[thick, orange, knot, rounded corners=5pt] (.7,1.7) -- (.7,2.1) -- (.9,1.9) -- (.9,1.3) -- (.7,1.5) -- (.7,1.7);
\draw[thick, orange, densely dotted, rounded corners=5pt] (1.7,1.7) -- (1.7,2.1) -- (1.9,1.9) -- (1.9,1.3) -- (1.7,1.5) -- (1.7,1.7);
\draw[thick, orange, knot, rounded corners=5pt] (1.5,2.1) -- (1.7,2.1) -- (1.5,2.3) -- (.9,2.3) -- (1.1,2.1) -- (1.5,2.1);
\draw[thick, orange, densely dotted, rounded corners=5pt] (1.5,1.1) -- (1.7,1.1) -- (1.5,1.3) -- (.9,1.3) -- (1.1,1.1) -- (1.5,1.1);
\draw[thick, knot] (.9,.9) grid (2.1,2.1);
\foreach \x in {1,2}{
\foreach \y in {1,2}{
\fill (\x,\y) circle (.025cm);
}}
\draw[thick, blue, knot, rounded corners=5pt] (1.2,1.2) rectangle (1.8,1.8);
}
\quad
&=
\quad
\tikzmath[scale=2]{
\fill[gray!30, rounded corners=5pt] (.8,1.6) rectangle (1.4,2.2);
\draw[thick, red, xshift=-.4cm, yshift=.4cm] (.9,.9) grid (2.1,2.1);
\foreach \x in {1,2}{
\foreach \y in {1,2}{
\draw[thick, red, knot] (\x,\y) -- ($ (\x,\y) + (-.45,.45) $);
}}
\foreach \x in {.6,1.6}{
\foreach \y in {1.4,2.4}{
\fill[red] (\x,\y) circle (.025cm);
}}
\draw[thick, orange, knot, rounded corners=5pt] (.7,1.7) -- (.7,2.1) -- (.9,1.9) -- (.9,1.3) -- (.7,1.5) -- (.7,1.7);
\draw[thick, orange, densely dotted, rounded corners=5pt] (1.7,1.7) -- (1.7,2.1) -- (1.9,1.9) -- (1.9,1.3) -- (1.7,1.5) -- (1.7,1.7);
\draw[thick, orange, knot, rounded corners=5pt] (1.5,2.1) -- (1.7,2.1) -- (1.5,2.3) -- (.9,2.3) -- (1.1,2.1) -- (1.5,2.1);
\draw[thick, orange, densely dotted, rounded corners=5pt] (1.5,1.1) -- (1.7,1.1) -- (1.5,1.3) -- (.9,1.3) -- (1.1,1.1) -- (1.5,1.1);
\draw[thick, knot] (.9,.9) grid (2.1,2.1);
\foreach \x in {1,2}{
\foreach \y in {1,2}{
\fill (\x,\y) circle (.025cm);
}}
\draw[thick, orange, knot, rounded corners=5pt] (1.1,1.1) rectangle (1.9,1.9);
\draw[thick, blue, knot, rounded corners=5pt] (1.2,1.2) rectangle (1.8,1.8);
}
\quad
=
\quad
\tikzmath[scale=2]{
\fill[gray!30, rounded corners=5pt] (.8,1.6) rectangle (1.4,2.2);
\draw[thick, orange, knot, rounded corners=5pt] (.9,1.5) -- (.9,2.1) -- (1.5,2.1);
\draw[thick, red, xshift=-.4cm, yshift=.4cm] (.9,.9) grid (2.1,2.1);
\foreach \x in {1,2}{
\foreach \y in {1,2}{
\draw[thick, red, knot] (\x,\y) -- ($ (\x,\y) + (-.45,.45) $);
}}
\foreach \x in {.6,1.6}{
\foreach \y in {1.4,2.4}{
\fill[red] (\x,\y) circle (.025cm);
}}
\draw[thick, orange, knot, rounded corners=5pt] (.7,1.7) -- (.7,2.1) -- (.8,1.9) -- (.8,1.3) -- (.7,1.5) -- (.7,1.7);
\draw[thick, orange, densely dotted, rounded corners=5pt] (1.7,1.7) -- (1.7,2.1) -- (1.8,1.9) -- (1.8,1.3) -- (1.7,1.5) -- (1.7,1.7);
\draw[thick, orange, knot, rounded corners=5pt] (1.5,2.2) -- (1.7,2.2) -- (1.5,2.3) -- (.9,2.3) -- (1.1,2.2) -- (1.5,2.2);
\draw[thick, orange, densely dotted, rounded corners=5pt] (1.5,1.2) -- (1.7,1.2) -- (1.5,1.3) -- (.9,1.3) -- (1.1,1.2) -- (1.5,1.2);
\draw[thick, orange, knot, rounded corners=5pt] (1.5,2.1) -- (1.9,2.1) -- (1.9,1.1) -- (.9,1.1) -- (.9,1.5);
\draw[thick, knot] (.9,.9) grid (2.1,2.1);
\foreach \x in {1,2}{
\foreach \y in {1,2}{
\fill (\x,\y) circle (.025cm);
}}
\draw[thick, blue, knot, rounded corners=5pt] (1.2,1.2) rectangle (1.8,1.8);
}
\notag
\\&=
\quad
\tikzmath[scale=2]{
\fill[gray!30, rounded corners=5pt] (.8,1.6) rectangle (1.4,2.2);
\draw[thick, orange, knot, rounded corners=5pt, xshift=-.2cm, yshift=.2cm] (.9,1.5) -- (.9,2.1) -- (1.5,2.1);
\draw[thick, red, xshift=-.4cm, yshift=.4cm] (.9,.9) grid (2.1,2.1);
\foreach \x in {1,2}{
\foreach \y in {1,2}{
\draw[thick, red, knot] (\x,\y) -- ($ (\x,\y) + (-.45,.45) $);
}}
\foreach \x in {.6,1.6}{
\foreach \y in {1.4,2.4}{
\fill[red] (\x,\y) circle (.025cm);
}}
\draw[thick, orange, knot, rounded corners=5pt, xshift=-.2cm, yshift=.2cm] (1.5,2.1) -- (1.9,2.1) -- (1.9,1.1) -- (.9,1.1) -- (.9,1.5);
\draw[thick, orange, knot, rounded corners=5pt, xshift=.1cm, yshift=-.1cm] (.7,1.7) -- (.7,2.1) -- (.8,1.9) -- (.8,1.3) -- (.7,1.5) -- (.7,1.7);
\draw[thick, orange, densely dotted, rounded corners=5pt, xshift=.1cm, yshift=-.1cm] (1.7,1.7) -- (1.7,2.1) -- (1.8,1.9) -- (1.8,1.3) -- (1.7,1.5) -- (1.7,1.7);
\draw[thick, orange, knot, rounded corners=5pt, xshift=.1cm, yshift=-.1cm] (1.5,2.2) -- (1.7,2.2) -- (1.5,2.3) -- (.9,2.3) -- (1.1,2.2) -- (1.5,2.2);
\draw[thick, orange, densely dotted, rounded corners=5pt, xshift=.1cm, yshift=-.1cm] (1.5,1.2) -- (1.7,1.2) -- (1.5,1.3) -- (.9,1.3) -- (1.1,1.2) -- (1.5,1.2);
\draw[thick, knot] (.9,.9) grid (2.1,2.1);
\foreach \x in {1,2}{
\foreach \y in {1,2}{
\fill (\x,\y) circle (.025cm);
}}
\draw[thick, blue, knot, rounded corners=5pt] (1.2,1.2) rectangle (1.8,1.8);
}
\quad
=
\quad
\tikzmath[scale=2]{
\filldraw[thick, orange, fill=gray!30, rounded corners=5pt] (.8,1.6) rectangle (1.4,2.2);
\draw[thick, red, xshift=-.4cm, yshift=.4cm] (.9,.9) grid (2.1,2.1);
\foreach \x in {1,2}{
\foreach \y in {1,2}{
\draw[thick, red, knot] (\x,\y) -- ($ (\x,\y) + (-.45,.45) $);
}}
\foreach \x in {.6,1.6}{
\foreach \y in {1.4,2.4}{
\fill[red] (\x,\y) circle (.025cm);
}}
\draw[thick, orange, knot, rounded corners=5pt] (.7,1.7) -- (.7,2.1) -- (.9,1.9) -- (.9,1.3) -- (.7,1.5) -- (.7,1.7);
\draw[thick, orange, densely dotted, rounded corners=5pt] (1.7,1.7) -- (1.7,2.1) -- (1.9,1.9) -- (1.9,1.3) -- (1.7,1.5) -- (1.7,1.7);
\draw[thick, orange, knot, rounded corners=5pt] (1.5,2.1) -- (1.7,2.1) -- (1.5,2.3) -- (.9,2.3) -- (1.1,2.1) -- (1.5,2.1);
\draw[thick, orange, densely dotted, rounded corners=5pt] (1.5,1.1) -- (1.7,1.1) -- (1.5,1.3) -- (.9,1.3) -- (1.1,1.1) -- (1.5,1.1);
\draw[thick, knot] (.9,.9) grid (2.1,2.1);
\foreach \x in {1,2}{
\foreach \y in {1,2}{
\fill (\x,\y) circle (.025cm);
}}
\draw[thick, blue, knot, rounded corners=5pt] (1.2,1.2) rectangle (1.8,1.8);
}
\label{eq:5=6}
\end{align}

The case when the front face is not an $\cX$-boundary plaquette follows from the same manipulation, using that $\cA$-plaquette operator is an idempotent. 
\end{proof}

Suppose we have a 3-ball $\cB$ in our lattice, which may intersect the $\cX$-boundary.
We have an \emph{evaluation map} $\eval$ from the tensor product over the $\cH_v$ corresponding to sites $v$ in the ball to the $\cA$-enriched skein module with the appropriate number of boundary points.
We let $N_p(\partial \cB)$ denote the number of plaquettes in the interior 
of $\cB$ on the $\cX$-boundary,
and we let $N_p(\cB^\circ)$ denote the number of plaquettes in the interior of $\cB$ which are not on the $\cX$-boundary.
(The $B_p$ terms for $p\in \partial \cB$ sum over simples in $\cX$, while the $B_p$ terms for $p\in \cB^\circ$ sum over simples in $\cA$.)
We let $N_c(\cB^\circ)$ denote the number of cubes in the interior of $\cB$.
The following proposition is the enriched analog of Theorem~\ref{thm:ProjectToSkeinModule} above.

\begin{thm}
Denoting the global dimensions of $\cA,\cX$ by $D_\cA, D_\cX$ respectively,
$$
\prod_{p\subset \cB} B_p= 
D_\cA^{-\big(N_p(\cB^\circ)-N_c(\cB^\circ)\big)}
D_\cX^{-N_p(\partial\cB)}
\cdot \eval^\dag\circ \eval.
$$ 
\end{thm}
\begin{proof}
The proof proceeds by adjoining adjacent plaquettes by induction.
We provide the detailed argument for 2 adjacent plaquettes on the $\cX$-boundary, and we leave the remaining details to the reader.
That is, we will prove that the matrix coefficients of $B_oB_p$ for adjacent plaquettes $o,p$ are given by
\begin{align*}
C_{o,p}(\xi',\xi)
&=
\left\langle
\tikzmath{
\begin{scope}
\draw[step=1.0,black, knot] (.5,.5) grid (3.5,2.5);
\end{scope}
\foreach \i in {1,2,3} {
\foreach \j in {1,2} {
        \draw[thick, red] ($ (\i,\j) $) -- ($ (\i,\j) + (-.3,.3)$);
        \node at ($ (\i,\j) + (.25,-.15) $) {$\scriptstyle \xi_{\i,\j}$};
}}
\node[red] at (.7,1.5) {$\scriptstyle a$};
\node[red] at (1.7,1.5) {$\scriptstyle b$};
\node[red] at (2.7,1.5) {$\scriptstyle c$};
\node[red] at (2.7,2.5) {$\scriptstyle d$};
\node[red] at (1.7,2.5) {$\scriptstyle e$};
\node[red] at (.7,2.5) {$\scriptstyle f$};
\node at (.3,1) {$\scriptstyle q$};
\node at (1,.3) {$\scriptstyle r$};
\node at (2,.3) {$\scriptstyle s$};
\node at (3,.3) {$\scriptstyle t$};
\node at (3.7,1) {$\scriptstyle u$};
\node at (3.7,2) {$\scriptstyle v$};
\node at (3,2.7) {$\scriptstyle w$};
\node at (2,2.7) {$\scriptstyle x$};
\node at (1,2.7) {$\scriptstyle y$};
\node at (.3,2) {$\scriptstyle z$};
\node at (1.4,1.4) {$\scriptstyle o$};
\node at (2.4,1.4) {$\scriptstyle p$};
\node at (1.15,1.5) {$\scriptstyle g$};
\node at (2.15,1.5) {$\scriptstyle h$};
\node at (3.15,1.5) {$\scriptstyle i$};
\node at (1.7,.9) {$\scriptstyle j$};
\node at (2.7,.9) {$\scriptstyle k$};
\node at (2.7,1.9) {$\scriptstyle \ell$};
\node at (1.7,1.9) {$\scriptstyle m$};
}
\Bigg|
B_o B_p
\tikzmath{
\begin{scope}
\draw[step=1.0,black, knot] (.5,.5) grid (3.5,2.5);
\end{scope}
\foreach \i in {1,2,3} {
\foreach \j in {1,2} {
        \draw[thick, red] ($ (\i,\j) $) -- ($ (\i,\j) + (-.3,.3)$);
        \node at ($ (\i,\j) + (.25,-.15) $) {$\scriptstyle \xi'_{\i,\j}$};
}}
\node[red] at (.7,1.5) {$\scriptstyle a$};
\node[red] at (1.7,1.5) {$\scriptstyle b$};
\node[red] at (2.7,1.5) {$\scriptstyle c$};
\node[red] at (2.7,2.5) {$\scriptstyle d$};
\node[red] at (1.7,2.5) {$\scriptstyle e$};
\node[red] at (.7,2.5) {$\scriptstyle f$};
\node at (.3,1) {$\scriptstyle q$};
\node at (1,.3) {$\scriptstyle r$};
\node at (2,.3) {$\scriptstyle s$};
\node at (3,.3) {$\scriptstyle t$};
\node at (3.7,1) {$\scriptstyle u$};
\node at (3.7,2) {$\scriptstyle v$};
\node at (3,2.7) {$\scriptstyle w$};
\node at (2,2.7) {$\scriptstyle x$};
\node at (1,2.7) {$\scriptstyle y$};
\node at (.3,2) {$\scriptstyle z$};
\node at (1.4,1.4) {$\scriptstyle o$};
\node at (2.4,1.4) {$\scriptstyle p$};
\node at (1.15,1.5) {$\scriptstyle g'$};
\node at (2.15,1.5) {$\scriptstyle h'$};
\node at (3.15,1.5) {$\scriptstyle i'$};
\node at (1.7,.9) {$\scriptstyle j'$};
\node at (2.7,.9) {$\scriptstyle k'$};
\node at (2.7,1.9) {$\scriptstyle \ell'$};
\node at (1.7,1.9) {$\scriptstyle m'$};
}
\right\rangle
\\&=
\frac{1}{D^2_\cX\sqrt{\textcolor{red}{d_a\cdots d_f}d_q\cdots d_z}}
\tikzmath{
\begin{scope}
\clip (1.2,1.8) rectangle (4.5,3.6);
\begin{scope}[shift={(-.3,.3)}]
\draw[step=1.0,black] (0.5,0.5) grid (5,4);
\end{scope}
\draw[step=1.0,black, knot] (1,1) grid (5,4);
\end{scope}
\foreach \j in {2,3} {
\draw (1.2,\j) arc (270:90:.15cm);
\draw (4.5,\j) arc (-90:90:.15cm);
};
\foreach \i in {2,3,4} {
    \draw (\i,3.6) arc (0:180:.15cm);
    \draw (\i,1.8) arc (0:-180:.15cm);
};
\foreach \i in {1,2,3} {
\foreach \j in {1,2} {
        \draw[thick, red] ($ (\i,\j) + (1,1)$) -- ($ (\i,\j) + (.7,1.3)$);
        \node at ($ (\i,\j) + (1.25,.85) $) {$\scriptstyle \xi'_{\i,\j}$};
        \node at ($ (\i,\j) + (.45,1.45) $) {$\scriptstyle \overline{\xi_{\i,\j}}$};
}}
}\,,
\end{align*}
where $\xi,\xi'$ are choices of 6 vectors from our distinguished ONB of $\cH_v$.
Again, we have suppressed shadings on squares for readability.
This will exactly show that for every $\xi, \xi'$ in an ONB for 
$\cS^{\textcolor{red}{\cA}}_\cX(\bbD, \textcolor{red}{6},10)$,
$$
\langle \xi | B_oB_p \xi'\rangle_{\bigotimes_{v\in \cB} \cH_v} 
= 
D_\cX^{-2}
\langle \eval\xi | \eval\xi'\rangle_{\cS^{\textcolor{red}{\cA}}_\cX(\bbD, \textcolor{red}{6},10)}
\qquad\Longleftrightarrow\qquad
B_oB_p=D_\cX^{-2}\eval^\dag\circ \eval
$$
on the ball $\cB$ which only contains the 2 boundary plaquettes $o$ and $p$.

To ease the notation, we set
\begin{align*}
K:= \frac{1}{\sqrt{\textcolor{red}{d_a\cdots d_f}d_q\cdots d_z}},
\displaybreak[1]\\
K_o&:=\frac{\sqrt{d_jd_hd_md_g}}{\sqrt{d_{j'}d_{h''}d_{m'}d_{g'}}}
&
K_o'&:=\frac{1}{\sqrt{d_jd_{j'}d_{h}d_md_{m'}d_gd_{g'}}},
\displaybreak[1]\\
K_p(h'')&:=\frac{\sqrt{d_kd_id_\ell d_{h''}}}{\sqrt{d_{k'}d_{i'}d_{\ell'}d_{h'}}}
&
K_p'&:=\frac{1}{\sqrt{d_kd_{k'}d_id_{i'}d_\ell d_{\ell'}d_{h'}}}.
\end{align*}
In the calculation below, we label strings and vertices as much as possible, but sometimes we omit labels that can be determined from the previous diagram when the diagram is getting very intricate.
Moreover, since we have fixed $\eta_{i,j}\in \operatorname{ONB}$, some of the sums over simples from the plaquette operators collapse as in the proofs of Lemma \ref{lem:PlaquetteCoefficients}, Theorem \ref{thm:ProjectToSkeinModule}, and Proposition \ref{prop:MatrixCoefficientsOfBp}.
Again, we suppress shaded squares for readability.
\begin{align*}
C_{o,p}(\xi',\xi)
&=
\left\langle
\tikzmath{
\begin{scope}
\draw[step=1.0,black, knot] (.5,.5) grid (3.5,2.5);
\end{scope}
\foreach \i in {1,2,3} {
\foreach \j in {1,2} {
        \draw[thick, red] ($ (\i,\j) $) -- ($ (\i,\j) + (-.3,.3)$);
        \node at ($ (\i,\j) + (.25,-.15) $) {$\scriptstyle \xi_{\i,\j}$};
}}
\node[red] at (.7,1.5) {$\scriptstyle a$};
\node[red] at (1.7,1.5) {$\scriptstyle b$};
\node[red] at (2.7,1.5) {$\scriptstyle c$};
\node[red] at (2.7,2.5) {$\scriptstyle d$};
\node[red] at (1.7,2.5) {$\scriptstyle e$};
\node[red] at (.7,2.5) {$\scriptstyle f$};
\node at (.3,1) {$\scriptstyle q$};
\node at (1,.3) {$\scriptstyle r$};
\node at (2,.3) {$\scriptstyle s$};
\node at (3,.3) {$\scriptstyle t$};
\node at (3.7,1) {$\scriptstyle u$};
\node at (3.7,2) {$\scriptstyle v$};
\node at (3,2.7) {$\scriptstyle w$};
\node at (2,2.7) {$\scriptstyle x$};
\node at (1,2.7) {$\scriptstyle y$};
\node at (.3,2) {$\scriptstyle z$};
\node at (1.4,1.4) {$\scriptstyle o$};
\node at (2.4,1.4) {$\scriptstyle p$};
\node at (1.15,1.5) {$\scriptstyle g$};
\node at (2.15,1.5) {$\scriptstyle h$};
\node at (3.15,1.5) {$\scriptstyle i$};
\node at (1.7,.9) {$\scriptstyle j$};
\node at (2.7,.9) {$\scriptstyle k$};
\node at (2.7,1.9) {$\scriptstyle \ell$};
\node at (1.7,1.9) {$\scriptstyle m$};
}
\Bigg|
B_o B_p
\tikzmath{
\begin{scope}
\draw[step=1.0,black, knot] (.5,.5) grid (3.5,2.5);
\end{scope}
\foreach \i in {1,2,3} {
\foreach \j in {1,2} {
        \draw[thick, red] ($ (\i,\j) $) -- ($ (\i,\j) + (-.3,.3)$);
        \node at ($ (\i,\j) + (.25,-.15) $) {$\scriptstyle \xi'_{\i,\j}$};
}}
\node[red] at (.7,1.5) {$\scriptstyle a$};
\node[red] at (1.7,1.5) {$\scriptstyle b$};
\node[red] at (2.7,1.5) {$\scriptstyle c$};
\node[red] at (2.7,2.5) {$\scriptstyle d$};
\node[red] at (1.7,2.5) {$\scriptstyle e$};
\node[red] at (.7,2.5) {$\scriptstyle f$};
\node at (.3,1) {$\scriptstyle q$};
\node at (1,.3) {$\scriptstyle r$};
\node at (2,.3) {$\scriptstyle s$};
\node at (3,.3) {$\scriptstyle t$};
\node at (3.7,1) {$\scriptstyle u$};
\node at (3.7,2) {$\scriptstyle v$};
\node at (3,2.7) {$\scriptstyle w$};
\node at (2,2.7) {$\scriptstyle x$};
\node at (1,2.7) {$\scriptstyle y$};
\node at (.3,2) {$\scriptstyle z$};
\node at (1.4,1.4) {$\scriptstyle o$};
\node at (2.4,1.4) {$\scriptstyle p$};
\node at (1.15,1.5) {$\scriptstyle g'$};
\node at (2.15,1.5) {$\scriptstyle h'$};
\node at (3.15,1.5) {$\scriptstyle i'$};
\node at (1.7,.9) {$\scriptstyle j'$};
\node at (2.7,.9) {$\scriptstyle k'$};
\node at (2.7,1.9) {$\scriptstyle \ell'$};
\node at (1.7,1.9) {$\scriptstyle m'$};
}
\right\rangle
\displaybreak[1]\\&=
\frac{1}{\cD_\cX^2}
\sum_{n',h''\in \Irr(\cX)} 
\frac{K_p(h'')}{d_{n'}}
\left\langle
\tikzmath{
\begin{scope}
\draw[step=1.0,black, knot] (.5,.5) grid (3.5,2.5);
\end{scope}
\foreach \i in {1,2,3} {
\foreach \j in {1,2} {
        \draw[thick, red] ($ (\i,\j) $) -- ($ (\i,\j) + (-.3,.3)$);
}}
\node[red] at (.7,1.5) {$\scriptstyle a$};
\node[red] at (1.7,1.5) {$\scriptstyle b$};
\node[red] at (2.7,1.5) {$\scriptstyle c$};
\node[red] at (2.7,2.5) {$\scriptstyle d$};
\node[red] at (1.7,2.5) {$\scriptstyle e$};
\node[red] at (.7,2.5) {$\scriptstyle f$};
\node at (.3,1) {$\scriptstyle q$};
\node at (1,.3) {$\scriptstyle r$};
\node at (2,.3) {$\scriptstyle s$};
\node at (3,.3) {$\scriptstyle t$};
\node at (3.7,1) {$\scriptstyle u$};
\node at (3.7,2) {$\scriptstyle v$};
\node at (3,2.7) {$\scriptstyle w$};
\node at (2,2.7) {$\scriptstyle x$};
\node at (1,2.7) {$\scriptstyle y$};
\node at (.3,2) {$\scriptstyle z$};
\node at (1.15,1.5) {$\scriptstyle g$};
\node at (2.15,1.5) {$\scriptstyle h$};
\node at (3.15,1.5) {$\scriptstyle i$};
\node at (1.5,.9) {$\scriptstyle j$};
\node at (2.5,.9) {$\scriptstyle k$};
\node at (2.5,1.9) {$\scriptstyle \ell$};
\node at (1.5,1.9) {$\scriptstyle m$};
}
\Bigg|
B_o
\tikzmath{
\begin{scope}
\draw[step=1.0,black, knot] (.5,.5) grid (3.5,2.5);
\end{scope}
\foreach \i in {1,2,3} {
\foreach \j in {1,2} {
        \draw[thick, red] ($ (\i,\j) $) -- ($ (\i,\j) + (-.3,.3)$);
}}
\draw[thick, cyan] (2.3,1) -- (2,1.3);
\draw[knot, thick, cyan] (2.7,1) -- (3,1.3);
\draw[thick, cyan] (2.3,2) -- (2,1.7);
\draw[thick, cyan] (2.7,2) -- (3,1.7);
\node[cyan] at (2.3,1.3) {$\scriptstyle n'$};
\fill[fill=green] (2.3,1) circle (.05cm);
\fill[fill=green] (2.7,1) circle (.05cm);
\fill[fill=purple] (3,1.3) circle (.05cm);
\fill[fill=purple] (3,1.7) circle (.05cm);
\fill[fill=yellow] (2.3,2) circle (.05cm);
\fill[fill=yellow] (2.7,2) circle (.05cm);
\fill[fill=orange] (2,1.3) circle (.05cm);
\fill[fill=orange] (2,1.7) circle (.05cm);
\node[red] at (.7,1.5) {$\scriptstyle a$};
\node[red] at (1.6,1.4) {$\scriptstyle b$};
\node[red] at (2.7,1.5) {$\scriptstyle c$};
\node[red] at (2.7,2.5) {$\scriptstyle d$};
\node[red] at (1.7,2.5) {$\scriptstyle e$};
\node[red] at (.7,2.5) {$\scriptstyle f$};
\node at (.3,1) {$\scriptstyle q$};
\node at (1,.3) {$\scriptstyle r$};
\node at (2,.3) {$\scriptstyle s$};
\node at (3,.3) {$\scriptstyle t$};
\node at (3.7,1) {$\scriptstyle u$};
\node at (3.7,2) {$\scriptstyle v$};
\node at (3,2.7) {$\scriptstyle w$};
\node at (2,2.7) {$\scriptstyle x$};
\node at (1,2.7) {$\scriptstyle y$};
\node at (.3,2) {$\scriptstyle z$};
\node at (1.85,1.15) {$\scriptstyle h'$};
\node at (1.85,1.5) {$\scriptstyle h''$};
\node at (1.85,1.85) {$\scriptstyle h'$};
\node at (3.15,1.85) {$\scriptstyle i'$};
\node at (3.15,1.5) {$\scriptstyle i$};
\node at (3.15,1.15) {$\scriptstyle i'$};
\node at (2.15,.9) {$\scriptstyle k'$};
\node at (2.5,.9) {$\scriptstyle k$};
\node at (2.85,.9) {$\scriptstyle k'$};
\node at (2.15,2.1) {$\scriptstyle \ell'$};
\node at (2.5,2.1) {$\scriptstyle \ell$};
\node at (2.85,2.1) {$\scriptstyle \ell'$};
}
\right\rangle
\displaybreak[1]\\&=
\frac{1}{\cD_\cX^2}
\sum_{n,n',h''\in \Irr(\cX)} 
\frac{K_oK_p(h'')}{d_nd_{n'}}
\left\langle
\tikzmath{
\begin{scope}
\draw[step=1.0,black, knot] (.5,.5) grid (3.5,2.5);
\end{scope}
\foreach \i in {1,2,3} {
\foreach \j in {1,2} {
        \draw[thick, red] ($ (\i,\j) $) -- ($ (\i,\j) + (-.3,.3)$);
}}
\node[red] at (.7,1.5) {$\scriptstyle a$};
\node[red] at (1.7,1.5) {$\scriptstyle b$};
\node[red] at (2.7,1.5) {$\scriptstyle c$};
\node[red] at (2.7,2.5) {$\scriptstyle d$};
\node[red] at (1.7,2.5) {$\scriptstyle e$};
\node[red] at (.7,2.5) {$\scriptstyle f$};
\node at (.3,1) {$\scriptstyle q$};
\node at (1,.3) {$\scriptstyle r$};
\node at (2,.3) {$\scriptstyle s$};
\node at (3,.3) {$\scriptstyle t$};
\node at (3.7,1) {$\scriptstyle u$};
\node at (3.7,2) {$\scriptstyle v$};
\node at (3,2.7) {$\scriptstyle w$};
\node at (2,2.7) {$\scriptstyle x$};
\node at (1,2.7) {$\scriptstyle y$};
\node at (.3,2) {$\scriptstyle z$};
\node at (1.15,1.5) {$\scriptstyle g$};
\node at (2.15,1.5) {$\scriptstyle h$};
\node at (3.15,1.5) {$\scriptstyle i$};
\node at (1.5,.9) {$\scriptstyle j$};
\node at (2.5,.9) {$\scriptstyle k$};
\node at (2.5,1.9) {$\scriptstyle \ell$};
\node at (1.5,1.9) {$\scriptstyle m$};
}
\Bigg|
\tikzmath{
\begin{scope}
\draw[step=1.0,black, knot] (.5,.5) grid (3.5,2.5);
\end{scope}
\foreach \i in {1,2,3} {
\foreach \j in {1,2} {
        \draw[thick, red] ($ (\i,\j) $) -- ($ (\i,\j) + (-.3,.3)$);
}}
\draw[thick, blue] (1.3,1) -- (1,1.3);
\draw[knot, thick, blue] (1.7,1) -- (2,1.3);
\draw[thick, blue] (1.3,2) -- (1,1.7);
\draw[thick, blue] (1.7,2) -- (2,1.7);
\fill[fill=green] (1.3,1) circle (.05cm);
\fill[fill=green] (1.7,1) circle (.05cm);
\fill[fill=purple] (2,1.3) circle (.05cm);
\fill[fill=purple] (2,1.7) circle (.05cm);
\fill[fill=yellow] (1.3,2) circle (.05cm);
\fill[fill=yellow] (1.7,2) circle (.05cm);
\fill[fill=orange] (1,1.3) circle (.05cm);
\fill[fill=orange] (1,1.7) circle (.05cm);
\node[blue] at (1.3,1.3) {$\scriptstyle n$};
\draw[thick, cyan] (2.3,1) -- (2,1.2);
\draw[knot, thick, cyan] (2.7,1) -- (3,1.3);
\draw[thick, cyan] (2.3,2) -- (2,1.8);
\draw[thick, cyan] (2.7,2) -- (3,1.7);
\node[cyan] at (2.3,1.3) {$\scriptstyle n'$};
\node[red] at (.7,1.5) {$\scriptstyle a$};
\node[red] at (1.6,1.4) {$\scriptstyle b$};
\node[red] at (2.7,1.5) {$\scriptstyle c$};
\node[red] at (2.7,2.5) {$\scriptstyle d$};
\node[red] at (1.7,2.5) {$\scriptstyle e$};
\node[red] at (.7,2.5) {$\scriptstyle f$};
\node at (.3,1) {$\scriptstyle q$};
\node at (1,.3) {$\scriptstyle r$};
\node at (2,.3) {$\scriptstyle s$};
\node at (3,.3) {$\scriptstyle t$};
\node at (3.7,1) {$\scriptstyle u$};
\node at (3.7,2) {$\scriptstyle v$};
\node at (3,2.7) {$\scriptstyle w$};
\node at (2,2.7) {$\scriptstyle x$};
\node at (1,2.7) {$\scriptstyle y$};
\node at (.3,2) {$\scriptstyle z$};
\node at (.85,1.85) {$\scriptstyle g'$};
\node at (.85,1.5) {$\scriptstyle g$};
\node at (.85,1.15) {$\scriptstyle g'$};
\node at (1.85,1.5) {$\scriptstyle h$};
\node at (3.15,1.85) {$\scriptstyle i'$};
\node at (3.15,1.5) {$\scriptstyle i$};
\node at (3.15,1.15) {$\scriptstyle i'$};
\node at (1.15,.9) {$\scriptstyle j'$};
\node at (1.5,.9) {$\scriptstyle j$};
\node at (1.85,.9) {$\scriptstyle j'$};
\node at (2.15,.9) {$\scriptstyle k'$};
\node at (2.5,.9) {$\scriptstyle k$};
\node at (2.85,.9) {$\scriptstyle k'$};
\node at (2.15,2.1) {$\scriptstyle \ell'$};
\node at (2.5,2.1) {$\scriptstyle \ell$};
\node at (2.85,2.1) {$\scriptstyle \ell'$};
\node at (1.15,2.1) {$\scriptstyle m'$};
\node at (1.5,2.1) {$\scriptstyle m$};
\node at (1.85,2.1) {$\scriptstyle m'$};
}
\right\rangle
\displaybreak[1]\\&=
\frac{K}{\cD_\cX^2}
\sum_{n,n',h''\in \Irr(\cX)} 
\frac{K_o' K_p'}{d_nd_{n'}}
\tikzmath{
\clip (-1,-1) rectangle (9,5);
\draw[thick, red, mid<] (1,1) .. controls ++(135:.8cm) and ++(135:.8cm) .. node[above, yshift=.1cm]{$\scriptstyle a$} (0,0);
\draw[thick, red, mid<] (3,1) .. controls ++(135:.7cm) and ++(135:.8cm) .. (3,0) node[left, xshift=-.1cm]{$\scriptstyle b$} .. controls ++(-45:.8cm) and ++(-45:.7cm) .. (4,0);
\draw[thick, red, mid<] (3,3) .. controls ++(135:.8cm) and ++(135:.8cm) .. node[above, yshift=.1cm]{$\scriptstyle e$} (4,4);
\draw[thick, red, mid<] (1,3) .. controls ++(135:.8cm) and ++(-45:.8cm) .. node[right, xshift=.05cm]{$\scriptstyle f$} (0,4);
\draw[mid>] (1,1) -- node[below]{$\scriptstyle j'$} (1.5,1);
\draw[mid>] (1.5,1) .. controls ++(0:1.5cm) and ++(-90:2cm) .. node[right]{$\scriptstyle j$} (0,0);
\draw[mid>] (1,1) -- node[left]{$\scriptstyle g'$} (1,1.5);
\draw[mid>] (1,1.5) .. controls ++(90:1.5cm) and ++(180:2cm) .. node[above]{$\scriptstyle g$} (0,0);
\draw[mid>] (1,2.5) -- node[left]{$\scriptstyle g'$} (1,3);
\draw[mid>] (0,4) .. controls ++(180:2cm) and ++(-90:1.5cm) .. node[below]{$\scriptstyle g$} (1,2.5);
\draw[mid>] (1,3) -- node[above]{$\scriptstyle m'$} (1.5,3);
\draw[mid>] (1.5,3) .. controls ++(0:1.5cm) and ++(90:2cm) .. node[right]{$\scriptstyle m$} (0,4);
\draw[knot, mid>] (2.5,1) -- node[below]{$\scriptstyle j'$} (3,1);
\draw[knot, mid<] (2.5,1) .. controls ++(180:1.5cm) and ++(-90:2cm) .. node[below]{$\scriptstyle j$} (4,0);
\draw[mid>] (3,1) -- (3,1.5); 
\draw[mid>] (3,1.5) .. controls ++(90:1.5cm) and ++(0:2cm) .. node[above]{$\scriptstyle h$} (4,0);
\draw[mid>] (3,2.5) -- (3,3); 
\draw[mid>] (4,4) .. controls ++(0:2cm) and ++(-90:1.5cm) .. node[below]{$\scriptstyle h$} (3,2.5);
\draw[mid>] (2.5,3) -- node[above]{$\scriptstyle m'$} (3,3);
\draw[mid>] (4,4) .. controls ++(90:2cm) and ++(180:1.5cm) .. node[left]{$\scriptstyle m$} (2.5,3);
\draw[mid>] (0,0) to[out=90,in=180] node[above]{$\scriptstyle q$} (1,1);
\draw[mid>] (0,0) to[out=0,in=-90] node[below]{$\scriptstyle r$} (1,1);
\draw[mid>] (4,0) to[out=180,in=-90] node[below]{$\scriptstyle s$} (3,1);
\draw[mid<] (4,0) to[out=90,in=0] node[right]{$\scriptstyle k$} (3.5,1) -- (3,1);
\draw[mid<] (4,4) to[out=180,in=90] node[left]{$\scriptstyle x$} (3,3);
\draw[mid<] (4,4) to[out=-90,in=0] node[right]{$\scriptstyle \ell$} (3.5,3) -- (3,3);
\draw[mid<] (0,4) to[out=0,in=90] node[right]{$\scriptstyle y$} (1,3);
\draw[mid>]  (0,4) to[out=-90,in=180] node[below]{$\scriptstyle z$} (1,3);
\node at (.4,2.85) {$\scriptstyle \xi_{1,2}$};
\node at (.7,.7) {$\scriptstyle \xi_{1,1}$};
\node at (3.25,.8) {$\scriptstyle \xi_{2,1}$};
\node at (3.4,3.3) {$\scriptstyle \xi_{2,2}$};
\node at (-.3,-.2) {$\scriptstyle \eta^\dag_{1,1}$};
\node at (4.6,-.2) {$\scriptstyle \overline{\eta_{2,1}}$};
\node at (4.3,4.2) {$\scriptstyle \eta^\dag_{2,2}$};
\node at (-.3,4.2) {$\scriptstyle \overline{\eta_{1,2}}$};
\draw[thick, blue, mid>] (1.5,1) --node[above]{$\scriptstyle n$} (1,1.5);
\draw[thick, blue, mid>] (1,2.5) --node[right]{$\scriptstyle n$} (1.5,3);
\draw[thick, blue, mid>] (3,1.5) --node[left]{$\scriptstyle n$} (2.5,1);
\draw[thick, blue, mid>] (2.5,3) --node[below]{$\scriptstyle n$} (3,2.5);
\draw[thick, cyan, mid>] (3.5,1) --node[above]{$\scriptstyle n'$} (3,1.3);
\draw[thick, cyan, mid>] (3,2.7) --node[right]{$\scriptstyle n'$} (3.5,3);
\fill[fill=green] (1.5,1) circle (.05cm);
\fill[fill=green] (2.5,1) circle (.05cm);
\fill[fill=purple] (3,1.5) circle (.05cm);
\fill[fill=purple] (3,2.5) circle (.05cm);
\fill[fill=yellow] (1.5,3) circle (.05cm);
\fill[fill=yellow] (2.5,3) circle (.05cm);
\fill[fill=orange] (1,1.5) circle (.05cm);
\fill[fill=orange] (1,2.5) circle (.05cm);
\draw[thick, red, mid<] (6,1) .. controls ++(135:.7cm) and ++(135:.8cm) .. (6,0) node[left, xshift=-.1cm]{$\scriptstyle c$} .. controls ++(-45:.8cm) and ++(-45:.7cm) .. (7,0);
\draw[thick, red, mid<] (6,3) .. controls ++(135:.8cm) and ++(135:.8cm) .. node[above, yshift=.1cm]{$\scriptstyle d$} (7,4);
\draw[knot, mid>] (5.5,1) -- node[below]{$\scriptstyle k'$} (6,1);
\draw[knot, mid<] (5.5,1) .. controls ++(180:1.5cm) and ++(-90:2cm) .. node[below]{$\scriptstyle k$} (7,0);
\draw[mid>] (6,1) -- node[right]{$\scriptstyle i'$} (6,1.5);
\draw[mid>] (6,1.5) .. controls ++(90:1.5cm) and ++(0:2cm) .. node[above]{$\scriptstyle i$} (7,0);
\draw[mid>] (6,2.5) -- node[right]{$\scriptstyle i'$} (6,3);
\draw[mid>] (7,4) .. controls ++(0:2cm) and ++(-90:1.5cm) .. node[below]{$\scriptstyle i$} (6,2.5);
\draw[mid>] (5.5,3) -- node[above]{$\scriptstyle \ell'$} (6,3);
\draw[mid>] (7,4) .. controls ++(90:2cm) and ++(180:1.5cm) .. node[left]{$\scriptstyle \ell$} (5.5,3);
\draw[mid>] (7,0) to[out=180,in=-90] node[below]{$\scriptstyle t$} (6,1);
\draw[mid<] (7,0) to[out=90,in=0] node[right]{$\scriptstyle u$} (6.5,1) -- (6,1);
\draw[mid<] (7,4) to[out=-90,in=0] node[right]{$\scriptstyle v$} (6.5,3) -- (6,3);
\draw[mid<] (7,4) to[out=180,in=90] node[left]{$\scriptstyle w$} (6,3);
\draw[thick, cyan, mid>] (6,1.5) --node[left]{$\scriptstyle n'$} (5.5,1);
\draw[thick, cyan, mid>] (5.5,3) --node[below]{$\scriptstyle n'$} (6,2.5);
\node at (6.25,.8) {$\scriptstyle \xi_{3,1}$};
\node at (6.4,3.3) {$\scriptstyle \xi_{3,2}$};
\node at (7.6,-.2) {$\scriptstyle \overline{\eta_{3,1}}$};
\node at (7.3,4.2) {$\scriptstyle \eta^\dag_{3,2}$};
}
\displaybreak[1]\\&=
\frac{K}{\cD_\cX^2}
\sum_{n',h''\in \Irr(\cX)} 
\frac{K_p'\sqrt{d_{h''}}}{d_{n'}}
\tikzmath{
\clip (-1,-1.5) rectangle (9,4.5);
\draw[thick, red, mid<] (1,1) .. controls ++(135:1cm) and ++(135:1cm) .. node[above, yshift=.1cm]{$\scriptstyle a$} (0,0);
\draw[thick, red, mid<] (2,1) .. controls ++(135:1cm) and ++(135:1cm) .. (2,0) node[left, xshift=-.1cm]{$\scriptstyle b$} .. controls ++(-45:1cm) and ++(-45:1cm) .. (3,0);
\draw[thick, red, mid<] (2,2) .. controls ++(135:1cm) and ++(135:1cm) .. node[above, yshift=.1cm]{$\scriptstyle e$} (3,3);
\draw[thick, red, mid<] (1,2) .. controls ++(135:1cm) and ++(-45:1cm) .. node[right, xshift=.05cm]{$\scriptstyle f$} (0,3);
\draw[knot, mid>] (1,1) -- node[above, xshift=-.1cm]{$\scriptstyle j'$} (2,1);
\draw[mid>] (2,1) -- (2,2); 
\draw[mid>] (1,2) -- node[above]{$\scriptstyle m'$} (2,2);
\draw[mid>] (1,1) -- node[left]{$\scriptstyle g'$} (1,2);
\draw[knot, mid<] (0,0) to[out=-90,in=-90] node[above]{$\scriptstyle j$} (3,0);
\draw[mid<] (3,0) to[out=0,in=0] node[left]{$\scriptstyle h$} (3,3);
\draw[mid<] (0,3) to[out=90,in=90] node[below]{$\scriptstyle m$} (3,3);
\draw[mid<] (0,0) to[out=180,in=180] node[right]{$\scriptstyle g$} (0,3);
\draw[mid>] (0,0) to[out=90,in=180] node[above]{$\scriptstyle q$} (1,1);
\draw[mid>] (0,0) to[out=0,in=-90] node[below]{$\scriptstyle r$} (1,1);
\draw[mid>] (3,0) to[out=180,in=-90] node[below]{$\scriptstyle s$} (2,1);
\draw[mid<] (3,0) to[out=90,in=0] node[right]{$\scriptstyle k$} (2.5,1) -- (2,1);
\draw[mid<] (3,3) to[out=180,in=90] node[left]{$\scriptstyle x$} (2,2);
\draw[mid<] (3,3) to[out=-90,in=0] node[right]{$\scriptstyle \ell$} (2.5,2) -- (2,2);
\draw[mid<] (0,3) to[out=0,in=90] node[right]{$\scriptstyle y$} (1,2);
\draw[mid>]  (0,3) to[out=-90,in=180] node[below]{$\scriptstyle z$} (1,2);
\draw[thick, cyan, mid>] (2.5,1) --node[above]{$\scriptstyle n'$} (2,1.3);
\draw[thick, cyan, mid>] (2,1.7) --node[right]{$\scriptstyle n'$} (2.5,2);
\node at (1.3,1.8) {$\scriptstyle \xi_{1,2}$};
\node at (1.3,.8) {$\scriptstyle \xi_{1,1}$};
\node at (2.3,.8) {$\scriptstyle \xi_{2,1}$};
\node at (2.3,2.2) {$\scriptstyle \xi_{2,2}$};
\node at (-.3,-.2) {$\scriptstyle \eta^\dag_{1,1}$};
\node at (3.6,-.2) {$\scriptstyle \overline{\eta_{2,1}}$};
\node at (3.3,3.2) {$\scriptstyle \eta^\dag_{2,2}$};
\node at (-.3,3.2) {$\scriptstyle \overline{\eta_{1,2}}$};
\begin{scope}[yshift=-.5cm, xshift=-.5cm]
\draw[thick, red, mid<] (6,1) .. controls ++(135:.7cm) and ++(135:.8cm) .. (6,0) node[left, xshift=-.1cm]{$\scriptstyle c$} .. controls ++(-45:.8cm) and ++(-45:.7cm) .. (7,0);
\draw[thick, red, mid<] (6,3) .. controls ++(135:.8cm) and ++(135:.8cm) .. node[above, yshift=.1cm]{$\scriptstyle d$} (7,4);
\draw[knot, mid>] (5.5,1) -- node[below]{$\scriptstyle k'$} (6,1);
\draw[knot, mid<] (5.5,1) .. controls ++(180:1.5cm) and ++(-90:2cm) .. node[below]{$\scriptstyle k$} (7,0);
\draw[mid>] (6,1) -- node[right]{$\scriptstyle i'$} (6,1.5);
\draw[mid>] (6,1.5) .. controls ++(90:1.5cm) and ++(0:2cm) .. node[above]{$\scriptstyle i$} (7,0);
\draw[mid>] (6,2.5) -- node[right]{$\scriptstyle i'$} (6,3);
\draw[mid>] (7,4) .. controls ++(0:2cm) and ++(-90:1.5cm) .. node[below]{$\scriptstyle i$} (6,2.5);
\draw[mid>] (5.5,3) -- node[above]{$\scriptstyle \ell'$} (6,3);
\draw[mid>] (7,4) .. controls ++(90:2cm) and ++(180:1.5cm) .. node[left]{$\scriptstyle \ell$} (5.5,3);
\draw[mid>] (7,0) to[out=180,in=-90] node[below]{$\scriptstyle t$} (6,1);
\draw[mid<] (7,0) to[out=90,in=0] node[right]{$\scriptstyle u$} (6.5,1) -- (6,1);
\draw[mid<] (7,4) to[out=-90,in=0] node[right]{$\scriptstyle v$} (6.5,3) -- (6,3);
\draw[mid<] (7,4) to[out=180,in=90] node[left]{$\scriptstyle w$} (6,3);
\draw[thick, cyan, mid>] (6,1.5) --node[left]{$\scriptstyle n'$} (5.5,1);
\draw[thick, cyan, mid>] (5.5,3) --node[below]{$\scriptstyle n'$} (6,2.5);
\node at (6.25,.8) {$\scriptstyle \xi_{3,1}$};
\node at (6.4,3.3) {$\scriptstyle \xi_{3,2}$};
\node at (7.6,-.2) {$\scriptstyle \overline{\eta_{3,1}}$};
\node at (7.3,4.2) {$\scriptstyle \eta^\dag_{3,2}$};
\end{scope}
}
\displaybreak[1]\\&=
\frac{K}{\cD_\cX^2}
\sum_{n',h''\in \Irr(\cX)} 
\frac{K_p'\sqrt{d_{h''}}}{d_{n'}}
\cdot
\tikzmath{
\clip (-1.7,-1.7) rectangle (9,4.7);
\draw[thick, red, mid<] (1,1) .. controls ++(135:1cm) and ++(135:1cm) .. node[above, yshift=.1cm]{$\scriptstyle a$} (0,0);
\draw[thick, red, mid<] (2,1) .. controls ++(135:1cm) and ++(135:1cm) .. (2,0) node[left, xshift=-.1cm]{$\scriptstyle b$} .. controls ++(-45:1cm) and ++(-45:1cm) .. (3,0);
\draw[thick, red, mid<] (2,2) .. controls ++(135:1cm) and ++(135:1cm) .. node[above, yshift=.1cm]{$\scriptstyle e$} (3,3);
\draw[thick, red, mid<] (1,2) .. controls ++(135:1cm) and ++(-45:1cm) .. node[right, xshift=.05cm]{$\scriptstyle f$} (0,3);
\draw[knot, mid>] (1,1) -- node[above, xshift=-.1cm]{$\scriptstyle j'$} (2,1);
\draw[mid>] (2,1) -- (2,2); 
\draw[mid>] (1,2) -- node[above]{$\scriptstyle m'$} (2,2);
\draw[mid>] (1,1) -- node[left]{$\scriptstyle g'$} (1,2);
\draw[knot, mid<] (0,0) to[out=-90,in=-90] node[above]{$\scriptstyle j$} (3,0);
\draw[mid<] (3,0) .. controls ++(0:1cm) and ++(0:1cm) .. (3,-1.5) -- (0,-1.5) to[out=180,in=-90] (-1.5,1.5) node[right]{$\scriptstyle h$} to[out=90,in=180] (0,4.5) -- (3,4.5) .. controls ++(0:1cm) and ++(0:1cm) .. (3,3);
\draw[mid<] (0,3) to[out=90,in=90] node[below]{$\scriptstyle m$} (3,3);
\draw[mid<] (0,0) to[out=180,in=180] node[right]{$\scriptstyle g$} (0,3);
\draw[mid>] (0,0) to[out=90,in=180] node[above]{$\scriptstyle q$} (1,1);
\draw[mid>] (0,0) to[out=0,in=-90] node[below]{$\scriptstyle r$} (1,1);
\draw[mid>] (3,0) to[out=180,in=-90] node[below]{$\scriptstyle s$} (2,1);
\draw[mid<] (3,0) to[out=90,in=0] node[right]{$\scriptstyle k$} (2.5,1) -- (2,1);
\draw[mid<] (3,3) to[out=180,in=90] node[left]{$\scriptstyle x$} (2,2);
\draw[mid<] (3,3) to[out=-90,in=0] node[right]{$\scriptstyle \ell$} (2.5,2) -- (2,2);
\draw[mid<] (0,3) to[out=0,in=90] node[right]{$\scriptstyle y$} (1,2);
\draw[mid>]  (0,3) to[out=-90,in=180] node[below]{$\scriptstyle z$} (1,2);
\draw[thick, cyan, mid>] (2.5,1) --node[above]{$\scriptstyle n'$} (2,1.3);
\draw[thick, cyan, mid>] (2,1.7) --node[right]{$\scriptstyle n'$} (2.5,2);
\node at (1.3,1.8) {$\scriptstyle \xi_{1,2}$};
\node at (1.3,.8) {$\scriptstyle \xi_{1,1}$};
\node at (2.3,.8) {$\scriptstyle \xi_{2,1}$};
\node at (2.3,2.2) {$\scriptstyle \xi_{2,2}$};
\node at (-.3,-.2) {$\scriptstyle \eta^\dag_{1,1}$};
\node at (3.3,.2) {$\scriptstyle \overline{\eta_{2,1}}$};
\node at (3.3,2.8) {$\scriptstyle \eta^\dag_{2,2}$};
\node at (-.3,3.2) {$\scriptstyle \overline{\eta_{1,2}}$};
\fill[fill=green] (2.5,1) circle (.05cm);
\fill[fill=yellow] (2.5,2) circle (.05cm);
\fill[fill=orange] (2,1.3) circle (.05cm);
\fill[fill=orange] (2,1.7) circle (.05cm);
\begin{scope}[yshift=-.5cm, xshift=-1cm]
\draw[thick, red, mid<] (6,1) .. controls ++(135:.7cm) and ++(135:.8cm) .. (6,0) node[left, xshift=-.1cm]{$\scriptstyle c$} .. controls ++(-45:.8cm) and ++(-45:.7cm) .. (7,0);
\draw[thick, red, mid<] (6,3) .. controls ++(135:.8cm) and ++(135:.8cm) .. node[above, yshift=.1cm]{$\scriptstyle d$} (7,4);
\draw[knot, mid>] (5.5,1) -- node[below]{$\scriptstyle k'$} (6,1);
\draw[knot, mid<] (5.5,1) .. controls ++(180:1.5cm) and ++(-90:2cm) .. node[below]{$\scriptstyle k$} (7,0);
\draw[mid>] (6,1) -- node[right]{$\scriptstyle i'$} (6,1.5);
\draw[mid>] (6,1.5) .. controls ++(90:1.5cm) and ++(0:2cm) .. node[above]{$\scriptstyle i$} (7,0);
\draw[mid>] (6,2.5) -- node[right]{$\scriptstyle i'$} (6,3);
\draw[mid>] (7,4) .. controls ++(0:2cm) and ++(-90:1.5cm) .. node[below]{$\scriptstyle i$} (6,2.5);
\draw[mid>] (5.5,3) -- node[above]{$\scriptstyle \ell'$} (6,3);
\draw[mid>] (7,4) .. controls ++(90:2cm) and ++(180:1.5cm) .. node[left]{$\scriptstyle \ell$} (5.5,3);
\draw[mid>] (7,0) to[out=180,in=-90] node[below]{$\scriptstyle t$} (6,1);
\draw[mid<] (7,0) to[out=90,in=0] node[right]{$\scriptstyle u$} (6.5,1) -- (6,1);
\draw[mid<] (7,4) to[out=-90,in=0] node[right]{$\scriptstyle v$} (6.5,3) -- (6,3);
\draw[mid<] (7,4) to[out=180,in=90] node[left]{$\scriptstyle w$} (6,3);
\draw[thick, cyan, mid>] (6,1.5) --node[left]{$\scriptstyle n'$} (5.5,1);
\draw[thick, cyan, mid>] (5.5,3) --node[below]{$\scriptstyle n'$} (6,2.5);
\fill[fill=green] (5.5,1) circle (.05cm);
\fill[fill=purple] (6,1.5) circle (.05cm);
\fill[fill=purple] (6,2.5) circle (.05cm);
\fill[fill=yellow] (5.5,3) circle (.05cm);
\node at (6.25,.8) {$\scriptstyle \xi_{3,1}$};
\node at (6.4,3.3) {$\scriptstyle \xi_{3,2}$};
\node at (7.6,-.2) {$\scriptstyle \overline{\eta_{3,1}}$};
\node at (7.3,4.2) {$\scriptstyle \eta^\dag_{3,2}$};
\end{scope}
}
\displaybreak[1]\\&=
\frac{K}{\cD_\cX^2}
\cdot
\tikzmath{
\clip (-1.7,-1.7) rectangle (9,4.7);
\draw[thick, red, mid<] (1,1) .. controls ++(135:1cm) and ++(135:1cm) .. node[above, yshift=.1cm]{$\scriptstyle a$} (0,0);
\draw[thick, red, mid<] (2,1) .. controls ++(135:1cm) and ++(135:1cm) .. (2,0) node[left, xshift=-.1cm]{$\scriptstyle b$} .. controls ++(-45:1cm) and ++(-45:1cm) .. (3,0);
\draw[thick, red, mid<] (5,1) .. controls ++(135:1cm) and ++(135:1cm) .. (5,0) node[left, xshift=-.1cm]{$\scriptstyle c$} .. controls ++(-45:1cm) and ++(-45:1cm) .. (6,0);
\draw[thick, red, mid<] (5,2) .. controls ++(135:1cm) and ++(135:1cm) .. node[above, yshift=.1cm]{$\scriptstyle d$} (6,3);
\draw[thick, red, mid<] (2,2) .. controls ++(135:1cm) and ++(135:1cm) .. node[above, yshift=.1cm]{$\scriptstyle e$} (3,3);
\draw[thick, red, mid<] (1,2) .. controls ++(135:1cm) and ++(-45:1cm) .. node[right, xshift=.05cm]{$\scriptstyle f$} (0,3);
\draw[knot, mid>] (1,1) -- node[above, xshift=-.1cm]{$\scriptstyle j'$} (2,1);
\draw[mid>] (2,1) -- node[left]{$\scriptstyle h'$} (2,2);
\draw[mid>] (1,2) -- node[above]{$\scriptstyle m'$} (2,2);
\draw[mid>] (1,1) -- node[left]{$\scriptstyle g'$} (1,2);
\draw[mid>] (2,2) -- node[above]{$\scriptstyle \ell'$} (5,2);
\draw[knot, mid>] (2,1) -- node[above]{$\scriptstyle k'$} (5,1);
\draw[mid>] (5,1) -- node[left]{$\scriptstyle i'$} (5,2);
\draw[knot, mid<] (0,0) to[out=-90,in=-90] node[above]{$\scriptstyle j$} (3,0);
\draw[mid<] (3,0) .. controls ++(0:1cm) and ++(0:1cm) .. (3,-1.5) -- (0,-1.5) to[out=180,in=-90] (-1.5,1.5) node[right]{$\scriptstyle h$} to[out=90,in=180] (0,4.5) -- (3,4.5) .. controls ++(0:1cm) and ++(0:1cm) .. (3,3);
\draw[mid<] (0,3) to[out=90,in=90] node[below]{$\scriptstyle m$} (3,3);
\draw[mid<] (0,0) to[out=180,in=180] node[right]{$\scriptstyle g$} (0,3);
\draw[mid>] (6,3) to[out=0,in=0] node[right]{$\scriptstyle i$} (6,0);
\draw[mid>] (0,0) to[out=90,in=180] node[above]{$\scriptstyle q$} (1,1);
\draw[mid>] (0,0) to[out=0,in=-90] node[below]{$\scriptstyle r$} (1,1);
\draw[mid>] (3,0) to[out=180,in=-90] node[below]{$\scriptstyle s$} (2,1);
\draw[knot, mid<] (3,0) to[out=90,in=90] (4.5,0) to[out=-90,in=-90] node[below]{$\scriptstyle k$} (6,0);
\draw[mid<] (3,3) to[out=180,in=90] node[left]{$\scriptstyle x$} (2,2);
\draw[mid<] (3,3) to[out=-90,in=-90] (4,3) to[out=90, in=90] node[above]{$\scriptstyle \ell$} (6,3);
\draw[mid>] (6,0) to[out=180,in=-90] node[below]{$\scriptstyle t$} (5,1);
\draw[mid<] (6,0) to[out=90,in=0] node[right]{$\scriptstyle u$} (5,1);
\draw[mid<] (6,3) to[out=-90,in=0] node[right]{$\scriptstyle v$} (5,2);
\draw[mid<] (6,3) to[out=180,in=90] node[left]{$\scriptstyle w$} (5,2);
\draw[mid<] (0,3) to[out=0,in=90] node[right]{$\scriptstyle y$} (1,2);
\draw[mid>]  (0,3) to[out=-90,in=180] node[below]{$\scriptstyle z$} (1,2);
\node at (1.3,1.8) {$\scriptstyle \xi_{1,2}$};
\node at (1.3,.8) {$\scriptstyle \xi_{1,1}$};
\node at (2.3,.8) {$\scriptstyle \xi_{2,1}$};
\node at (2.3,2.2) {$\scriptstyle \xi_{2,2}$};
\node at (5.3,.8) {$\scriptstyle \xi_{3,1}$};
\node at (5.3,2.2) {$\scriptstyle \xi_{3,2}$};
\node at (-.3,-.2) {$\scriptstyle \eta^\dag_{1,1}$};
\node at (-.3,3.2) {$\scriptstyle \overline{\eta_{1,2}}$};
\node at (2.8,.2) {$\scriptstyle \overline{\eta_{2,1}}$};
\node at (2.8,2.7) {$\scriptstyle \eta^\dag_{2,2}$};
\node at (6.3,3.2) {$\scriptstyle \eta^\dag_{3,2}$};
\node at (6.3,.2) {$\scriptstyle \overline{\eta_{3,1}}$};
}
\end{align*}
This final diagram is exactly the $D_\cX^{-2}$ times the skein module $\cS^{\textcolor{red}{\cA}}_\cX(\bbD,\textcolor{red}{6},10)$ inner product.

It remains to account for the exponent of $\cD_\cA$ for the bulk plaquettes.  
To get the correct exponent, we introduce a subset $\cB'\subset \cB^\circ$ that we obtain by removing exactly one non $\cX$-boundary plaquette from each cube in $\cB^\circ$.  
We do so by picking a maximal spanning tree for the dual lattice intersected with $\cB$, which has a vertex for each cube and an edge for each face between cubes, plus one additional (non $\cX$-boundary)~edge connecting our maximal spanning tree to a cube in the bulk outside of $\cB$.  
By Lemma \ref{lem:5=6}, it follows that
$$
\prod\limits_{p\subset \cB'}B_p=\prod\limits_{p\subset \cB^\circ}B_p\,.
$$

The subset $\cB'\subset \cB$ is exactly the set of $N_p(\cB^\circ) - N_c(\cB^\circ)$ plaquettes of $\cB$ which do not intersect this spanning tree.  
Just as with the plaquettes on the boundary, we can multiply two plaquette operators which share an edge and compute the matrix coefficients of this new operator. 
Since the union of the plaquettes in $\cB'$ is simply connected, at each inductive step, the resulting enriched skein module is unitarily isomorphic to one of the form $\cS^{\textcolor{red}{\cA}}_\cX(\bbD, \textcolor{red}{m},n)$, via a deformation retraction onto the boundary. 
This topological map gives a well defined map on skein modules since $\cA$ is braided; i.e the red strings in the bulk may be unambigously projected onto the boundary via the deformation retraction. 
Therefore we may combine all $N_p(\cB^\circ) - N_c(\cB^\circ)$ non $\cX$-boundary plaquettes and compute matrix coefficients analogously to the two plaquette case. 
The matrix coefficients have a factor of $D_\cA^{-1}$ for each bulk plaquette in $\cB'$. 
Further multiplying the bulk plaquettes in $\cB'$ with the $\cX$-boundary plaquettes which each contribute a factor of $D_\cX$, 
the result follows.
\end{proof}

\begin{rem}
\label{rem:HostExcitationsAtVertices}
We can use enriched skein modules to give another string-net model for a(n uneriched) UFC $\cX$ due to Corey Jones where excitations can be localized on a single vertex rather than an edge and its neighboring plaquettes.
This is similar to Corey Jones' model presented in \cite{MR4642306} with extra strands corresponding to a condensable algebra
in order to implement anyon condensation.

We replace the local Hilbert space $\cH_v$ with
$$
\tikzmath{
\draw[thick] (-.5,0) node[left]{$\scriptstyle w$} -- (.5,0) node[right]{$\scriptstyle z$};
\draw[thick] (0,-.5) node[below]{$\scriptstyle x$} -- (0,.5) node[above]{$\scriptstyle y$};
\filldraw (0,0) circle (.05cm);
\draw[thick, blue, decorate, decoration={snake, segment length=1mm, amplitude=.2mm}] (0,0) -- (.4,-.4) node[right]{$\scriptstyle A$}; 
\node at (-.2,-.2) {$\scriptstyle v$};
\draw[red!50, very thin] (-.25,.5) -- (.25,-.5);
{\draw[red!50, very thin, -stealth] ($ (-.15,.3) - (.12,.06)$) to ($ (-.15,.3) + (.12,.06)$);}
{\draw[red!50, very thin, -stealth] ($ (.15,-.3) - (.12,.06)$) to ($ (.15,-.3) + (.12,.06)$);}
}
\leftrightarrow
\cH_v:=
\bigoplus_{\substack{
x\in \Irr(\cX)
\\
A\in \Irr(Z(\cX))
}}
\cX(wx \to yzF(A))
$$
where $F: Z(\cX)\to \cX$ is the forgetful functor.
This space caries the $\cS_\cX^{\textcolor{blue}{Z(\cX)}}(\bbD,\textcolor{blue}{1},4)$ enriched skein module inner product.
The total Hilbert space is the tensor product of local spaces.
Again, we have edge operators $A_\ell$ which require simple labels to match along edges and plaquette $B_p$ operators on faces.
Buulding on Proposition \ref{prop:DescriptionOfIm(PA)},
we can describe the ground state space $\im(P_A)$ of $-\sum_\ell A_\ell$ on a contractible patch $\Lambda$ of lattice as
$$
\bigoplus_{\vec{x}\in \partial \Lambda}
\cX(\vec{x} \to F(Z)^{\otimes \# v\in \Lambda} FI(1_\cX)^{\#p \in \Lambda})
$$
where $Z:= \bigoplus_{A\in \Irr(Z(\cX))} A$.

On $\im(P_A)$, the plaquette operator $B_p$ now uses the half-braiding for $X\in Z(\cX)$:
\begin{align*}
\frac{1}{D_\cX}\sum_{r\in \Irr(\cX)}
d_r\cdot
\tikzmath{
\draw[step=1.0,black,thin] (0.5,0.5) grid (2.5,2.5);
\filldraw[thick, DarkGreen, rounded corners=5pt, fill=gray!30] (1.15,1.15) rectangle (1.85,1.85);
\draw[thick, blue, decorate, decoration={snake, segment length=1mm, amplitude=.2mm}] (2,1) -- ($ (2,1) + (.4,-.4)$);
\draw[thick, blue, decorate, decoration={snake, segment length=1mm, amplitude=.2mm}] (1,1) -- ($ (1,1) + (.4,-.4)$);
\draw[knot, thick, blue, decorate, decoration={snake, segment length=1mm, amplitude=.2mm}] (1,2) -- ($ (1,2) + (.4,-.4)$);
\draw[thick, blue, decorate, decoration={snake, segment length=1mm, amplitude=.2mm}] (2,2) -- ($ (2,2) + (.4,-.4)$);
\node at (2.3,1.2) {$\scriptstyle \xi_{2,1}$};
\node at (.7,2.2) {$\scriptstyle \xi_{1,2}$};
\node at (.7,1.2) {$\scriptstyle \xi_{1,1}$};
\node at (2.3,2.2) {$\scriptstyle \xi_{2,2}$};
\node[DarkGreen] at (1.3,1.5) {$\scriptstyle r$};
\node at (1.5,.85) {$\scriptstyle g$};
\node at (2.15,1.5) {$\scriptstyle h$};
\node at (1.5,2.15) {$\scriptstyle i$};
\node at (.85,1.5) {$\scriptstyle j$};
}
&=
\frac{1}{D_\cX}\sum_{r,s,t,u,v \in \Irr(\cX)}
\frac{\sqrt{d_kd_\ell d_md_n}}{d_r\sqrt{d_gd_hd_id_j}}
\cdot
\tikzmath{
\draw[step=1.0,black,thin] (0.5,0.5) grid (2.5,2.5);
\fill[gray!30, rounded corners=5pt] (1.15,1.15) rectangle (1.85,1.85);
\draw[thick, DarkGreen] (1.3,1) -- (1,1.3);
\draw[thick, DarkGreen] (1.7,1) -- (2,1.3);
\draw[thick, DarkGreen] (1.3,2) -- (1,1.7);
\draw[thick, DarkGreen] (1.7,2) -- (2,1.7);
\draw[thick, blue, decorate, decoration={snake, segment length=1mm, amplitude=.2mm}] (2,1) -- ($ (2,1) + (.4,-.4)$);
\draw[thick, blue, decorate, decoration={snake, segment length=1mm, amplitude=.2mm}] (1,1) -- ($ (1,1) + (.4,-.4)$);
\draw[knot, thick, blue, decorate, decoration={snake, segment length=1mm, amplitude=.2mm}] (1,2) -- ($ (1,2) + (.4,-.4)$);
\draw[thick, blue, decorate, decoration={snake, segment length=1mm, amplitude=.2mm}] (2,2) -- ($ (2,2) + (.4,-.4)$);
\fill[fill=red] (1.3,1) circle (.05cm);
\fill[fill=red] (1.7,1) circle (.05cm);
\fill[fill=purple] (2,1.3) circle (.05cm);
\fill[fill=purple] (2,1.7) circle (.05cm);
\fill[fill=yellow] (1.3,2) circle (.05cm);
\fill[fill=yellow] (1.7,2) circle (.05cm);
\fill[fill=orange] (1,1.3) circle (.05cm);
\fill[fill=orange] (1,1.7) circle (.05cm);
\node at (2.3,.8) {$\scriptstyle \xi_{2,1}$};
\node at (.7,2.2) {$\scriptstyle \xi_{1,2}$};
\node at (.7,.8) {$\scriptstyle \xi_{1,1}$};
\node at (2.3,2.2) {$\scriptstyle \xi_{2,2}$};
\node at (1.15,.85) {$\scriptstyle g$};
\node at (1.85,.85) {$\scriptstyle g$};
\node at (2.15,1.15) {$\scriptstyle h$};
\node at (2.15,1.85) {$\scriptstyle h$};
\node at (1.15,2.15) {$\scriptstyle i$};
\node at (1.85,2.15) {$\scriptstyle i$};
\node at (.85,1.15) {$\scriptstyle j$};
\node at (.85,1.85) {$\scriptstyle j$};
\node at (1.5,.85) {$\scriptstyle k$};
\node at (2.15,1.5) {$\scriptstyle \ell$};
\node at (1.5,2.15) {$\scriptstyle m$};
\node at (.85,1.5) {$\scriptstyle n$};
}
\\&=
\sum_{\eta}
C(\xi,\eta)
\tikzmath{
\draw[step=1.0,black,thin] (0.5,0.5) grid (2.5,2.5);
\fill[gray!30, rounded corners=5pt] (1.1,1.1) rectangle (1.9,1.9);
\draw[thick, blue, decorate, decoration={snake, segment length=1mm, amplitude=.2mm}] (2,1) -- ($ (2,1) + (.4,-.4)$);
\draw[thick, blue, decorate, decoration={snake, segment length=1mm, amplitude=.2mm}] (1,1) -- ($ (1,1) + (.4,-.4)$);
\draw[thick, blue, decorate, decoration={snake, segment length=1mm, amplitude=.2mm}] (1,2) -- ($ (1,2) + (.4,-.4)$);
\draw[thick, blue, decorate, decoration={snake, segment length=1mm, amplitude=.2mm}] (2,2) -- ($ (2,2) + (.4,-.4)$);
\node at (2.3,1.2) {$\scriptstyle \eta_{2,1}$};
\node at (.7,2.2) {$\scriptstyle \eta_{1,2}$};
\node at (.7,1.2) {$\scriptstyle \eta_{1,1}$};
\node at (2.3,2.2) {$\scriptstyle \eta_{2,2}$};
}\,.
\end{align*}
Similar to Proposition \ref{prop:MatrixCoefficientsOfBp} in the $\cA$-enriched setting,
when all edge/link labels on simple tensors $\eta,\xi$ match
$$
\tikzmath{
\draw[step=1.0,black,thin] (0.5,0.5) grid (2.5,2.5);
\fill[gray!30, rounded corners=5pt] (1.1,1.1) rectangle (1.9,1.9);
\draw[thick, blue, decorate, decoration={snake, segment length=1mm, amplitude=.2mm}] (2,1) -- ($ (2,1) + (.4,-.4)$) node[right]{$\scriptstyle D$};
\draw[thick, blue, decorate, decoration={snake, segment length=1mm, amplitude=.2mm}] (1,1) -- ($ (1,1) + (.4,-.4)$) node[right]{$\scriptstyle A$};
\draw[thick, blue, decorate, decoration={snake, segment length=1mm, amplitude=.2mm}] (1,2) -- ($ (1,2) + (.4,-.4)$) node[right]{$\scriptstyle B$};
\draw[thick, blue, decorate, decoration={snake, segment length=1mm, amplitude=.2mm}] (2,2) -- ($ (2,2) + (.4,-.4)$) node[right]{$\scriptstyle C$};
\node at (1.7,1.2) {$\scriptstyle \eta_{2,1}$};
\node at (.7,2.2) {$\scriptstyle \eta_{1,2}$};
\node at (.7,1.2) {$\scriptstyle \eta_{1,1}$};
\node at (1.7,2.2) {$\scriptstyle \eta_{2,2}$};
\node at (.3,1) {$\scriptstyle s$};
\node at (1,.3) {$\scriptstyle t$};
\node at (2,.3) {$\scriptstyle u$};
\node at (2.7,1) {$\scriptstyle v$};
\node at (2.7,2) {$\scriptstyle w$};
\node at (2,2.7) {$\scriptstyle x$};
\node at (1,2.7) {$\scriptstyle y$};
\node at (.3,2) {$\scriptstyle z$};
}
\qquad\qquad
\tikzmath{
\draw[step=1.0,black,thin] (0.5,0.5) grid (2.5,2.5);
\fill[gray!30, rounded corners=5pt] (1.1,1.1) rectangle (1.9,1.9);
\draw[thick, blue, decorate, decoration={snake, segment length=1mm, amplitude=.2mm}] (2,1) -- ($ (2,1) + (.4,-.4)$) node[right]{$\scriptstyle D$};
\draw[thick, blue, decorate, decoration={snake, segment length=1mm, amplitude=.2mm}] (1,1) -- ($ (1,1) + (.4,-.4)$) node[right]{$\scriptstyle A$};
\draw[thick, blue, decorate, decoration={snake, segment length=1mm, amplitude=.2mm}] (1,2) -- ($ (1,2) + (.4,-.4)$) node[right]{$\scriptstyle B$};
\draw[thick, blue, decorate, decoration={snake, segment length=1mm, amplitude=.2mm}] (2,2) -- ($ (2,2) + (.4,-.4)$) node[right]{$\scriptstyle C$};
\node at (1.7,1.2) {$\scriptstyle \eta_{2,1}$};
\node at (.7,2.2) {$\scriptstyle \eta_{1,2}$};
\node at (.7,1.2) {$\scriptstyle \eta_{1,1}$};
\node at (1.7,2.2) {$\scriptstyle \eta_{2,2}$};
\node at (.3,1) {$\scriptstyle s$};
\node at (1,.3) {$\scriptstyle t$};
\node at (2,.3) {$\scriptstyle u$};
\node at (2.7,1) {$\scriptstyle v$};
\node at (2.7,2) {$\scriptstyle w$};
\node at (2,2.7) {$\scriptstyle x$};
\node at (1,2.7) {$\scriptstyle y$};
\node at (.3,2) {$\scriptstyle z$};
}\,,
$$
the matrix coefficient $C(\xi,\eta)$ is given by
$$
C(\xi,\eta)
=
\frac{1}{D_\cX\sqrt{\textcolor{blue}{d_Ad_Bd_Cd_D}d_sd_td_ud_vd_wd_xd_yd_z}}
\cdot
\tikzmath{
\begin{scope}
\clip (1.2,1.8) rectangle (3.5,3.6);
\begin{scope}[shift={(-.3,.3)}]
\draw[step=1.0,black] (0.5,0.5) grid (4,4);
\end{scope}
\draw[step=1.0,black, knot] (1,1) grid (4,4);
\end{scope}
\foreach \j in {2,3} {
\draw (1.2,\j) arc (270:90:.15cm);
\draw (3.5,\j) arc (-90:90:.15cm);
};
\foreach \i in {2,3} {
    \draw (\i,3.6) arc (0:180:.15cm);
    \draw (\i,1.8) arc (0:-180:.15cm);
};
\foreach \i in {1,2} {
\foreach \j in {1,2} {
        \draw[thick, blue, decorate, decoration={snake, segment length=1mm, amplitude=.2mm}] ($ (\i,\j) + (1,1)$) -- ($ (\i,\j) + (.7,1.3)$);
        \node at ($ (\i,\j) + (1.25,.8) $) {$\scriptstyle \overline{\eta_{\i,\j}}$};
        \node at ($ (\i,\j) + (.45,1.45) $) {$\scriptstyle \xi_{\i,\j}$};
}}
}\,.
$$
Otherwise, $C(\xi,\eta)=0$.
In particular, the $B_p$ are self-adjoint, commuting orthogonal projections.

Finally, the local Hamiltonian now has \emph{vertex terms} $C_v$ which enforce that the simple label on each blue edge is $1_{Z(\cX)}$.
Hence the ground states of
$$
H := -\sum_\ell A_\ell - \sum_p B_p - \sum_v C_v
$$
are exactly the ground states of the usual string-net model for $\cX$ from \S\ref{sec:StringNet}.

Now excitations can live on the blue edges, corresponding to a violation of a single $C_v$ term.
We can detect the type of the excitation as it corresponds to a representation of the abelian $\rm C^*$-algebra $\bbC^{\Irr(Z(\cX))}$ which acts on each blue edge.
Note that this algebra is Morita equivalent to $\Tube(\cX)$ as $\Rep(\Tube(\cX)) \cong Z(\cX) \cong \Rep(\bbC^{\Irr(Z(\cX))})$.
These vertex excitations are equivalent to the excitations discussed above in \S\ref{sec:StringOperators}; 
indeed, for each $(X,\sigma_X)\in \Irr(Z(\cX))$ and $\psi : F(X)\to x$ for $x\in \Irr(\cX)$ corresponding to the localized excitation \eqref{eq:BpvqProtection},
we have the hopping operator
$$
H^{(X,\sigma_X)}_\pi(\psi)
:=
B_pB_qT_\pi^{(X,\sigma_X)}(\psi)
$$
as in the cartoon below: 
$$
\frac{1}{D_\cX^2}
\sum_{\psi \in \operatorname{ONB}(Y\to x)}
\sum_{s,t\in \Irr(\cX)}
d_sd_t
\tikzmath{
\draw[thick] (.5,.5) grid (4.5,3.5);
\fill[white] (1.8,1.05) rectangle (2.2,1.95);
\filldraw[blue, thick, fill=gray!30, rounded corners=5pt] (1.1,1.1) rectangle (1.7,1.9);
\filldraw[blue, thick, fill=gray!30, rounded corners=5pt] (2.4,1.1) rectangle (2.9,1.9);
\node[blue] at (1.25,1.5) {$\scriptstyle s$};
\node[blue] at (2.55,1.5) {$\scriptstyle t$};
\foreach \x in {1,2,3,4}{
\foreach \y in {1,2,3}{
\filldraw (\x,\y) circle (.05cm);
}}
\foreach \x in {0.5,1.5,2.5}{
\foreach \y in {0.5,1.5}{
}}
\draw[thick] (2.15,1.3) --node[right, xshift=-.1cm]{$\scriptstyle x$} (2.15,1.7);
\draw[thick, knot, cyan, decorate, decoration={snake, segment length=1mm, amplitude=.2mm}] (2.15,.5) --node[right]{$\scriptstyle X$} (2.15,1) -- (2.15,1.3);
\draw[thick, knot, cyan, decorate, decoration={snake, segment length=1mm, amplitude=.2mm}] (2.15,1.7) -- (2.15,2.15) --node[above]{$\scriptstyle X$} (2.85,2.15) -- (3.15,2.15) -- (3.15,2.85) -- (3.4,2.6);
\filldraw[cyan] (2.15,1.7) node[left, xshift=.1cm]{$\scriptstyle \psi^\dag$} circle (.05cm);
\filldraw[cyan] (2.15,1.3) node[left]{$\scriptstyle \psi$} circle (.05cm);
}
$$
We apply the modified string operator $T_\pi^{(X,\sigma_X)}(\psi)$ which glues $\psi^\dag$ at the site $\ell$ of the excitation and repair the neighboring plaquettes to $\ell$ by applying $B_p$ and $B_q$.
By \eqref{eq:MendStringOperators}, this extends the excitation all the way to the boundary blue edge.
As we used $\psi$, we know that we have not obtained a zero string operator:
$$
0
\neq
\frac{1}{D_\cX}
\sum_{r\in \Irr(\cX)}
d_r
\cdot
\tikzmath[scale=1.5]{
\draw[blue, thick, rounded corners=5pt] (1.7,1.1) rectangle (2.3,1.9);
\node[blue] at (1.8,1.5) {$\scriptstyle r$};
\draw[thick] (2.15,1.3) --node[left]{$\scriptstyle x$} (2.15,1.7);
\draw[thick, knot, cyan, decorate, decoration={snake, segment length=1mm, amplitude=.2mm}] (2.15,.8) --node[right]{$\scriptstyle X$} (2.15,1) -- (2.15,1.3);
\draw[thick, knot, cyan, decorate, decoration={snake, segment length=1mm, amplitude=.2mm}] (2.15,1.7) -- (2.15,2) --node[right]{$\scriptstyle X$} (2.15,2.2);
\filldraw[cyan] (2.15,1.7) node[left, xshift=.1cm]{$\scriptstyle \psi^\dag$} circle (.05cm);
\filldraw[cyan] (2.15,1.3) node[left]{$\scriptstyle \psi$} circle (.05cm);
}
\in
\bbR_{>0} 
\id_{X}.
$$

By \eqref{eq:IzumiTubeAction}, the hopping operator $H^{(X,\sigma_X)}_\pi(\psi)$ takes vectors in the $\Tube(\cX)$-representation corresponding to $(X,\sigma_X)\in \Irr(Z(\cX))$
to vectors in the $\bbC^{\Irr(Z(\cX))}$-representation corresponding to $(X,\sigma_X)$, establishing the equivalence of the excitations.
\end{rem}

\subsection{Dome algebra and boundary excitations}
\label{sec:DomeAlgebra}

In this section, we discuss the analog of $\Tube(\cX)$ which will help us classify localized excitations in our enriched model.
For this section, $\cX$ is a unitary fusion category and $\cA$ is a unitary braided tensor subcategory of $Z(\cX)$ (which need not be modular).
This equips $\cX$ with the structure of an $\cA$-enriched fusion category, i.e., a unitary braided central functor $\Phi^{Z}:\cA\to Z(\cX)$, which in this setting is given by the fully faithful inclusion.
We write $\Phi: \cA\to \cX$ for the composite of $\Forget\circ \Phi^{Z}$, where $\Forget: Z(\cX)\to \cX$ is the forgetful functor.

The following definition was suggested to us by Corey Jones.

\begin{defn}
The \emph{dome algebra} $\Dome^{\textcolor{red}{\cA}}(\cX)$
is the quotient of $\Tube(\cX)$ by the $*$-ideal
$$
\cJ^{\textcolor{red}{\cA}}:=
\left\langle\;\;\sum_{v\in\Irr(\cX)}
\tikzmath{
\roundNbox{}{(0,0)}{.3}{.2}{.2}{$\phi$}

\draw[thick, red] (-.3,-.3) --node[left]{$\scriptstyle \Phi(a)$} (-.3,-.7);
\draw (.3,.7) --node[left]{$\scriptstyle v$} (.3,1.1);
\draw (-.3,-.7) --node[left]{$\scriptstyle v$} (-.3,-1.1);
\draw (.3,-.3) --node[right]{$\scriptstyle x$} (.3,-1.1);
\draw (-.3,.3) --node[left]{$\scriptstyle y$} (-.3,1.1);
\draw[thick, red] (.3,.3) --node[right]{$\scriptstyle \Phi(a)$} (.3,.7);
\fill[fill=blue] (.3,.7) circle (.05cm);
\fill[fill=blue] (-.3,-.7) circle (.05cm);
}
-
\tikzmath{
\draw[thick, red] (-.3,-.3) node[left, yshift=-.2cm]{$\scriptstyle \Phi(a)$} arc (-180:0:.6cm) -- (.9,.3) node[right]{$\scriptstyle \overline{\Phi(a)}$} arc (0:180:.3cm);
\draw[knot] (.3,-.3) -- (.3,-.9) --node[right]{$\scriptstyle x$} (.3,-1.3);
\draw (-.3,.3) --node[left]{$\scriptstyle y$} (-.3,1.1);
\roundNbox{}{(0,0)}{.3}{.2}{.2}{$\phi$}
}
\;\;\right\rangle.
$$
Here, the shaded nodes represent summing over an ONB of $\cX(\Phi(a)\to c)$ with the isometry inner product \eqref{eq:IsometryInnerProdDefn} and its adjoint.
\end{defn}

\begin{rem}
\label{rem:OtherDomeIdealGenerators}
By factoring through simples as in \eqref{eq:FactorTubeThroughSimples}, we see that we also have
$$
\cJ^{\textcolor{red}{\cA}}:=
\left\langle
\sum_{v\in \Irr(\cX)}
\tikzmath{
\draw (.3,-2.1) -- node[right]{$\scriptstyle y$} (.3,-1.3);
\draw (-.3,.3) -- node[left]{$\scriptstyle z$} (-.3,1.1);
\draw (.3,.7) --node[left]{$\scriptstyle v$} (.3,1.1);
\draw (-.3,-1.7) --node[left]{$\scriptstyle v$} (-.3,-2.1);
\draw[thick, red] (.3,.3) --  node[right]{$\scriptstyle \Phi(a)$} (.3,.7);
\draw[thick, red] (-.3,-1.3) -- node[left]{$\scriptstyle \Phi(a)$} (-.3,-1.7);
\draw (0,-.7) --node[right]{$\scriptstyle x$} (0,-.3);
\roundNbox{fill=white}{(0,0)}{.3}{.3}{.3}{$g^\dag$}
\roundNbox{fill=white}{(0,-1)}{.3}{.3}{.3}{$f$}
\fill[fill=blue] (.3,.7) circle (.05cm);
\fill[fill=blue] (-.3,-1.7) circle (.05cm);
}
-
\tikzmath{
\draw (-.3,.3) -- node[left]{$\scriptstyle z$} (-.3,.7);
\draw[thick, red] (-.3,-1.3) node[left, yshift=-.2cm]{$\scriptstyle \Phi(a)$} arc (-180:0:.6cm) -- (.9,.3) node[right]{$\scriptstyle \overline{\Phi(a)}$} arc (0:180:.3cm);
\draw (0,-.7) --node[right]{$\scriptstyle x$} (0,-.3);
\draw[knot] (.3,-1.3) -- node[right]{$\scriptstyle y$} (.3,-1.7) --(.3,-2.3);
\roundNbox{fill=white}{(0,0)}{.3}{.3}{.3}{$g^\dag$}
\roundNbox{fill=white}{(0,-1)}{.3}{.3}{.3}{$f$}
}
\right\rangle.
$$
\end{rem}

\begin{rem}
That $\cJ^{\textcolor{red}{\cA}}$ is a $*$-ideal follows from the fact that
$$
\tikzmath{
\draw[thick, red] (-.3,-.3) node[left, yshift=-.2cm]{$\scriptstyle \Phi(a)$} arc (-180:0:.6cm) -- (.9,.3) node[right]{$\scriptstyle \overline{\Phi(a)}$} arc (0:180:.3cm);
\draw[knot] (.3,-.3) -- (.3,-.9) --node[right]{$\scriptstyle x$} (.3,-1.3);
\draw (-.3,.3) --node[left]{$\scriptstyle y$} (-.3,.7);
\roundNbox{}{(0,0)}{.3}{.2}{.2}{$\phi$}
}
=\tikzmath{
\draw[thick, red] (-.3,-.3) -- (-.3,-.5)    arc (0:-180:.2cm) -- (-.7,.4) arc (0:180:.2cm)--node[left]{$\scriptstyle \Phi(a)$} (-1.1,-.2) arc (0:180:-1.cm) --  (.9,.3) node[right]{$\scriptstyle \overline{\Phi(a)}$} arc (0:180:.3cm);
\draw[knot] (.3,-.3) --node[right]{$\scriptstyle x$} (.3,-.9) -- (.3,-1.6);
\draw (-.3,.3) --node[left]{$\scriptstyle y$} (-.3,1);
\roundNbox{}{(0,0)}{.3}{.2}{.2}{$\phi$}
}
=
\tikzmath{
\draw[thick, red] (-1.2,1.3) .. controls ++(270:1cm) and ++(270:1cm) .. (.6,1.3);
\draw[thick, red, knot] (.6,1.3) arc (0:180:.15cm) -- (.3,.3);
\draw[thick, red, knot] (-.3,-.3) arc (0:-180:.3cm) --node[left]{$\scriptstyle \overline{\Phi(a)}$} (-.9,.7) -- (-.9,1.3) arc (0:180:.15cm);
\draw (.3,-.3) -- (.3,-.9) --node[right]{$\scriptstyle x$} (.3,-.7);
\draw[knot] (-.3,1.7) --node[left]{$\scriptstyle y$} (-.3,.3);
\roundNbox{}{(0,0)}{.3}{.2}{.2}{$\phi$}
}
=
\tikzmath[xscale=-1, yscale=-1]{
\draw[thick, red] (-.3,-.3) node[right, yshift=.2cm]{$\scriptstyle \Phi(a)$} arc (-180:0:.6cm) -- (.9,.3) node[left]{$\scriptstyle \overline{\Phi(a)}$} arc (0:180:.3cm);
\draw[knot] (.3,-.3) -- (.3,-.9) --node[right]{$\scriptstyle y$} (.3,-1.3);
\draw (-.3,.3) --node[left]{$\scriptstyle x$} (-.3,.7);
\roundNbox{}{(0,0)}{.3}{.2}{.2}{$\phi$}
}
\,.
$$
In the final equality above, we used that $Z(\cX)$ is ribbon.
(Note that unitary tensor functors between UFCs are automatically pivotal.)
\end{rem}

The name `dome algebra' comes from thinking about the $\cA$-enriched UFC $\cX$ as a unitary module tensor category; 
$\cX$ corresponds to a 2D boundary of a 3D $\cA$-bulk.
Pushing a morphism from the $\cA$-bulk into the boundary corresponds to applying $\Phi=\Forget\circ \Phi^Z:\cA\to \cX$.
Morphisms in the image of this action map can be lifted back out into the bulk.
Hence the annuli in the defining relation can be viewed as `domes' (or rather `bowls') where the 2D boundary is an annulus and the $\cA$-strings lift into the 3D bulk.\IN{$\cJ^{\textcolor{red}{\cA}}$ the ideal of $\Tube(\cX)$ corresponding to $\textcolor{red}{\cA}$}
$$
\cJ^{\textcolor{red}{\cA}}=
\left\langle
\,\,
\tikzmath{
	\pgfmathsetmacro{\xDome}{.4};	
	\pgfmathsetmacro{\yDome}{.15};	
	\coordinate (a) at (0,0);
	\coordinate (b) at ($ (a) + 4*(\xDome,0) $);
	\coordinate (c) at ($ (a) + (\xDome,0) $);	
	\halfDomeNoDots{(a)}{5*\xDome}{5*\yDome}{}
	\draw (a) --node[below=-.1cm] {\scriptsize{$x$}} ($ (a) + (\xDome,0) $);
	\draw ($ (a) + 3*(\xDome,0) $) --node[below=-.1cm] {\scriptsize{$y$}} ($ (a) + 4*(\xDome,0) $);
	\draw[thick, red] ($ (a) + 5*(\xDome,0) $) ellipse ({3.2*\xDome} and {3.2*\yDome});
    \node at ($ (a) + 7.1*(\xDome,0) + 8*(0,0) $) {\textcolor{red}{\scriptsize{$\Phi(a)$}}};
	\halfDome{(b)}{\xDome}{\yDome}{}
        \fill[white] ($ (c) + (\xDome,0) $) ellipse ({\xDome} and {\yDome});
	\halfDome{(c)}{\xDome}{\yDome}{$\scriptstyle \phi$}
}
\,\,\,\,-\,\,\,\,
\tikzmath{
	\pgfmathsetmacro{\xDome}{.4};	
	\pgfmathsetmacro{\yDome}{.15};	
	\coordinate (a) at (0,0);
	\coordinate (b) at ($ (a) + 4*(\xDome,0) $);
	\coordinate (c) at ($ (a) + (\xDome,0) $);	
	\draw[thick, red] ($ (a) + 2*(\xDome,0) + (240:{\xDome}) $) .. controls ++(240:1cm) and ++(-60:1cm) .. ($ (a) + 2*(\xDome,0) + (-60:\xDome)$) ;
	\node at ($ (a) + 2*(\xDome,0) - 8*(0,\yDome) $) {\textcolor{red}{\scriptsize{$a$}}};
        \draw[super thick, white] (a) arc (-180:0:{5*\xDome} and {5*\yDome});
	\halfDomeNoDots{(a)}{5*\xDome}{5*\yDome}{}
	\draw (a) --node[below=-.1cm] {\scriptsize{$x$}} ($ (a) + (\xDome,0) $);
	\draw ($ (a) + 3*(\xDome,0) $) --node[below=-.1cm] {\scriptsize{$y$}} ($ (a) + 4*(\xDome,0) $);
	\halfDome{(b)}{\xDome}{\yDome}{}
	\halfDome{(c)}{\xDome}{\yDome}{$\scriptstyle \phi$}
}
\,\,
\right\rangle
$$

\begin{thm}
\label{thm:DomeReps=EnrichedCenter}
Suppose $\cA\subset Z(\cX)$ is a braided subcategory and $(X, \sigma_X)\in Z(\cX)$.
The following are equivalent.
\begin{enumerate}
\item 
$(X,\sigma_X)\in Z^\cA(\cX)$, the M\"uger centralizer of $\cA$ in $Z(\cX)$,
\item
The half-braidings $\sigma_{X,a}$ and $\sigma_{a,X}^{-1}$
agree in $\cX(X\Phi(a)\to \Phi(a)X)$,
\item
The half-braidings $\sigma_{X,\Phi(a)}$ and $\sigma_{\Phi(a),X}^{-1}$ have the same matrix elements, i.e., for every
$m\in \cX(X\to y)$,
$n\in \cX(X\to z)$,
$f\in \cX(\Phi(a)y\to x)$,
and
$g\in \cX(z\Phi(a)\to x)$
for $x,y,z\in \Irr(\cX)$,

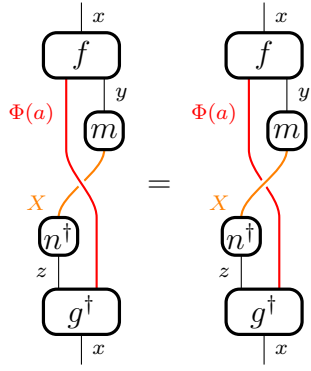
$$
\tikzmath{
\draw[thick, orange] (-.3,-.45) node[left, yshift=.2cm]{$\scriptstyle X$} .. controls ++(90:.3cm) and ++(-90:.3cm) .. (.3,.45);
\draw[knot, thick, red] (.2,-1.4) -- (.2,-.45) .. controls ++(90:.3cm) and ++(-90:.3cm) .. (-.2,.45) --node[left]{$\scriptstyle \Phi(a)$} (-.2,1.4);
\draw (.3,.95) --node[right]{$\scriptstyle y$} (.3,1.4);
\draw (-.3,-.95) --node[left]{$\scriptstyle z$} (-.3,-1.4);
\draw (0,2) --node[right]{$\scriptstyle x$} (0,2.4);
\draw (0,-2) --node[right]{$\scriptstyle x$} (0,-2.4);
\roundNbox{fill=white}{(0,1.7)}{.3}{.2}{.2}{$f$}
\roundNbox{fill=white}{(.3,.7)}{.25}{0}{0}{$m$}
\roundNbox{fill=white}{(-.3,-.7)}{.25}{0}{0}{$n^\dag$}
\roundNbox{fill=white}{(0,-1.7)}{.3}{.2}{.2}{$g^\dag$}
}
=
\tikzmath{
\draw[thick, red] (.2,-1.4) -- (.2,-.45) .. controls ++(90:.3cm) and ++(-90:.3cm) .. (-.2,.45) --node[left]{$\scriptstyle \Phi(a)$} (-.2,1.4);
\draw[thick, orange, knot] (-.3,-.45) node[left, yshift=.2cm]{$\scriptstyle X$} .. controls ++(90:.3cm) and ++(-90:.3cm) .. (.3,.45);
\draw (.3,.95) --node[right]{$\scriptstyle y$} (.3,1.4);
\draw (-.3,-.95) --node[left]{$\scriptstyle z$} (-.3,-1.4);
\draw (0,2) --node[right]{$\scriptstyle x$} (0,2.4);
\draw (0,-2) --node[right]{$\scriptstyle x$} (0,-2.4);
\roundNbox{fill=white}{(0,1.7)}{.3}{.2}{.2}{$f$}
\roundNbox{fill=white}{(.3,.7)}{.25}{0}{0}{$m$}
\roundNbox{fill=white}{(-.3,-.7)}{.25}{0}{0}{$n^\dag$}
\roundNbox{fill=white}{(0,-1.7)}{.3}{.2}{.2}{$g^\dag$}
}
$$
\item
The representation of the tube algebra $\Tube(\cX)$ associated to $(X,\sigma_X)$ descends to a representation of the dome algebra $\Dome^{\textcolor{red}{\cA}}(\cX)$.
\end{enumerate}
\end{thm}
\begin{proof}
That (1), (2), and (3) are equivalent is clear. 
For $(3) \Leftrightarrow(4)$, we use Remark \ref{rem:OtherDomeIdealGenerators}. 
Consider
$$
\phi =
\sum_{v\in \Irr(\cX)}
\tikzmath{
\draw (.3,-2.1) -- node[right]{$\scriptstyle y$} (.3,-1.3);
\draw (-.3,.3) -- node[left]{$\scriptstyle z$} (-.3,1.1);
\draw (.3,.7) --node[left]{$\scriptstyle v$} (.3,1.1);
\draw (-.3,-1.7) --node[left]{$\scriptstyle v$} (-.3,-2.1);
\draw[thick, red] (.3,.3) --  node[right]{$\scriptstyle \Phi(a)$} (.3,.7);
\draw[thick, red] (-.3,-1.3) -- node[left]{$\scriptstyle \Phi(a)$} (-.3,-1.7);
\draw (0,-.7) --node[right]{$\scriptstyle x$} (0,-.3);
\roundNbox{fill=white}{(0,0)}{.3}{.3}{.3}{$g^\dag$}
\roundNbox{fill=white}{(0,-1)}{.3}{.3}{.3}{$f$}
\fill[fill=blue] (.3,.7) circle (.05cm);
\fill[fill=blue] (-.3,-1.7) circle (.05cm);
}
\qquad\text{and}\qquad 
\psi = 
\tikzmath{
\draw (-.3,.3) -- node[left]{$\scriptstyle z$} (-.3,.7);
\draw[thick, red] (-.3,-1.3) node[left, yshift=-.2cm]{$\scriptstyle \Phi(a)$} arc (-180:0:.6cm) -- (.9,.3) node[right]{$\scriptstyle \overline{\Phi(a)}$} arc (0:180:.3cm);
\draw (0,-.7) --node[right]{$\scriptstyle x$} (0,-.3);
\draw[knot] (.3,-1.3) -- node[right]{$\scriptstyle y$} (.3,-1.7) --(.3,-2.3);
\roundNbox{fill=white}{(0,0)}{.3}{.3}{.3}{$g^\dag$}
\roundNbox{fill=white}{(0,-1)}{.3}{.3}{.3}{$f$}
}.
$$
Given $(X,\sigma_X)\in Z(\cX)$, the representation $\cH_X$ of $\Tube(\cX)$ descends to a representation of $\Dome^{\textcolor{red}{\cA}}(\cX)$ if and only if for all $m\in \cX(X\to y)$ and $n\in \cX(X\to z)$,
$$
\langle n | \phi \rhd m \rangle_{\cX(X\to z)}
=
\langle n | \psi \rhd m \rangle_{\cX(X\to z)}.
$$
We calculate that
\begin{align*}
\langle n | \phi\rhd m \rangle 
&=
\frac{1}{d_X}
\tr^\cX_{X}\left(
\tikzmath{
\draw[thick, orange] (.3,-2) -- node[right]{$\scriptstyle X$} (.3,-1.6) --(.3,-1.3);
\draw[thick, orange] (-.3,1.7) -- node[right]{$\scriptstyle X$} (-.3,1.3);
\draw (.3,-.7) -- node[right]{$\scriptstyle y$} (.3,-.3);
\draw (-.3,.3) -- node[left]{$\scriptstyle z$} (-.3,.7);
\draw[knot, red, thick] (.3,.3) node[left, yshift=.15cm]{} arc (180:0:.3cm) -- node[right]{$\scriptstyle \overline{\Phi(a)}$} (.9,-1) arc (0:-180:.6cm) -- node[left]{$\scriptstyle \Phi(a)$} (-.3,-.3);
\roundNbox{fill=white}{(0,0)}{.3}{.3}{.3}{$\phi$}
\roundNbox{fill=white}{(.3,-1)}{.3}{0}{0}{$m$}
\roundNbox{fill=white}{(-.3,1)}{.3}{0}{0}{$n^\dag$}
}
\right)
=
\frac{1}{d_X}\tr^\cX_x\left( \tikzmath{
\draw[thick, orange] (-.3,-.45) node[left, yshift=.2cm]{$\scriptstyle X$} .. controls ++(90:.3cm) and ++(-90:.3cm) .. (.3,.45);
\draw[knot, thick, red] (.2,-1.4) -- (.2,-.45) .. controls ++(90:.3cm) and ++(-90:.3cm) .. (-.2,.45) --node[left]{$\scriptstyle \Phi(a)$} (-.2,1.4);
\draw (.3,.95) --node[right]{$\scriptstyle y$} (.3,1.4);
\draw (-.3,-.95) --node[left]{$\scriptstyle z$} (-.3,-1.4);
\draw (0,2) --node[right]{$\scriptstyle x$} (0,2.4);
\draw (0,-2) --node[right]{$\scriptstyle x$} (0,-2.4);
\roundNbox{fill=white}{(0,1.7)}{.3}{.2}{.2}{$f$}
\roundNbox{fill=white}{(.3,.7)}{.25}{0}{0}{$m$}
\roundNbox{fill=white}{(-.3,-.7)}{.25}{0}{0}{$n^\dag$}
\roundNbox{fill=white}{(0,-1.7)}{.3}{.2}{.2}{$g^\dag$}
}\right) 
&&\text{and}
\displaybreak[1]\\
\langle n | \psi\rhd m \rangle 
&=
\frac{1}{d_X}
\tr^\cX_{X}\left(
\tikzmath{
\draw[thick, orange] (0,-1.7) --node[right]{$\scriptstyle X$} (0,-1.3);
\draw[thick, orange] (0,1.7) --node[right]{$\scriptstyle X$} (0,1.3);
\draw (0,-.7) -- node[right]{$\scriptstyle y$} (0,-.3);
\draw (0,.3) -- node[right]{$\scriptstyle z$} (0,.7);
\roundNbox{fill=white}{(0,0)}{.3}{0}{0}{$\psi$}
\roundNbox{fill=white}{(0,-1)}{.3}{0}{0}{$m$}
\roundNbox{fill=white}{(0,1)}{.3}{0}{0}{$n^\dag$}
}
\right)
= 
\frac{1}{d_X}\tr^\cX_x\left( \tikzmath{
\draw[thick, red] (.2,-1.4) -- (.2,-.45) .. controls ++(90:.3cm) and ++(-90:.3cm) .. (-.2,.45) --node[left]{$\scriptstyle \Phi(a)$} (-.2,1.4);
\draw[thick, orange, knot] (-.3,-.45) node[left, yshift=.2cm]{$\scriptstyle X$} .. controls ++(90:.3cm) and ++(-90:.3cm) .. (.3,.45);
\draw (.3,.95) --node[right]{$\scriptstyle y$} (.3,1.4);
\draw (-.3,-.95) --node[left]{$\scriptstyle z$} (-.3,-1.4);
\draw (0,2) --node[right]{$\scriptstyle x$} (0,2.4);
\draw (0,-2) --node[right]{$\scriptstyle x$} (0,-2.4);
\roundNbox{fill=white}{(0,1.7)}{.3}{.2}{.2}{$f$}
\roundNbox{fill=white}{(.3,.7)}{.25}{0}{0}{$m$}
\roundNbox{fill=white}{(-.3,-.7)}{.25}{0}{0}{$n^\dag$}
\roundNbox{fill=white}{(0,-1.7)}{.3}{.2}{.2}{$g^\dag$}
}\right).
\end{align*}
These two expressions are equal for all $m,n$ if and only if (3) holds.
\end{proof}

\subsection{Localized enriched excitations, punctured enriched skein modules, and enriched string operators}
\label{sec:localEnrichedExcitations}

We now focus on localized boundary excitations of our $\cA$-enriched $\cX$-string net model.
Here, localized means that excitations violate (at most) 4 terms in the Hamiltonian, corresponding to a single boundary edge and its 3 adjacent plaquettes.
These excitations can thus be viewed as vectors in a punctured enriched skein module $\cS^{\textcolor{red}{\cA}}_\cX(\bbA,\textcolor{red}{m},n,1)$.
As above, this space carries an action of the $\rm C^*$-algebra $\fA$ from \eqref{eq:BigTube}, which by an isotopy/Morita equivalence argument, can be reduced to an action similar to \eqref{eq:TubeAction} of $\Tube(\cX)$ after applying $\pi^1_{\ell,v}$.
$$
\fA\curvearrowright
\tikzmath{
\draw[thick, red, xshift=-.4cm, yshift=.4cm] (.7,.7) grid (3.3,2.3);
\draw[thick, red, knot, xshift=-.2cm, yshift=.2cm] (.7,.7) grid (3.3,2.3);
\foreach \x in {1,2,3}{
\foreach \y in {1,2}{
\draw[thick, red] (\x,\y) -- ($ (\x,\y) + (-.6,.6) $);
}}
\foreach \x in {.8,1.8,2.8}{
\foreach \y in {1.2,2.2}{
\fill[red] (\x,\y) circle (.05cm);
}}
\foreach \x in {.6,1.6,2.6}{
\foreach \y in {1.4,2.4}{
\fill[red] (\x,\y) circle (.05cm);
}}
\fill[gray!60, opacity=.5, rounded corners=5pt] (.9,1.5) -- (.9,1.3) -- (1.1,1.1) -- (2.9,1.1) -- (2.9,1.9) -- (2.8,2.1) -- (.9,2.1) -- (.9,1.5);
\draw[thick, knot] (.7,1) -- (3.3,1);
\draw[thick, knot] (.7,2) -- (3.3,2);
\draw[thick, knot] (1,.7) -- (1,2.3);
\draw[thick, knot] (3,.7) -- (3,2.3);
\draw[thick, knot] (2,.7) -- (2,1.3);
\draw[thick, knot] (2,2.3) -- (2,1.7);
\foreach \x in {1,2,3}{
\foreach \y in {1,2}{
\fill (\x,\y) circle (.05cm);
}}
\node at (2.15,2.15) {$\scriptstyle v$};
\node at (2.15,1.75) {$\scriptstyle \ell$};
}
\begin{tikzcd}
\mbox{}
\arrow[r,squiggly,"\pi^1_{\ell,v}"]
&
\mbox{}
\end{tikzcd}
\Tube(\cX)\curvearrowright
\tikzmath{
\draw[thick, red, xshift=-.4cm, yshift=.4cm] (.7,.7) grid (3.3,2.3);
\draw[thick, red, knot, xshift=-.2cm, yshift=.2cm] (.7,.7) grid (3.3,2.3);
\foreach \x in {1,2,3}{
\foreach \y in {1,2}{
\draw[thick, red] (\x,\y) -- ($ (\x,\y) + (-.6,.6) $);
}}
\foreach \x in {.8,1.8,2.8}{
\foreach \y in {1.2,2.2}{
\fill[red] (\x,\y) circle (.05cm);
}}
\foreach \x in {.6,1.6,2.6}{
\foreach \y in {1.4,2.4}{
\fill[red] (\x,\y) circle (.05cm);
}}
\fill[gray!60, opacity=.5, rounded corners=5pt] (.9,1.5) -- (.9,1.3) -- (1.1,1.1) -- (2.9,1.1) -- (2.9,1.9) -- (2.8,2.1) -- (.9,2.1) -- (.9,1.5);
\draw[thick, knot] (.7,1) -- (3.3,1);
\draw[thick, knot] (.7,2) -- (3.3,2);
\draw[thick, knot] (1,.7) -- (1,2.3);
\draw[thick, knot] (3,.7) -- (3,2.3);
\draw[thick, knot] (2,.7) -- (2,1.3);
\draw[thick, knot] (2,2.3) -- (2,2);
\foreach \x in {1,2,3}{
\foreach \y in {1,2}{
\fill (\x,\y) circle (.05cm);
}}
}
$$

\begin{prop}
The ideal $\cJ^{\textcolor{red}{\cA}}\subset \Tube(\cX)$ acts as zero.
Thus the $\Tube(\cX)$ action descends to an action of $\Dome^{\textcolor{red}{\cA}}(\cX)$.
\end{prop}
\begin{proof}
We sketch the proof, which is similar to the calculation \eqref{eq:5=6}.
After applying $\pi^1_{\ell,v}$, we have effectively glued two cubes together, resulting in a large rectangular prism with 8 $\cA$-faces and one $\cX$-face.
The image of our excited state under $\pi^1_{\ell,v}$ still lies in the images of the 8 plaquette operators which act along the boundary of our rectangular prism.
We now compute 
\begin{align*}
\tikzmath{
\draw[thick, orange] (-.3,-.3) --node[left]{$\scriptstyle \Phi(a)$} (-.3,-.7);
\draw (.3,.7) --node[left]{$\scriptstyle c$} (.3,1.1);
\draw (-.3,-.7) --node[left]{$\scriptstyle c$} (-.3,-1.1);
\draw (.3,-.3) --node[right]{$\scriptstyle y$} (.3,-1.1);
\draw (-.3,.3) --node[left]{$\scriptstyle z$} (-.3,1.1);
\draw[thick, orange] (.3,.3) --node[right]{$\scriptstyle \Phi(a)$} (.3,.7);
\fill[fill=blue] (.3,.7) circle (.05cm);
\fill[fill=blue] (-.3,-.7) circle (.05cm);
\roundNbox{}{(0,0)}{.3}{.2}{.2}{$f$}
}
\rhd
\tikzmath{
\draw[thick, red, xshift=-.2cm, yshift=.2cm] (.7,.7) grid (3.3,2.3);
\foreach \x in {1,2,3}{
\foreach \y in {1,2}{
\draw[thick, red] (\x,\y) -- ($ (\x,\y) + (-.4,.4) $);
}}
\foreach \x in {.8,1.8,2.8}{
\foreach \y in {1.2,2.2}{
\fill[red] (\x,\y) circle (.05cm);
}}
\fill[gray!60, opacity=.5, rounded corners=5pt] (.9,1.5) -- (.9,1.3) -- (1.1,1.1) -- (2.9,1.1) -- (2.9,1.9) -- (2.8,2.1) -- (.9,2.1) -- (.9,1.5);
\draw[thick, knot] (.7,1) -- (3.3,1);
\draw[thick, knot] (.7,2) -- (3.3,2);
\draw[thick, knot] (1,.7) -- (1,2.3);
\draw[thick, knot] (3,.7) -- (3,2.3);
\draw[thick, knot] (2,.7) -- (2,1.3);
\draw[thick, knot] (2,2.3) -- (2,2);
\foreach \x in {1,2,3}{
\foreach \y in {1,2}{
\fill (\x,\y) circle (.05cm);
}}
\node at (2.15,1.1) {$\scriptstyle x$};
}
&=
\delta_{x=y}
\left(\frac{d_z}{d_x}\right)^{1/4}
\cdot
\tikzmath{
\draw[thick, red, xshift=-.2cm, yshift=.2cm] (.7,.7) grid (3.3,2.3);
\foreach \x in {1,2,3}{
\foreach \y in {1,2}{
\draw[thick, red] (\x,\y) -- ($ (\x,\y) + (-.4,.4) $);
}}
\foreach \x in {.8,1.8,2.8}{
\foreach \y in {1.2,2.2}{
\fill[red] (\x,\y) circle (.05cm);
}}
\draw[thick, knot] (.7,.7) grid (3.3,2.3);
\fill[white] (1.9,1.6) rectangle (2.1,1.95);
\filldraw[orange, thick, fill=gray!30] (1.5,1.1) -- (1.8,1.1) to[out=0,in=180] (2,1.3) to[out=0,in=180] (2.2,1.1) -- (2.8,1.1) arc (-90:0:.1cm) -- (2.9,1.8) arc (0:90:.1cm) -- (1.2,1.9) arc (90:180:.1cm) -- (1.1,1.2) arc (180:270:.1cm) -- (1.5,1.1);
\draw[thick] (2,1.3) -- (2,1.5);
\filldraw[orange] (2,1.3) node[right, yshift=.1cm]{$\scriptstyle f$} circle (.05cm);
\foreach \x in {1,2,3}{
\foreach \y in {1,2}{
\filldraw (\x,\y) circle (.05cm);
}}
\foreach \x in {0.5,1.5,2.5}{
\foreach \y in {0.5,1.5}{
}}
\node[orange] at (1.5,1.25) {$\scriptstyle \Phi(a)$};
}
\displaybreak[1]\\&=
\delta_{x=y}
\left(\frac{d_z}{d_x}\right)^{1/4}
\cdot
\tikzmath[scale=2]{
\draw[thick, red, xshift=-.4cm, yshift=.4cm] (.9,.9) grid (3.1,2.1);
\foreach \x in {1,2,3}{
\foreach \y in {1,2}{
\draw[thick, red] (\x,\y) -- ($ (\x,\y) + (-.5,.5) $);
}}
\foreach \x in {.6,1.6,2.6}{
\foreach \y in {1.4,2.4}{
\fill[red] (\x,\y) circle (.025cm);
}}
\draw[thick, orange, densely dotted, rounded corners=5pt] (.8,1.6) rectangle (1.4,2.2) ;
\draw[thick, orange, densely dotted, rounded corners=5pt, xshift=1cm] (.8,1.6) rectangle (1.4,2.2) ;
\draw[thick, orange, densely dotted, rounded corners=5pt, xshift=1cm] (1.7,1.7) -- (1.7,2.1) -- (1.9,1.9) -- (1.9,1.3) -- (1.7,1.5) -- (1.7,1.7);
\draw[thick, orange, densely dotted, rounded corners=5pt] (1.5,1.1) -- (1.7,1.1) -- (1.5,1.3) -- (.9,1.3) -- (1.1,1.1) -- (1.5,1.1);
\draw[thick, orange, densely dotted, rounded corners=5pt, xshift=1cm] (1.5,1.1) -- (1.7,1.1) -- (1.5,1.3) -- (.9,1.3) -- (1.1,1.1) -- (1.5,1.1);
\draw[thick, knot] (.9,.9) grid (3.1,2.1);
\fill[white] (1.9,1.6) rectangle (2.1,1.975);
\filldraw[orange, thick, fill=gray!30] (1.5,1.1) -- (1.8,1.1) to[out=0,in=180] (2,1.3) to[out=0,in=180] (2.2,1.1) -- (2.8,1.1) arc (-90:0:.1cm) -- (2.9,1.8) arc (0:90:.1cm) -- (1.2,1.9) arc (90:180:.1cm) -- (1.1,1.2) arc (180:270:.1cm) -- (1.5,1.1);
\draw[thick] (2,1.3) -- (2,1.5);
\filldraw[orange] (2,1.3) node[left]{$\scriptstyle f$} circle (.025cm);
\foreach \x in {1,2,3}{
\foreach \y in {1,2}{
\filldraw (\x,\y) circle (.025cm);
}}
\foreach \x in {0.5,1.5,2.5}{
\foreach \y in {0.5,1.5}{
}}
\draw[thick, orange, knot, rounded corners=5pt] (.7,1.7) -- (.7,2.1) -- (.9,1.9) -- (.9,1.3) -- (.7,1.5) -- (.7,1.7);
\draw[thick, orange, knot, rounded corners=5pt] (1.5,2.1) -- (1.7,2.1) -- (1.5,2.3) -- (.9,2.3) -- (1.1,2.1) -- (1.5,2.1);
\draw[thick, orange, knot, rounded corners=5pt, xshift=1cm] (1.5,2.1) -- (1.7,2.1) -- (1.5,2.3) -- (.9,2.3) -- (1.1,2.1) -- (1.5,2.1);
\node[orange] at (1.5,1.2) {$\scriptstyle \Phi(a)$};
}
\displaybreak[1]\\&=
\delta_{x=y}
\left(\frac{d_z}{d_x}\right)^{1/4}
\cdot
\tikzmath[scale=2]{
\draw[thick, red, xshift=-.4cm, yshift=.4cm] (.9,.9) grid (3.1,2.1);
\foreach \x in {1,2,3}{
\foreach \y in {1,2}{
\draw[thick, red] (\x,\y) -- ($ (\x,\y) + (-.5,.5) $);
}}
\foreach \x in {.6,1.6,2.6}{
\foreach \y in {1.4,2.4}{
\fill[red] (\x,\y) circle (.025cm);
}}
\draw[thick, orange, densely dotted, rounded corners=5pt] (.8,1.6) rectangle (1.4,2.2) ;
\draw[thick, orange, densely dotted, rounded corners=5pt, xshift=1cm] (.8,1.9) -- (.8,2.2) -- (1.4,2.2) -- (1.4,1.9);
\draw[thick, orange, densely dotted, rounded corners=5pt, xshift=1cm] (1.7,1.7) -- (1.7,2.1) -- (1.9,1.9) -- (1.9,1.3) -- (1.7,1.5) -- (1.7,1.7);
\draw[thick, orange, densely dotted, rounded corners=5pt] (1.5,1.15) -- (1.7,1.15) -- (1.5,1.25) -- (.9,1.25) -- (1.1,1.15) -- (1.5,1.15);
\draw[thick, orange, densely dotted, rounded corners=5pt, xshift=1cm] (1.5,1.15) -- (1.7,1.15) -- (1.5,1.25) -- (.9,1.25) -- (1.1,1.15) -- (1.5,1.15);
\draw[orange, thick, knot] (1.5,1.1) -- (1.8,1.1) to[out=0,in=180] (2,1.2) to[out=0,in=180] (2.2,1.1) -- (2.7,1.1) to[out=0,in=-45] (2.7,1.2) to[out=135, in=0] (2.5,1.3) -- (.9,1.3) to[out=180, in = 135] (.9,1.2) to[out=-45, in=180] (1.1,1.1) -- (1.5,1.1);
\draw[thick, knot] (.9,.9) grid (3.1,2.1);
\fill[white] (1.9,1.5) rectangle (2.1,1.975);
\draw[thick, orange, densely dotted, rounded corners=5pt, xshift=1cm] (.8,1.9) -- (.8,1.6) -- (1.4,1.6) -- (1.4,1.9) ;
\filldraw[orange] (2,1.2) node[right, xshift=.1cm]{$\scriptstyle f$} circle (.025cm);
\foreach \x in {1,2,3}{
\foreach \y in {1,2}{
\filldraw (\x,\y) circle (.025cm);
}}
\foreach \x in {0.5,1.5,2.5}{
\foreach \y in {0.5,1.5}{
}}
\draw[thick, orange, knot, rounded corners=5pt] (.7,1.7) -- (.7,2.1) -- (.9,1.9) -- (.9,1.3) -- (.7,1.5) -- (.7,1.7);
\draw[thick, orange, knot, rounded corners=5pt] (1.5,2.1) -- (1.7,2.1) -- (1.5,2.3) -- (.9,2.3) -- (1.1,2.1) -- (1.5,2.1);
\draw[thick, orange, knot, rounded corners=5pt, xshift=1cm] (1.5,2.1) -- (1.7,2.1) -- (1.5,2.3) -- (.9,2.3) -- (1.1,2.1) -- (1.5,2.1);
}
\\&=
\delta_{x=y}
\left(\frac{d_z}{d_x}\right)^{1/4}
\cdot
\tikzmath[scale=1.2]{
\draw[thick, red, xshift=-.2cm, yshift=.2cm] (.7,.7) grid (3.3,2.3);
\foreach \x in {1,2,3}{
\foreach \y in {1,2}{
\draw[thick, red] (\x,\y) -- ($ (\x,\y) + (-.4,.4) $);
}}
\foreach \x in {.8,1.8,2.8}{
\foreach \y in {1.2,2.2}{
\fill[red] (\x,\y) circle (.05cm);
}}
\fill[gray!60, opacity=.5, rounded corners=5pt] (.9,1.5) -- (.9,1.3) -- (1.1,1.1) -- (2.9,1.1) -- (2.9,1.9) -- (2.8,2.1) -- (.9,2.1) -- (.9,1.5);
\draw[thick, orange, knot] (2,1.3) circle (.15cm);
\draw[thick, knot] (.7,1) -- (3.3,1);
\draw[thick, knot] (.7,2) -- (3.3,2);
\draw[thick, knot] (1,.7) -- (1,2.3);
\draw[thick, knot] (3,.7) -- (3,2.3);
\draw[thick, knot] (2,.7) -- (2,1.6);
\draw[thick, knot] (2,2.3) -- (2,2);
\foreach \x in {1,2,3}{
\foreach \y in {1,2}{
\fill (\x,\y) circle (.05cm);
}}
\filldraw[orange] (2,1.45) node[right, yshift=.1cm]{$\scriptstyle f$} circle (.05cm);
}
\\&=
\tikzmath{
\draw[thick, orange] (-.3,-.3) node[left, yshift=-.2cm]{$\scriptstyle \Phi(a)$} arc (-180:0:.6cm) -- (.9,.3) node[right]{$\scriptstyle \overline{\Phi(a)}$} arc (0:180:.3cm);
\draw[knot] (.3,-.3) -- (.3,-.9) --node[right]{$\scriptstyle y$} (.3,-1.3);
\draw (-.3,.3) --node[left]{$\scriptstyle z$} (-.3,1.1);
\roundNbox{}{(0,0)}{.3}{.2}{.2}{$f$}
}
\rhd
\tikzmath{
\draw[thick, red, xshift=-.2cm, yshift=.2cm] (.7,.7) grid (3.3,2.3);
\foreach \x in {1,2,3}{
\foreach \y in {1,2}{
\draw[thick, red] (\x,\y) -- ($ (\x,\y) + (-.4,.4) $);
}}
\foreach \x in {.8,1.8,2.8}{
\foreach \y in {1.2,2.2}{
\fill[red] (\x,\y) circle (.05cm);
}}
\fill[gray!60, opacity=.5, rounded corners=5pt] (.9,1.5) -- (.9,1.3) -- (1.1,1.1) -- (2.9,1.1) -- (2.9,1.9) -- (2.8,2.1) -- (.9,2.1) -- (.9,1.5);
\draw[thick, knot] (.7,1) -- (3.3,1);
\draw[thick, knot] (.7,2) -- (3.3,2);
\draw[thick, knot] (1,.7) -- (1,2.3);
\draw[thick, knot] (3,.7) -- (3,2.3);
\draw[thick, knot] (2,.7) -- (2,1.3);
\draw[thick, knot] (2,2.3) -- (2,2);
\foreach \x in {1,2,3}{
\foreach \y in {1,2}{
\fill (\x,\y) circle (.05cm);
}}
\node at (2.15,1.1) {$\scriptstyle x$};
}
\end{align*}
This immediately implies that anything in $\cJ^{\textcolor{red}{\cA}}$ acts as zero.     
\end{proof}

Looking at our string operators from \S\ref{sec:StringOperators} above,
it is clear from \eqref{eq:IzumiTubeAction} that the dome ideal acts as zero whenever the anyon comes from $Z^\cA(\cX)$.
Indeed, if $X\in Z^\cA(\cX)$, then calculating in $Z(\cX)$, remembering that squiggly-drawn anyons represent $Z^\cA(\cX)^{\rev}$ as opposed to $Z^\cA(\cX)\subset Z(\cX)$ (and suppressing 4th roots),
$$
\sum_{x\in \Irr(\cX)}
\tikzmath{
\draw[red, thick] (-.4,0) arc (180:0:.4cm);
\draw (-.4,0) arc (-180:0:.4cm);
\filldraw[blue] (-.4,0) circle (.05cm);
\filldraw[blue] (.4,0) circle (.05cm);
\draw[thick, knot, orange, decorate, decoration={snake, segment length=1mm, amplitude=.2mm}](0,-.7) --node[right]{$\scriptstyle X$} (0,-.4) -- (0,0);
\node[red] at (-.7,.3) {$\scriptstyle \Phi(a)$};
\node at (-.45,-.3) {$\scriptstyle x$};
\draw[thick] (0,0) -- (0,.7);
\filldraw[red] (0,.4) node[right, yshift=.15cm]{$\scriptstyle f$} circle (.05cm);
\filldraw[orange] (0,0) node[right,xshift=-.05cm]{$\scriptstyle \psi$} circle (.05cm);
}
:=
\sum_{x\in \Irr(\cX)}
\tikzmath{
\draw[thick, orange](0,-.7) --node[right]{$\scriptstyle X$} (0,-.4) -- (0,0);
\draw[red, thick] (-.4,0) arc (180:0:.4cm);
\draw[knot] (-.4,0) arc (-180:0:.4cm);
\filldraw[blue] (-.4,0) circle (.05cm);
\filldraw[blue] (.4,0) circle (.05cm);
\node[red] at (-.7,.3) {$\scriptstyle \Phi(a)$};
\node at (-.45,-.3) {$\scriptstyle x$};
\draw[thick] (0,0) -- (0,.7);
\filldraw[red] (0,.4) node[right, yshift=.15cm]{$\scriptstyle f$} circle (.05cm);
\filldraw[orange] (0,0) node[right,xshift=-.05cm]{$\scriptstyle \psi$} circle (.05cm);
}
=
\tikzmath{
\draw[thick, orange](0,-.7) --node[right]{$\scriptstyle X$} (0,-.4) -- (0,0);
\draw[red, thick] (-.4,0) arc (180:0:.4cm);
\draw[red, thick, knot] (-.4,0) arc (-180:0:.4cm);
\node[red] at (-.7,.3) {$\scriptstyle \Phi(a)$};
\draw[thick] (0,0) -- (0,.7);
\filldraw[red] (0,.4) node[right, yshift=.15cm]{$\scriptstyle f$} circle (.05cm);
\filldraw[orange] (0,0) node[right,xshift=-.05cm]{$\scriptstyle \psi$} circle (.05cm);
}
=
\tikzmath{
\draw[red, thick] (-.4,0) arc (180:0:.4cm);
\draw[red, thick] (-.4,0) arc (-180:0:.4cm);
\draw[thick, knot, orange](0,-.7) --node[right]{$\scriptstyle X$} (0,-.4) -- (0,0);
\node[red] at (-.7,.3) {$\scriptstyle \Phi(a)$};
\draw[thick] (0,0) -- (0,.7);
\filldraw[red] (0,.4) node[right, yshift=.15cm]{$\scriptstyle f$} circle (.05cm);
\filldraw[orange] (0,0) node[right,xshift=-.05cm]{$\scriptstyle \psi$} circle (.05cm);
}
=
\tikzmath{
\draw[thick, orange](0,-.7) --node[right]{$\scriptstyle X$} (0,-.2);
\draw[red, thick] (-.2,.2) arc (180:0:.2cm);
\draw[red, thick] (-.2,.2) arc (-180:0:.2cm);
\draw[thick, knot] (0,-.2) -- (0,.7);
\filldraw[red] (0,.4) node[right, yshift=.15cm]{$\scriptstyle f$} circle (.05cm);
\filldraw[orange] (0,-.2) node[right,xshift=-.05cm]{$\scriptstyle \psi$} circle (.05cm);
}
=:
\tikzmath{
\draw[thick, orange, decorate, decoration={snake, segment length=1mm, amplitude=.2mm}](0,-.7) --node[right]{$\scriptstyle X$} (0,-.2);
\draw[red, thick] (-.2,.2) arc (180:0:.2cm);
\draw[red, thick] (-.2,.2) arc (-180:0:.2cm);
\draw[thick, knot] (0,-.2) -- (0,.7);
\node[red] at (-.5,.3) {$\scriptstyle \Phi(a)$};
\filldraw[red] (0,.4) node[right, yshift=.15cm]{$\scriptstyle f$} circle (.05cm);
\filldraw[orange] (0,-.2) node[right,xshift=-.05cm]{$\scriptstyle \psi$} circle (.05cm);
}
\,.
$$

\begin{rem}
 As in Remark~\ref{rem:TubeAlgebraMoritaCorrect}, we can additionally see that isomorphism classes of $\Dome^{\textcolor{red}{\cA}}(\cX)$ representations are in bijection with boundary anyon types.
 In the case of a twice-punctured sphere with internal $\cA$-edges, 
 $$
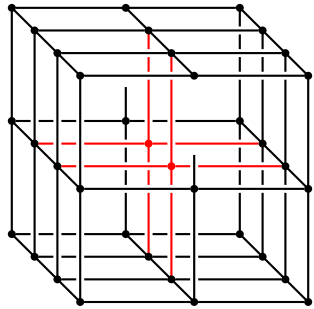
\tikzmath[scale=1.5]{
\begin{scope}[xshift=-.4cm, yshift=.4cm]
\draw[thick, knot, thick] (1,1) -- (3,1);
\draw[thick, knot, thick] (1,2) -- (3,2);
\draw[thick, knot, thick] (1,3) -- (3,3);
\draw[thick, knot, thick] (1,1) -- (1,3);
\draw[thick, knot, thick] (3,1) -- (3,3);
\draw[thick, knot, thick] (2,1) -- (2,2.3);
\end{scope}
\begin{scope}[xshift=-.2cm, yshift=.2cm]
\draw[thick, knot, thick] (1,1) -- (3,1);
\draw[thick, knot, thick, red] (1,2) -- (3,2);
\draw[thick, knot, thick] (1,3) -- (3,3);
\draw[thick, knot, thick] (1,1) -- (1,3);
\draw[thick, knot, thick] (3,1) -- (3,3);
\draw[thick, knot, thick, red] (2,1) -- (2,3);
\end{scope}
\begin{scope}
\draw[thick, knot, thick] (1,1) -- (3,1);
\draw[thick, knot, thick, red] (1,2) -- (3,2);
\draw[thick, knot, thick] (1,3) -- (3,3);
\draw[thick, knot, thick] (1,1) -- (1,3);
\draw[thick, knot, thick] (3,1) -- (3,3);
\draw[thick, knot, thick, red] (2,1) -- (2,3);
\end{scope}
\begin{scope}[xshift=.2cm, yshift=-.2cm]
\draw[thick, knot, thick] (1,1) -- (3,1);
\draw[thick, knot, thick] (1,2) -- (3,2);
\draw[thick, knot, thick] (1,3) -- (3,3);
\draw[thick, knot, thick] (1,1) -- (1,3);
\draw[thick, knot, thick] (3,1) -- (3,3);
\draw[thick, knot, thick] (2,1) -- (2,2.3);
\end{scope}
\foreach \x in {1,2,3}{
\foreach \y in {1,2,3}{
\fill ($ (\x,\y) + (-.4,.4) $) circle (.035cm);
\fill ($ (\x,\y) + (.2,-.2) $) circle (.035cm);
}}
\foreach \x in {1,2,3}{
\foreach \y in {1,3}{
\fill ($ (\x,\y) + (-.2,.2) $) circle (.035cm);
\fill (\x,\y) circle (.035cm);
}}
\fill (1,2) circle (.035cm);
\fill ($ (1,2) + (-.2,.2) $) circle (.035cm);
\fill[red] (2,2) circle (.035cm);
\fill[red] ($ (2,2) + (-.2,.2) $) circle (.035cm);
\fill (3,2) circle (.035cm);
\fill ($ (3,2) + (-.2,.2) $) circle (.035cm);
\foreach \x in {1,2,3}{
\foreach \y in {1,3}{
\draw[thick] ($ (\x,\y) + (.2,-.2) $) -- ($ (\x,\y) + (-.4,.4) $);
}}
\draw[thick] ($ (1,2) + (.2,-.2) $) -- ($ (1,2) + (-.4,.4) $);
\draw[thick] ($ (3,2) + (.2,-.2) $) -- ($ (3,2) + (-.4,.4) $);
}
$$
 the space of ground states is now the subspace of $\Tube(\cX)$ which is compatible with the $\textcolor{red}{\cA}$-bulk.
 Since $\Dome^{\textcolor{red}{\cA}}(\cX)$ is just the quotient of $\Tube(\cX)$ obtained by imposing the compatibility relations $\cJ^{\textcolor{red}{\cA}}$, this is exactly the regular representation of $\Dome^{\textcolor{red}{\cA}}(\cX)$, {i.e.} $\Dome^{\textcolor{red}{\cA}}(\cX)$ as a $\Dome^{\textcolor{red}{\cA}}(\cX)-\Dome^{\textcolor{red}{\cA}}(\cX)$ bimodule.
An argument similar to Remark~\ref{rem:TubeAlgebraMoritaCorrect} shows that anyon types are in bijective correspondence with $\Irr(\Rep(\Dome^{\textcolor{red}{\cA}}(\cX)))=\Irr(Z^{\cA}(\cX))$.
\end{rem}

\subsection{Sphere algebra, bulk excitations, and bulk enrichment of boundary}
\label{sec:SphereAlgebra}

Suppose $\cA$ is an arbitrary UBFC and $\cX$ is an $\cA$-enriched UFC such that the unitary braided central functor $\Phi^Z:\cA\to Z(\cX)$ is fully faithful.
In this section, we analyze the (3+1)D bulk $\cA$-excitations and show they are $Z_2(\cA)$ by constructing representations of the \emph{sphere algebra} $\Sphere(\cA)$, a further quotient of $\Tube(\cA)$.
This is consistent with \cite[Lem.~2.16]{MR4654609} which states that $\Omega Z(\Sigma \cA)= Z_2(\cA)$.
By \cite[Prop.~4.3]{MR3022755}, $Z^\cA(\cX)$ is canonically $Z_2(\cA)$-enriched as a UBFC, i.e., $Z_2(Z^\cA(\cX))=Z_2(\cA)$.
We realize this enrichment by using hopping operators to bring bulk excitations into the boundary.

The following definition was suggested to us by Corey Jones.

\begin{defn}
The \emph{sphere algebra} $\Sphere(\cA)$ \IN{$\Sphere(\cA)$ the sphere algebra of $\cA$}
is the quotient of $\Tube(\cA)$ by the $*$-ideal
$$
\cJ:=
\left\langle\;\;
\tikzmath{
\draw[thick, orange] (-.3,-.3) --node[left]{$\scriptstyle c$} (-.3,-.7);
\draw[thick, orange] (.3,-.3) --node[right]{$\scriptstyle a$} (.3,-.7);
\draw[thick, orange] (-.3,.3) --node[left]{$\scriptstyle b$} (-.3,.7);
\draw[thick, orange] (.3,.3) --node[right]{$\scriptstyle c$} (.3,.7);
\roundNbox{}{(0,0)}{.3}{.2}{.2}{$\phi$}
}
-
\tikzmath{
\draw[thick, orange] (-.3,-.3) node[left, yshift=-.2cm]{$\scriptstyle c$} arc (-180:0:.6cm) -- (.9,.3) node[right]{$\scriptstyle \overline{c}$} arc (0:180:.3cm);
\draw[knot, thick, orange] (.3,-.3) -- (.3,-.9) --node[right]{$\scriptstyle a$} (.3,-1.3);
\draw[thick, orange] (-.3,.3) --node[left]{$\scriptstyle b$} (-.3,1.1);
\roundNbox{}{(0,0)}{.3}{.2}{.2}{$\phi$}
}
\,,\,
\tikzmath{
\draw[thick, orange] (-.3,-.3) --node[left]{$\scriptstyle c$} (-.3,-.7);
\draw[thick, orange] (.3,-.3) --node[right]{$\scriptstyle a$} (.3,-.7);
\draw[thick, orange] (-.3,.3) --node[left]{$\scriptstyle b$} (-.3,.7);
\draw[thick, orange] (.3,.3) --node[right]{$\scriptstyle c$} (.3,.7);
\roundNbox{}{(0,0)}{.3}{.2}{.2}{$\psi$}
}
-
\tikzmath{
\draw[thick, orange] (.3,-.3) -- (.3,-.9) --node[right]{$\scriptstyle a$} (.3,-1.3);
\draw[knot, thick, orange] (-.3,-.3) node[left, yshift=-.2cm]{$\scriptstyle c$} arc (-180:0:.6cm) -- (.9,.3) node[right]{$\scriptstyle \overline{c}$} arc (0:180:.3cm);
\draw[thick, orange] (-.3,.3) --node[left]{$\scriptstyle b$} (-.3,1.1);
\roundNbox{}{(0,0)}{.3}{.2}{.2}{$\psi$}
}
\;\;\right\rangle.
$$
\end{defn}
Similar to the dome algebra, $\Sphere(\cA)$ corresponds to thinking of $\cA$ as an $\cA\boxtimes \cA^{\rev}$-enriched UFC.
$$
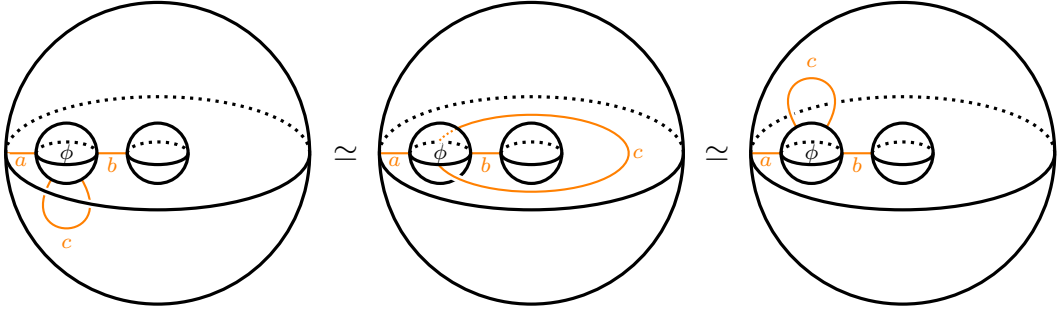
\tikzmath{
	\pgfmathsetmacro{\xDome}{.4};	
	\pgfmathsetmacro{\yDome}{.15};	
        \coordinate (a) at (0,0);
	\coordinate (b) at ($ (a) + 4*(\xDome,0) $);
	\coordinate (c) at ($ (a) + (\xDome,0) $);
	\draw[thick, orange] ($ (a) + 2*(\xDome,0) + (240:{\xDome}) $) .. controls ++(240:1cm) and ++(-60:1cm) .. ($ (a) + 2*(\xDome,0) + (-60:\xDome)$) ;
	\node at ($ (a) + 2*(\xDome,0) - 8*(0,\yDome) $) {\textcolor{orange}{\scriptsize{$c$}}};
    \draw[thick, orange] (a) --node[below=-.1cm] {\scriptsize{$a$}} ($ (a) + (\xDome,0) $);
	\draw[thick, orange] ($ (a) + 3*(\xDome,0) $) --node[below=-.1cm] {\scriptsize{$b$}} ($ (a) + 4*(\xDome,0) $);
        \draw[very thick, dotted] (a) arc (180:0:{5*\xDome} and {5*\yDome});
        \draw[very thick, knot] (a) arc (-180:0:{5*\xDome} and {5*\yDome});
        \draw[very thick] (a) arc (-180:180:{5*\xDome});
        \draw[very thick, dotted] (b) arc (180:0:{\xDome} and {\yDome});
        \draw[very thick] (b) arc (-180:0:{\xDome} and {\yDome});
        \draw[very thick] (b) arc (-180:180:{\xDome});
        \draw[very thick, dotted] (c) arc (180:0:{\xDome} and {\yDome});
        \draw[very thick] (c) arc (-180:0:{\xDome} and {\yDome});
        \draw[very thick] (c) arc (-180:180:{\xDome});
        \node at ($ (c) + (\xDome,0) $) {$\scriptstyle \phi$};
}
\,\,\simeq\,\,
\tikzmath{
	\pgfmathsetmacro{\xDome}{.4};	
	\pgfmathsetmacro{\yDome}{.15};	
        \coordinate (a) at (0,0);
	\coordinate (b) at ($ (a) + 4*(\xDome,0) $);
	\coordinate (c) at ($ (a) + (\xDome,0) $);
    \node at ($ (a) + 8.5*(\xDome,0) $) {\textcolor{orange}{\scriptsize{$c$}}};
    \draw[thick, orange] (a) --node[below=-.1cm] {\scriptsize{$a$}} ($ (a) + (\xDome,0) $);
	\draw[thick, orange] ($ (a) + 3*(\xDome,0) $) --node[below=-.1cm] {\scriptsize{$b$}} ($ (a) + 4*(\xDome,0) $);
        \draw[very thick, dotted] (a) arc (180:0:{5*\xDome} and {5*\yDome});
        \draw[very thick] (a) arc (-180:0:{5*\xDome} and {5*\yDome});
        \draw[very thick] (a) arc (-180:180:{5*\xDome});
        \draw[very thick, dotted] (b) arc (180:0:{\xDome} and {\yDome});
        \draw[very thick] (b) arc (-180:0:{\xDome} and {\yDome});
        \draw[very thick] (b) arc (-180:180:{\xDome});
        \draw[very thick, dotted] (c) arc (180:0:{\xDome} and {\yDome});
        \draw[very thick] (c) arc (-180:0:{\xDome});
        \node at ($ (c) + (\xDome,0) $) {$\scriptstyle \phi$};
	\draw[thick, orange, knot] ($ (b) + (\xDome,0) + (-163:{3.2*\xDome} and {3.4*\yDome}) $) arc (-163:140:{3.2*\xDome} and {3.4*\yDome});
	\draw[thick, orange, knot, densely dotted] ($ (b) + (\xDome,0) + (140:{3.2*\xDome} and {3.4*\yDome}) $) arc (140:163:{3.2*\xDome} and {3.4*\yDome});
    \draw[very thick] (c) arc (-180:0:{\xDome} and {\yDome});
        \draw[very thick] (c) arc (180:0:{\xDome});
}
\,\,\simeq\,\,
\tikzmath{
	\pgfmathsetmacro{\xDome}{.4};	
	\pgfmathsetmacro{\yDome}{.15};	
        \coordinate (a) at (0,0);
	\coordinate (b) at ($ (a) + 4*(\xDome,0) $);
	\coordinate (c) at ($ (a) + (\xDome,0) $);
        \draw[very thick, dotted] (a) arc (180:0:{5*\xDome} and {5*\yDome});
	\draw[thick, orange, knot] ($ (a) + 2*(\xDome,0) + (60:{\xDome}) $) .. controls ++(60:1cm) and ++(120:1cm) .. ($ (a) + 2*(\xDome,0) + (120:\xDome)$) ;
	\node at ($ (a) + 2*(\xDome,0) + 8*(0,\yDome) $) {\textcolor{orange}{\scriptsize{$c$}}};
    \draw[thick, orange] (a) --node[below=-.1cm] {\scriptsize{$a$}} ($ (a) + (\xDome,0) $);
	\draw[thick, orange] ($ (a) + 3*(\xDome,0) $) --node[below=-.1cm] {\scriptsize{$b$}} ($ (a) + 4*(\xDome,0) $);
        \draw[very thick] (a) arc (-180:0:{5*\xDome} and {5*\yDome});
        \draw[very thick] (a) arc (-180:180:{5*\xDome});
        \draw[very thick, dotted] (b) arc (180:0:{\xDome} and {\yDome});
        \draw[very thick] (b) arc (-180:0:{\xDome} and {\yDome});
        \draw[very thick] (b) arc (-180:180:{\xDome});
        \draw[very thick, dotted] (c) arc (180:0:{\xDome} and {\yDome});
        \draw[very thick] (c) arc (-180:0:{\xDome} and {\yDome});
        \draw[very thick] (c) arc (-180:180:{\xDome});
        \node at ($ (c) + (\xDome,0) $) {$\scriptstyle \phi$};
}
$$

The algebra $\Sphere(\cA)$ is actually abelian.
Indeed, consider the algebra maps:
\begin{equation}
\label{eqn:SphereAlgMaps}
\tikzmath{
\draw[thick, orange] (-.3,-.3) --node[left]{$\scriptstyle c$} (-.3,-.7);
\draw[thick, orange] (.3,-.3) --node[right]{$\scriptstyle a$} (.3,-.7);
\draw[thick, orange] (-.3,.3) --node[left]{$\scriptstyle b$} (-.3,.7);
\draw[thick, orange] (.3,.3) --node[right]{$\scriptstyle c$} (.3,.7);
\roundNbox{}{(0,0)}{.3}{.2}{.2}{$\phi$}
}
\overset{\lambda}{\longmapsto}
\delta_{a=b}
\cdot
\tikzmath{
\draw[thick, orange] (-.3,-.3) node[left, yshift=-.2cm]{$\scriptstyle c$} arc (-180:0:.6cm) -- (.9,.3) node[right]{$\scriptstyle \overline{c}$} arc (0:180:.3cm);
\draw[knot, thick, orange] (.3,-.3) -- (.3,-.9) --node[right]{$\scriptstyle a$} (.3,-1.3);
\draw[thick, orange] (-.3,.3) --node[left]{$\scriptstyle b$} (-.3,1.1);
\roundNbox{}{(0,0)}{.3}{.2}{.2}{$\phi$}
}
\qquad\text{and}\qquad
\tikzmath{
\draw[thick, orange] (-.3,-.3) --node[left]{$\scriptstyle c$} (-.3,-.7);
\draw[thick, orange] (.3,-.3) --node[right]{$\scriptstyle a$} (.3,-.7);
\draw[thick, orange] (-.3,.3) --node[left]{$\scriptstyle b$} (-.3,.7);
\draw[thick, orange] (.3,.3) --node[right]{$\scriptstyle c$} (.3,.7);
\roundNbox{}{(0,0)}{.3}{.2}{.2}{$\psi$}
}
\overset{\rho}{\longmapsto}
\delta_{a=b}
\cdot
\tikzmath{
\draw[thick, orange] (.3,-.3) -- (.3,-.9) --node[right]{$\scriptstyle a$} (.3,-1.3);
\draw[knot, thick, orange] (-.3,-.3) node[left, yshift=-.2cm]{$\scriptstyle c$} arc (-180:0:.6cm) -- (.9,.3) node[right]{$\scriptstyle \overline{c}$} arc (0:180:.3cm);
\draw[thick, orange] (-.3,.3) --node[left]{$\scriptstyle b$} (-.3,1.1);
\roundNbox{}{(0,0)}{.3}{.2}{.2}{$\psi$}
}\,.
\end{equation}
We see that if $\phi: ca\to ac$ and $\psi: da\to ad$, then $\phi\psi - \psi\phi \in \cJ$:
\begin{align*}
\phi\psi - \psi\phi
&=
\phi\psi 
- \lambda(\phi)\psi + \lambda(\phi)\psi
- \lambda(\phi)\lambda(\psi) + \underbrace{\lambda(\phi)\lambda(\psi)}_{=\lambda(\psi)\lambda(\phi)}
- \lambda(\psi)\phi + \lambda(\psi)\phi
-\psi\phi
\\
&=
\big(\phi - \lambda(\phi)\big)\psi 
+ 
\lambda(\phi)\big(\psi - \lambda(\psi)\big) 
+ 
\lambda(\psi)\big(\lambda(\phi)-\phi\big) 
+ 
\big(\lambda(\psi)-\psi\big)\phi
\in \cJ.
\end{align*}
The next lemma gives an abstract characterization of $\Sphere(\cA)$.

\begin{lem}
\label{lem:SphereAlgebraIso}
The following $*$-algebras are isomorphic:
\begin{enumerate}
    \item $\Sphere(\cA)$,
    \item $T_2/(T_2\cap \cJ)$ where $T_2\subset \Tube(\cX)$ is the subalgebra spanned by $\phi:ca\to bc$ for $a,b,c\in \Irr(\cA)$ such that $a=b\in Z_2(\cA)$, and
    \item $\bbC^{\Irr(Z_2(\cA))}$.
\end{enumerate}
\end{lem}
\begin{proof}
First, given 
$a,b,c\in \Irr(\cA)$
and
$\phi : ca\to bc$ in $\Tube(\cA)$,
if $a\neq b$, then $\phi \in \cJ$.
If $a\neq b$, then $\lambda(\phi)=0$, so 
$\phi = \phi-\lambda(\phi)\in \cJ$.

Second, suppose $a\notin \Irr(Z_2(\cA))$, and pick $c\in \Irr(\cA)$ whose $S$-matrix entry $S_{a,c}\neq d_ad_c$.
We claim that $\id_a\in \cJ$.
Indeed, 
$$
\frac{S_{a,c}
-
d_ad_c}{d_a}
\cdot
\id_a
=
\tikzmath{
\draw[thick, orange] (0,0) -- (0,.5);
\draw[thick, red, knot] (0,0) circle (.3cm);
\draw[thick, orange, knot] (0,0) -- (0,-.5);
\node[red] at (-.45,0) {$\scriptstyle c$};
\node[orange] at (.15,0) {$\scriptstyle a$};
}
-
\tikzmath{
\draw[thick, orange] (0,-.5) -- (0,.5);
\draw[thick, red, knot] (0,0) circle (.3cm);
\node[red] at (-.45,0) {$\scriptstyle c$};
\node[orange] at (.15,0) {$\scriptstyle a$};
}
=
\lambda(\beta_{a,c})
-
\rho(\beta_{a,c})
=
\lambda(\beta_{a,c})
-
\beta_{a,c}
+
\beta_{a,c}
-
\rho(\beta_{a,c})
\in\cJ.
$$

\item[\underline{$(1)\cong(2)$:}]
By the second isomorphism theorem, $T_2/(T_2\cap \cJ) \cong (T_2+\cJ)/\cJ$.
By the two facts proven above, every representative of an element of $\Sphere(\cA)$ in $\Tube(\cA)$ can be taken from $T_2$.
This means the $*$-algebra map $(T_2+\cJ)/\cJ \to \Sphere(\cA)$
given by
$x+\cJ\mapsto x+\cJ$ 
is onto and thus an isomorphism.

\item[\underline{$(2)\cong(3)$:}]
By the two facts proven above, the map $\delta_a\mapsto \id_a+T_2\cap \cJ$ is a surjective $*$-algebra map $\bbC^{\Irr(Z_2(\cA))}\to T_2/(T_2\cap \cJ)$.
The inverse is given by $\lambda|_{T_2}$, which descends to $T_2/(T_2\cap \cJ)$ since $T_2\cap \cJ \subset \ker(\lambda|_{T_2})$.
\end{proof}

The proof of the following theorem is similar to Theorem \ref{thm:DomeReps=EnrichedCenter} and omitted.
The equivalences of (1)-(4) below are all well known; e.g., see \cite[Lem.~7.5]{MR1966525}.

\begin{thm}
Suppose $\cA$ is a braided fusion category.
The following are equivalent.
\begin{enumerate}
\item 
$a\in Z_2(\cA)$, the M\"uger center of $\cA$,
\item
The braidings $\beta_{a,b}$ and $\beta_{b,a}^{-1}$
agree for all $b\in \cA$,
\item
The two canonical lifts of $a\in\cA$ to the Drinfeld center $Z(\cA)$ agree,
\item
The braidings $\beta_{a,b}$ and $\beta_{b,a}^{-1}$ have the same matrix elements, i.e., 
for every
$m\in \cA(a\to c)$,
$n\in \cA(a\to d)$,
$f\in \cA(b c\to e)$,
and
$g\in \cA(d b\to e)$
for $c,d,e\in \Irr(\cA)$,
$$
\tikzmath{
\draw[thick, red] (-.3,-.45) node[left, yshift=.2cm]{$\scriptstyle a$} .. controls ++(90:.3cm) and ++(-90:.3cm) .. (.3,.45);
\draw[knot, thick, orange] (.2,-1.4) -- (.2,-.45) .. controls ++(90:.3cm) and ++(-90:.3cm) .. (-.2,.45) --node[left]{$\scriptstyle b$} (-.2,1.4);
\draw (.3,.95) --node[right]{$\scriptstyle c$} (.3,1.4);
\draw (-.3,-.95) --node[left]{$\scriptstyle d$} (-.3,-1.4);
\draw (0,2) --node[right]{$\scriptstyle e$} (0,2.4);
\draw (0,-2) --node[right]{$\scriptstyle e$} (0,-2.4);
\roundNbox{fill=white}{(0,1.7)}{.3}{.2}{.2}{$f$}
\roundNbox{fill=white}{(.3,.7)}{.25}{0}{0}{$m$}
\roundNbox{fill=white}{(-.3,-.7)}{.25}{0}{0}{$n^\dag$}
\roundNbox{fill=white}{(0,-1.7)}{.3}{.2}{.2}{$g^\dag$}
}
=
\tikzmath{
\draw[thick, orange] (.2,-1.4) -- (.2,-.45) .. controls ++(90:.3cm) and ++(-90:.3cm) .. (-.2,.45) --node[left]{$\scriptstyle b$} (-.2,1.4);
\draw[thick, red, knot] (-.3,-.45) node[left, yshift=.2cm]{$\scriptstyle a$} .. controls ++(90:.3cm) and ++(-90:.3cm) .. (.3,.45);
\draw (.3,.95) --node[right]{$\scriptstyle c$} (.3,1.4);
\draw (-.3,-.95) --node[left]{$\scriptstyle d$} (-.3,-1.4);
\draw (0,2) --node[right]{$\scriptstyle e$} (0,2.4);
\draw (0,-2) --node[right]{$\scriptstyle e$} (0,-2.4);
\roundNbox{fill=white}{(0,1.7)}{.3}{.2}{.2}{$f$}
\roundNbox{fill=white}{(.3,.7)}{.25}{0}{0}{$m$}
\roundNbox{fill=white}{(-.3,-.7)}{.25}{0}{0}{$n^\dag$}
\roundNbox{fill=white}{(0,-1.7)}{.3}{.2}{.2}{$g^\dag$}
}
$$
\item
The representation of the tube algebra $\Tube(\cA)$ associated to $(a,\beta_a)$ descends to a representation of the sphere algebra $\Sphere(\cA)$.
\end{enumerate}
\end{thm}

\begin{rem}
While the sphere algebra is similar to the dome algebra, representations of the former correspond to $Z_2(\cA)$, while representations of the latter gives $Z^\cA(\cA)$, the centralizer $Z(\cA)$ of $A$ in $Z(\cA)$.
The canonical quotient map $\Dome^\cA(\cA)\twoheadrightarrow \Sphere(\cA)$ corresponds to the inclusion $Z_2(\cA)\hookrightarrow Z^\cA(\cA)$.
\end{rem}

Similar to \eqref{eq:TubeAction} above, we have a $\Tube(\cA)$-action on any state with a localized bulk excitation.
Here, localized means that the bulk Hamiltonian is violated at (at most) 5 terms, namely one edge term and the 4 neighboring plaquette terms.
By a Morita equivalence argument similar to that in \S\ref{sec:StringOperators}, it again suffices to consider the $\Tube(\cA)$-action at a single edge $\ell$ connected to one vertex of the form
$$
\tikzmath{
\draw[thick, orange] (-.3,-.3) --node[left]{$\scriptstyle d$} (-.3,-.7);
\draw[thick, orange] (.3,-.3) --node[right]{$\scriptstyle b$} (.3,-.7);
\draw[thick, orange] (-.3,.3) --node[left]{$\scriptstyle c$} (-.3,.7);
\draw[thick, orange] (.3,.3) --node[right]{$\scriptstyle d$} (.3,.7);
\roundNbox{}{(0,0)}{.3}{.2}{.2}{$f$}
}
\rhd
\tikzmath[scale=1.5]{
\draw[thick, red, xshift=-.4cm, yshift=.4cm] (.7,.7) grid (3.3,2.3);
\begin{scope}[xshift=-.2cm, yshift=.2cm]
\draw[thick, knot, thick, red] (.7,1) -- (3.3,1);
\draw[thick, knot, thick, red] (.7,2) -- (3.3,2);
\draw[thick, knot, thick, red] (1,.7) -- (1,2.3);
\draw[thick, knot, thick, red] (3,.7) -- (3,2.3);
\draw[thick, knot, thick, red] (2,.7) -- (2,1.3);
\draw[thick, knot, thick, red] (2,2.3) -- (2,2);
\end{scope}
\draw[thick, knot, red] (.7,.7) grid (3.3,2.3);
\foreach \x in {1,2,3}{
\foreach \y in {1,2}{
\draw[thick, red] ($ (\x,\y) + (.2,-.2) $) -- ($ (\x,\y) + (-.6,.6) $);
}}
\foreach \x in {.8,1.8,2.8}{
\foreach \y in {1.2,2.2}{
\fill[red] (\x,\y) circle (.03cm);
}}
\foreach \x in {.6,1.6,2.6}{
\foreach \y in {1.4,2.4}{
\fill[red] (\x,\y) circle (.03cm);
}}
\foreach \x in {1,2,3}{
\foreach \y in {1,2}{
\fill[red] (\x,\y) circle (.03cm);
}}
\node[red] at (1.9,1.5) {$\scriptstyle a$};
}
:=
\delta_{a=b}
\left(\frac{d_c}{d_a}\right)^{1/4}
\cdot
\tikzmath[scale=1.5]{
\draw[thick, red, xshift=-.4cm, yshift=.4cm] (.7,.7) grid (3.3,2.3);
\begin{scope}[xshift=-.2cm, yshift=.2cm]
\draw[thick, knot, thick, red] (.7,1) -- (3.3,1);
\draw[thick, knot, thick, red] (.7,2) -- (3.3,2);
\draw[thick, knot, thick, red] (1,.7) -- (1,2.3);
\draw[thick, knot, thick, red] (3,.7) -- (3,2.3);
\draw[thick, knot, thick, red] (2,.7) -- (2,1.3);
\draw[thick, knot, thick, red] (2,2.3) -- (2,2);
\end{scope}
\foreach \x in {.6,1.6,2.6}{
\foreach \y in {1.4,2.4}{
\fill[red] (\x,\y) circle (.03cm);
}}
\draw[thick, knot, red] (.7,.7) grid (3.3,2.3);
\foreach \x in {1,2,3}{
\foreach \y in {1}{
\draw[thick, red] ($ (\x,\y) + (.2,-.2) $) -- ($ (\x,\y) + (-.6,.6) $);
}}
\filldraw[orange, thick, fill=gray!30, xshift=-.2cm, yshift=.2cm] (1.5,1.1) -- (1.8,1.1) to[out=0,in=180] (2,1.3) to[out=0,in=180] (2.2,1.1) -- (2.8,1.1) arc (-90:0:.1cm) -- (2.9,1.8) arc (0:90:.1cm) -- (1.2,1.9) arc (90:180:.1cm) -- (1.1,1.2) arc (180:270:.1cm) -- (1.5,1.1);
\draw[thick, knot, red] (.7,.7) grid (3.3,2.3);
\foreach \x in {1,2,3}{
\foreach \y in {2}{
\draw[thick, red] ($ (\x,\y) + (.2,-.2) $) -- ($ (\x,\y) + (-.6,.6) $);
}}
\foreach \x in {.8,1.8,2.8}{
\foreach \y in {1.2,2.2}{
\fill[red] (\x,\y) circle (.03cm);
}}
\foreach \x in {1,2,3}{
\foreach \y in {1,2}{
\fill[red] (\x,\y) circle (.03cm);
}}
\node[orange] at (1.3,1.4) {$\scriptstyle d$};
\draw[thick, red] (1.8,1.5) -- (1.8,1.7);
\filldraw[orange] (1.8,1.5) node[right, yshift=.1cm]{$\scriptstyle f$} circle (.05cm);
}\,.
$$
One now applies the argument of \eqref{eq:5=6} to see that this equal to the action of
$$
\tikzmath{
\draw[thick, orange] (-.3,-.3) node[left, yshift=-.2cm]{$\scriptstyle d$} arc (-180:0:.6cm) -- (.9,.3) node[right]{$\scriptstyle \overline{d}$} arc (0:180:.3cm);
\draw[knot, thick, orange] (.3,-.3) -- (.3,-.9) --node[right]{$\scriptstyle b$} (.3,-1.3);
\draw[thick, orange] (-.3,.3) --node[left]{$\scriptstyle c$} (-.3,1.1);
\roundNbox{}{(0,0)}{.3}{.2}{.2}{$f$}
}
\qquad\text{ or }\qquad
\tikzmath{
\draw[thick, orange] (.3,-.3) -- (.3,-.9) --node[right]{$\scriptstyle b$} (.3,-1.3);
\draw[knot, thick, orange] (-.3,-.3) node[left, yshift=-.2cm]{$\scriptstyle d$} arc (-180:0:.6cm) -- (.9,.3) node[right]{$\scriptstyle \overline{d}$} arc (0:180:.3cm);
\draw[thick, orange] (-.3,.3) --node[left]{$\scriptstyle c$} (-.3,1.1);
\roundNbox{}{(0,0)}{.3}{.2}{.2}{$f$}
}
\,,
$$
and thus this $\Tube(\cA)$-action descends to an action of $\Sphere(\cA)$.
This argument also proves that any other similar $\Tube(\cA)$-action one writes down at this edge
(for example, we could glue in the $d$ string around $f$ to a different loop around the boundary)
will be equal to this one.


\begin{fact}[{\cite[Prop.~4.3]{MR3022755}}]
Suppose $\cX$ is an $\cA$-enriched UFC, where the braided central functor $\Phi^Z:\cA\to Z(\cX)$ is full.
The M\"uger center of the enriched center $Z^\cA(\cX)$ is $Z_2(\cA)$.
\end{fact}

\begin{rem}
The above lemma clearly fails when the braided functor $\cA\to Z(\cX)$ is not full.
For example, consider $\cA=\Rep(G)$ and $\cX=\Vect$ and $\cA\to Z(\cX)=\Vect$ is the forgetful functor.
\end{rem}

Given a bulk excitation $a\in \Irr(Z_2(\cA))$,
for every 
$\psi: \Phi(a)\to x$ where $x\in \Irr(\cX)$,
we can define the \emph{bulk-to-boundary hopping operator} 
$$
H^{(a,\beta_a)}_\pi(\psi)
:=
B_oB_pB_qB_r T_\pi^{(a,\beta_a)}(\psi)\pi_k^1
$$
as in the cartoon below.
$$
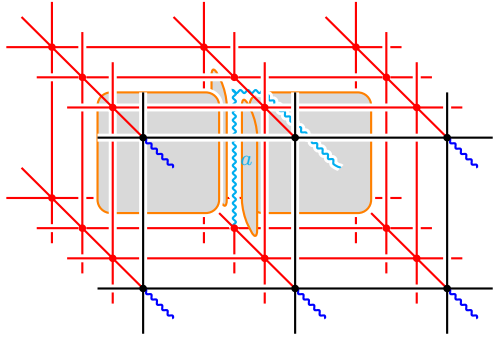
\tikzmath[scale=2]{
\draw[thick, red, xshift=-.4cm, yshift=.4cm] (.7,.7) grid (3.3,3.3);
\filldraw[thick, orange, rounded corners=5pt, fill=gray!30] (1.65,1.7) -- (1.65,2.3) -- (1.75,2.1) -- (1.75,1.3) -- (1.65,1.5) -- (1.65,1.7);
\begin{scope}[xshift=-.2cm, yshift=.2cm]
\draw[thick, knot, thick, red] (.7,1) -- (3.3,1);
\draw[thick, knot, thick, red] (.7,2) -- (3.3,2);
\draw[thick, knot, thick, red] (.7,3) -- (3.3,3);
\draw[thick, knot, thick, red] (1,.7) -- (1,3.3);
\draw[thick, knot, thick, red] (3,.7) -- (3,3.3);
\draw[thick, knot, thick, red] (2,.7) -- (2,1);
\draw[thick, knot, thick, red] (2,3.3) -- (2,2);
\end{scope}
\foreach \x in {1,2,3}{
\foreach \y in {1,2,3}{
\fill[red] ($ (\x,\y) + (-.4,.4) $) circle (.025cm);
}}
\filldraw[thick, orange, knot, rounded corners=5pt, fill=gray!30] (.9,1.3) rectangle (1.7,2.1);
\filldraw[thick, orange, knot, rounded corners=5pt, fill=gray!30] (1.9,1.3) rectangle (2.7,2.1);
\filldraw[thick, orange, rounded corners=5pt, xshift=.2cm, yshift=-.2cm, fill=gray!30] (1.65,1.7) -- (1.65,2.3) -- (1.75,2.1) -- (1.75,1.3) -- (1.65,1.5) -- (1.65,1.7);
\draw[thick, cyan, knot, decorate, decoration={snake, segment length=1mm, amplitude=.2mm}] (1.8,1.2) --node[right, xshift=-.07cm]{$\scriptstyle a$} (1.8,2.1) -- (2,2.1) -- (2.2,1.9);
\draw[thick, knot, red] (.7,.7) grid (3.3,3.3);
\draw[thick, cyan, knot, decorate, decoration={snake, segment length=1mm, amplitude=.2mm}] (2.2,1.9) -- (2.2,2.15) ;
\begin{scope}[xshift=.2cm, yshift=-.2cm]
\draw[thick, knot, thick] (.7,1) -- (3.3,1);
\draw[thick, knot, thick] (.7,2) -- (3.3,2);
\draw[thick, knot, thick] (.7,3) -- (3.3,3);
\draw[thick, knot, thick] (1,.7) -- (1,3.3);
\draw[thick, knot, thick] (3,.7) -- (3,3.3);
\draw[thick, knot, thick] (2,.7) -- (2,2);
\draw[thick, knot, thick] (2,3.3) -- (2,3);
\draw[thick, knot, thick] (2,2.35) -- (2,2.55);
\end{scope}
\foreach \x in {1,2,3}{
\foreach \y in {2,3}{
\draw[thick, red] ($ (\x,\y) + (.2,-.2) $) -- ($ (\x,\y) + (-.6,.6) $);
}}
\draw[thick, red] ($ (1,1) + (.2,-.2) $) -- ($ (1,1) + (-.6,.6) $);
\draw[thick, red] ($ (2,1) + (.2,-.2) $) -- ($ (2,1) + (-.3,.3) $);
\draw[thick, red] ($ (3,1) + (.2,-.2) $) -- ($ (3,1) + (-.3,.3) $);
\foreach \x in {1,2,3}{
\foreach \y in {1,2,3}{
\fill[red] ($ (\x,\y) + (-.2,.2) $) circle (.025cm);
\fill[red] (\x,\y) circle (.025cm);
\fill ($ (\x,\y) + (.2,-.2) $) circle (.025cm);
}}
\filldraw[cyan] (2.2,2.15) circle (.025cm) node[right]{$\scriptstyle \psi$};
}
$$
That is, since $\Sphere(\cA)\cong \bbC^{\Irr(Z_2(\cA))}$, the $\Sphere(\cA)$-representation at our bulk edge $\ell$ and vertex $v$ must be labelled by the simple object $a$.
We apply $\pi^1_{k}$ at the boundary edge $k$, we extend $a$ along the path $\pi$
placing $\psi$ at the end to get a simple $x\in \Irr(\cX)$ on the boundary,
and finally we apply the 4 bulk plaquette operators $B_oB_pB_qB_r$ which meet the bulk edge $\ell$.

By \eqref{eq:IzumiTubeAction}, the bulk-to-boundary hopping operator $H^{(a,\beta_a)}_\pi(\psi)$ takes vectors in the $\Sphere(\cA)$-representation corresponding to $a\in \Irr(Z_2(\cA))$
to vectors in the $\Dome^{\textcolor{red}{\cA}}(\cX)$-representation corresponding to $\Phi^Z(a)$.
Thus the hopping operator physically demonstrates that $Z^\cA(\cX)$ is a $Z_2(\cA)$-enriched UBFC \cite[Prop.~4.3]{MR3022755}.

The above discussion is sub-optimal, as the hopping operator depends on a choice of $\psi$.
This problem occurs because we have pushed $a\in\cA$ into $\cX$ via $\Phi$, and $\Sphere(\cA)$ is not actually a quotient of $\Dome^{\textcolor{red}{\cA}}(\cX)$, but only Morita equivalent to a quotient.
If we instead use a model for the boundary where boundary excitations can be hosted at a single vertex as in Remark \ref{rem:HostExcitationsAtVertices} above, 
$$
\tikzmath{
\draw[thick] (-.5,0) node[left]{$\scriptstyle w$} -- (.5,0) node[right]{$\scriptstyle z$};
\draw[thick] (0,-.5) node[below]{$\scriptstyle x$} -- (0,.5) node[above]{$\scriptstyle y$};
\filldraw (0,0) circle (.05cm);
\draw[thick, red] (0,0) -- (-.4,.4) node[left]{$\scriptstyle a$}; 
\draw[thick, blue, decorate, decoration={snake, segment length=1mm, amplitude=.2mm}] (0,0) -- (.4,-.4) node[right]{$\scriptstyle X$}; 
\node at (-.2,-.2) {$\scriptstyle v$};
\draw[purple!50, very thin] (-.25,.5) -- (.25,-.5);
{\draw[purple!50, very thin, -stealth] ($ (-.15,.3) - (.12,.06)$) to ($ (-.15,.3) + (.12,.06)$);}
{\draw[purple!50, very thin, -stealth] ($ (.15,-.3) - (.12,.06)$) to ($ (.15,-.3) + (.12,.06)$);}
}
\leftrightarrow
\cH_v:=
\bigoplus_{\substack{
x\in \Irr(\cX)
\\
X\in \Irr(Z^\cA(\cX))
}}
\cX(\Phi(a)wx \to yzF(X))
$$
we can avoid the choice of $\psi$ altogether.
Indeed, $\bbC^{\Irr(Z_2(\cA))}$ is an honest quotient of $\bbC^{\Irr(Z^{\cA}(\cX))}$, which is the algebra of local operators (Morita equivalent to $\Dome^{\textcolor{red}{\cA}}(\cX)$) acting at this type of ancilla.
In this setting, we can define the bulk-to-boundary hopping operator by
$$
H_\pi^{(a,\beta_a)}:=
B_oB_pB_qB_r T_\pi^{(a,\beta_a)}
$$
as in the cartoon below:
$$
\tikzmath[scale=2]{
\draw[thick, red, xshift=-.4cm, yshift=.4cm] (.7,.7) grid (3.3,2.3);
\filldraw[thick, orange, rounded corners=5pt, fill=gray!30] (1.65,1.7) -- (1.65,2.3) -- (1.75,2.1) -- (1.75,1.3) -- (1.65,1.5) -- (1.65,1.7);
\begin{scope}[xshift=-.2cm, yshift=.2cm]
\draw[thick, knot, thick, red] (.7,1) -- (3.3,1);
\draw[thick, knot, thick, red] (.7,2) -- (3.3,2);
\draw[thick, knot, thick, red] (1,.7) -- (1,2.3);
\draw[thick, knot, thick, red] (3,.7) -- (3,2.3);
\draw[thick, knot, thick, red] (2,.7) -- (2,1);
\draw[thick, knot, thick, red] (2,2.3) -- (2,2);
\end{scope}
\foreach \x in {1,2,3}{
\foreach \y in {1,2}{
\fill[red] ($ (\x,\y) + (-.4,.4) $) circle (.025cm);
}}
\filldraw[thick, orange, knot, rounded corners=5pt, fill=gray!30] (.9,1.3) rectangle (1.7,2.1);
\filldraw[thick, orange, knot, rounded corners=5pt, fill=gray!30] (1.9,1.3) rectangle (2.7,2.1);
\filldraw[thick, orange, rounded corners=5pt, xshift=.2cm, yshift=-.2cm, fill=gray!30] (1.65,1.7) -- (1.65,2.3) -- (1.75,2.1) -- (1.75,1.3) -- (1.65,1.5) -- (1.65,1.7);
\draw[thick, cyan, knot, decorate, decoration={snake, segment length=1mm, amplitude=.2mm}] (1.8,1.2) --node[right, xshift=-.07cm]{$\scriptstyle a$} (1.8,2.1) -- (2,2.1) -- (2.5,1.6);
\draw[thick, knot, red] (.7,.7) grid (3.3,2.3);
\draw[thick, knot, xshift=.2cm, yshift=-.2cm] (.7,.7) grid (3.3,2.3);
\foreach \x in {1,2,3}{
\foreach \y in {2}{
\draw[thick, red] ($ (\x,\y) + (.2,-.2) $) -- ($ (\x,\y) + (-.6,.6) $);
}}
\draw[thick, red] ($ (1,1) + (.2,-.2) $) -- ($ (1,1) + (-.6,.6) $);
\draw[thick, red] ($ (2,1) + (.2,-.2) $) -- ($ (2,1) + (-.3,.3) $);
\draw[thick, red] ($ (3,1) + (.2,-.2) $) -- ($ (3,1) + (-.3,.3) $);
\foreach \x in {1,2,3}{
\foreach \y in {1}{
\draw[thick, blue, decorate, decoration={snake, segment length=1mm, amplitude=.2mm}] ($ (\x,\y) + (.2,-.2) $) -- ($ (\x,\y) + (.4,-.4) $);
}}
\draw[thick, blue, decorate, decoration={snake, segment length=1mm, amplitude=.2mm}] ($ (1,2) + (.2,-.2) $) -- ($ (1,2) + (.4,-.4) $);
\draw[thick, blue, decorate, decoration={snake, segment length=1mm, amplitude=.2mm}] ($ (3,2) + (.2,-.2) $) -- ($ (3,2) + (.4,-.4) $);
\foreach \x in {1,2,3}{
\foreach \y in {1,2}{
\fill[red] ($ (\x,\y) + (-.2,.2) $) circle (.025cm);
\fill[red] (\x,\y) circle (.025cm);
\fill ($ (\x,\y) + (.2,-.2) $) circle (.025cm);
}}
}
$$
Now using the isomorphism $\Sphere(\cA)\cong \bbC^{\Irr(Z_2(\cA))}$ from Remark \ref{lem:SphereAlgebraIso}, we see that this hopping operator intertwines the $\Sphere(\cA)$-action coming from the bulk with the $\bbC^{\Irr(Z_2(\cA))}\subset \bbC^{\Irr(Z^\cA(\cX))}$-action hosted on the blue edges.

\section{Conclusion}

In this manuscript, we have reviewed the explicit construction of the 2D boundary boundary theory of the 3D Walker-Wang model given from \cite{MR4640433} by a UMTC $\mathcal{A}$ as a sort of enriched Levin-Wen model based on an $\mathcal{A}$-enriched UFC $\cX$. 
Using techniques from TQFT, we have confirmed the prediction of \cite{MR4640433} that the excitations of this 2D model are given by the enriched center/M\"uger centralizer $Z^\mathcal{A}(\cX)$ of $\mathcal{A}$ in $Z(\cX)$. 
We achieved this by defining the algebra $\Dome^{\cA}(\cX)$ and showing that its irreducible representations are in one-to-one correspondence with the excitations of the boundary theory.  
This is analogous to how the irreducible representations of the tube algebra $\Tube(\cX)$ are in one-to-one correspondence of the excitations of an unenriched Levin-Wen model given by a UFC $\cX$.

All of these arguments are made on finite lattices. 
There is also an operator algebraic approach \cite{MR3617688} for identifying excitation types on infinite lattices. 
It would be interesting to check that the excitation types on the infinite lattice align with those from our analysis.

In our construction, we have not considered protection by 0-form symmetries. 
A direction for future research would be to endow the bulk theory with such a symmetry and to study the excitations in the resulting anomalous $G$-enriched topological order in the 2D boundary theory.

\printindex

\bibliographystyle{alpha}
{\footnotesize{
\bibliography{bibliography}
}}
\end{document}